\documentclass[11pt]{article}

\usepackage[sc]{mathpazo}

\usepackage[margin=1in]{geometry}
\usepackage[english]{babel}
\usepackage[utf8x]{inputenc}
\usepackage[compact]{titlesec}
\usepackage[usenames,dvipsnames]{xcolor}

\usepackage{cmap}
\usepackage[T1]{fontenc}
\usepackage{bm}
\usepackage{amsmath}
\usepackage{amsfonts}
\usepackage{amssymb}
\usepackage{amsbsy}
\usepackage{amsthm}
\usepackage{dsfont}

\usepackage{mathtools}
\usepackage{xspace}

\usepackage{graphicx, ucs}

\usepackage{subcaption}
\usepackage{rotating}
\usepackage{float}
\usepackage{tikz}

\usepackage{algorithm}
\usepackage[noend]{algpseudocode}
\usepackage{listings}

\usepackage{enumitem}
\usepackage{hyperref}
\hypersetup{
  colorlinks = true,
  urlcolor = {blueGrotto},
  linkcolor = {royalBlue},
  citecolor = {navyBlueP}
}
\usepackage{multirow}
\usepackage{array}

\usepackage{accents}

\usepackage[outline]{contour}

\usepackage[%
linewidth=3pt,
linecolor=gray,
middlelinecolor= black,
middlelinewidth=0.4pt,
roundcorner=1pt,
topline = false,
rightline = false,
bottomline = false,
rightmargin=0pt,
skipabove=0pt,
skipbelow=0pt,
leftmargin=0pt,
innerleftmargin=4pt,
innerrightmargin=0pt,
innertopmargin=0pt,
innerbottommargin=0pt,
]{mdframed}

\usepackage{chngcntr}
\usepackage{soul}
\usepackage{arydshln}

\usepackage{thm-restate}
\usepackage{cleveref}

\usepackage{color-edits}
\addauthor{cd}{niceRed}
\addauthor{sk}{limeGreen}
\addauthor{mz}{blueGrotto}

\definecolor{Gred}{RGB}{219, 50, 54}
\definecolor{Ggreen}{RGB}{60, 186, 84}
\definecolor{Gblue}{RGB}{72, 133, 237}
\definecolor{Gyellow}{RGB}{247, 178, 16}
\definecolor{ToCgreen}{RGB}{0, 128, 0}
\definecolor{myGold}{RGB}{231,141,20}
\definecolor{myBlue}{rgb}{0.19,0.41,.65}
\definecolor{myPurple}{RGB}{175,0,124}

\usepackage{sidecap}
\usepackage{environ}
\sidecaptionvpos{figure}{t}
\usepackage{comment}

\usepackage{ulem}

\usepackage{cancel}
\usepackage{wrapfig}

\newcommand{\F}{\ensuremath{\mathbb{F}}}

\newcommand{\N}{\ensuremath{\mathbb{N}}}

\newcommand{\R}{\ensuremath{\mathbb{R}}}

\newcommand{\calC}{\ensuremath{\mathcal{C}}}

\newcommand{\calI}{\ensuremath{\mathcal{I}}}

\newcommand{\calO}{\ensuremath{\mathcal{O}}}
\newcommand{\calP}{\ensuremath{\mathcal{P}}}
\newcommand{\calQ}{\ensuremath{\mathcal{Q}}}

\renewcommand{\vec}[1]{\boldsymbol{#1}}

\newcommand{\matA}{\ensuremath{\boldsymbol{A}}}

\newcommand{\vecb}{\ensuremath{\boldsymbol{b}}}
\newcommand{\vecc}{\ensuremath{\boldsymbol{c}}}

\newcommand{\vecg}{\ensuremath{\boldsymbol{g}}}

\newcommand{\vecp}{\ensuremath{\boldsymbol{p}}}
\newcommand{\vecq}{\ensuremath{\boldsymbol{q}}}

\newcommand{\vecs}{\ensuremath{\boldsymbol{s}}}
\newcommand{\vect}{\ensuremath{\boldsymbol{t}}}
\newcommand{\vecu}{\ensuremath{\boldsymbol{u}}}
\newcommand{\vecv}{\ensuremath{\boldsymbol{v}}}
\newcommand{\vecw}{\ensuremath{\boldsymbol{w}}}
\newcommand{\vecx}{\ensuremath{\boldsymbol{x}}}
\newcommand{\vecy}{\ensuremath{\boldsymbol{y}}}
\newcommand{\vecz}{\ensuremath{\boldsymbol{z}}}

\theoremstyle{plain}            
\newtheorem{theorem}{Theorem}[section]
\newtheorem{inftheorem}{Informal Theorem}
\newtheorem{lemma}[theorem]{Lemma}
\newtheorem{corollary}[theorem]{Corollary}

\newtheorem{claim}[theorem]{Claim}

\theoremstyle{definition}       
\newtheorem{definition}[theorem]{Definition}

\theoremstyle{remark}           
\newtheorem{remark}[theorem]{Remark}
\newtheorem{example}[theorem]{Example}

\numberwithin{equation}{section}

\contourlength{0.1pt}
\contournumber{10}

\newenvironment{nproblem}[1][\unskip]{%
\medskip
\begin{mdframed}
\noindent
\contour{black}{\underline{$#1$.}} \\
\noindent}
{\end{mdframed}}

\newcommand{\class}[1]{\ensuremath{\mathsf{#1}}\xspace}

\renewcommand{\P}{\class{P}}
\newcommand{\FP}{\class{FP}}

\newcommand{\NP}{\class{NP}}
\newcommand{\FNP}{\class{FNP}}

\newcommand{\PLS}{\class{PLS}}
\newcommand{\PPAD}{\class{PPAD}}

\newcommand{\ball}{\mathcal{B}}

\newif\ifnotes\notestrue

\ifnotes
\usepackage{color}
\definecolor{mygrey}{gray}{0.50}
\newcommand{\notename}[2]{{\textcolor{red}{\footnotesize{\bf (#1:} {#2}{\bf
) }}}}

\else

\newcommand{\notename}[2]{{}}

\fi

\DeclareMathOperator*{\Prob}{{\mathbb{P}}}

\mathchardef\mdash="2D 

\newcommand{\eps}{\varepsilon}

\renewcommand{\epsilon}{\varepsilon}

\def\compactify{\itemsep=0pt \topsep=0pt \partopsep=0pt \parsep=0pt}
\let\latexusecounter=\usecounter

\newenvironment{Enumerate}
  {\def\usecounter{\compactify\latexusecounter}
   \begin{enumerate}}
  {\end{enumerate}\let\usecounter=\latexusecounter}

{\begin{itemize}%
\setlength{\itemsep}{0pt}%
\setlength{\topsep}{0pt}%
\setlength{\partopsep}{0 in}%
\setlength{\parskip}{0 pt}}%
{\end{itemize}}

\DeclareMathOperator*{\argmin}{argmin}

\def\ceil#1{\mathop{\left\lceil#1\right\rceil}}
\def\abs#1{\left|#1\right|}
\def\p#1{\left(#1\right)}
\def\b#1{\left[#1\right]}
\def\set#1{\left\{#1\right\}}

\newcommand{\paragr}[1]{\noindent \textbf{#1}}

\def\norm#1{\left\|#1\right\|}

\def\poly{\mathrm{poly}}

\newcommand{\diam}[2]{\mathrm{diam}_{#1}\left[ #2 \right]}

\definecolor{niceRed}{RGB}{190,38,38}
\definecolor{blueGrotto}{HTML}{059DC0}
\definecolor{royalBlue}{HTML}{057DCD}
\definecolor{navyBlueP}{HTML}{0B579C}
\definecolor{limeGreen}{HTML}{81B622}

\definecolor{lred}{RGB}{236,33,39}
\definecolor{lyellow}{RGB}{251,168,26}
\definecolor{lblue}{RGB}{94,202,231}
\definecolor{lgreen}{RGB}{124,194,67}
\definecolor{lgray}{RGB}{179,181,184}
\definecolor{lbrown}{RGB}{139,94,60}
\definecolor{lredBlue}{RGB}{165,118,135}
\definecolor{lredGreen}{RGB}{186,105,50}
\definecolor{lyellowGreen}{RGB}{188,181,50}
\definecolor{lyellowBlue}{RGB}{166,187,139}

\renewcommand{\bar}[1]{\overline{#1}}
\renewcommand{\diam}{\mathrm{diam}}

\renewcommand{\ball}{\class{B}}

\newcommand{\stat}{\textsc{StationaryPoint}}
\newcommand{\lmin}{\textsc{LocalMin}}
\newcommand{\lNash}{\textsc{LocalMinMax}}
\newcommand{\lrlNash}{\text{\sc LR-\lNash}}

\newcommand{\gdFixed}{\textsc{GDFixedPoint}}
\newcommand{\gdaFixed}{\textsc{GDAFixedPoint}}

\newcommand{\Sperner}{\textsc{Sperner}}
\newcommand{\BiSperner}{\textsc{BiSperner}}
\newcommand{\TD}{2\mathrm{D}}
\newcommand{\HD}{\textsc{High}\mathrm{D}}

\newcommand{\nm}[1]{\b{#1} - 1}

\title{The Complexity of Constrained Min-Max Optimization}
\author{
  Constantinos Daskalakis \\
  MIT \\
  \url{costis@csail.mit.edu} \normalsize
  \and Stratis Skoulakis \\
  SUTD \\
  \url{efstratios@sutd.edu.sg} \normalsize
  \and Manolis Zampetakis \\
  MIT \\
  \url{mzampet@mit.edu} \normalsize
}
\begin{document}
\thispagestyle{empty}
\maketitle

\begin{abstract}
    Despite its important applications in Machine Learning, min-max optimization
  of objective functions that are \textit{nonconvex-nonconcave} remains elusive.
  Not only are there no known first-order methods converging even to approximate
  local min-max points, but the computational complexity of identifying them is
  also poorly understood. In this paper, we provide a characterization of the
  computational complexity of the problem, as well as of the limitations of
  first-order methods in constrained min-max optimization problems with
  nonconvex-nonconcave objectives and linear constraints.

    As a warm-up, we show that, even when the objective is a Lipschitz and
  smooth differentiable function, deciding whether a min-max point exists, in
  fact even deciding whether an approximate min-max point exists, is $\NP$-hard.
  More importantly, we show that an approximate local min-max point of large
  enough approximation is guaranteed to exist, but finding one such point is
  $\PPAD$-complete. The same is true of computing an approximate fixed point of
  the (Projected) Gradient Descent/Ascent update dynamics.

    An important byproduct of our proof is to establish an unconditional
  hardness result in the~Nemirovsky-Yudin~\cite{nemirovsky1983problem} oracle
  optimization model. We show that, given oracle access to some function
  $f : \calP \to [-1, 1]$ and its gradient $\nabla f$, where
  $\calP \subseteq [0, 1]^d$ is a known convex polytope, every algorithm that
  finds a $\eps$-approximate local min-max point needs to make a number of
  queries that is exponential in at least one of $1/\eps$, $L$, $G$, or $d$,
  where $L$ and $G$ are respectively the smoothness and Lipschitzness of $f$ and
  $d$ is the dimension. This comes in sharp contrast to minimization problems,
  where finding approximate local minima in the same setting can be done
  with Projected Gradient Descent using  $O(L/\eps)$ many queries. Our result is
  the first to show an exponential separation between these two fundamental
  optimization problems in the oracle model.
\end{abstract}
\thispagestyle{empty}
\newpage
\setcounter{page}{1}

\thispagestyle{empty}
\tableofcontents
\thispagestyle{empty}
\newpage
\setcounter{page}{1}

\section{Introduction} \label{sec:intro}

  {\em Min-Max Optimization} has played a central role in the development of
Game Theory~\cite{VN28}, Convex Optimization~\cite{Dantzig1951,Adler13}, and
Online Learning
\cite{B56, C06, shalev2012online, bubeck2012regret, shalev2014understanding, hazan2016introduction}.
In its general constrained form, it can be written down as follows:
\begin{align}
  \min_{\vecx \in \mathbb{R}^{d_1}} \max_{\vecy \in \mathbb{R}^{d_2}} f(\vecx,\vecy); \label{eq:costis1} \\
  ~~~~~~~~\text{s.t.}~~g(\vecx,\vecy) \le 0. \notag
\end{align}
Here, $f : \R^{d_1} \times \R^{d_2} {\to [-B, B]}$ with $B \in \R_+$, and
$g : \R^{d_1} \times \R^{d_2} {\to \R}$ is typically taken to be a convex
function so that the constraint set $g(\vecx, \vecy) \le 0$ is convex. In this
paper, we only use linear functions $g$ so the constraint set is a polytope,
thus projecting on this set and checking feasibility of a point with respect to
this set can both be done in polynomial time.

  The goal in \eqref{eq:costis1} is to find a feasible pair
$(\vecx^{\star}, \vecy^{\star})$, i.e., $g(\vecx^{\star}, \vecy^{\star}) \le 0$,
that satisfies the following
\begin{align}
    & f(\vecx^{\star}, \vecy^{\star}) \le f(\vecx, \vecy^{\star}), ~~\text{for all~$\vecx$ s.t.~$g(\vecx,\vecy^{\star}) \le 0$}; \label{eq:local min costis}\\
    & f(\vecx^{\star}, \vecy^{\star}) \ge f(\vecx^{\star},\vecy), ~~\text{for all~$\vecy$ s.t.~$g(\vecx^{\star},\vecy) \le 0$}. \label{eq:local max costis}
\end{align}

  It is well-known that, when $f(\vecx,\vecy)$ is a convex-concave function,
i.e., $f$ is convex in $\vecx$ for all $\vecy$ and it is concave in $\vecy$ for
all $\vecx$, then Problem~\eqref{eq:costis1} is guaranteed to have a solution,
under compactness of the constraint set~\cite{VN28,rosen1965existence}, while
computing a solution is amenable to convex programming. In fact, if $f$ is
$L$-smooth, the problem can be solved via first-order methods, which are
iterative, only access $f$ through its gradient,\footnote{In general, the access
to the constraints $g$ by these methods is more involved, namely through an
optimization oracle that optimizes convex functions
(in fact, quadratic suffices) over $g(\vecx,\vecy)\le 0$. In the settings
considered in this paper $g$ is linear and these tasks are computationally straightforward.} and achieve an approximation error of $\poly(L,1/T)$ in $T$
iterations; see e.g.~\cite{korpelevich1976extragradient,nemirovski2004interior}.\footnote{In the stated error rate, we are suppressing factors that depend on the diameter of the feasible set. Moreover, the stated error of $\varepsilon(L,T) \triangleq \poly(L, 1/T)$ reflects that these methods return an approximate min-max solution, wherein the inequalities on the LHS of~\eqref{eq:local min costis} and~\eqref{eq:local max costis} are satisfied to within an additive $\varepsilon(L,T)$.
}
When the function is strongly convex-strongly concave, the rate becomes
geometric~\cite{facchinei2007finite}.

  Unfortunately, our ability to solve Problem~\eqref{eq:costis1} remains rather
poor in settings where our objective function $f$ is {\em not} convex-concave.
This is emerging as a major challenge in Deep Learning, where min-max
optimization has recently found many important applications, such as training
Generative Adversarial Networks (see e.g.~\cite{GAN14,arjovsky2017wasserstein}),
and robustifying deep neural network-based models against adversarial attacks
(see e.g.~\cite{madry2017towards}). These applications are indicative of a
broader deep learning paradigm wherein robustness properties of a deep learning
system are tested and enforced by another deep learning system. In these
applications, it is very common to encounter min-max problems with objectives that are
nonconvex-nonconcave, and thus evade treatment by the classical algorithmic
toolkit targeting convex-concave objectives.

  Indeed, the optimization challenges posed by objectives that are
nonconvex-nonconcave are not just theoretical frustration. Practical experience
with first-order methods is rife with frustration as well. A common experience
is that the training dynamics of first-order methods is unstable, oscillatory
or divergent, and the quality of the points encountered in the course of
training can be poor; see e.g.~\cite{goodfellow2016nips,metz2016unrolled,daskalakis2017training,mescheder2018training,daskalakis2018limit,mazumdar2018convergence,MertikopoulosPP18,adolphs2018local}.
This experience is in stark contrast to minimization (resp.~maximization)
problems, where even for nonconvex (resp.~nonconcave) objectives, first-order
methods have been found to efficiently converge to approximate local optima or
stationary points (see e.g.~\cite{agarwal2017finding,jin2017escape,LeePPSJR19}),
while practical methods such Stochastic Gradient Descent, Adagrad, and
Adam~\cite{duchi2011adaptive,kingma2014adam,ReddiKK18} are driving much of the
recent progress in Deep Learning.
\medskip

  The goal of this paper is to \textit{shed light on the complexity of min-max
optimization problems}, and \textit{elucidate its difference to minimization and
maximization problems}---as far as the latter is concerned without loss of
generality we focus on minimization problems, as maximization problems behave
exactly the same; we will also think of minimization problems in the framework
of~\eqref{eq:costis1}, where the variable $\vec y$ is absent, that is
$d_2 = 0$. An important driver of our comparison between min-max optimization
and minimization is, of course, the nature of the objective. So let us discuss:
\bigskip

\noindent $\triangleright$ \textit{Convex-Concave Objective.} The benign
setting for min-max optimization is that where the objective function is
convex-concave, while the benign setting for minimization is that where the
objective function is convex. In their corresponding benign settings, the two
problems behave quite similarly from a computational perspective in that they
are amenable to convex programming, as well as first-order methods which only
require gradient information about the objective function. Moreover, in their
benign settings, both problems have guaranteed existence of a solution under
compactness of the constraint set. Finally, it is clear how to define
approximate solutions. We just relax the inequalities on the left
hand side of~\eqref{eq:local min costis} and~\eqref{eq:local max costis} by some
$\eps > 0$.
\medskip

\noindent $\triangleright$ \textit{Nonconvex-Nonconcave Objective.} By
contrapositive, the challenging setting for min-max optimization is that where
the objective is \textit{not} convex-concave, while the challenging setting for minimization is that where the objective is not convex. In these
challenging settings, the behavior of the two problems diverges significantly.
The first difference is that, while a solution to a minimization problem is
still guaranteed to exist under compactness of the constraint set even when the
objective is not convex, a solution to a min-max problem is \textit{not}
guaranteed to exist when the objective is not convex-concave, even under
compactness of the constrained set. A trivial example is this:
$\min_{x \in [0,1]} \max_{y \in [0,1]} (x-y)^2$. Unsurprisingly, we show that
checking whether a min-max optimization problem has a solution is $\NP$-hard. In
fact, we show that checking whether there is an approximate min-max solution is
$\NP$-hard, even when the function is Lispchitz and smooth and the desired
approximation error is an absolute constant (see Theorem~\ref{t:local_nash}).
\bigskip

  Since min-max solutions may not exist, what could we plausibly hope to
compute? There are two obvious targets:

\begin{enumerate}
  \item[(I)] approximate stationary points of $f$, as
considered e.g.~by~\cite{abernethy2019last}; and \label{item:costas1}

  \item[(II)] some type of approximate {\em local} min-max solution. \label{item:costas2}
\end{enumerate}

Unfortunately, as far as (I) is concerned, it is still
possible that (even approximate) stationary points may not exist, and we show
that checking if there is one is $\NP$-hard, even when the constraint set is
$[0,1]^d$, the objective has Lipschitzness and smoothness polynomial in $d$, and
the desired approximation is an absolute constant
(Theorem~\ref{thm:approximateStationaryPointsHardness}). So we focus on~(II),
i.e.~(approximate) local min-max solutions. Several kinds of those have been
proposed in the literature~\cite{daskalakis2018limit,mazumdar2018convergence,jin2019minmax}.
We consider a generalization of the concept of local min-max equilibria,
proposed in~\cite{daskalakis2018limit,mazumdar2018convergence}, that also
accommodates approximation.

\begin{definition}[Approximate Local Min-Max Equilibrium]\label{def:general constraint local min-max}
    Given $f$, $g$ as above, and $\varepsilon,\delta>0$, some point
  $(\vecx^{\star}, \vecy^{\star})$ is an
  {\em $(\varepsilon, \delta)$-local min-max solution of~\eqref{eq:costis1}}, or
  a {\em $(\varepsilon, \delta)$-local min-max equilibrium}, if it is feasible,
  i.e.~$g(\vecx^{\star}, \vecy^{\star})\le 0$, and satisfies:
  \begin{align}
    &f(\vecx^{\star}, \vecy^{\star}) < f(\vecx, \vecy^{\star}) + \eps, ~\text{for all~$\vecx$ such that $\norm{\vecx - \vecx^{\star}} \le \delta$ and $g(\vecx, \vecy^{\star}) \le 0$}; \label{eq:local min costis2}\\
    &f(\vecx^{\star}, \vecy^{\star}) > f(\vecx^{\star},\vecy) - \eps, ~\text{for all~$\vecy$ such that $\norm{\vecy - \vecy^{\star}} \le \delta$ and~$g(\vecx^{\star},\vecy) \le 0$}. \label{eq:local max costis2}
  \end{align}
\end{definition}

\noindent In words, $(\vecx^{\star}, \vecy^{\star})$ is an
$(\eps, \delta)$-local min-max equilibrium, whenever the min player cannot
update $\vecx$ to a feasible point within $\delta$ of $\vecx^{\star}$ to reduce
$f$ by at least $\eps$, and symmetrically the max player cannot change $\vecy$
locally to increase $f$ by at least $\eps$.
\medskip

  We show that the existence and complexity of computing such approximate local
min-max equilibria depends on the relationship of $\varepsilon$ and $\delta$
with the smoothness, $L$, and the Lipschitzness, $G$, of the objective function
$f$. We distinguish the following regimes, also shown in Figure
\ref{fig:regimes} together with a summary of our associated results.
\smallskip

\paragr{$\blacktriangleright$ Trivial Regime.} This occurs when
$\delta < \frac{\eps}{G}$. This regime is trivial because the $G$-Lipschitzness
of $f$ guarantees that all feasible points are $(\varepsilon, \delta)$-local
min-max solutions.
\smallskip

\paragr{$\blacktriangleright$ Local Regime.} This occurs when
$\delta < \sqrt{\frac{2 \eps}{L}}$, and it represents the interesting regime for
min-max optimization. In this regime, we use the smoothness of $f$ to show that
$(\varepsilon, \delta)$-local min-max solutions always exist. Indeed, we show
(Theorem~\ref{t:LocalMaxMin-GDA}) that computing them is computationally
equivalent to the following variant of (I) which is more suitable for the
constrained setting:
\begin{enumerate}
    \item[(I')] (approximate) fixed points of the projected gradient descent-ascent dynamics (Section~\ref{sec:computational:syntactic}).
\end{enumerate}
We show via an application of Brouwer's fixed point theorem to the iteration map
of the projected gradient descent-ascent dynamics that (I)' are guaranteed to
exist. In fact, not only do they exist, but computing them is in $\PPAD$, as can
be shown by bounding the Lipschitzness of  the projected gradient descent-ascent
dynamics (Theorem~\ref{t:PPAD_inclusion}).
\smallskip

\paragr{$\blacktriangleright$ Global Regime.} This occurs when $\delta$ is
comparable to the diameter of the constraint set. In this case, the existence of
$(\varepsilon, \delta)$-local min-max solutions is not guaranteed, and
determining their existence is $\NP$-hard, even if $\varepsilon$ is an absolute
constant (Theorem~\ref{t:local_nash}).
\medskip

\definecolor{themePurple}{RGB}{146,39,143}
\definecolor{themeBlue}{RGB}{76,90,106}
\definecolor{themeBrown1}{RGB}{134,86,64}
\begin{figure}[!p]
  \centering
  \includegraphics[width=\textwidth]{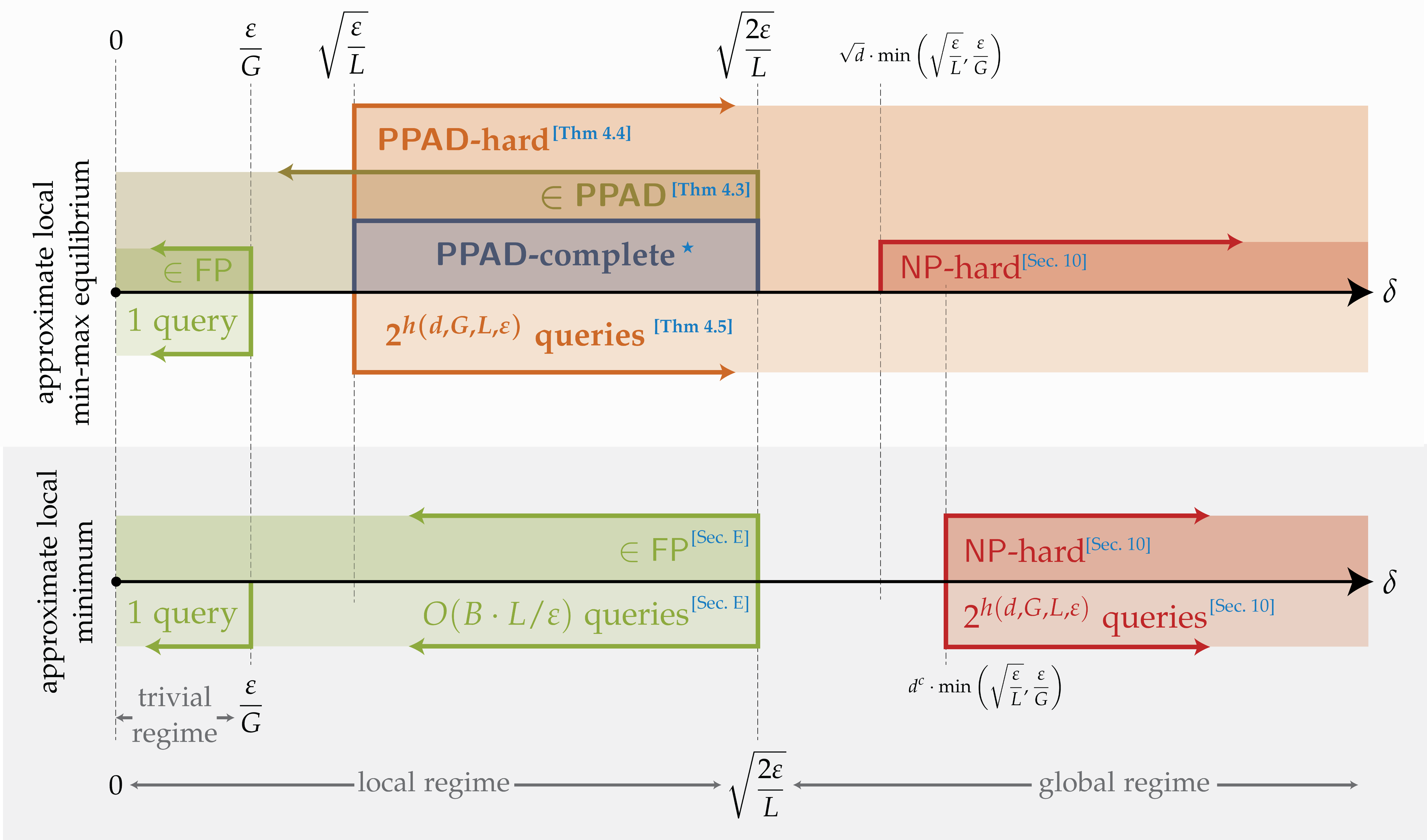}

  \caption{Overview of the results proven in this paper and comparison
  between the complexity of computing an $(\varepsilon,\delta)$-approximate
  local minimum and an $(\varepsilon,\delta)$-approximate local min-max
  equilibrium of a $G$-Lipschitz and $L$-smooth function over a $d$-dimensional
  polytope taking values in the interval $[-B, B]$. We assume that
  $\eps < G^2/L$, thus the trivial regime is a strict subset of the local
  regime. Moreover, we assume that the approximation parameter $\eps$ is
  provided in unary representation in the input to these problems, which makes
  our hardness results stronger and the comparison to the upper bounds known for
  finding approximate local minima fair, as these require time/oracle queries
  that are polynomial in $1/\eps$. We note that the unary representation is not
  required for our results proving inclusion in ${\PPAD}$. The figure portrays a
  sharp contrast between the computational complexity of approximate local
  minima and approximate local min-max equilibria in the local regime. Above the
  black lines, tracking the value of $\delta$, we state our ``white box''
  results and below the black lines we state our ``black-box'' results. The main
  result of this paper is the $\PPAD$-hardness of approximate local min-max
  equilibrium for $\delta \ge \sqrt{\eps/L}$ and the corresponding query lower
  bound. In the query lower bound the function $h$ is defined as
  $h(d, G, L, \eps) = \p{\min(d, \sqrt{L/\eps}, G/\eps)}^p$ for some universal
  constant $p \in \R_+$. With \textcolor{royalBlue}{$\star$} we indicate our
  $\PPAD$-completeness result which directly follows from Theorems
  \ref{thm:localNashExistence} and \ref{thm:localNashHardness}. The
  $\NP$-hardess results in the global regime are presented in Section
  \ref{sec:localNash}. Finally, the folklore result showing the tractability of
  finding approximate local minima is presented
  for completeness of  exposition in Appendix \ref{sec:gdStationary}. The
  claimed results for the trivial regime follow from the definition of
  Lipschitzness.}
  \label{fig:regimes}
\end{figure}

  The main results of this paper, summarized in Figure \ref{fig:regimes}, are to
characterize the complexity of computing local min-max solutions in the local
regime. Our first main theorem is the following:
\begin{inftheorem}[see Theorems~\ref{thm:localNashExistence},  \ref{thm:localNashHardness} and~\ref{t:LocalMaxMin-GDA}] \label{informal:ppad hardness}
    Computing  $(\varepsilon, \delta)$-local min-max solutions of Lipschitz
  and smooth objectives over convex compact domains in the local regime is
  $\PPAD$-complete. The hardness holds even when the constraint set is a
  polytope that is a subset of $[0, 1]^d$, {the objective takes values in
  $[-1, 1]$} and the smoothness, Lipschitzness, $1/\varepsilon$ and $1/\delta$
  are polynomial in the dimension. Equivalently, computing $\alpha$-approximate
  fixed points of the Projected Gradient Descent-Ascent dynamics on smooth and
  Lipschitz objectives is $\PPAD$-complete, and the hardness holds even when the
  the constraint set is a polytope that is a subset of $[0, 1]^d$, {the
  objective takes values in $[-d, d]$} and smoothness, Lipschitzness, and
  $1/\alpha$ are polynomial in the dimension.
\end{inftheorem}

  For the above complexity result we assume that we have
``white box'' access to the objective function. An important byproduct of our
proof, however, is to also establish an {\em unconditional hardness result} in
the Nemirovsky-Yudin~\cite{nemirovsky1983problem} oracle optimization model,
wherein we are given black-box access to oracles computing the objective
function and its gradient. Our second main result is informally stated in
Informal Theorem \ref{informal:blackBoxhardness}.

\begin{inftheorem}[see Theorem~\ref{thm:localNashBlackBoxLowerBound}] \label{informal:blackBoxhardness}
    Assume that we have black-box access to an oracle computing a $G$-Lipschitz
  and $L$-smooth objective function {$f : \calP \to [-1, 1]$, where
  $\calP \subseteq [0, 1]^d$ is a known polytope,} and its gradient $\nabla f$. Then,
  computing an $(\eps, \delta)$-local min-max solution in the local regime
  (i.e., when $\delta < \sqrt{2 \eps/L}$) requires a number of oracle queries
  that is exponential in at least one of the following: $1/\varepsilon$, $L$,
  $G$, or $d$. In fact, exponential in $d$-many queries are required even when $L$, $G$, $1/\varepsilon$ and $1/\delta$ are all
  polynomial in $d$.
\end{inftheorem}
\noindent Importantly, the above lower bounds, in both the white-box and the
black-box setting, come in sharp contrast to minimization problems, given that
finding approximate local minima of smooth non-convex objectives {
ranging in $[-B, B]$ in the local regime can be done using first-order methods using
$O(B \cdot L / \eps)$ time/queries (see Section \ref{sec:gdStationary})}. Our
results are the first to show an exponential separation between these two
fundamental problems in optimization in the black-box setting, and a
super-polynomial separation in the white-box setting assuming $\PPAD \neq \FP$.

\subsection{Brief Overview of the Techniques} \label{sec:intro:techniques}

  We very briefly outline some of the main ideas for the $\PPAD$-hardness proof
that we present in Sections \ref{sec:hardness:2D} and \ref{sec:hardness}. Our
starting point as in many $\PPAD$-hardness results is a discrete analog of the
problem of finding Brouwer fixed points of a continuous map. Departing from
previous work, however, we do not use Sperner's lemma as the discrete analog of
Brouwer's fixed point theorem. Instead, we define a new problem, called
$\BiSperner$, which is useful for showing our hardness results. $\BiSperner$ is
closely related to the problem of finding panchromatic simplices guaranteed by
Sperner's lemma except, roughly speaking, that the vertices of the
simplicization of a $d$-dimensional hypercube are colored with $2d$ rather than
$d + 1$ colors,  every point of the simplicization is colored with $d$ colors
rather than one, and we are seeking a vertex of the simplicization so that the
union of colors on the vertices in its neighborhood covers the full set of
colors. The first step of our proof is to show that $\BiSperner$ is
$\PPAD$-hard. This step follows from the hardness of computing Brouwer fixed
points.
\smallskip

  The step that we describe next is only implicitly done by our proof, but it
serves as useful intuition for reading and understanding it. We want to define
a {\em discrete} two-player zero-sum game whose local equilibrium points
correspond to solutions of a given $\BiSperner$ instance. Our two players,
called  ``minimizer'' and  ``maximizer,'' each choose a vertex of the
simplicization of the $\BiSperner$ instance. For every pair of strategies in our
discrete game, i.e.~vertices, chosen by our players, we define a function value
and gradient values. Note that, at this point, we treat these values at
different vertices of the simplicization as independent choices, i.e.~are not
defining a function over the continuum whose function values and gradient values
are consistent with these choices. It is our intention, however, that in the
{\em continuous} two-player zero-sum game that we  obtain in the next paragraph
via our interpolation scheme, wherein the minimizer and maximizer  may choose
any point in the continuous hypercube, the function value determines the payment
of the minimizer to the maximizer, and the gradient value determines the
direction of the best-response dynamics of the game. Before getting to that
continuous game in the next paragraph, the main technical step of this discrete
part of our construction is showing that every local equilibrium of the discrete
game corresponds to a solution of the $\BiSperner$ instance we are reducing
from. In order to achieve this we need to add some constraints to couple the
strategies of the minimizer and the maximizer player. This step is the reason
that the constraints $g(\vecx, \vecy) \le 0$ appear in the final  min-max
problem that we produce.
\smallskip

  The third and quite  challenging step of the proof is to show that we
can interpolate in a \textit{smooth} and \textit{computationally efficient} way
the discrete zero-sum game of the previous step. In low dimensions (treated in
Section \ref{sec:hardness:2D}) such smooth and efficient interpolation can be
done in a relatively simple way using single-dimensional smooth step functions.
In high dimensions, however, the smooth and efficient interpolation becomes a
challenging problem and to the best of our knowledge no simple solution exists.
For this reason we construct our novel
\textit{smooth and efficient interpolation coefficients} of Section
\ref{sec:seic}. These are a technically involved construction that we believe
will prove to be very useful for characterizing the complexity of approximate
solutions of other  optimization problems.
\smallskip

  The last part of our proof is to show that all the previous steps  can be
implemented in an efficient way both with respect to computational but also with
respect to query complexity. This part is essential for both our white-box and
black-box results. Although this seems like a relatively easy step,
it becomes more difficult due to the complicated expressions in our smooth and
efficient interpolation coefficients used in our previous step.
\smallskip

  Closing this section we mention that all our $\NP$-hardness results are proven
using a cute application of Lov\'asz Local Lemma \cite{ErdosL1973}, which
provides a  powerful rounding tool that can drive the inapproximability all the
way up to an absolute constant.

\subsection{Local Minimization vs Local Min-Max Optimization} \label{sec:intro:minimization}

  Because our proof is convoluted, involving multiple steps, it is difficult to
discern from it why finding local min-max solutions is so much harder than
finding local minima. For this reason, we illustrate in this section a
fundamental difference between local minimization and local min-max
optimization. This provides good intuition about why our hardness construction
would fail if we tried to apply it to prove hardness results for finding local
minima (which we know don't exist).

  So let us illustrate a key difference between min-max problems that can be
expressed in the form $\min_{x \in {\cal X}} \max_{y \in {\cal Y}} f(x,y)$,
i.e.~two-player zero-sum games wherein the players optimize opposing objectives,
and min-min problems of the form
$\min_{x \in {\cal X}} \min_{y \in {\cal Y}} f(x,y)$, i.e., two-player
coordination games wherein the players optimize the same objective. For
simplicity, suppose ${\cal X} = {\cal Y} = \R$ and let us consider long paths of
best-response dynamics in the strategy space, ${\cal X} \times {\cal Y}$, of the
two players; these are paths along which at least one of the players improves their payoff. For our illustration, suppose that the derivative of the function
with respect to either variable is either $1$ or $-1$. Consider a long path of
best-response dynamics starting at a pair of strategies $(x_0, y_0)$ in either a
min-min problem or a min-max problem, and a specific point $(x, y)$ along that
path. We claim that in min-min  problems the function value at $(x, y)$ will
have to reveal how far from $(x_0, y_0)$ point $(x, y)$ lies within the path in $\ell_1$ distance. On
the other hand, in min-max problems the function value at $(x, y)$ may reveal
very little about how far $(x, y)$ lies from $(x_0, y_0)$. We
illustrate this in Figure \ref{fig:easyMinimum}. While in our min-min example
the function value must be monotonically decreasing inside the best-response
path, in the min-max example the function values repeat themselves in every
straight line segment of length $3$, without revealing where in the
path each segment is.

  Ultimately a key difference between min-min  and min-max optimization is that
best-response paths in min-max optimization problems can be closed, i.e., can
form a cycle, as shown in Figure \ref{fig:easyMinimum}, Panel (b). On the other
hand, this is impossible in min-min problems as the function value must
monotonically decrease along  best-response paths, thus cycles may not exist.
\medskip

\begin{figure}[t]
  \centering
  \begin{subfigure}[b]{0.42\textwidth}
    \centering
    \includegraphics[width=\textwidth]{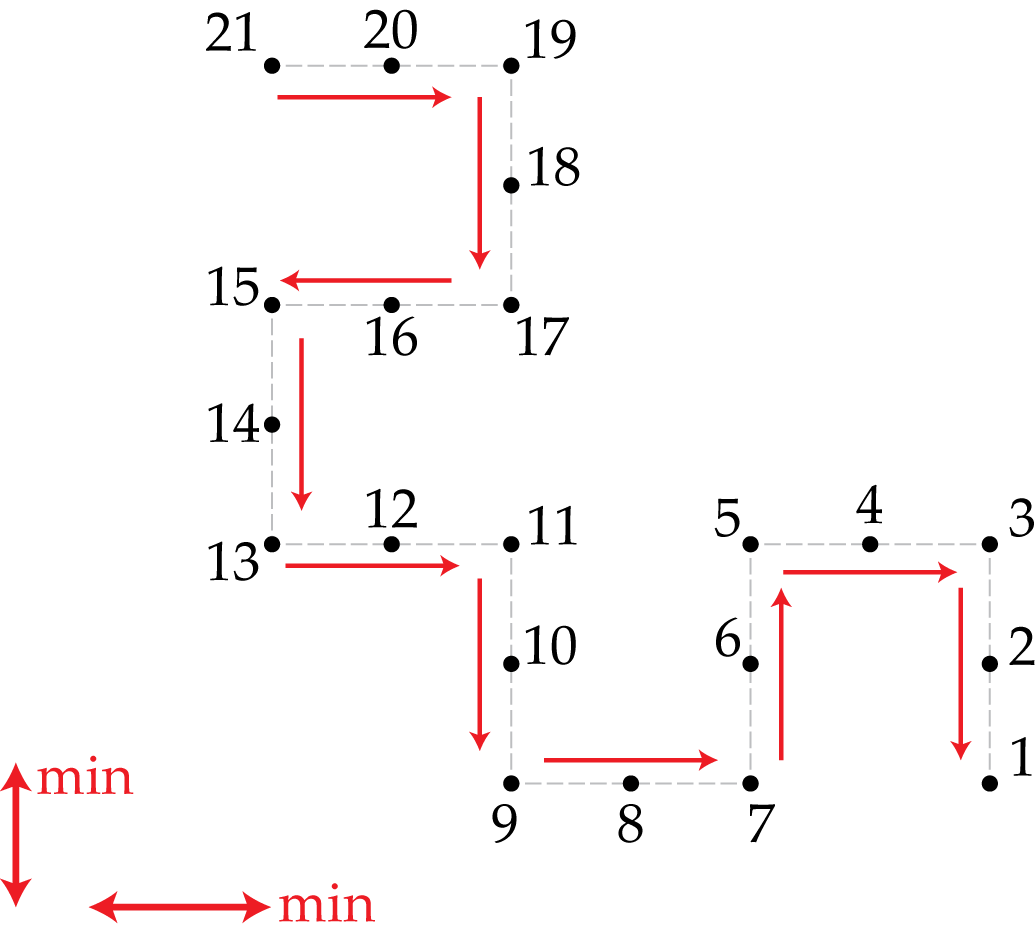}
    \caption{Min-min problem; the function values reveal the location of the
    points within best response path.}
  \end{subfigure}
  ~~~~~~~
  \begin{subfigure}[b]{0.42\textwidth}
    \centering
    \includegraphics[width=\textwidth]{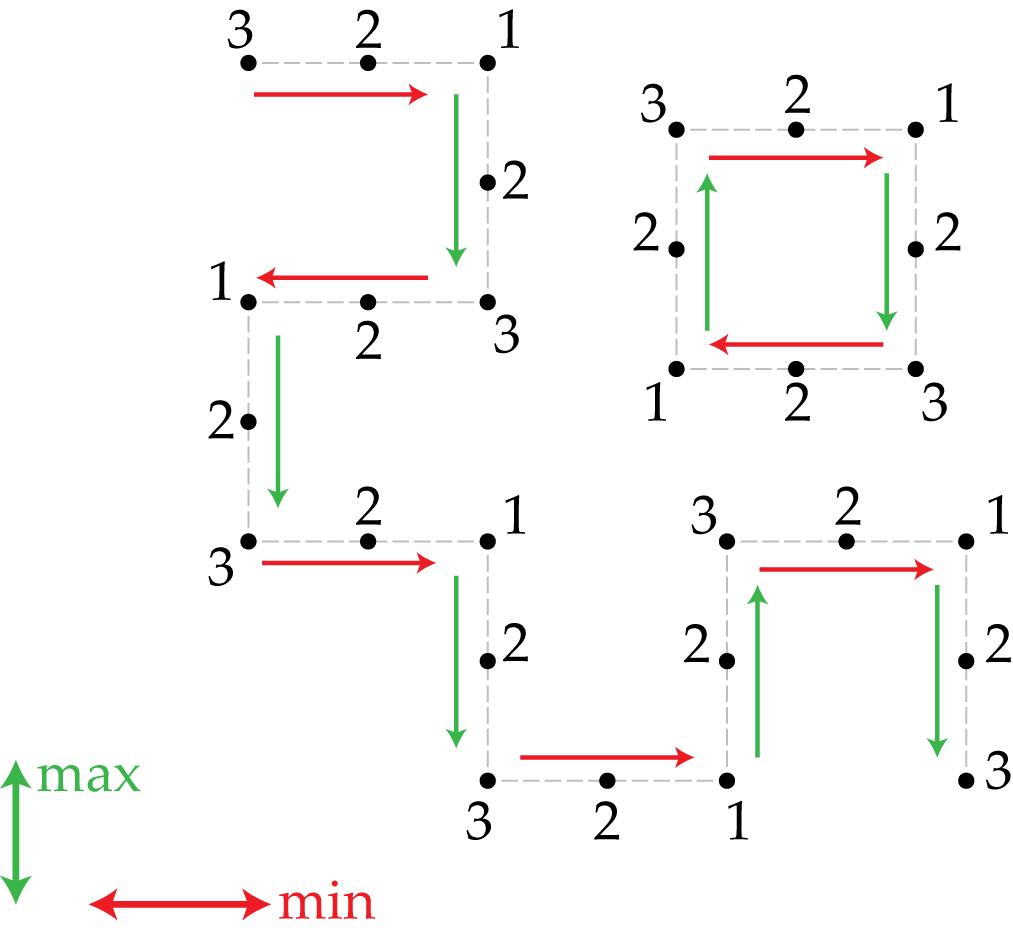}
    \caption{Min-max problem; the function values do not reveal the location of
    the points within best response path.}
  \end{subfigure}
  \caption{Long paths of best-response dynamics in min-min problems (Panel (a))
  and min-max problems (Panel (b)), where horizontal moves correspond to one
  player (who is a minimizer in both (a) and (b)) and vertical moves correspond
  to the other player (who is minimizer in (a) but a maximizer in (b)). In
  Panels (a) and (b), we show the function value at a subset of discrete points
  in a 2D grid along a long path of best-response dynamics, where for our
  illustration we assumed that the derivative of the objective with respect to
  either variable always has absolute value $2$. As we see in Panel (a), the
  function value at some point along a long path of the best-response dynamics
  in a min-min problem reveals information about where in the path that point
  lies. This is in sharp contrast to min-max problems where only local
  information is revealed about the objective as shown in Panel (b), due to the
  frequent turns of the path. In Panel (b) we also show that the best-response
  dynamics in min-max problems can form closed paths. This cannot happen in
  min-min problems as the function value must decrease along paths of
  best-response dynamics, and hence it is impossible in min-min problems to
  build long best-response paths with function values that can be computed
  locally.}
  \label{fig:easyMinimum}
\end{figure}

  The above discussion offers qualitative differences between min-min and
min-max optimization, which lie in the heart of why our computational
intractability results are possible to prove for min-max but not min-min
problems. For the precise step in our construction that breaks if we were to
switch from a min-max to a min-min problem we refer the reader to Remark
\ref{rem:easyMinimum:2D}.
\bigskip

\subsection{Further Related Work} \label{sec:intro:related}

  There is a broad literature on the complexity of equilibrium computation.
Virtually all these results are obtained within the computational complexity
formalism of \textit{total search problems} in $\NP$, which was spearheaded
by~\cite{JohnsonPY1988, MeggidoP1989, Papadimitriou1994} to capture the
complexity of search problems that are guaranteed to have a solution. Some key
complexity classes in this landscape are shown in Figure \ref{fig:tfnp}. We give
a non-exhaustive list of intractability results for equilibrium computation:
\cite{FabrikantPT04} prove that  computing pure Nash equilibria in congestion
games is $\PLS$-complete; \cite{daskalakis2009complexity} and later
\cite{chen2009settling} show that computing approximate Nash equilibria in
normal-form games is $\PPAD$-complete; \cite{etessami2010complexity} study the
complexity of computing exact Nash equilibria (which may use irrational
probabilities), introducing the complexity class $\class{FIXP}$;
\begin{wrapfigure}{l}{0.45\textwidth}
  \begin{center}
    \includegraphics[width=0.32\textwidth]{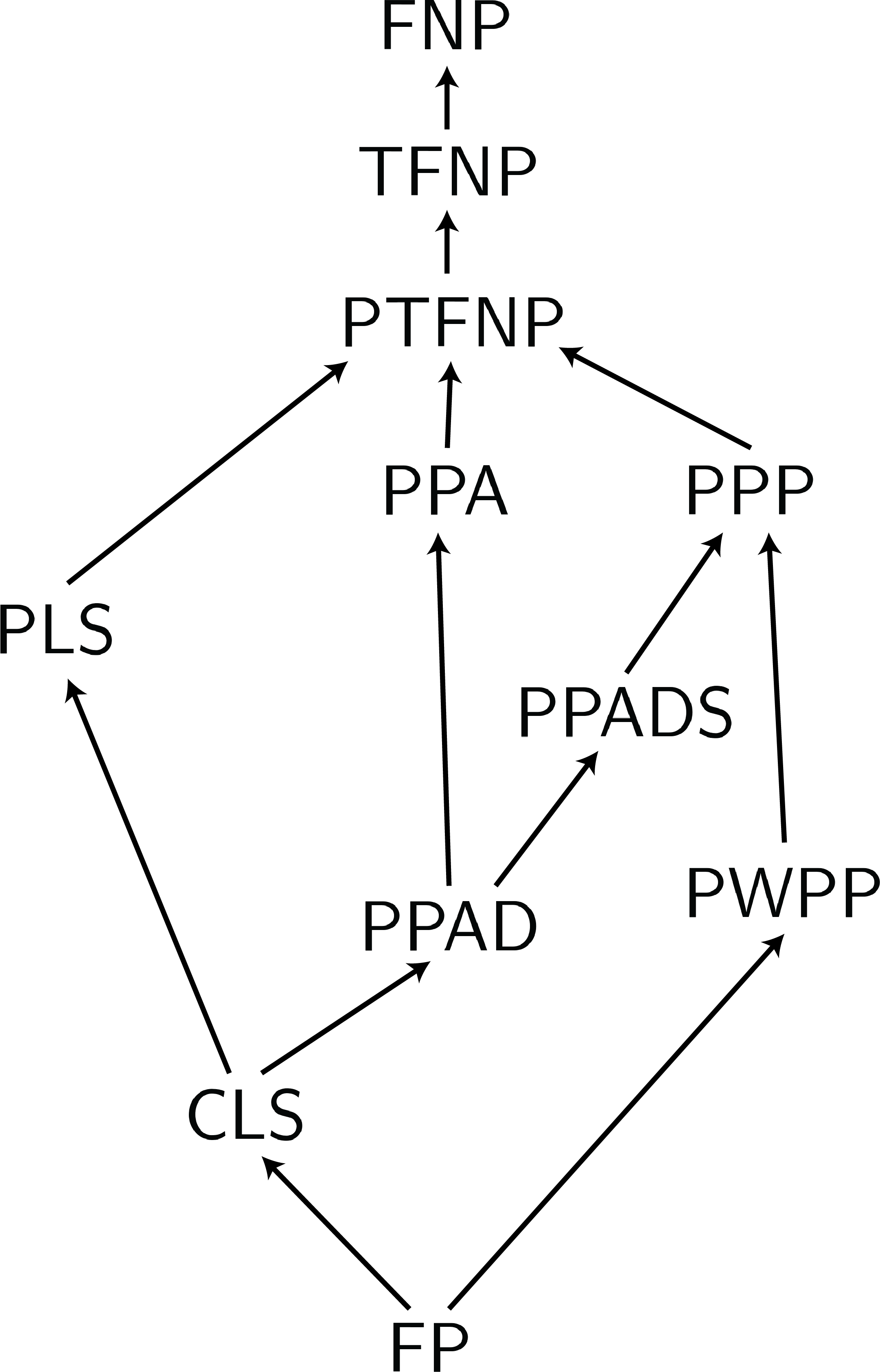}
  \end{center}
  \caption{The complexity-theoretic landscape of total search problems in $\NP$.}
  \label{fig:tfnp}
\end{wrapfigure}
 \cite{VaziraniY11, ChenPY17}
consider the complexity of computing Market equilibria;
\cite{daskalakis2013complexity,R15, R16} consider the complexity of computing
approximate Nash equilibria of constant approximation; \cite{KM18} establish a
connection between approximate Nash equilibrium computation and the SoS
hierarchy; \cite{M14, DFS19} study the complexity of computing Nash equilibria
in specially structured games. A result that is particularly useful for our work
is the result of \cite{HPV88} which shows black-box query lower bounds for
computing Brouwer fixed points of a continuous function. We use this result in
Section~\ref{sec:blackBox} as an ingredient for proving our black-box lower
bounds for computing approximate local min-max solutions.

  Beyond equilibrium computation and its applications to Economics and Game
Theory, the study of total search problems has found profound connections to
many scientific fields, including continuous  optimization
\cite{daskalakis2011continuous, DaskalakisTZ18}, combinatorial  optimization
\cite{SchafferY91}, query complexity \cite{BeameCEIP95}, topology~\cite{GH19},
topological combinatorics and social choice theory
\cite{FG18, GF19, FilosHSZ20, FilosHSZ20c}, algebraic combinatorics
\cite{BelovsIQSY17, GoosKSZ19}, and cryptography
\cite{jevrabek2016integer, BitanskyPR15, SotirakiZZ18}. For a more extensive
overview of total search problems we refer the reader to the recent survey by
Daskalakis~\cite{DaskalakisNevanlinna19}.
\medskip

  As already discussed, min-max optimization has intimate connections to the
foundations of Game Theory, Mathematical Programming, Online Learning,
Statistics, and several other fields. Recent applications of min-max
optimization to Machine Learning, such as Generative Adversarial Networks and
Adversarial Training, have motivated a slew of recent work targeting first-order
(or other light-weight online learning) methods for solving min-max optimization
problems for convex-concave, nonconvex-concave, as well as nonconvex-nonconcave
objectives. Work on convex-concave and nonconvex-concave objectives has focused
on obtaining online learning methods with improved rates
\cite{kong2019accelerated, lin2019gradient, thekumparampil2019efficient, nouiehed2019solving, lu2019hybrid, ouyang2019lower, zhao2019optimal, alkousa2019accelerated, AzizianMLG20, GolowichPDO, lin2020near}
and last-iterate convergence guarantees
\cite{daskalakis2017training, daskalakis2018limit, mazumdar2018convergence, MertikopoulosPP18, rafique2018non, hamedani2018primal, adolphs2018local, daskalakis2019last, liang2019interaction, gidel2019negative, mokhtari2019unified, abernethy2019last},
while work on nonconvex-nonconcave problems has focused on identifying different
notions of local min-max solutions \cite{jin2019minmax, MangoubiV20} and studying
the existence and (local) convergence properties of learning methods at these
points \cite{wang2019solving, MangoubiV20, MangoubiSV20}.

\section{Preliminaries} \label{sec:model}

\noindent \textbf{Notation.} For any compact and convex $K \subseteq \R^d$ and
$B \in \R_+$, we define $L_{\infty}(K, B)$ to be the set of all continuous
functions $f : K \to \R$ such that
$\max_{\vec{x} \in K} \abs{f(\vec{x})} \le B$. When $K = [0, 1]^d$, we use
$L_{\infty}(B)$ instead of $L_{\infty}([0, 1]^d, B)$ for ease of notation. For
$p > 0$, we define
$\diam_{p}(K) = \max_{\vec{x}, \vec{y} \in K} \norm{\vec{x} - \vec{y}}_{p}$,
where $\norm{\cdot}_p$ is the usual $\ell_p$-norm of vectors. For an alphabet
set $\Sigma$, the set $\Sigma^*$, called the Kleene star of $\Sigma$, is equal
to $\cup_{i = 0}^{\infty} \Sigma^i$. For any string $\vecq \in \Sigma$ we use
$\abs{\vecq}$ to denote the length of $\vecq$. We use the symbol $\log(\cdot)$
for base $2$ logarithms and $\ln(\cdot)$ for the natural logarithm. We use
$[n] \triangleq \{1, \ldots, n\}$, $\nm{n} \triangleq \{0, \dots, n - 1\}$, and
$[n]_0 \triangleq \{0, \dots, n\}$.
\medskip

\paragr{Lipschitzness, Smoothness, and Normalization.} Our main objects of study are
continuously differentiable Lipschitz and smooth functions $f : {\cal P}  \to \R$, where ${\cal P} \subseteq [0, 1]^d$ is some polytope. A continuously
differentiable function $f$ is called \textit{$G$-Lipschitz}
if $\abs{f(\vecx) - f(\vecy)} \le G \norm{\vecx - \vecy}_2$, for all $\vecx,\vecy$, and
\textit{$L$-smooth} if
$\norm{\nabla f(\vecx) - \nabla f(\vecy)}_2 \le L \norm{\vecx - \vecy}_2$, for all $\vecx,\vecy$.
\smallskip

\begin{remark}[Function Normalization] \label{rem:boundOnValues}
  \em
    Note that the $G$-Lipschitzness of a function $f: {\cal P} \to \mathbb{R}$,
  where ${\cal P} \subseteq [0,1]^d$ implies that for any $\vecx$ and $\vecy$ it
  holds that $\abs{f(\vecx) - f(\vecy)} \le G \sqrt{d}$. Whenever the range of a
  $G$-Lipschitz function is taken to be $[-B,B]$, for some $B$, we always assume
  that $B \le G \sqrt{d}$. This can be accomplished by setting
  $\tilde{f}(\vecx) = f(\vecx) - f(\vecx_0)$ for some fixed $\vecx_0$ in the
  domain of $f$. For all the problems that we consider in this paper any
  solution for $\tilde{f}$ is also a solution for $f$ and vice-versa.
\end{remark}
\smallskip

\paragraph{Function Access.} We study optimization problems involving
real-valued functions, considering two access models to such functions.

\begin{enumerate}
  \item[$\blacktriangleright$] \textbf{Black Box Model.} In this model we are
  given access to an oracle $\calO_f$ such that given a point
  $\vecx \in [0, 1]^d$ the oracle $\calO_f$ returns the values $f(\vecx)$ and
  $\nabla f(\vecx)$. In this model we assume that we can perform real number
  arithmetic operations. This is the traditional model used to prove lower
  bounds in Optimization and Machine Learning \cite{nemirovsky1983problem}.
  \item[$\blacktriangleright$] \textbf{White Box Model.} In this model we are
  given the description of a polynomial-time Turing machine $\calC_f$ that
  computes $f(\vecx)$ and $\nabla f(\vecx)$. More precisely, given some input
  $\vecx \in [0, 1]^d$, described using $B$ bits, and some accuracy $\eps$,
  $\calC_f$ runs in time upper bounded by some polynomial in $B$ and
  $\log(1/\eps)$ and outputs approximate values for $f(\vecx)$ and
  $\nabla f(\vecx)$, with approximation error that is at most $\eps$ in $\ell_2$
  distance. We note that a running time upper bound on a given Turing Machine
  can be enforced syntactically by stopping the computation and outputting a
  fixed output whenever the computation exceeds the bound. See also
  Remark~\ref{rem:representation} for an important remark about how to formally
  study the computational complexity of problems that take as input a
  polynomial-time Turing Machine.
\end{enumerate}

\paragr{Promise Problems.} To simplify the exposition of our paper, make the
definitions of our computational problems and theorem statements clearer, and
make our intractability results stronger, we choose to enforce the following
constraints on our function access, $\calO_f$ or $\calC_f$, as a
\textit{promise}, rather than enforcing these constraints in some syntactic
manner.
\begin{enumerate}
  \item \textbf{Consistency of Function Values and Gradient Values.} Given some
  oracle $\calO_f$ or Turing machine $\calC_f$, it is difficult to determine by
  querying the oracle or examining the description of the Turing machine whether
  the function and gradient values output on different inputs  are consistent
  with some differentiable function. In all our computational problems, we will
  only consider instances where this is promised to be the case.  Moreover, for
  all our computational hardness results, the instances of the problems arising
  from our reductions satisfy these constraints, which are guaranteed
  syntactically by our reduction.

  \item \textbf{Lipschitzness, Smoothness and Boundedness.} Similarly, given
  some oracle $\calO_f$ or Turing machine $\calC_f$, it is difficult to
  determine, by querying the oracle or examining the description of the Turing
  machine, whether the function and gradient values output by $\calO_f$ or
  $\calC_f$ are consistent with some Lipschitz, smooth and bounded function with
  some prescribed Lipschitzness, smoothness, and  bound on its absolute value.
  In all our computational problems, we only consider instances where the
  $G$-Lipschitzness, $L$-smoothness and $B$-boundedness of the function are
  promised to hold for the prescribed, in the input of the problem, parameters
  $G$, $L$ and $B$. Moreover, for all our computational hardness results, the
  instances of the problems arising from our reductions satisfy this constraint,
  which is guaranteed syntactically by our reduction.
\end{enumerate}

  In summary, in the rest of this paper, whenever we prove
\textit{an upper bound} for some computational problem, namely an upper bound on
the number of steps or queries to the function oracle required to solve the
problem in  the black-box model, or the containment of the problem in some
complexity class in the white-box model, we assume that the afore-described
properties are satisfied by the $\calO_f$ or $\calC_f$ provided in the input. On
the other hand, whenever we prove a \textit{lower bound} for some computational
problem, namely a lower bound on the number of steps/queries required to solve
it in the black-box model, or its hardness for some complexity class in the
white-box model, the instances arising in our lower bounds are guaranteed to
satisfy the above properties syntactically by our constructions. As such, our
hardness results will not exploit the difficulty in checking whether $\calO_f$
or $\calC_f$ satisfy the above constraints in order to infuse computational
complexity into our problems, but will faithfully target the computational
problems pertaining to min-max optimization of smooth and Lipschitz objectives
that we aim to understand in this paper.

\addtocontents{toc}{\protect\setcounter{tocdepth}{1}}
\subsection{Complexity Classes and Reductions} \label{sec:model:complexity}
\addtocontents{toc}{\protect\setcounter{tocdepth}{2}}

  In this section we define the main complexity classes that we use in this
paper, namely $\NP$, $\FNP$ and $\PPAD$, as well as the notion of reduction used
to show containment or hardness of a problem for one of these complexity
classes.

\begin{definition}[Search Problems, ${\NP}$, ${\FNP}$] \label{def:FNP}
    A binary relation $\calQ \subseteq \set{0, 1}^* \times \set{0, 1}^*$ is in
  the class $\FNP$ if (i) for every $\vecx, \vecy \in \set{0, 1}^*$ such that
  $(\vecx, \vecy) \in \calQ$, it holds that
  $\abs{\vecy} \le \poly(\abs{\vecx})$; and (ii) there exists an algorithm that
  verifies whether $(\vecx, \vecy) \in \calQ$ in time
  $\poly(\abs{\vecx}, \abs{\vecy})$. The \textit{search problem} associated with
  a binary relation $\calQ$ takes some $\vecx$ as input and requests as output
  some $\vecy$ such that $(\vecx, \vecy) \in \calQ$ or outputting $\bot$ if no
  such $\vecy$ exists. The \textit{decision problem} associated with $\calQ$
  takes some $\vecx$ as input and requests as output the bit $1$, if there
  exists some $\vecy$ such that $(\vecx, \vecy) \in \calQ$, and the bit $0$,
  otherwise. The class $\NP$ is defined as the set of decision problems
  associated with relations $\calQ \in \FNP$.
\end{definition}

  To define the complexity class $\PPAD$ we first define the notion of
polynomial-time reductions between search problems\footnote{In this paper we
only define and consider Karp-reductions between search problems.}, and the computational problem
\textsc{End-of-a-Line}\footnote{This problem is sometimes called
\textsc{End-of-the-Line}, but we adopt the nomenclature proposed by~\cite{R16}
since we agree that it describes the problem better.}.

\begin{definition}[Polynomial-Time Reductions]
    A search problem $P_1$ is \textit{polynomial-time reducible} to a search
  problem $P_2$ if there exist polynomial-time computable functions
  $f : \set{0, 1}^* \to \set{0, 1}^*$ and
  $g : \set{0, 1}^* \times \set{0, 1}^* \times \set{0, 1}^* \to \set{0, 1}^*$
  with the following properties: (i) if $\vecx$ is an input to $P_1$, then
  $f(\vecx)$ is an input to $P_2$; and (ii) if $\vecy$ is a solution to $P_2$ on
  input $f(\vecx)$, then $g(\vecx, f(\vecx), \vecy)$ is a solution to $P_1$ on
  input $\vecx$.
\end{definition}

\begin{nproblem}[\textsc{End-of-a-Line}]
  \textsc{Input:} Binary circuits $\calC_S$ (for successor) and $\calC_P$ (for
  predecessor) with $n$ inputs and $n$ outputs.
  \smallskip

  \noindent \textsc{Output:} One of the following:
  \begin{Enumerate}
    \item[0.] $\vec{0}$ if either both $\calC_P(\calC_S(\vec{0}))$ and
              $\calC_S(\calC_P(\vec{0}))$ are equal to $\vec{0}$, or if they are
              both different than $\vec{0}$, where $\vec{0}$ is the all-$0$
              string.
    \item[1.] a binary string $\vecx \in \set{0, 1}^n$ such that
              $\vecx \neq \vec{0}$ and $\calC_P(\calC_S(\vecx)) \neq \vecx$ or
              $\calC_S(\calC_P(\vecx)) \neq \vecx$.
  \end{Enumerate}
\end{nproblem}
\medskip

  To make sense of the above definition, we envision that the circuits $\calC_S$
and $\calC_P$ implicitly define a directed graph, with vertex set $\{0, 1\}^n$,
such that the directed edge
$(\vecx, \vecy) \in \set{0, 1}^n \times \set{0, 1}^n$ belongs to the graph if
and only if $\calC_S(\vecx) = \vecy$ and $\calC_P(\vecy) = \vecx$. As such, all
vertices in the implicitly defined graph have in-degree and out-degree at most
$1$. The above problem permits an output of $\vec{0}$ if $\vec{0}$ has equal
in-degree and out-degree in this graph. Otherwise it permits an output
$\vecx \neq \vec{0}$ such that $\vecx$ has in-degree or out-degree equal to $0$.
It follows by the parity argument on directed graphs, namely that in every
directed graph the sum of in-degrees equals the sum of out-degrees, that
\textsc{End-of-a-Line} is a \textit{total problem}, i.e. that for any possible
binary circuits $\calC_S$ and $\calC_P$ there exists a solution of the ``0.''
kind or the ``1.'' kind in the definition of our problem (or both). Indeed, if
$\vec{0}$ has unequal in- and out-degrees, there must exist another vertex
$\vec{x} \neq \vec{0}$ with unequal in- and out-degrees, thus one of these
degrees must be $0$ (as all vertices in the graph have in- and out-degrees
bounded by $1$).
\medskip

  We are finally ready to define the complexity class $\PPAD$ introduced by
\cite{Papadimitriou1994}.

\begin{definition}[$\boldsymbol{\PPAD}$] \label{def:PPAD}
    The complexity class $\PPAD$ contains all search problems that are
  polynomial time reducible to the \textsc{End-of-a-Line} problem.
\end{definition}

  The complexity class $\PPAD$ is of particular importance, since it contains
lots of fundamental problems in Game Theory, Economics, Topology and several
other fields~\cite{daskalakis2009complexity, DaskalakisNevanlinna19}. A
particularly important $\PPAD$-complete problem is finding fixed points of
continuous functions, whose existence is guaranteed by Brouwer's fixed point
theorem.

\begin{nproblem}[\textsc{Brouwer}]
  \textsc{Input:} Scalars $L$ and $\gamma$ and a polynomial-time Turing machine
  $\calC_M$ evaluating a $L$-Lipschitz function $M : [0, 1]^d \to [0, 1]^d$.
  \smallskip

  \noindent \textsc{Output:} A point $\vecz^{\star} \in [0, 1]^d$ such that
  $\norm{\vecz^{\star} - M(\vecz^{\star})}_2 < \gamma$.
\end{nproblem}

\smallskip \noindent While not stated exactly in this form, the following is a
straightforward implication of the results presented in \cite{chen2009settling}.

\begin{lemma}[\cite{chen2009settling}] \label{lem:BrouwerPPAD}
    $\textsc{Brouwer}$ is $\PPAD$-complete even when $d = 2$.
  Additionally, $\textsc{Brouwer}$ is $\PPAD$-complete even when
  $\gamma = \poly(1/d)$ and $L = \poly(d)$.
\end{lemma}

\begin{remark}[Respresentation of a polynomial-time Turing Machine] \label{rem:representation}
  \em
  In the definition of the problem $\textsc{Brouwer}$ we assume that we are
given in the input the description of a Turing Machine $\calC_M$ that computes
the map $M$. In order for polynomial-time reductions to and from this problem to
be meaningful we need to have an upper bound on the running time of this Turing
Machine which we want to be polynomial in the input of the Turing Machine. The
formal way to ensure this and derive meaningful complexity results is to define
a different problem, say $k$-$\textsc{Brouwer}$, for every $k \in \N$. In the
problem $k$-$\textsc{Brouwer}$ the input Turing Machine $\calC_M$ has running
time bounded by $n^k$ in the size $n$ of its input. In the rest of the paper
whenever we say that a polynomial-time Turing Machine is required in the input
to a computational problem $\textsc{Pr}$, we formally mean that we define a
hierarchy of problems $k$-$\textsc{Pr}$, $k\in \mathbb{N}$, such that
$k$-$\textsc{Pr}$ takes as input Turing Machines with running time bounded by
$n^k$, and we interpret computational complexity results for $\textsc{Pr}$ in
the following way: whenever we prove that $\textsc{Pr}$ belongs to some
complexity class, we prove that $k$-$\textsc{Pr}$ belongs to the complexity
class for all $k \in \mathbb{N}$; whenever we prove that $\textsc{Pr}$ is hard
for some complexity class, we prove that, for some absolute constant $k_0$
determined in the hardness proof, $k$-$\textsc{Pr}$ is hard for that class,
for all $k \ge k_0$. For simplicity of exposition of our problems and results
we do not repeat this discussion in the rest of this paper.
\end{remark}

\section{Computational Problems of Interest} \label{sec:computational}

  In this section, we define the computational problems that we study in this
paper and discuss our main results, postponing formal statements to Section
\ref{sec:summary}. We start in Section \ref{sec:computational:math} by defining
the mathematical objects of our study, and proceed in Section
\ref{sec:computational:problems} to define our main computational problems,
namely: (1) finding approximate stationary points; (2) finding approximate local
minima; and (3) finding approximate local min-max equilibria. In Section
\ref{sec:computational:syntactic}, we present some bonus problems, which are
intimately related, as we will see, to problems (2) and (3). As discussed in
Section \ref{sec:model}, for ease of presentation, we define our problems as
promise problems.

\subsection{Mathematical Definitions} \label{sec:computational:math}

  We define the concepts of \textit{stationary points}, \textit{local minima},
and \textit{local min-max equilibria} of real valued functions, and make some
remarks about their existence, as well as their computational complexity. The
formal discussion of the latter is postponed to Sections
\ref{sec:computational:problems} and \ref{sec:summary}.

  Before we proceed with our definitions, recall that the goal of this paper is
to study constrained optimization. Our domain will be the hypercube $[0, 1]^d$,
which we might intersect with the set $\{\vecx~|~\vecg(\vecx) \le \vec{0}\}$,
for some convex (potentially multivariate) function $\vecg$. Although most of
the definitions and results that we explore in this paper can be extended to
arbitrary convex functions, we will focus on the case where $\vecg$ is linear,
and the feasible set is thus a polytope. Focusing on this case avoids additional
complications related to the representation of $\vecg$ in the input to the
computational problems that we define in the next section, and avoids also
issues related to verifying the convexity of $\vecg$.

\begin{definition}[Feasible Set and Refutation of Feasibility] \label{def:feasibleSet}
    Given $\matA \in \R^{d \times m}$ and $\vecb \in \R^m$, we define the set of
  feasible solutions to be
  $\calP(\matA, \vecb) = \{\vecz \in [0, 1]^d \mid  \matA^T \vecz  \le \vecb\}$.
  Observe that testing whether $\calP(\matA, \vecb)$ is empty can be done in
  polynomial time in the bit complexity of $\matA$ and $\vecb$.
\end{definition}

\begin{definition}[Projection Operator] \label{def:projectionOperator}
    For a nonempty, closed, and convex set $K \subset \R^d$, we define the projection operator
  $\Pi_K : \R^d \to K$ as follows
  $\Pi_K ~ \vecx = \argmin_{\vecy \in K} \norm{\vecx - \vecy}_2$. It is well-known that for any nonempty, closed, and convex set $K$ the
  $\argmin_{\vecy \in K} \norm{\vecx - \vecy}_2$ exists and is unique, hence
  $\Pi_K$ is well defined.
\end{definition}

  Now that we have defined the domain of the real-valued functions that we
consider in this paper we are ready to define a notion of approximate stationary
points.

\begin{definition}[$\eps$-Stationary Point] \label{def:math:stationary}
    Let $f : [0, 1]^d \to \R$ be a $G$-Lipschitz and $L$-smooth function and
  $\matA \in \R^{d \times m}$, $\vecb \in \R^m$. We call a point
  $\vecx^{\star} \in \calP(\matA, \vecb)$ a $\eps$-\textit{stationary point} of
  $f$ if $\norm{\nabla f(\vecx^{\star})}_2 < \eps$.
\end{definition}

\noindent It is easy to see that there exist continuously differentiable
functions $f$ that do not have any (approximate) stationary points, e.g. linear
functions. As we will see later in this paper, deciding whether a given function
$f$ has a stationary point is $\NP$-hard and, in fact, it is even $\NP$-hard to
decide whether a function has an approximate stationary point of a very gross
approximation. At the same time, verifying whether a given point is
(approximately) stationary can be done efficiently given access to a
polynomial-time Turing machine that computes $\nabla f$, so the problem of
deciding whether an (approximate) stationary point exists lies in $\NP$, as long
as we can guarantee that, if there is such a point, there will also be one with
polynomial bit complexity. We postpone a formal discussion of the computational
complexity of finding (approximate) stationary points or deciding their
existence  until we have formally defined our corresponding computational
problem  and settled the bit complexity of its solutions.
\smallskip

  For the definition of local minima and local min-max equilibria we need the
notion of closed $d$-dimensional Euclidean balls.

\begin{definition}[Euclidean Ball] \label{def:EuclideanBall}
    For $r \in \R_+$ we define the
  \textit{closed Euclidean ball of radius $r$} to be the set
  $\class{B}_d(r) = \set{\vecx \in \R^d \mid \norm{\vecx}_2 \le r}$. We also
  define the
  \textit{closed Euclidean ball of radius $r$ centered at $\vecz \in \R^d$} to
  be the set
  $\class{B}_d(r; \vecz) = \set{\vecx \in \R^d \mid \norm{\vecx - \vecz}_2 \le r}$.
\end{definition}

\begin{definition}[$(\eps, \delta)$-Local Minimum] \label{def:math:localMinimum}
    Let $f : [0, 1]^d \to \R$ be a $G$-Lipschitz and $L$-smooth function,
  $\matA \in \R^{d \times m}$, $\vecb \in \R^m$, and $\eps,\delta>0$. A point
  $\vecx^{\star} \in \calP(\matA, \vecb)$ is an
  $(\eps, \delta)$-\textit{local minimum} of $f$ constrained on $\calP(\matA, \vecb)$
  if and only if $f(\vecx^{\star}) < f(\vecx) + \eps$ for every
  $\vecx \in \calP(\matA, \vecb)$ such that
  $\vecx \in \class{B}_d(\delta; \vecx^{\star})$.
\end{definition}

\noindent To be clear, using the term ``local minimum'' in Definition
\ref{def:math:localMinimum} is a bit of a misnomer, since for large enough
values of $\delta$ the definition captures global minima as well. As $\delta$
ranges from large to small, our notion of $(\eps, \delta)$-local minimum
transitions from being an $\eps$-globally optimal point to being an
$\eps$-locally optimal point. Importantly, unlike (approximate) stationary
points, a $(\eps, \delta)$-local minimum is guaranteed to exist for all
$\eps,\delta > 0$ due to the compactness of $[0, 1]^d \cap \calP(\matA, \vecb)$ and
the continuity of $f$. Thus the problem of finding an $(\eps, \delta)$-local
minimum is \textit{total} for arbitrary values of $\eps$ and $\delta$. On the
negative side, for arbitrary values of $\eps$ and $\delta$, there is no
polynomial-size and polynomial-time verifiable witness for certifying that a
point $\vecx^{\star}$ is an $(\eps, \delta)$-local minimum. Thus the problem of
finding an $(\eps, \delta)$-local minimum is not known to lie in $\FNP$. As we
will see in Section \ref{sec:summary}, this issue can be circumvented if we
focus on particular settings of $\eps$ and $\delta$, in relationship to the
Lipschitzness and smoothness of $f$ and the dimension $d$.
\smallskip

  Finally we define $(\eps, \delta)$-local min-max equilibrium as follows,
recasting Definition \ref{def:general constraint local min-max} to the
constraint set $\calP(\matA, \vecb)$.

\begin{definition}[$(\eps, \delta)$-Local Min-Max Equilibrium] \label{def:math:localNash}
    Let $f : [0, 1]^{d_1} \times [0, 1]^{d_2} \to \R$ be a $G$-Lipschitz and
  $L$-smooth function, $\matA \in \R^{d \times m}$ and $\vecb \in \R^m$, where
  $d = d_1 + d_2$, and $\eps,\delta>0$. A point
  $(\vecx^{\star}, \vecy^{\star}) \in \calP(\matA, \vecb)$ is an
  $(\eps, \delta)$-\textit{local min-max equilibrium} of $f$ if and only if the
  following hold:
  \begin{enumerate}
    \item[$\blacktriangleright$]
      $f(\vecx^{\star}, \vecy^{\star}) < f(\vecx, \vecy^{\star}) + \eps$ for
      every $\vecx \in \class{B}_{d_1}(\delta; \vecx^{\star})$ with
      $(\vecx, \vecy^{\star}) \in \calP(\matA, \vecb)$; and
    \item[$\blacktriangleright$]
      $f(\vecx^{\star}, \vecy^{\star}) > f(\vecx^{\star}, \vecy) - \eps$ for
      every $\vecy \in \class{B}_{d_2}(\delta; \vecy^{\star})$ with
      $(\vecx^{\star}, \vecy) \in \calP(\matA, \vecb)$.
  \end{enumerate}
\end{definition}
\smallskip

\noindent Similarly to Definition \ref{def:math:localMinimum}, for large enough
values of $\delta$, Definition \ref{def:math:localNash} captures global min-max
equilibria as well. As $\delta$ ranges from large to small, our notion of
$(\eps, \delta)$-local min-max equilibrium transitions from being an
$\eps$-approximate min-max equilibrium to being an $\eps$-approximate local
min-max equilibrium. Moreover, in comparison to local minima and stationary
points, the problem of finding an $(\eps, \delta)$-local min-max equilibrium is
neither total nor can its solutions be verified efficiently for all values of
$\eps$ and $\delta$, even when $\calP(\matA, \vecb) = [0, 1]^d$. Again, this
issue can be circumvented if we focus on particular settings of $\eps$ and
$\delta$ values, as we will see in Section \ref{sec:summary}.

\subsection{First-Order Local Optimization Computational Problems} \label{sec:computational:problems}

  In this section, we define the search problems associated with our
aforementioned definitions of approximate stationary points, local minima, and
local min-max equilibria. We state our problems in terms of white-box access to
the function $f$ and its gradient. Switching to the black-box variants of our
computational problems amounts to simply replacing the Turing machines provided
in the input of the problems with oracle access to the function and its
gradient, as discussed in Section \ref{sec:model}. As per our discussion in the
same section, we define our computational problems as \textit{promise problems},
the promise being that the Turing machine (or oracle) provided in the input to
our problems outputs function values and gradient values that are consistent
with a smooth and Lipschitz function with the prescribed in the input smoothness
and Lipschitzness. Besides making the presentation cleaner, as we discussed in
Section \ref{sec:model}, the motivation for doing so is to prevent the
possibility that computational complexity is tacked into our problems due to the
possibility that the Turing machines/oracles provided in the input do not output
function and gradient values that are consistent with a Lipschitz and smooth
function. Importantly, all our computational hardness results syntactically
guarantee that the Turing machines/oracles provided as input to our constructed
hard instances satisfy these constraints.

  Before stating our main computational problems below, we note that, for each
problem, the dimension $d$ (in unary representation) is also an implicit input,
as the description of the Turing machine $\calC_f$ (or the interface to the
oracle $\calO_f$ in the black-box counterpart of each problem below) has size at
least linear in $d$. We also refer to Remark \ref{rem:representation} for how we
may formally study complexity problems that take a polynomial-time Turing
Machine in their input.

\medskip
\begin{nproblem}[\textsc{StationaryPoint}]
  \textsc{Input:} Scalars $\eps, G, L, B > 0$ and a polynomial-time Turing
  machine $\calC_f$ evaluating a $G$-Lipschitz and $L$-smooth function
  $f : [0, 1]^d \to [-B, B]$ and its gradient $\nabla f: [0, 1]^d \to \R^d$; a
  matrix $\matA \in \R^{d \times m}$ and vector $\vecb \in \R^m$ such that
  $\calP(\matA, \vecb) \neq \emptyset$.
  \smallskip

  \noindent \textsc{Output:} If there exists some point
  $\vecx \in \calP(\matA, \vecb)$ such that
  $\norm{\nabla f(\vecx)}_2 < \eps/2$, output some point
  $\vecx^{\star} \in \calP(\matA, \vecb)$ such that
  $\norm{\nabla f(\vecx^{\star})}_{2} < \eps$; if, for all
  $\vecx \in \calP(\matA, \vecb)$, $\norm{\nabla f(\vecx)}_{2} > \eps$, output
  $\bot$; otherwise, it is allowed to either output
  $\vecx^{\star} \in \calP(\matA, \vecb)$ such that
  $\norm{\nabla f(\vecx^{\star})}_2 < \eps$ or to output $\bot$.
\end{nproblem}
\medskip

\noindent It is easy to see that $\stat$ lies in $\FNP$. Indeed, if there exists
some point $\vecx \in \calP(\matA, \vecb)$ such that
$\norm{\nabla f(\vecx)}_2 < \eps/2$, then by the $L$-smoothness of $f$ there
must exist some point $\vecx^{\star} \in \calP(\matA, \vecb)$ of bit complexity
polynomial in the size of the input such that
$\norm{\nabla f(\vecx^{\star})}_2 < \eps$.  On the other hand, it is clear that
no such point exists if for all $\vecx \in \calP(\matA, \vecb)$,
$\norm{\nabla f(\vecx)}_{2} > \eps$. We note that the looseness of the output
requirement in our problem for functions $f$ that do not have points
$x \in \calP(\matA, \vecb)$ such that $\norm{\nabla f(\vecx)}_{2} < \eps/2$ but
do have points $x \in \calP(\matA, \vecb)$ such that
$\norm{\nabla f(\vecx)}_{2} \le \eps$ is introduced for the sole purpose of
making the problem lie in $\FNP$, as otherwise we would not be able to guarantee
that the solutions to our search problem have polynomial bit complexity. As we
show in Section \ref{sec:summary}, $\stat$ is also $\FNP$-hard, even when $\eps$
is a constant, the constraint set is very simple, namely
$\calP(\matA, \vecb) = [0, 1]^d$, and $G, L$ are both polynomial in $d$.

\smallskip
  Next, we define the computational problems associated with local minimum and
local min-max equilibrium. Recall that the first is guaranteed to have a
solution, because, in particular, a global minimum exists due to the continuity
of $f$ and the compactness of $\calP(\matA, \vecb)$.

\smallskip
\begin{nproblem}[\textsc{LocalMin}]
  \textsc{Input:} Scalars $\eps, \delta, G, L, B > 0$ and a polynomial-time
  Turing machine $\calC_f$ evaluating a $G$-Lipschitz and $L$-smooth function
  $f : [0, 1]^d \to [-B, B]$ and its gradient $\nabla f: [0, 1]^d \to \R^d$; a
  matrix $\matA \in \R^{d \times m}$ and vector $\vecb \in \R^m$ such that
  $\calP(\matA, \vecb) \neq \emptyset$.
  \smallskip

  \noindent \textsc{Output:} A point $\vecx^{\star} \in \calP(\matA, \vecb)$
  such that $f(\vecx^{\star}) < f(\vecx) + \eps$ for all
  $\vecx \in \ball_d(\delta; \vecx^{\star}) \cap \calP(\matA, \vecb)$.
\end{nproblem}
\medskip

\begin{nproblem}[\textsc{LocalMinMax}]
  \textsc{Input:} Scalars $\eps, \delta, G, L, B > 0$; a polynomial-time Turing
  machine $\calC_f$ evaluating a  $G$-Lipschitz and $L$-smooth function
  $f : [0, 1]^{d_1} \times [0, 1]^{d_2} \to [-B, B]$ and its gradient
  $\nabla f: [0, 1]^{d_1}\times [0, 1]^{d_2} \to \R^{d_1 + d_2}$; a matrix
  $\matA \in \R^{d \times m}$ and vector $\vecb \in \R^m$ such that
  $\calP(\matA, \vecb) \neq \emptyset$, where $d = d_1 + d_2$.
  \smallskip

  \noindent \textsc{Output:} A point
  $(\vecx^{\star}, \vecy^{\star}) \in \calP(\matA, \vecb)$ such that
  \begin{Enumerate}
    \item[$\triangleright$]
    $f(\vecx^{\star}, \vecy^{\star}) < f(\vecx, \vecy^{\star}) + \eps$ for all
    $\vecx \in B_{d_1}(\delta; \vecx^{\star})$ with
    $(\vecx, \vecy^{\star}) \in \calP(\matA, \vecb)$ and
    \item[$\triangleright$]
    $f(\vecx^{\star}, \vecy^{\star}) > f(\vecx^{\star}, \vecy) - \eps$ for all
    $\vecy \in B_{d_2}(\delta; \vecy^{\star})$ with
    $(\vecx^{\star}, \vecy) \in \calP(\matA, \vecb)$,
  \end{Enumerate}
  or $\bot$ if no such point exists.
\end{nproblem}
\medskip

\noindent Unlike $\stat$ the problems $\lmin$ and $\lNash$ exhibit vastly
different behavior, depending on the values of the inputs $\eps$ and $\delta$ in
relationship to $G$, $L$ and $d$, as we will see in Section \ref{sec:summary}
where we summarize our computational complexity results. This range of behaviors
is rooted at our earlier remark that, depending on the value of $\delta$
provided in the input to these problems, they capture the complexity of finding
\textit{global} minima/min-max equilibria, for large values of $\delta$, as well
as finding \textit{local} minima/min-max equilibria, for small values of
$\delta$.

\subsection{Bonus Problems: Fixed Points of Gradient Descent/Gradient Descent-Ascent} \label{sec:computational:syntactic}

  Next we present a couple of bonus problems, $\gdFixed$ and $\gdaFixed$, which
respectively capture the computation of fixed points of the (projected) gradient
descent and the (projected) gradient descent-ascent dynamics, with learning rate
$ = 1$. As we see in Section \ref{sec:existence}, these problems are intimately
related, indeed equivalent under polynomial-time reductions, to problems $\lmin$ and $\lNash$ respectively, in certain regimes of the approximation parameters.
Before stating problems $\gdFixed$ and $\gdaFixed$, we define the mappings
$F_{GD}$ and $F_{GDA}$ whose fixed points these problems are targeting.

\begin{definition}[Projected Gradient Descent] \label{def:FGD}
    For a closed and convex $K \subseteq \R^d$ and some continuously
  differentiable function $f : K \to \R$, we define the
  \textit{Projected Gradient Descent Dynamics with learning rate $1$} as the map
  $F_{GD} : K \to K$, where $F_{GD}(\vecx) = \Pi_K (\vecx - \nabla f(\vecx))$.
\end{definition}

\begin{definition}[Projected Gradient Descent/Ascent] \label{def:FGDA}
    For a closed and convex $K \subseteq \R^{d_1} \times \R^{d_2}$ and some
  continuously differentiable function $f : K \to \R$, we define the
  \textit{Unsafe Projected Gradient Descent/Ascent Dynamic with learning rate $1$}
  as the map  $F_{GDA} : K \to \R^{d_1} \times \R^{d_2}$ defined as follows
  \[ F_{GDA}(\vecx, \vecy) \triangleq \begin{bmatrix} \Pi_{K(\vecy)} (\vecx - \nabla_{\vecx} f(\vecx, \vecy)) \\ \Pi_{K(\vecx)} (\vecy + \nabla_{\vecy} f(\vecx, \vecy)) \end{bmatrix}  \triangleq \begin{bmatrix} F_{GDAx} (\vecx, \vecy)\\F_{GDAy} (\vecx, \vecy) \end{bmatrix} \]
  \noindent for all $(\vecx, \vecy) \in K$, where
  $K(\vecy) = \{\vecx' \mid (\vecx', \vecy) \in K\}$ and
  $K(\vecx) = \{\vecy' \mid (\vecx, \vecy') \in K\}$.
\end{definition}

\noindent Note that $F_{GDA}$ is called ``unsafe'' because the projection
happens individually for $\vecx - \nabla_{\vecx} f(\vecx, \vecy)$ and
$\vecy + \nabla_{\vecy} f(\vecx, \vecy)$, thus $F_{GDA}(\vecx, \vecy)$ may not
lie in $K$. We also define the ``safe'' version $F_{sGDA}$, which projects the
pair
$(\vecx - \nabla_{\vecx} f(\vecx, \vecy), \vecy + \nabla_{\vecy} f(\vecx, \vecy))$
jointly onto $K$. As we show in Section \ref{sec:existence} (in particular inside the proof of Theorem~\ref{t:PPAD_inclusion}), computing fixed
points of $F_{GDA}$ and $F_{sGDA}$ are computationally equivalent so we stick to
$F_{GDA}$ which makes the presentation slightly cleaner.

  We are now ready to define $\gdFixed$ and $\gdaFixed$. As per earlier
discussions, we  define these computational problems as
\textit{promise problems}, the promise being that the Turing machine provided in
the input to these problems outputs function values and gradient values that are
consistent with a smooth and Lipschitz function with the prescribed, in the
input to these problems, smoothness and Lipschitzness.

\medskip
\begin{nproblem}[\gdFixed]
  \textsc{Input:} Scalars $\alpha, G, L, B > 0$ and a polynomial-time Turing
  machine $\calC_f$ evaluating a $G$-Lipschitz and $L$-smooth function
  $f : [0, 1]^d \to [-B, B]$ and its gradient $\nabla f: [0, 1]^d \to \R^d$; a
  matrix $\matA \in \R^{d \times m}$ and vector $\vecb \in \R^m$ such that
  $\calP(\matA, \vecb) \neq \emptyset$.
  \smallskip

  \noindent \textsc{Output:} A point $\vecx^{\star}\in \calP(\matA, \vecb)$ such
  that $\norm{\vecx^{\star} - F_{GD}(\vecx^{\star})}_2 < \alpha$, where
  $K = \calP(\matA, \vecb)$ is the projection set used in the definition of
  $F_{GD}$.
\end{nproblem}
\medskip

\begin{nproblem}[\gdaFixed]
  \textsc{Input:} Scalars $\alpha, G, L, B > 0$ and a polynomial-time Turing
  machine $\calC_f$ evaluating a  $G$-Lipschitz and $L$-smooth function
  $f : [0, 1]^{d_1}\times [0, 1]^{d_2} \to [-B, B]$ and its gradient
  $\nabla f: [0, 1]^{d_1}\times [0, 1]^{d_2} \to \R^{d_1 + d_2}$; a matrix
  $\matA \in \R^{d \times m}$ and vector $\vecb \in \R^m$ such that
  $\calP(\matA, \vecb) \neq \emptyset$, where $d = d_1 + d_2$.
  \smallskip

  \noindent \textsc{Output:} A point
  $(\vecx^{\star},\vecy^{\star})\in \calP(\matA, \vecb)$ such that
  $\norm{(\vecx^{\star}, \vecy^{\star}) - F_{GDA}(\vecx^{\star},\vecy^{\star})}_2 < \alpha$,
  where $K = \calP(\matA, \vecb)$ is the projection set used in the definition
  of $F_{GDA}$.
\end{nproblem}
\medskip

\noindent In Section \ref{sec:existence} we show that the problems $\gdFixed$ and $\lmin$ are equivalent under polynomial-time reductions, and the problems $\gdaFixed$ and $\lNash$ are equivalent under polynomial-time reductions, in certain regimes of the approximation parameters.

\section{Summary of Results} \label{sec:summary}

  In this section we summarize our results for the optimization problems that we
defined in the previous section. We start with our theorem about the complexity
of finding approximate stationary points, which we show to be $\FNP$-complete
even for large values of the approximation.

\begin{theorem}[Complexity of Finding Approximate Stationary Points] \label{thm:approximateStationaryPointsHardness}
    The computational problem $\stat$ is $\FNP$-complete, even when $\eps$ is
  set to any value $\le 1/24$, and even when $\calP(\matA, \vecb) = [0, 1]^d$,
  $G = \sqrt{d}$, $L = d$, and $B = 1$.
\end{theorem}

  It is folklore and easy to verify that  approximate stationary points always
exist and can be found in time $\poly(B, 1/\eps, L)$ when the domain of $f$ is
unconstrained, i.e.~it is the whole $\R^d$, and the range of $f$ is bounded,
i.e., when $f(\R^d) \subseteq [-B, B]$. Theorem
\ref{thm:approximateStationaryPointsHardness} implies that such a guarantee
should not be expected in the bounded domain case, where the existence of
approximate stationary points is not guaranteed and must also be verified. In
particular, it follows from our theorem that any algorithm that verifies the
existence of and computes approximate stationary points in the constrained case
should take time that is super-polynomial in at least one of $G$, $L$, or $d$,
unless $\P = \NP$. The proof of Theorem
\ref{thm:approximateStationaryPointsHardness} is based on an elegant
construction for converting (real valued) stationary points of an appropriately
constructed function to (binary) solutions of a target $\textsc{Sat}$ instance.
This conversion involves the use of Lov\'asz Local Lemma \cite{ErdosL1973}. The
details of the proof can be found in Appendix
\ref{sec:proof:approximateStationaryPointsHardness}.
\medskip

  The complexity of $\lmin$ and $\lNash$ is more difficult to characterize, as
the nature of these problems changes drastically depending on the relationship
of $\delta$ with with $\eps$, $G$, $L$ and $d$, which determines whether these
problems ask for a \textit{globally} vs \textit{locally} approximately optimal
solution. In particular, there are two regimes wherein the complexity of both
problems is simple to characterize.
\begin{enumerate}
  \item[$\triangleright$] \textbf{Global Regime.} When $\delta \ge \sqrt{d}$
  then both $\lmin$ and $\lNash$ ask for a \textit{globally} optimal solution.
  In this regime it is not difficult to see that both problems are $\FNP$-hard
  to solve even when $\eps = \Theta(1)$ and $G$, $L$ are $O(d)$
  (see Section \ref{sec:localNash}).

  \item[$\triangleright$] \textbf{Trivial Regime.} When $\delta$ satisfies
  $\delta < \eps/G$, then for every point $\vecz \in \calP(\matA, \vecb)$ it
  holds that $\abs{f(\vecz) - f(\vecz')} < \eps$ for every
  $\vecz' \in \ball_d(\delta; \vecz)$ with $\vecz' \in \calP(\matA, \vecb)$.
  Thus, every point $\vecz$ in the domain $\calP(\matA, \vecb)$ is a solution to
  both $\lmin$ and $\lNash$.
\end{enumerate}

\noindent It is clear from our discussion above, and in earlier sections, that,
to really capture the complexity of finding local as opposed to global
minima/min-max equilibria, we should restrict the value of $\delta$. We identify
the following regime, which we call the ``\textit{local regime}.'' As we argue shortly,
this regime is markedly different from the global regime identified above in
that (i) a solution is guaranteed to exist for both our problems of interest,
where in the global regime only $\lmin$ is guaranteed to have a solution; and
(ii) their computational complexity transitions to lower complexity classes.

\begin{enumerate}
  \item[$\triangleright$] \textbf{Local Regime.} Our main focus in this paper
  is the regime defined by $\delta < \sqrt{2\eps / L}$. In this regime it is
  well known that Projected Gradient Descent can solve $\lmin$ in time
  $O(B \cdot L / \eps)$ (see Appendix \ref{sec:gdStationary}). Our main
  interest is understanding the complexity of $\lNash$, which is not well
  understood in this regime. We note that the use of the constant $2$ in the
  constraint $\delta < \sqrt{2\eps / L}$ which defines the local regime has a
  natural motivation: consider a point $\vecz$ where a $L$-smooth function $f$
  has $\nabla f(\vecz) = 0$; it follows from the definition of smoothness that
  $\vecz$ is both an $(\eps, \delta)$-local min and an $(\eps, \delta)$-local
  min-max equilibrium, as long as $\delta < \sqrt{2 \eps / L}$.
\end{enumerate}

\noindent The following  theorems provide tight upper and lower bounds on the
computational complexity of solving $\lNash$ in the local regime. For
compactness, we define the following problem:

\begin{definition}[Local Regime $\lNash$] \label{def:local regime local Nash}
    We define the
  \textit{local-regime local min-max equilibrium computation problem}, in short
  $\lrlNash$, to be the search problem $\lNash$ restricted to instances in the
  local regime, i.e. satisfying $\delta < \sqrt{2 \eps/L}$.
\end{definition}

\begin{theorem}[Existence of Approximate Local Min-Max Equilibrium] \label{thm:localNashExistence}
    The computational problem $\lrlNash$ belongs to $\PPAD$. As a byproduct, if
  some function $f$ is $G$-Lipschitz and $L$-smooth, then an
  $(\eps, \delta)$-local min-max equilibrium is guaranteed to exist when
  $\delta < \sqrt{2 \eps/L}$, i.e.~in the local regime.
\end{theorem}

\begin{theorem}[Hardness of Finding Approximate Local Min-Max Equilibrium] \label{thm:localNashHardness}
    The search problem $\lrlNash$ is $\PPAD$-hard, for any
  $\delta \ge \sqrt{\eps/L}$, and even when it holds that
  $1/\eps = \poly(d)$, $G = \poly(d)$, $L = \poly(d)$, and $B = d$.
\end{theorem}

\noindent Theorem \ref{thm:localNashHardness} implies that any algorithm that
computes an $(\eps, \delta)$-local min-max equilibrium of a $G$-Lipschitz and
$L$-smooth function $f$ in the local regime should take time that is
super-polynomial in at least one of $1/\eps$, $G$, $L$ or $d$, unless
$\FP = \PPAD$. As such, the complexity of computing local min-max equilibria in
the local regime is markedly different from the complexity of computing local
minima, which can be found using Projected Gradient Descent in
$\poly(G, L, 1/\eps, d)$ time and function/gradient evaluations
(see Appendix \ref{sec:gdStationary}).
\medskip

  An important property of our reduction in the proof of Theorem
\ref{thm:localNashHardness} is that it is a \textit{black-box reduction}. We can
hence prove the following unconditional lower bound in the black-box model.

\begin{theorem}[Black-Box Lower Bound for Finding Approximate Local Min-Max Equilibrium] \label{thm:localNashBlackBoxLowerBound}
    Suppose $\matA \in \R^{d \times m}$ and $\vecb \in \R^m$ are given together
  with an oracle $\calO_f$ that outputs a $G$-Lipschtz and $L$-smooth function
  $f : \calP(\matA, \vecb) \to [-1, 1]$ and its gradient $\nabla f$. Let also
  $\delta \ge \sqrt{L/\eps}$, $\eps \le G^2/L$, and let all the parameters
  $1/\eps$, $1/\delta$, $L$, $G$ be upper bounded by $\poly(d)$. Then any
  algorithm that has access to $f$ only through $\calO_f$ and computes an
  $(\eps, \delta)$-local min-max equilibrium has to make a number of queries to
  $\calO_f$ that is exponential in at least one of the parameters: $1/\eps$,
  $G$, $L$ or $d$ even when $\calP(\matA, \vecb) \subseteq [0, 1]^d$.
\end{theorem}

  Our main goal in the rest of the paper is to provide the proofs of Theorems
\ref{thm:localNashExistence}, \ref{thm:localNashHardness}
and \ref{thm:localNashBlackBoxLowerBound}. In Section \ref{sec:existence}, we
show how to use Brouwer's fixed point theorem to prove the existence of
approximate local min-max equilibrium in the local regime. Moreover, we
establish an equivalence between $\lNash$ and $\gdaFixed$, in the local regime,
and show that both belong to $\PPAD$. In Sections \ref{sec:hardness:2D} and
\ref{sec:hardness}, we provide a detailed proof of our main result, i.e. Theorem
\ref{thm:localNashHardness}. Finally, in Section \ref{sec:blackBox}, we show how
our proof from Section \ref{sec:hardness} produces as a byproduct the black-box,
unconditional lower bound of Theorem \ref{thm:localNashBlackBoxLowerBound}. In
Section \ref{sec:seic}, we outline a useful interpolation technique which allows
as to interpolate a function given its values and the values of its gradient on
a hypergrid, so as to enforce the Lipschitzness and smoothness of the
interpolating function. We make heavy use of this technically involved result in
all our hardness proofs.

\section{Existence of Approximate Local Min-Max Equilibrium} \label{sec:existence}

  In this section, we establish the totality of \lrlNash, i.e.~\lNash\ for
instances satisfying $\delta < \sqrt{2\eps/L}$ as defined in
Definition~\ref{def:local regime local Nash}. In particular, we prove that every
$G$-Lipschitz and $L$-smooth function admits an $(\eps, \delta)$-local min-max
equilibrium, as long as $\delta < \sqrt{2\eps/L}$. A byproduct of our proof is in
fact that $\lrlNash$ lies inside $\PPAD$. Specifically the main tool that we use to
prove our result is a computational equivalence between the problem of finding
fixed points of the Gradient Descent/Ascent dynamic, i.e.~$\gdaFixed$, and the
problem $\lrlNash$. A similar equivalence between \gdFixed\ and \lmin\ also holds,
but the details of that are left to the reader as a simple exercise. Next, we first
present the equivalence between \gdaFixed\ and \lrlNash, and we then show that
\gdaFixed\ is in \PPAD, which then also establishes that \lrlNash\ is in \PPAD.

\begin{theorem}\label{t:LocalMaxMin-GDA}
    The search problems \lrlNash\ and \gdaFixed\ are equivalent under
  polynomial-time reductions. That is, there is a polynomial-time reduction from
  \lrlNash\ to \gdaFixed\ and vice versa. In particular, given some
  $\matA \in \R^{d \times m}$ and $\vecb \in \R^m$ such that
  $\calP(\matA, \vecb) \neq \emptyset$, along with a $G$-Lipschitz and $L$-smooth
  function $f : \calP(\matA, \vecb) \to \R$:
  \begin{enumerate}
    \item For arbitrary $\eps >0$ and $0 < \delta < \sqrt{2 \eps / L}$, suppose
    that $(\vecx^{\ast},\vecy^{\ast})\in \calP(\matA, \vecb)$ is an
    $\alpha$-approximate fixed point of $F_{GDA}$, i.e.,
    $\norm{(\vecx^{\ast}, \vecy^{\ast}) - F_{GDA}(\vecx^{\ast},\vecy^{\ast})}_2 < \alpha$, where
    $\alpha \le \frac{\sqrt{(G+\delta)^2+4(\eps - \frac{L}{2}\delta^2)} - (G+\delta)}{2}$.
    Then $(\vecx^\ast, \vecy^\ast)$ is also a $(\eps,\delta)$-local min-max equilibrium of $f$.
    \item For arbitary $\alpha>0$, suppose that $(\vecx^\ast,\vecy^\ast)$ is an $(\eps,\delta)$-local min-max equilibrium of $f$ for $\eps = \frac{\alpha^2 \cdot L}{(5 L + 2)^2}$ and $\delta = \sqrt{\eps/L}$. Then $(\vecx^\ast, \vecy^\ast)$ is also an $\alpha$-approximate fixed point of $F_{GDA}$.
    \end{enumerate}
\end{theorem}

\noindent The  proof of Theorem~\ref{t:LocalMaxMin-GDA} is presented in
Appendix~\ref{sec:proof:LocalMaxMin-GDA}. As already discussed, we use \gdaFixed\ as an intermediate step to establish the totality of \lrlNash\ and to show its inclusion in $\PPAD$. This leads to the
following theorem.

\begin{theorem}\label{t:PPAD_inclusion}
    The computational problems $\gdaFixed$ and $\lrlNash$ are both total search
  problems and they both lie in~$\PPAD$.
\end{theorem}

\noindent Observe that Theorem \ref{thm:localNashExistence} is implied by
Theorem \ref{t:PPAD_inclusion} whose proof is presented in Appendix
\ref{sec:proof:t:PPAD_inclusion}.

\section{Hardness of Local Min-Max Equilibrium -- Four-Dimensions} \label{sec:hardness:2D}

  In Section~\ref{sec:existence}, we established that $\lrlNash$ belongs to
$\PPAD$. Our proof is via the intermediate problem $\gdaFixed$ which we showed
that it is computationally equivalent to $\lrlNash$. Our next step is to prove
the $\PPAD$-hardness of $\lrlNash$ using again $\gdaFixed$ as an intermediate
problem.
\smallskip

  In this section we prove that $\gdaFixed$ is $\PPAD$-hard in four dimensions.
To establish this hardness result we introduce a variant of the classical
$\TD$-$\Sperner$ problem which we call $\TD$-$\BiSperner$ which we show is
$\PPAD$-hard. The main technical part of our proof is to show that
$\TD$-$\BiSperner$ with input size $n$ reduces to $\gdaFixed$, with input
size $\poly(n)$, $\alpha = \exp(-\poly(n))$, $G = L = \exp(\poly(n))$, and
$B = 2$. This reduction proves the hardness of $\gdaFixed$. Formally, our main
result of this section is the following theorem.

\begin{theorem} \label{thm:hardness2D}
    The problem $\gdaFixed$ is $\PPAD$-complete even in dimension $d = 4$
  and $B = 2$. Therefore, $\lrlNash$ is $\PPAD$-complete even in dimension
  $d = 4$ and $B = 2$.
\end{theorem}

\noindent The above result excludes the existence of an algorithm for
$\gdaFixed$ whose running time is $\poly(\log G,\log L,\log(1/\alpha), B)$ and,
equivalently, the existence of an algorithm for the problem $\lrlNash$ with
running time $\poly(\log G,\log L,\log(1/\eps), \log(1/\delta), B)$, unless
$\FP = \PPAD$. Observe that it would  not be possible to get a stronger hardness
result for the four dimensional $\gdaFixed$ problem since it is simple to
construct brute-force search algorithms with running time
$\poly(1/\alpha, G, L, B)$. We elaborate more on such algorithms towards the end
of this section. In order to prove the hardness of $\gdaFixed$ for polynomially
(rather than exponentially) bounded (in the size of the input) values of
$1/\alpha$, $G$, and $L$ (See Theorem \ref{thm:localNashHardness}) we need to
consider optimization problems in higher dimensions. This is the problem that we
explore in Section \ref{sec:hardness}. Beyond establishing the hardness of the
problem for $d = 4$ dimensions, the purpose of this section is to provide a
simpler reduction that helps in the understanding of our main result in the next
section.

\subsection{The 2D Bi-Sperner Problem} \label{sec:2DBiSperner}

  We start by introducing the $\TD$-$\BiSperner$ problem. Consider a coloring of
the $N \times N$, $2$-dimensional grid, where instead of coloring each vertex of
the grid with a single color (as in Sperner's lemma), each vertex is colored
via a combination of two out of  four available colors. The four available
colors are $1^-, 1^+, 2^-, 2^+$. The five rules that define a proper coloring of
the $N \times N$ grid are the following.
\medskip

\begin{enumerate}
  \item The first color of every vertex is either $1^-$ or $1^+$ and the second
  color is either $2^-$ or $2^+$.
  \item The first color of all  vertices on the left boundary of the grid is
  $1^+$. \label{legal1}
  \item The first color of all  vertices on the right boundary of the grid is
  $1^-$. \label{legal2}
  \item The second color of all  vertices on the bottom boundary of the grid
  is $2^+$. \label{legal3}
  \item The second color of all  vertices on the top boundary of the grid is
  $2^-$. \label{legal4}
\end{enumerate}

\begin{figure}[t]
  \centering
  \includegraphics[scale=1.2]{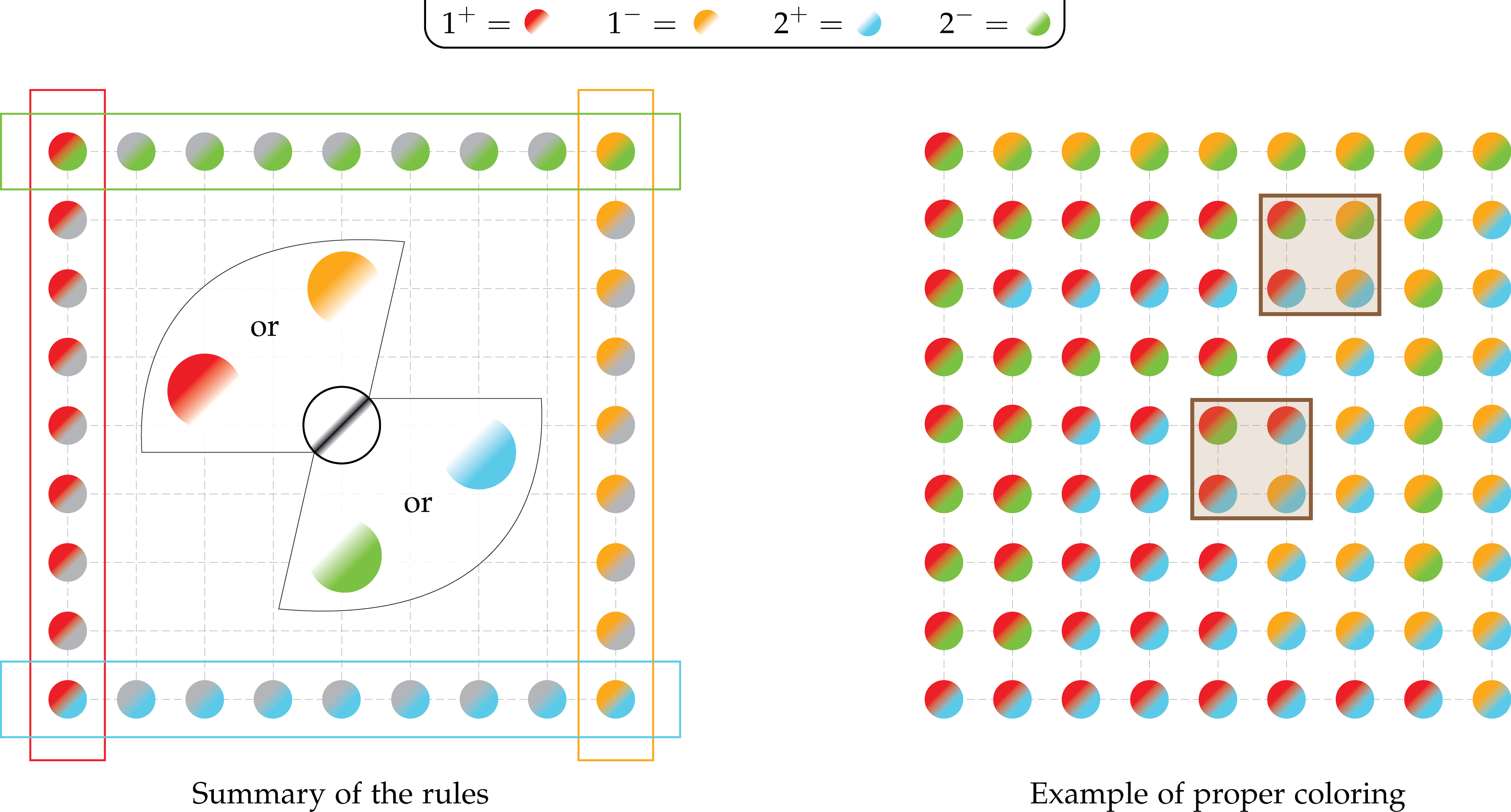}
  \caption{\textit{Left}: Summary of the rules from a proper coloring of the
  grid. The \textcolor{lgray}{\textbf{gray}} color on the left and the right
  side can be replaced with either \textcolor{lblue}{\textbf{blue}} or
  \textcolor{lgreen}{\textbf{green}}. Similarly the
  \textcolor{lgray}{\textbf{gray}} color on the top and the bottom side can be
  replaced with either \textcolor{lred}{\textbf{red}} or
  \textcolor{lyellow}{\textbf{yellow}}. \textit{Right:} An example of a proper
  coloring of a $9 \times 9$ grid. The \textcolor{lbrown}{\textbf{brown}} boxes
  indicate the two panchromatic cells, i.e., the cells where all the four available
  colors appear.}
  \label{fig:biSpernerRules&Example}
\end{figure}

\noindent Using similar proof ideas as in Sperner's lemma it is not
hard to establish via a combinatorial argument that, in every proper coloring of the $N \times N$ grid, there
exists a square cell where each of the four colors in $\{1^-, 1^+, 2^-, 2^+\}$
appears in at least one of its vertices. We call such a cell a \textit{panchromatic square}. In
the $\TD$-$\BiSperner$ problem, defined formally below, we are given the description of some coloring of the grid and are asked to find either a panchromatic square or the violation of the proper coloring conditions. In this paper,
we will not present a direct combinatorial argument guaranteeing the existence of panchromatic squares under proper colorings of the grid, since the existence of panchromatic squares will be implied by the totality of the
$\TD$-$\BiSperner$ problem, which will follow from our reduction from
$\TD$-$\BiSperner$ to $\gdaFixed$ as well as our proofs in Section \ref{sec:existence}
establishing the totality of $\gdaFixed$. In Figure
\ref{fig:biSpernerRules&Example} we summarize the five rules that define proper
colorings and we present an example of a proper coloring of the grid with $9$
discrete points on each side.
\medskip

  In order to formally define the computational problem $\TD$-$\BiSperner$ in a
way that is useful for our reductions we need to allow for colorings of the
$N \times N$ grid described in a succinct way, where the value $N$ can be
exponentially large compared to the size of the input to the problem. A standard
way to do this, introduced by \cite{Papadimitriou1994} in defining the
computational version of Sperner's lemma, is to describe a coloring via a binary
circuit $\calC_l$ that takes as input the coordinates of a vertex in the grid
and outputs the combination of colors that is used to color this vertex. In the
input, each one of the two coordinates of the input vertex is given via the
binary representation of a number in $\nm{N}$. Setting $N = 2^n$ we have that
the representation of each coordinate  belongs to $\{0, 1\}^n$. In the rest of
the section we abuse the notation and we use a coordinate $i \in \{0, 1\}^n$
both as a binary string and as a number in $\nm{2^n}$ and it is clear from the
context which of the two we use. The output of $\calC_l$ should be a combination
of one of the colors $\{1^-, 1^+\}$ and one of the colors $\{2^-, 2^+\}$. We
represent this combination as a pair of $\{-1, 1\}^2$. The first coordinate of this
pair refers to the choice of $1^-$ or $1^+$ and the second coordinate refers to the
choice of $2^-$ or $2^+$.
\medskip

  In the definition of the computational problem $\TD$-$\BiSperner$ the input is
a circuit $\calC_l$, as described above. One type of possible solutions to
$\TD$-$\BiSperner$ is providing a pair of coordinates
$(i, j) \in \{0, 1\}^n \times \{0, 1\}^n$ indexing a cell of the grid
whose bottom left vertex is $(i, j)$. For this type of solution to be valid it must
be that the output of $\calC_l$ when evaluated on all the vertices of this square
contains at least one negative and one positive value for each one of the two
output coordinates of $\calC_l$, i.e.~the cell must be panchromatic. Another type
of possible solution to $\TD$-$\BiSperner$ is a vertex whose coloring violates the
proper coloring conditions for the boundary, namely~\ref{legal1}--\ref{legal4}
above. For notational convenience we refer to the first coordinate of the output of
$\calC_l$ by $\calC_l^1$ and to the second coordinate by $\calC_l^2$. The formal
definition of the computational problem  $\TD$-$\BiSperner$ is then the following.
\medskip

\begin{nproblem}[\TD\text{-}\BiSperner]
  \textsc{Input:} A boolean circuit
  $\calC_l : \{0, 1\}^n \times \{0, 1\}^n \to \{-1, 1\}^2$.
  \smallskip

  \noindent \textsc{Output:} A vertex $(i, j) \in \{0, 1\}^n \times \{0, 1\}^n$
  such that one of the following holds
  \begin{Enumerate}
    \item $i \neq \vec{1}$, $j \neq \vec{1}$, and
    \[ \bigcup_{\substack{i' - i \in \{0, 1\} \\ j' - j \in \{0, 1\}}} \calC_l^1(i', j') = \{-1 , 1\} ~~~~ \text{ and } ~~~~ \bigcup_{\substack{i' - i \in \{0, 1\} \\ j' - j \in \{0, 1\}}} \calC_l^2(i', j') = \{-1 , 1\}, \text{ or} \]
    \item $i = \vec{0}$ and $\calC_l^1(i, j) = -1$, or
    \item $i = \vec{1}$ and $\calC_l^1(i, j) = +1$, or
    \item $j = \vec{0}$ and $\calC_l^2(i, j) = -1$, or
    \item $j = \vec{1}$ and $\calC_l^2(i, j) = +1$.
  \end{Enumerate}
\end{nproblem}
\medskip

  Our next step is to show that the problem $\TD$-$\BiSperner$ is $\PPAD$-hard.
Thus our reduction from $\TD$-$\BiSperner$ to $\gdaFixed$ in the next section establishes both the
$\PPAD$-hardness of $\gdaFixed$ and the inclusion of $\TD$-$\BiSperner$ to
$\PPAD$.

\begin{lemma} \label{lem:hardness2DBiSperner}
    The problem $\TD$-$\BiSperner$ is $\PPAD$-hard.
\end{lemma}

\begin{proof}
    To prove this Lemma we will use Lemma \ref{lem:BrouwerPPAD}. Let $\calC_M$
  be a polynomial-time Turing machine that computes a function
  $M : [0, 1]^2 \to [0, 1]^2$ that is $L$-Lipschitz. We know from Lemma
  \ref{lem:BrouwerPPAD} that finding $\gamma$-approximate fixed points of $M$
  is $\PPAD$-hard. We will use $\calC_M$ to define a circuit $\calC_l$ such that
  a solution of $\TD$-$\BiSperner$ with input $\calC_l$ will give us a
  $\gamma$-approximate fixed point of $M$.
  \smallskip

    Consider the function $g(\vecx) = M(\vecx) - \vecx$. Since $M$ is
  $L$-Lipschitz, the function $g: [0, 1]^2 \to [-1,1]^2$ is also
  $(L + 1)$-Lipschitz. Additionally $g$ can be easily computed via a
  polynomial-time Turing machine $\calC_g$ that uses $\calC_M$ as a subroutine.
  We construct a proper coloring of a fine grid of $[0, 1]^2$ using the signs of
  the outputs of $g$. Namely we set $n = \ceil{\log(L/\gamma) + 2}$ and this
  defines a $2^n \times 2^n$ grid over $[0,1]^2$ that is indexed by
  $\{0, 1\}^n \times \{0, 1\}^n$. Let $g_{\eta} : [0, 1]^2 \to [-1, 1]^2$ be the
  function that the Turing Machine $\calC_g$ evaluate when the requested
  accuracy is $\eta > 0$. Now we can define the circuit $\calC_l$ as
  follows, \footnote{We remind that we abuse the notation and we use a coordinate
  $i \in \{0, 1\}^n$ both as a binary string and as a number in
  $\p{\nm{2^n - 1}}$ and it is clear from the context which of the two we use.}
  \[ \calC_l^1(i, j) = \begin{cases}

                        1  &  i = 0 \\
                                              -1  &  i = 2^n - 1\\

                         1  & g_{\eta, 1}\p{\frac{i}{2^n - 1}, \frac{j}{2^n - 1}} \ge 0 \text{ and } i \neq -1\\
                         -1 & g_{\eta, 1}\p{\frac{i}{2^n - 1}, \frac{j}{2^n - 1}} < 0 \text{ and } i \neq 0
                       \end{cases}, \]
  \[ \calC_l^2(i, j) = \begin{cases}

                        1  &  i = 0 \\
                                              -1  &  i = 2^n - 1\\

                         1  & g_{\eta, 2}\p{\frac{i}{2^n - 2}, \frac{j}{2^n - 1}} \ge 0 \text{ and } i \neq -1  \\
                         -1 & g_{\eta, 2}\p{\frac{i}{2^n - 2}, \frac{j}{2^n - 1}} < 0 \text{ and } i \neq 0
                       \end{cases}, \]
  where $g_i$ is the $i$th output coordinate of $g$. It is not hard then to
  observe that the coloring $\calC_l$ is proper, i.e. it satisfies the boundary
  conditions due to the fact that the image of $M$ is always inside $[0, 1]^2$.
  Therefore the only possible solution to $\TD$-$\BiSperner$ with input
  $\calC_l$ is a cell that contains all the colors $\{1^-, 1^+, 2^-, 2^+\}$. Let
  $(i, j)$ be the bottom-left vertex of this cell which we denote by $R$, namely
  \[ R = \set{\vecx \in [0, 1]^2 \mid x_1 \in \b{\frac{i}{2^n - 1}, \frac{i + 1}{2^n - 1}}, x_2 \in \b{\frac{j}{2^n - 1}, \frac{j + 1}{2^n - 1}}}. \]

  \begin{claim} \label{clm:proof:lem:hardness2DBiSperner}
      Let $\eta = \frac{\gamma}{2 \sqrt{2}}$, there exists $\vecx \in R$ such
    that $\abs{g_1(\vecx)} \le \frac{\gamma}{2 \sqrt{2}}$ and $\vecy \in R$ such
    that $\abs{g_2(\vecy)} \le \frac{\gamma}{2 \sqrt{2}}$.
  \end{claim}

  \begin{proof}[Proof of Claim \ref{clm:proof:lem:hardness2DBiSperner}]
      We will prove the existence of $\vecx$ and the existence of $\vecy$
    follows using an identical argument. If there exists a corner $\vecx$ of $R$
    such that $g_1(\vecx)$ is in the range $[-\eta, \eta]$ then the claim
    follows. Suppose not. Using this together with the fact that the first color
    of one of the corners of $R$ is $1^-$ and also the first color of one of the
    corners of $R$ is $1^+$ we conclude that there exist points $\vecx, \vecx'$
    such that $g_{\eta, 1}(\vecx) \ge 0$ and $g_{\eta, 1}(\vecx') \le 0$ \footnote{ The latter is inaccurate for the cases where the vertex $(0,j)$ belongs to either facets $i=0$ or $i=2^{n}-1$. Notice that the coloring in such vertices does not depend on the value of $g_{\eta}$. However in case where the color of such a corner is not consistent with the value of $g_{\eta}$, i.e. $g_{\eta,1}( 0 , j) <0$ and  $\calC_l^1(0, j) = 1$ then this means that $|g_{1}(0,j)| \leq \eta$. This is due to the fact that $g_{1}(0,j) \geq 0$ and $|g_{1}(0,j) - g_{1,\eta}(0,j) |\leq \eta$.}.  But we
    have that $\norm{g_{\eta} - g}_2 \le \eta$. This together with the fact that
    $g_1(\vecx) \not\in [-\eta, \eta]$ and $g_1(\vecx') \not\in [-\eta, \eta]$
    implies that $g_1(\vecx) \ge 0$ and also $g_1(\vecx') \le 0$. But because of
    the $L$-Lipschitzness of $g$ and because the distance between $\vecx$ and
    $\vecx'$ is at most $\sqrt{2} \frac{\gamma}{4 L}$ we conclude that
    $\abs{g_1(\vecx) - g_1(\vecx')} \le \frac{\gamma}{2 \sqrt{2}}$. Hence due to
    the signs of $g_1(\vecx)$ and $g_1(\vecx')$ we conclude that
    $\abs{g_1(\vecx)} \le \frac{\gamma}{2 \sqrt{2}}$. The same way we can prove
    that $\abs{g_1(\vecy)} \le \frac{\gamma}{2 \sqrt{2}}$ and the claim follows.
  \end{proof}

  \noindent Using the Claim \ref{clm:proof:lem:hardness2DBiSperner} and the
  $L$-Lipschitzness of $g$ we get that for every $\vecz \in R$
  \begin{align*}
    \abs{g_1(\vecz) - g_1(\vecx)} & \le L \norm{\vecx - \vecz}_2 \le \sqrt{2} \cdot L \cdot \frac{\gamma}{4 L} \implies \abs{g_1(\vecz)} \le \frac{\gamma}{\sqrt{2}}, \text{ and } \\
    \abs{g_2(\vecz) - g_2(\vecy)} & \le L \norm{\vecy - \vecz}_2 \le \sqrt{2} \cdot L \cdot \frac{\gamma}{4 L} \implies \abs{g_2(\vecz)} \le \frac{\gamma}{\sqrt{2}}
  \end{align*}
  where we have used also the fact that for any two points $\vecz, \vecw$ it
  holds that $\norm{\vecz - \vecw}_{2} \le \sqrt{2} \frac{\gamma}{4 L}$ which
  follows from the definition of the size of the grid. Therefore we have that
  $\norm{g(\vecz)}_2 \le \gamma$ and hence
  $\norm{M(\vecz) - \vecz}_2 \le \gamma$ which implies that any point
  $\vecz \in R$ is a $\gamma$-approximate fixed point of $M$ and the lemma
  follows.
\end{proof}

\noindent Now that we have established the $\PPAD$-hardness of
$\TD$-$\BiSperner$ we are ready to present our main result of this section which
is a reduction from $\TD$-$\BiSperner$ to $\gdaFixed$.

\subsection{From 2D Bi-Sperner to Fixed Points of Gradient Descent/Ascent} \label{sec:2DmainReduction}

  We start with presenting a construction of a Lipschitz and smooth real-valued
function $f: [0, 1]^2 \times [0, 1]^2 \to \R$ based on a given coloring circuit
$\calC_l : \{0, 1\}^n \times \{0, 1\}^n \to \{-1, 1\}^2$. Then in Section
\ref{sec:2DmainReduction:Reduction} we will show that any solution to
$\gdaFixed$ with input the representation $\calC_f$ of $f$ is also a solution to
the $\TD$-$\BiSperner$ problem with input $\calC_l$. Constructing Lipschitz and
smooth functions based on only local information is a surprisingly challenging
task in high-dimensions as we will explain in detail in Section
\ref{sec:hardness}. Fortunately in the low-dimensional case that we consider in
this section the construction is much more simple and the main ideas of our
reduction are more clear.
\smallskip

  The basic idea of the construction of $f$ consists in interpreting the
coloring of a given point in the grid as the directions of the gradient of
$f(\vecx, \vecy)$  with respect to the variables $x_1, y_1$ and $x_2, y_2$
respectively. More precisely, following the ideas in the proof of Lemma
\ref{lem:hardness2DBiSperner}, we divide the $[0, 1]^2$ square in
\textit{square-cells} of length $1/(N - 1) = 1/(2^n - 1)$ where the corners of
these cells correspond to vertices of the $N \times N$ grid of the
$\TD$-$\BiSperner$ instance described by $\calC_l$. When $\vecx$ is on a vertex
of this grid, the first color of this vertex determines the direction of
gradient with respect to the variables $x_1$ and $y_1$, while the second color
of this vertex determines the direction of the gradient of the variables $x_2$
and $y_2$. As an example, if $\vecx = (x_1, x_2)$ is on a vertex of the
$N \times N$ grid, and the coloring of this vertex is $(1^-,2^+)$, i.e. the
output of $\calC_l$ on this vertex is $(-1, +1)$, then we would like to have
\[ \frac{\partial f}{\partial x_1}(\vecx, \vecy) \ge 0, \quad \frac{\partial f}{\partial y_1}(\vecx, \vecy) \le 0, \quad \frac{\partial f}{\partial x_2}(\vecx, \vecy) \le 0, \quad \frac{\partial f}{\partial y_2}(\vecx, \vecy) \ge 0. \]
The simplest way to achieve this is to define the function $f$ locally close to
$(\vecx, \vecy)$ to be equal to
\[ f(\vecx,\vecy) = (x_1 - y_1) - (x_2 - y_2). \]
Similarly, if $\vecx$ is on a vertex of the $N \times N$ grid, and the coloring
of this vertex is $(1^-,2^-)$, i.e. the output of $\calC_l$ on this vertex is
$(-1, -1)$, then we would like to have
\[ \frac{\partial f}{\partial x_1}(\vecx, \vecy) \ge 0, \quad \frac{\partial f}{\partial y_1}(\vecx, \vecy) \le 0, \quad \frac{\partial f}{\partial x_2}(\vecx, \vecy) \ge 0, \quad \frac{\partial f}{\partial y_2}(\vecx, \vecy) \le 0. \]
The simplest way to achieve this is to define the function $f$ locally close to
$(\vecx, \vecy)$ to be equal to
\[ f(\vecx,\vecy) = (x_1 - y_1) + (x_2 - y_2). \]
\noindent In Figure \ref{fig:colorsAssignment} we show pictorially the
correspondence of the colors of the vertices of the grid with the gradient of
the function $f$ that we design. As shown in the figure, any set of vertices
that share at least one of the colors $1^+$, $1^-$, $2^+$, $2^-$, agree on the
direction of the gradient with respect  the horizontal or the vertical
axis. This observation is one of the main ingredients in the proof of
correctness of our reduction that we present  later in this section.

\begin{figure}[t]
  \centering
  \includegraphics[scale=1.2]{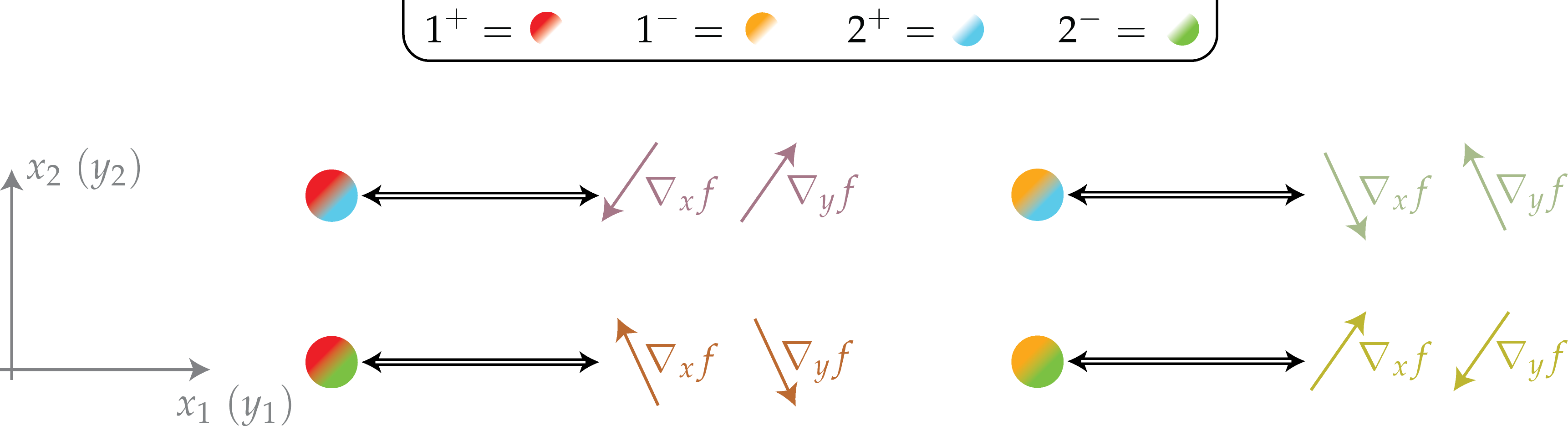}
  \caption{The correspondence of the colors of the vertices of the $N \times N$
  grid with the directions of the gradient of the function $f$ that we design.}
  \label{fig:colorsAssignment}
\end{figure}

  When $\vecx$ is not on a vertex of the $N \times N$ grid then our goal is to
define $f$ via interpolating the functions corresponding to the corners of the
cell in which $\vecx$ belongs. The reason that this interpolation is challenging
is that we need to make sure the following properties are satisfied
\begin{enumerate}
  \item[$\triangleright$] the resulting function $f$ is both Lipschitz and
  smooth inside every cell,
  \item[$\triangleright$] the resulting function $f$ is both Lipschitz and
  smooth even at the boundaries of every cell, where two differect cells stick
  together,
  \item[$\triangleright$] no solution to the $\gdaFixed$ problem is created
  inside cells that are not solutions to the $\TD$-$\BiSperner$ problem. In
  particular, it has to be true that if all the vertices of one cell agree on
  some color then the gradient of $f$ inside that cell has large enough gradient
  in the corresponding direction.
\end{enumerate}
\noindent For the low dimensional case, that we explore in this section,
satisfying the first two properties is not a very difficult task, whereas for
the third property we need to be careful and achieving this property is the main
technical contribution of this section. On the contrary, for the
high-dimensional case that we explore in Section \ref{sec:hardness} even
achieving the first two properties is very challenging and technical.

  As we will see in Section \ref{sec:2DmainReduction:Reduction}, if we
accomplish a construction of a function $f$ with the aforementioned properties,
then the fixed points of the projected Gradient Descent/Ascent can only appear
inside cells that have all of the colors $\{1^-, 1^+, 2^-, 2^+\}$ at their
corners. To see this consider a cell that misses some color, e.g. $1^+$. Then
all the corners of this cell have as first color $1^-$. Since $f$ is defined as
interpolation of the functions in the corners of the cells, with the
aforementioned properties, inside that cell there is always a direction with
respect to $x_1$ and $y_1$ for which the gradient is large enough. Hence any
point inside that cell cannot be a fixed point of the projected Gradient
Descent/Ascent. Of course this example provides just an intuition of our
construction and ignores case where the cell is on the boundary of the grid. We
provide a detailed explanation of this case in Section
\ref{sec:2DmainReduction:Reduction}.

  The above neat idea needs some technical adjustments in order to work. At
first, the interpolation of the function in the interior of the cell must be
smooth enough so that the resulting function is both Lipschitz and smooth. In
order to satisfy this, we need to choose appropriate coefficients of the
interpolation that interpolate smoothly not only the value of the function but
also its derivatives. For this purpose we use the following smooth step function
of order $1$.

\begin{definition}[Smooth Step Function of Order $1$]
    We define $S_1 : [0, 1] \to [0, 1]$ to be the
  \textit{smooth step function of order $1$} that is equal to
  $S_1(x) = 3 x^2 - 2 x^3$. Observe that the following hold $S_{1}(0) = 0$,
  $S_{1}(1) = 1$, $S_{1}'(0) = 0$, and $S_{1}'(1) = 0$.
\end{definition}

  As we have discussed, another issue is that since the interpolation
coefficients depend on the value of $\vecx$ it could be that the derivatives of
these coefficients overpower the derivatives of the functions that we
interpolate. In this case we could be potentially creating fixed points of
Gradient Descent/Ascent even in \textit{non} panchromatic squares. As we will
see later the magnitude of the derivatives from the interpolation coefficients
depends on the differences $x_1 - y_1$ and $x_2 - y_2$. Hence if we ensure that
these differences are small then the derivatives of the interpolation
coefficients will have to remain small and hence they can never overpower the
derivatives from the corners of every cell. This is the place in our reduction
where we add the constraints $\matA \cdot (\vecx, \vecy) \le \vecb$ that define
the domain of the function $f$ as we describe in Section
\ref{sec:computational}.
\smallskip

\noindent Now that we have summarized the main ideas of our construction we are
ready for the formal definition of $f$ based on the coloring circuit $\calC_l$.

\begin{definition}[Continuous and Smooth Function from Colorings of 2D-Bi-Sperner] \label{d:payoff}
    Given a binary circuit
  $\calC_l : \{0, 1\}^n \times \{0, 1\}^n \to \{-1, 1\}^2$, we
  define the function $f_{\calC_l} : [0, 1]^2 \times [0, 1]^2 \to \R$ as
  follows. For any $\vecx \in [0, 1]^2$, let $A = (i_A, j_A)$, $B = (i_B, j_B)$,
  $C = (i_C, j_C)$, $D = (i_D, j_D)$ be the vertices of the cell of the
  $N (= 2^n) \times N$ grid which contains $\vecx$ and $\vecx^A$, $\vecx^B$,
  $\vecx^C$ and $\vecx^C$ the corresponding points in the unit square
  $[0, 1]^2$, i.e. $x_1^A = i_A/(2^n - 1)$, $x_2^A = j_A/(2^n - 1)$ etc. Let also $A$ be
  down-left corner of this cell and $B$, $C$, $D$ be the rest of the vertices in
  clockwise order, then we define
  \[ f_{\calC_l}(\vecx, \vecy) = \alpha_1(\vecx) \cdot (y_1 - x_{1}) + \alpha_2(\vecx) \cdot (y_2 - x_2) \]
  \noindent where the coefficients
  $\alpha_1(\vecx), \alpha_2(\vecx) \in [-1, 1]$ are defined as follows
  \begin{eqnarray*}
    \alpha_i(\vecx)
    & = & S_1 \left(\frac{x^C_1 - x_1}{\delta} \right) \cdot S_1 \left(\frac{x^C_2 - x_2}{ \delta} \right) \cdot \calC_l^i(A) + S_1 \left(\frac{x^D_1 - x_1}{ \delta} \right) \cdot S_1 \left(\frac{x_2 - x^D_2}{\delta} \right) \cdot \calC_l^i(B) \\
    &   & + S_1 \left(\frac{x_1 - x^A_1}{ \delta} \right) \cdot S_1 \left(\frac{x_2 - x^A_2}{ \delta} \right) \cdot \calC_l^i(C) + S_1 \left(\frac{x_1 - x^B_1}{ \delta} \right) \cdot S_1 \left(\frac{x^B_2 - x_2}{ \delta} \right) \cdot \calC_l^i(D) \\
  \end{eqnarray*}
  where $\delta \triangleq 1/(N - 1) = 1/(2^n - 1)$.
\end{definition}

\noindent In Figure \ref{fig:interpolationExample} we present an example of the
application of Definition \ref{d:payoff} to a specific cell with some given
coloring on the corners.

\begin{figure}[t]
  \centering
  \includegraphics[scale=1.1]{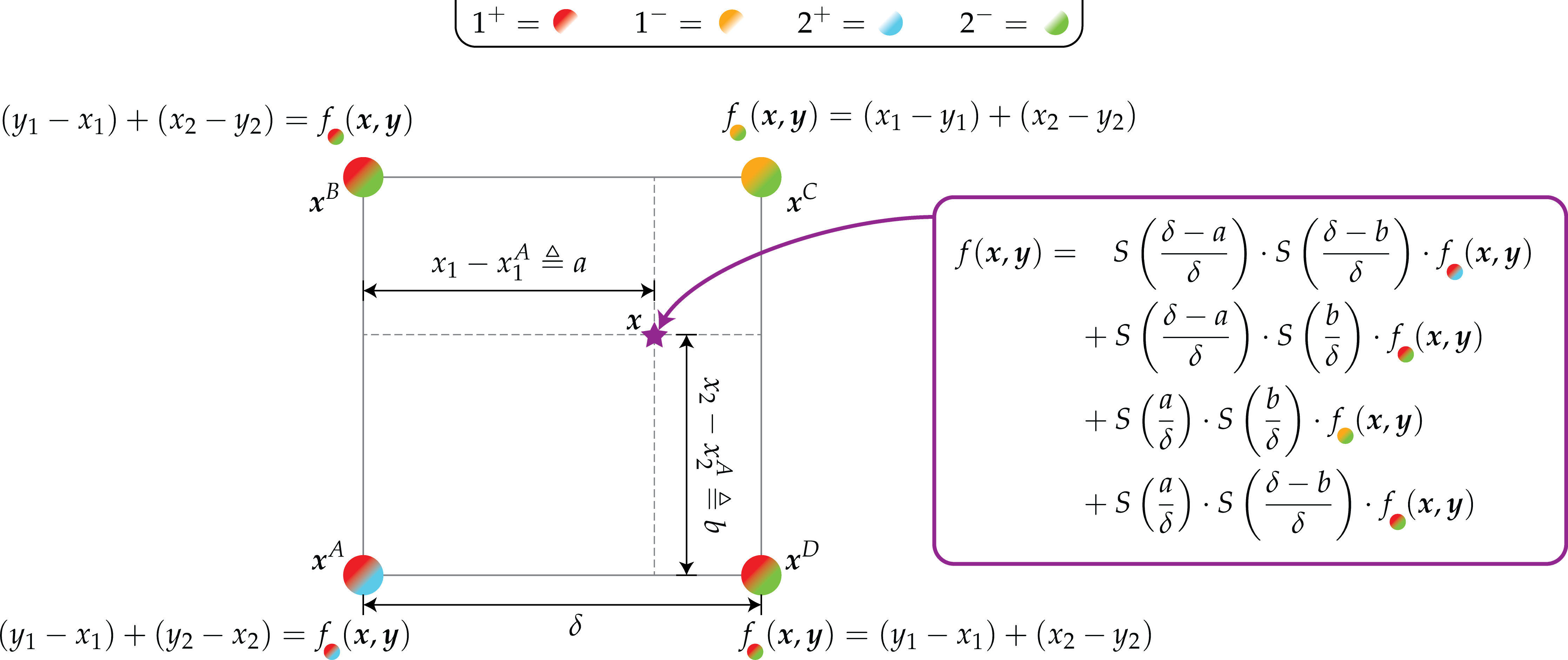}
  \caption{Example of the definition of the Lipschitz and smooth function $f$ on
  some cell given the coloring on the corners of the cell. For details see
  Definition \ref{d:payoff}.}
  \label{fig:interpolationExample}
\end{figure}

  An important property of the definition of the function $f_{\calC_l}$ is that
the coefficients used in the definition of $\alpha_i$ have the following two
properties
\[ S_1 \left(\frac{x^C_1 - x_1}{\delta} \right) \cdot S_1 \left(\frac{x^C_2 - x_2}{ \delta} \right) \ge 0, ~ S_1 \left(\frac{x^D_1 - x_1}{ \delta} \right) \cdot S_1 \left(\frac{x_2 - x^D_2}{\delta} \right) \ge 0, \]
\[ S_1 \left(\frac{x_1 - x^A_1}{ \delta} \right) \cdot S_1 \left(\frac{x_2 - x^A_2}{ \delta} \right) \ge 0, ~ S_1 \left(\frac{x_1 - x^B_1}{ \delta} \right) \cdot S_1 \left(\frac{x^B_2 - x_2}{ \delta} \right) \ge 0, \text{ and } \]
\[ S_1 \left(\frac{x^C_1 - x_1}{\delta} \right) \cdot S_1 \left(\frac{x^C_2 - x_2}{ \delta} \right) + S_1 \left(\frac{x^D_1 - x_1}{ \delta} \right) \cdot S_1 \left(\frac{x_2 - x^D_2}{\delta} \right) \]
\[ \text{ } ~~~~~~~~~~~ + S_1 \left(\frac{x_1 - x^A_1}{ \delta} \right) \cdot S_1 \left(\frac{x_2 - x^A_2}{ \delta} \right) + S_1 \left(\frac{x_1 - x^B_1}{ \delta} \right) \cdot S_1 \left(\frac{x^B_2 - x_2}{ \delta} \right) = 1. \]
Hence the function $f_{\calC_l}$ inside a cell is a smooth convex combination of
the functions on the corners of the cell, as is suggested from Figure
\ref{fig:interpolationExample}. Of course there are many ways to define such
convex combination but in our case we use the smooth step function $S_1$ to
ensure the Lipschitz continuous gradient of the overall function $f_{\calC_l}$.
We prove this formally in the next lemma.

\begin{lemma}\label{l:smoothness}
    Let $f_{\calC_l}$ be the function defined based on a coloring circuit
  $\calC_l$, as per Definition \ref{d:payoff}. Then $f_{\calC_l}$ is continuous
  and differentiable at any point $(\vecx, \vecy) \in [0, 1]^4$. Moreover,
  $f_{\calC_l}$ is $\Theta(1/\delta)$-Lipschitz and $\Theta(1/\delta^2)$-smooth
  in the whole 4-dimensional hypercube $[0, 1]^4$, where
  $\delta = 1/(N - 1) = 1/(2^n - 1)$.
\end{lemma}

\begin{proof}
    Clearly from Definition \ref{d:payoff}, $f_{\calC_l}$ is differentiable at
  any point $(\vecx, \vecy) \in [0, 1]^4$ in which $\vecx$ lies on the strict
  interior of its respective cell. In this case the derivative with respect to
  $x_1$ is
  \[ \frac{\partial f_{\calC_l}(\vecx,\vecy)}{\partial x_1} = \frac{\partial \alpha_1(\vecx)}{\partial x_1} \cdot (y_1 - x_1) - \alpha_1(\vecx) + \frac{\partial \alpha_2(\vecx)}{\partial x_1}\cdot(y_2 - x_2). \]
  \noindent where for $\partial \alpha_1(\vecx)/\partial x_1$ we have that
  \begin{align*}
    \frac{\partial \alpha_1(\vecx)}{\partial x_1}
      = & - \frac{1}{\delta} S'_1 \left(\frac{x^C_1 - x_1}{\delta} \right) \cdot S_1 \left(\frac{x^C_2 - x_2}{\delta} \right) \cdot \calC_l^{1}(A) \\
        & - \frac{1}{\delta} S'_1 \left(\frac{x^D_1 - x_1}{\delta} \right) \cdot S_1 \left(\frac{x_2 - x^D_2}{\delta} \right) \cdot \calC_l^{1}(B) \\
        & + \frac{1}{\delta} S'_1 \left(\frac{x_1 - x^A_1}{\delta} \right) \cdot S_1 \left(\frac{x_2 - x^A_2}{\delta} \right) \cdot \calC_l^{1}(C) \\
        & + \frac{1}{\delta} S'_1 \left(\frac{x_1 - x^B_1}{\delta} \right) \cdot S_1 \left(\frac{x^B_2 - x_2}{\delta} \right) \cdot \calC_l^{1}(D).
  \end{align*}
  \noindent Now since $\max_{z \in [0, 1]} \abs{S'_1(z)} \le 6$, we can conclude
  that $\abs{\frac{\partial \alpha_1(\vecx)}{\partial x_1}} \le 24/\delta$.
  Similarly we can prove that
  $\abs{\frac{\partial \alpha_2(\vecx)}{\partial x_1}} \le 24/\delta$, which
  combined with $\abs{\alpha_1(\vecx)} \le 1$ implies
  $\abs{\frac{\partial f_{\calC_l}(\vecx,\vecy)}{\partial x_1}} \le O(1/\delta)$. Using
  similar reasoning we can prove that
  $\abs{\frac{\partial f_{\calC_l}(\vecx,\vecy)}{\partial x_2}} \le O(1/\delta)$ and that
  $\abs{\frac{\partial f_{\calC_l}(\vecx,\vecy)}{\partial y_i}} \le 1$ for
  $i = 1, 2$. Hence
  \[ \norm{\nabla f_{\calC_l}(\vecx,\vecy)}_2 \le O(1/\delta). \]
  The only thing we are missing to prove the Lipschitzness of $f_{\calC_l}$ is
  to prove its continuity on the boundaries of the cells of our subdivision.
  Suppose $\vecx$ lies on the boundary of some cell, e.g.~let $\vecx$ lie on
  edge $(C, D)$ of one cell that is the same as the edge $(A',B')$ of the cell
  to the right of that cell. Since $S_1(0) = 0$, $S'_1(0) = 0$ and $S'_1(1) = 0$
  it holds that $\partial \alpha_1(\vecx)/\partial x_1 = 0$ and the same for
  $\alpha_2$. Therefore the value of $\partial f_{\calC_l}/\partial x_1$ remains
  the same no matter the cell according to which it was calculated. As a result,
  $f_{\calC_l}$ is differentiable with respect to $x_1$ even if $\vecx$ belongs
  in the boundary of its cell. Using the exact same reasoning for the rest of
  the variables, one can show that the function $f_{\calC_l}$ is differentiable
  at any point $(\vecx, \vecy) \in [0, 1]^4$ and because of the aforementioned
  bound on the gradient $\nabla f_{\calC_l}$ we can conclude that $f_{\calC_l}$
  is $O(1/\delta)$-Lipschitz.
  \medskip

  Using very similar calculations, we can compute the closed formulas of the
  second derivatives of $f_{\calC_l}$ and using the bounds
  $\abs{f_{\calC_l}(\cdot)} \le 2$, $\abs{S_1(\cdot)} \le 1$,
  $\abs{S'_1(\cdot)} \le 6$, and $\abs{S''_1(\cdot)} \le 6$, we can prove that
  each entry of the Hessian $\nabla^2 f_{\calC_l}(\vecx, \vecy)$ is bounded by
  $O(1/\delta^2)$ and thus
  \[\norm{\nabla^2 f_{\calC_l}(\vecx,\vecy)}_2 \leq O(1/\delta^2)\]
  \noindent which implies the $\Theta(1/\delta^2)$-smoothness of $f_{\calC_l}$.
\end{proof}

\subsubsection{Description and Correctness of the Reduction -- Proof of Theorem
\ref{thm:hardness2D}} \label{sec:2DmainReduction:Reduction}

  In this section, we present and prove the exact polynomial-time construction
of the instance of the problem $\gdaFixed$ from an instance $\calC_l$ of the
problem $\TD$-$\BiSperner$.
\medskip

\noindent \textbf{$\boldsymbol{(+)}$ Construction of Instance for Fixed Points
of Gradient Descent/Ascent.} \label{lbl:construction2D}

\noindent Our construction can be described via the following properties.
\begin{enumerate}
  \item[$\blacktriangleright$] The payoff function is the real-valued function
  $f_{\calC_l}(\vecx, \vecy)$ from the Definition \ref{d:payoff}.
  \item[$\blacktriangleright$] The domain is the polytope $\calP(\matA, \vecb)$
  that we described in Section \ref{sec:computational}. The matrix $\matA$ and
  the vector $\vecb$ have constant size and they are computed so that the
  following inequalities hold
  \begin{equation} \label{eq:2Dproof:domainDefinition}
    x_1 - y_1 \le \Delta, ~~ y_1 - x_1 \le \Delta, ~~ x_2 - y_2 \le \Delta, ~\text{ and }~ y_2 - x_2 \le \Delta
  \end{equation}
  where $\Delta = \delta / 12$ and $\delta = 1/(N - 1) = 1/(2^n - 1)$.
  \item[$\blacktriangleright$] The parameter $\alpha$ is set to be equal to
  $\Delta/3$.
  \item[$\blacktriangleright$] The parameters $G$ and $L$ are set to be equal to
  the upper bounds on the Lipschitzness and the smoothness of $f_{\calC_l}$
  respectively that we derived in Lemma \ref{l:smoothness}. Namely we have that
  $G = O(1/\delta) = O(2^n)$ and $L = O(1/\delta^2) = O(2^{2 n})$.
\end{enumerate}

  The first thing that is simple to observe in the above reduction is that it
runs in polynomial time with respect to the size of the the circuit $\calC_l$
which is the input to the $\TD$-$\BiSperner$ problem that we started with. To
see this, recall from the definition of $\gdaFixed$ that our reduction needs to
output: (1) a Turing machine $\calC_{f_{\calC_l}}$ that computes the value
and the gradient of the function $f_{\calC_l}$ in time polynomial in the number
of requested bits of accuracy; (2) the required scalars $\alpha$, $G$, and $L$.
For the first, we observe from the definition of $f_{\calC_l}$ that it is actually a piece-wise polynomial function with a closed form that
only depends on the values of the circuit $\calC_l$ on the corners of the
corresponding cell. Since the size of $\calC_l$ is the size of the input to
$\TD$-$\BiSperner$ we can easily construct a polynomial-time Turing machine that
computes both function value and the gradient of the piecewise polynomial
function $f_{\calC_l}$. Also, from the aforementioned description of the
reduction we have that $\log(G)$, $\log(L)$ and $\log(1/\alpha)$ are linear in
$n$ and hence we can construct the binary representation of all this scalars in
time $O(n)$. The same is true for the coefficients of $\matA$ and $\vecb$ as we
can see from their definition in \hyperref[lbl:construction2D]{$(+)$}. Hence we
conclude that our reduction runs in time that is polynomial in the size of the
circuit $\calC_l$.
\smallskip

  The next thing to observe is that, according to Lemma \ref{l:smoothness}, the
function $f_{\calC_l}$ is both $G$-Lipschitz and $L$-smooth and hence the output
of our reduction is a valid input for the promise problem $\gdaFixed$. So the
last step to complete the proof of Theorem \ref{thm:hardness2D} is to prove that
the vector $\vecx^{\star}$ of every solution $(\vecx^{\star}, \vecy^{\star})$ of
$\gdaFixed$ with input $\calC_{f_{\calC_l}}$, lies in a cell that is either
panchromatic or violates the rules for proper coloring, in any of these cases we
can find a solution to the $\TD$-$\BiSperner$ problem. This proves that our
construction reduces $\TD$-$\BiSperner$ to $\gdaFixed$.

  We prove this last statement in Lemma \ref{l:main_2}, but before that we need
the following technical lemma that will be useful to argue about solution on the
boundary of $\calP(\matA, \vecb)$.

\begin{lemma}\label{c:conditions}
    Let $\calC_l$ be an input to the $\TD$-$\BiSperner$ problem, let
  $f_{\calC_l}$ be the corresponding $G$-Lipschitz and $L$-smooth function
  defined in Definition \ref{d:payoff}, and let $\calP(\matA, \vecb)$ be the
  polytope defined by \eqref{eq:2Dproof:domainDefinition}. If
  $(\vecx^{\star}, \vecy^{\star})$ is any solution to the $\gdaFixed$ problem
  with inputs $\alpha$, $G$, $L$, $\calC_{f_{\calC_l}}$, $\matA$, and $\vecb$,
  defined in \hyperref[lbl:construction2D]{$(+)$} then the following statements
  hold, where recall that $\Delta = \delta/12$. For $i\in \{1,2\}$:
  \begin{enumerate}
    \item[$\diamond$] If $x_i^\star \in (\alpha, 1 - \alpha)$ and
    $x_i^{\star} \in (y_i^\star - \Delta + \alpha, y_i^\star + \Delta - \alpha)$
    then
    $\abs{\frac{\partial f_{\calC_l}(\vecx^\star,\vecy^\star)}{\partial x_i}}
    \leq \alpha$.
    \item[$\diamond$] If $x^\star_i \le \alpha$ or
    $x^\star_i \le y^\star_i - \Delta + \alpha$ then
    $\frac{\partial f_{\calC_l}(\vecx^\star,\vecy^\star)}{\partial x_i} \ge - \alpha$.
    \item[$\diamond$] If $x^\star_i \ge 1 - \alpha$ or
    $x^\star_i \ge y^\star_i + \Delta - \alpha$ then
    $\frac{\partial f_{\calC_l}(\vecx^\star, \vecy^\star)}{\partial x_i} \le \alpha$.
  \end{enumerate}
  The symmetric statements for $y_i^{\star}$ hold. For $i\in \{1,2\}$:
  \begin{enumerate}
    \item[$\diamond$] If $y_i^\star \in  (\alpha, 1 - \alpha)$ and
    $y_i^{\star} \in (x_i^\star - \Delta + \alpha, x_i^\star + \Delta - \alpha)$
    then
    $\abs{\frac{\partial f_{\calC_l}(\vecx^\star, \vecy^\star)}{\partial y_i}} \le \alpha$.
    \item[$\diamond$] If $y^\star_i \le \alpha$ or
    $y^\star_i \le x^\star_i - \Delta + \alpha$ then
    $\frac{\partial f_{\calC_l}(\vecx^\star, \vecy^\star)}{\partial y_i} \le \alpha$.
    \item[$\diamond$] If $y^\star_i \ge 1 - \alpha$ or
    $y^\star_i \ge x^\star_i + \Delta - \alpha$ then
    $\frac{\partial f_{\calC_l}(\vecx^\star,\vecy^\star)}{\partial y_i} \ge - \alpha$.
  \end{enumerate}
\end{lemma}

\begin{proof}
    For this proof it is convenient to define
  $\hat{\vecx} = \vecx^\star - \nabla_x f_{\calC_l}(\vecx^{\star}, \vecy^{\star})$,
  $K(\vecy^{\star}) =\{\vecx \mid (\vecx,\vecy^\star) \in \calP (\matA,\vecb))\}$,
  and $\vecz = \Pi_{K(\vecy^{\star})} \hat{\vecx}$.

    We first consider the first statement, so for the sake of contradiction
  let's assume that $x_i^\star \in (\alpha, 1 - \alpha)$, that
  $x_i^{\star} \in (y_i^\star - \Delta + \alpha, y_i^\star + \Delta - \alpha)$,
  and that
  $\abs{\frac{\partial f_{\calC_l}(\vecx^\star,\vecy^\star)}{\partial x_i}} > \alpha$.
  Due to the definition of $\calP(\matA, \vecb)$ in
  \eqref{eq:2Dproof:domainDefinition} the set $K(\vecy^{\star})$ is an axes
  aligned box of $\R^2$ and hence the projection of any vector $\vecx$ onto
  $K(\vecy^{\star})$ can be implemented independently for every coordinate $x_i$
  of $\vecx$. Therefore if it happens that
  $\hat{x}_i \in (0,1)\cap (y_i^\star - \Delta, y_i^\star + \Delta)$,
  then it holds that $\hat{x}_i = \vecz_i$. Now from the definition of
  $\hat{x}_i$ and $z_i$, and the fact that $K(\vecy^{\star})$ is an axes aligned
  box, we get that
  $\abs{x_i^\star -  z_i} = \abs{x_i^\star -  \hat{x}_i} = \abs{\frac{\partial f_{\calC_l}(\vecx^\star,\vecy^\star)}{\partial x_i}} > \alpha$
  which contradicts  the fact that $(\vecx^\star, \vecy^\star)$ is a
  solution to the problem $\gdaFixed$. On the other hand if
  $\hat{x}_i \not\in (y_i^\star - \Delta, y_i^\star + \Delta) \cap (0, 1)$ then
  $z_i$ has to be on the boundary of $K(\vecy^{\star})$ and hence $z_i$ has to
  be equal to either $0$, or $1$, or $y^{\star}_i - \Delta$, or
  $y^{\star}_i + \Delta$. In any of these cases since we assumed that
  $x_i^\star \in (\alpha, 1 - \alpha)$ and that
  $x_i^{\star} \in (y_i^\star - \Delta + \alpha, y_i^\star + \Delta - \alpha)$
  we conclude that $\abs{x^{\star}_i - z_i} > \alpha$ and hence we get again a
  contradiction with the fact that $(\vecx^{\star}, \vecy^{\star})$ is a
  solution to the problem $\gdaFixed$. Hence we have that
  $\abs{\frac{\partial f_{\calC_l}(\vecx^\star,\vecy^\star)}{\partial x_i}} \le \alpha$.

    For the second case, we assume for the sake of contradiction that
  $x^{\star}_i \le \alpha$ and
  $\frac{\partial f_{\calC_l}(\vecx^\star, \vecy^\star)}{\partial x_i} < - \alpha$.
  These imply that $\hat{x}_i > x_i^\star + \alpha$ and that
  $z_i = \min(y_i^{\star} + \Delta, \hat{x}_i, 1) > \min(\Delta, \hat{x}_i, 1) \ge \min(3 \alpha, x_i^{\star} + \alpha)$.
  As a result,
  $\abs{x_i^\star - z_i} = z_i - x_i^\star > \min(3 \alpha, \hat{x}_i + \alpha) - x_i^\star$
  which is greater than $\alpha$. The latter is a contradiction with the
  assumption that $(\vecx^{\star}, \vecy^{\star})$ is a solution to the
  $\gdaFixed$ problem. Also if we assume that
  $x_i^\star \leq y_i^\star - \Delta + \alpha$ using the same reasoning we get
  that $z_i = \min(\hat{x}_i, y_i^\star + \Delta - \alpha, 1)$. From this we can
  again prove that $\abs{x_i^\star -  z_i} > \alpha$ which contradicts the fact
  that $(\vecx^{\star}, \vecy^{\star})$ is a solution to $\gdaFixed$.

    The third case can be proved using the same arguments as the second case.
  Then using the corresponding arguments we can prove the corresponding
  statements for the $y$ variables.
\end{proof}

  We are now ready to prove that solutions of $\gdaFixed$ can only occur in
cells that are either panchromatic or violate the boundary conditions of a
proper coloring. For convenience in the rest of this section we define
$R(\vecx)$ to be the cell of the $2^n \times 2^n$ grid that contains $\vecx$.
\begin{equation} \label{eq:2Dhardnees:cellDefinition}
  R(\vecx) = \left[\frac{i}{2^n - 1}, \frac{i + 1}{2^n - 1}\right] \times \left[\frac{j}{2^n - 1}, \frac{j + 1}{2^n - 1}\right],
\end{equation}
\noindent for $i, j$ such that
$x_1 \in \left[\frac{i}{2^n - 1}, \frac{i + 1}{2^n - 1}\right] \text{ and } x_2 \in \left[\frac{j}{2^n - 1}, \frac{j + 1}{2^n - 1}\right]$
if there are multiple $i$, $j$ that satisfy the above condition then we choose
$R(\vecx)$ to be the cell that corresponds to the $i$, $j$ such that the pair
$(i, j)$ it the lexicographically first such that $i$, $j$ satisfy the above
condition. We also define the corners $R_c(\vecx)$ of $R(\vecx)$ as
\begin{equation} \label{eq:2Dhardnees:cornersDefinition}
  R_c(\vecx) = \set{(i, j), (i, j + 1), (i + 1, j), (i + 1), (j + 1)}
\end{equation}
where
$R(\vecx) = \left[\frac{i}{2^n - 1}, \frac{i + 1}{2^n - 1}\right] \times \left[\frac{j}{2^n - 1}, \frac{j + 1}{2^n - 1}\right]$.

\begin{lemma} \label{l:main_2}
    Let $\calC_l$ be an input to the $\TD$-$\BiSperner$ problem, let
  $f_{\calC_l}$ be the corresponding $G$-Lipschitz and $L$-smooth function
  defined in Definition \ref{d:payoff}, and let $\calP(\matA, \vecb)$ be the
  polytope defined by \eqref{eq:2Dproof:domainDefinition}. If
  $(\vecx^{\star}, \vecy^{\star})$ is any solution to the $\gdaFixed$ problem
  with inputs $\alpha$, $G$, $L$, $\calC_{f_{\calC_l}}$, $\matA$, and $\vecb$
  defined in \hyperref[lbl:construction2D]{$(+)$} then none of the following
  statements hold for the cell $R(\vecx^{\star})$.
  \begin{enumerate}
    \item $x_1^{\star} \ge 1/(2^n - 1)$ and, for all $\vecv \in R_c(\vecx^{\star})$,
      it holds that $\calC_l^1(\vecv) = -1$.
    \item $x_1^{\star} \le (2^n - 2)/(2^n - 1)$ and, for all
      $\vecv \in R_c(\vecx^{\star})$, it holds that $\calC_l^1(\vecv) = +1$.
    \item $x_2^{\star} \ge 1/(2^n - 1)$ and, for all $\vecv \in R_c(\vecx^{\star})$,
      it holds that $\calC_l^2(\vecv) = -1$.
    \item $x_2^{\star} \le (2^n - 2)/(2^n - 1)$ and, for all
      $\vecv \in R_c(\vecx^{\star})$, it holds that $\calC_l^2(\vecv) = +1$.
  \end{enumerate}
\end{lemma}

\begin{proof}
    We prove that there is no solution $(\vecx^\star, \vecy^\star)$ of
  $\gdaFixed$ that satisfies the statement 1. and the fact that
  $(\vecx^\star, \vecy^\star)$ cannot satisfy the other statements follows
  similarly. It is convenient for us to define
  $\hat{\vecx} = \vecx^\star - \nabla_x f_{\calC_l}(\vecx^{\star}, \vecy^{\star})$,
  $K(\vecy^{\star}) = \{\vecx \mid (\vecx,\vecy^\star) \in \calP (\matA,\vecb))\}$,
  $\vecz = \Pi_{K(\vecy^{\star})} \hat{\vecx}$, and
  $\hat{\vecy} = \vecy^\star + \nabla_y f_{\calC_l}(\vecx^{\star}, \vecy^{\star})$,
  $K(\vecx^{\star}) = \{\vecy \mid (\vecx^\star, \vecy) \in \calP (\matA,\vecb))\}$,
  $\vecw = \Pi_{K(\vecx^{\star})} \hat{\vecy}$.

    For the sake of contradiction we assume that there exists a solution of
  $(\vecx^\star,\vecy^\star)$ such that $x_1^{\star} \ge 1/(2^n - 1)$ and for all
  $\vecv \in R_c(\vecx^{\star})$ it holds that $\calC_l^1(\vecv) = -1$. Using
  the fact that the first color of all the corners of $R(\vecx^{\star})$ is
  $1^-$, we will prove that (1)
  $\frac{\partial f_{\calC_l}(\vecx^\star, \vecy^\star)}{\partial x_1} \geq 1/2$,
  and (2)
  $\frac{\partial f_{\calC_l}(\vecx^\star, \vecy^\star)}{\partial y_1} = -1$.

    Let
  $R(\vecx^{\star}) = \left[\frac{i}{2^n - 1}, \frac{i + 1}{2^n - 1}\right] \times \left[\frac{j}{2^n - 1}, \frac{j + 1}{2^n - 1}\right]$,
  then since all the corners $\vecv \in R_c(\vecx^{\star})$ have
  $\calC_l^1(\vecv) = -1$, from the Definition \ref{d:payoff} we have that
  \begin{align*}
    f_{\calC_l}(\vecx^{\star}, \vecy^{\star})
      = & ~ (x_1^{\star} - y_1^{\star}) - (x_2^{\star} - y_2^{\star}) \cdot S_1 \left(\frac{x^C_1 - x_1^{\star}}{\delta} \right) \cdot S_1 \left(\frac{x^C_2 - x_2^{\star}}{\delta} \right) \cdot \calC_l^2(i, j) \\
        & - (x_2^{\star} - y_2^{\star}) \cdot S_1 \left(\frac{x^D_1 - x_1^{\star}}{\delta} \right) \cdot S_1 \left(\frac{x_2^{\star} - x^D_2}{\delta} \right) \cdot \calC_l^2(i, j + 1) \\
        & - (x_2^{\star} - y_2^{\star}) \cdot S_1 \left(\frac{x_1^{\star} - x^A_1}{\delta} \right) \cdot S_1 \left(\frac{x_2^{\star} - x^A_2}{\delta} \right) \cdot \calC_l^2(i + 1, j + 1) \\
        & - (x_2^{\star} - y_2^{\star}) \cdot S_1 \left(\frac{x_1^{\star} - x^B_1}{\delta} \right) \cdot S_1 \left(\frac{x^B_2 - x_2^{\star}}{\delta} \right) \cdot \calC_l^2(i + 1, j)
  \end{align*}
  where
  $(x_1^A, x_2^A) = (i/(2^n - 1), j/(2^n - 1))$, $(x_1^B, x_2^B) = (i/(2^n - 1), (j + 1)/(2^n - 1))$,
  $(x_1^C, x_2^C) = ((i + 1)/(2^n - 1), (j + 1)/(2^n - 1))$, and
  $(x_1^D, x_2^D) = ((i + 1)/(2^n - 1), j/(2^n - 1))$. If we differentiate this with respect
  to $y_1$ we immediately get that
  $\frac{\partial f_{\calC_l}(\vecx^{\star}, \vecy^{\star})}{\partial y_1} = -1$.
  On the other hand if we differentiate with respect to $x_1$ we get
  \begin{align}
    \frac{\partial f_{\calC_l}(\vecx^{\star}, \vecy^{\star})}{\partial x_1} = ~
        & 1 + (x_2^{\star} - y_2^{\star}) \cdot \frac{1}{\delta} \cdot S'_1 \left(1 - \frac{x_1^{\star} - x^A_1}{\delta} \right) \cdot S_1 \left(1 - \frac{x_2^{\star} - x^A_2}{\delta} \right) \cdot \calC_l^2(i, j) \nonumber \\
        & + (x_2^{\star} - y_2^{\star}) \cdot \frac{1}{\delta} \cdot S'_1 \left(1 - \frac{x_1^{\star} - x^A_1}{\delta} \right) \cdot S_1 \left(\frac{x_2^{\star} - x^{A}_2}{\delta} \right) \cdot \calC_l^2(i, j + 1) \nonumber \\
        & - (x_2^{\star} - y_2^{\star}) \cdot \frac{1}{\delta} \cdot S'_1 \left(\frac{x_1^{\star} - x^A_1}{\delta} \right) \cdot S_1 \left(\frac{x_2^{\star} - x^A_2}{\delta} \right) \cdot \calC_l^2(i + 1, j + 1) \nonumber \\
        & - (x_2^{\star} - y_2^{\star}) \cdot \frac{1}{\delta} \cdot S'_1 \left(\frac{x_1^{\star} - x^A_1}{\delta} \right) \cdot S_1 \left(1 - \frac{x_2^{\star} - x^A_2}{\delta} \right) \cdot \calC_l^2(i + 1, j) \nonumber \\
    \ge & 1 - 4 \abs{x_2^{\star} - y_2^{\star}} \cdot \frac{3}{2 \delta} \nonumber \\
    \ge & 1 - 6 \cdot \frac{\Delta}{\delta} \ge 1/2 \label{eq:easyMinimum:2D}
  \end{align}
  \noindent where the last inequality follows from the fact that
  $\abs{S'_1(\cdot)} \le 3/2$ and the fact that, due to the constraints that
  define the polytope $\calP(\matA, \vecb)$, it holds that
  $\abs{x_2 - y_2} \le \Delta$.

    Hence we have established that if $x_1^{\star} \ge 1/(2^n - 1)$ and for all
  $\vecv \in R_c(\vecx^{\star})$ it holds that $\calC_l^1(\vecv) = -1$ then it
  holds that that (1)
  $\frac{\partial f_{\calC_l}(\vecx^\star, \vecy^\star)}{\partial x_1} \geq 1/2$,
  and (2)
  $\frac{\partial f_{\calC_l}(\vecx^\star, \vecy^\star)}{\partial y_1} = -1$.
  Now it is easy to see that the only way to satisfy both
  $\frac{\partial f_{\calC_l}(\vecx^\star, \vecy^\star)}{\partial x_1} \geq 1/2$
  and $\abs{z_1 - x_1^{\star}} \le \alpha$ is that either $x_1^\star \le \alpha$
  or $x_1^\star \le y_1^\star - \Delta + \alpha$. The first case is excluded by
  the assumption in the first statement of our lemma and our choice of
  $\alpha = \Delta/3 = 1/(36 \cdot (2^n - 1))$ thus it holds that
  $x_1^\star \le y_1^\star - \Delta + \alpha$. But then we can use the case 3
  for the $y$ variables of Lemma \ref{c:conditions} and we get that
  $\frac{\partial f_{\calC_l}(\vecx^\star, \vecy^\star)}{\partial y_1} \ge - \alpha$,
  which cannot be true since we proved that
  $\frac{\partial f_{\calC_l}(\vecx^\star, \vecy^\star)}{\partial y_1} = -1$.
  Therefore we have a contradiction and the first statement of the lemma holds.
  Using the same reasoning we prove the rest of the statements.
\end{proof}

\begin{remark} \label{rem:easyMinimum:2D}
    The computations presented in \eqref{eq:easyMinimum:2D} is the precise point
  where an attempt to prove the hardness of minimization problems would fail. In
  particular, if our goal was to construct a hard minimization instance then the
  function $f_{\calC_l}$ would need to have the terms $x_i + y_i$ instead of
  $x_i - y_i$ so that the fixed points of gradient descent coincide with
  approximate local minimum of $f_{\calC_l}$. In that case we cannot lower bound
  the gradient of \eqref{eq:easyMinimum:2D} below from $1/2$ because the term
  $\abs{x^{\star}_2 + y^{\star}_2}$ will be the dominant one and hence the
  sign of the derivative can change depending on the value
  $\abs{x^{\star}_2 + y^{\star}_2}$. For a more intuitive explanation of the
  reason why we cannot prove hardness of minimization problems we refer to the
  Introduction, at Section \ref{sec:intro:minimization}.
\end{remark}
\medskip

  We have now all the ingredients to prove Theorem \ref{thm:hardness2D}.

\begin{proof}[Proof of Theorem \ref{thm:hardness2D}]
    Let $(\vecx^{\star}, \vecy^{\star})$ be a solution to the $\gdaFixed$
  instance that we construct based on the instance $\calC_l$ of
  $\TD$-$\BiSperner$. Let also $R(\vecx^{\star})$ be the cell that contains
  $\vecx^{\star}$. If the corners $R_c(\vecx^{\star})$ contain all the colors
  $1^-$, $1^+$, $2^-$, $2^+$ then we have a solution to the $\TD$-$\BiSperner$
  instance and the Theorem \ref{thm:hardness2D} follows. Otherwise there is at
  least one color missing from $R_c(\vecx^{\star})$, let's assume without loss
  of generality that one of the missing colors is $1^-$, hence for every
  $\vecv \in R_c(\vecx^{\star})$ it holds that $\calC_l(\vecv) = + 1$. Now from
  Lemma \ref{l:main_2} the only way for this to happen is that
  $x^{\star}_1 > (2^n - 2)/(2^n - 1)$ which implies that in $R_c(\vecx^{\star})$
  there is at least one corner of the form $\vecv = (2^n - 1, j)$. But we have
  assumed that $\calC_l(\vecv) = + 1$, hence $\vecv$ is a violation of the
  proper coloring rules and hence a solution to the $\TD$-$\BiSperner$ instance.
  We can prove the corresponding statement if any other color from $1^+$, $2^-$,
  $2^+$ is missing. Finally, we observe that the function that we define has
  range $[-2, 2]$ and hence the Theorem \ref{thm:hardness2D} follows.
\end{proof}

\section{Hardness of Local Min-Max Equilibrium -- High-Dimensions} \label{sec:hardness}

  Although the results of Section \ref{sec:hardness:2D} are quite indicative
about the computational complexity of $\gdaFixed$ and $\lrlNash$, we have
not yet
excluded the possibility of the existence of algorithms running in
$\poly(d, G, L, 1/\eps)$ time. In this section we present a, significantly more
challenging, high dimensional version of the reduction that we presented in
Section \ref{sec:hardness:2D}. The advantage of this reduction is that it rules
out the existence even of algorithms running in $\poly(d, G, L, 1/\eps)$ steps
unless $\FP = \PPAD$, for details see Theorem \ref{thm:localNashHardness}. An
easy consequence of our result is an unconditional lower bound on the
\textit{black-box model} that states that the running time of any algorithm for
$\lrlNash$ that has only oracle access to $f$ and $\nabla f$ has to be
exponential in  $d$, or $G$, or $L$, or $1/\eps$, for details we refer to
the Theorem \ref{thm:localNashBlackBoxLowerBound} and Section
\ref{sec:blackBox}.

  The main reduction that we use to prove Theorem \ref{thm:localNashHardness} is
from the high dimensional generalization of the problem $\TD$-$\BiSperner$,
which we call $\HD$-$\BiSperner$, to $\gdaFixed$. Our reduction in this section
resembles some of the ideas of the reductions of Section \ref{sec:hardness:2D}
but it has many additional significant technical difficulties. The main
difficulty that we face is how to define a function on a $d$-dimensional simplex
that is: (1) both Lipschitz and smooth, (2) interpolated between some fixed
functions at the $d + 1$ corners of the simplex, and (3) remains Lipschitz and
smooth even if we glue together different simplices. It is well understood from
previous works how to construct such a function if we are interested only in
achieving the Lipschitz continuity. Surprisingly adding the smoothness
requirement makes the problem very different and significantly more difficult.
Our proof overcomes this technical difficulty by introducing a novel but very
technically involved way to define interpolation within a simplex of some fixed
functions on the corners of the simplex. We believe that our novel interpolation
technique is of independent interest and we hope that it will be at the heart of
other computational hardness results of optimization problems in continuous
optimization.

\subsection{The High Dimensional Bi-Sperner Problem} \label{sec:highDBiSperner}

  We start by presenting the $\HD$-$\BiSperner$ problem. The $\HD$-$\BiSperner$
is a straightforward $d$-dimensional generalization of the $\TD$-$\BiSperner$
that we defined in the Section \ref{sec:hardness:2D}. Assume that we have a
$d$-dimensional grid $N \times \cdots (d \text{ times}) \cdots \times N$. We
assign to every vertex of this grid a sequence of $d$ colors and we say that a
coloring is \textit{proper} if the following rules are satisfied.
\begin{enumerate}
  \item The $i$th color of every vertex is either the color $i^+$ or the color
  $i^-$.
  \item All the vertices whose $i$th coordinate is $0$, i.e. they are at the
  lower boundary of the $i$th direction, should have the $i$th color equal to
  $i^+$.
  \item All the vertices whose $i$th coordinate is $1$, i.e. they are at the
  higher boundary of the $i$th direction, should have the $i$th color equal to
  $i^-$.
\end{enumerate}

\noindent Using proof ideas similar to the proof of the original Sperner's Lemma
it is not hard to prove via a combinatorial argument that in every proper
coloring of a $d$-dimensional grid, there exists a cubelet of the grid where all
the $2 \cdot d$ colors $\{1^-,1^+, \dots, d^-, d^+\}$ appear in some of its
vertices, we call such a cubelet \textit{panchromatic}. In the
$\HD$-$\BiSperner$ problem we are asked to find such a cubelet, or a violation
of the rules of proper coloring. As in Section \ref{sec:2DBiSperner} we do not
present this combinatorial argument in this paper since the totality of the
$\HD$-$\BiSperner$ problem will follow from our reduction from
$\HD$-$\BiSperner$ to $\gdaFixed$ and our proofs in Section \ref{sec:existence}
that establish the totality of $\gdaFixed$.
\smallskip

  As in the case of $\TD$-$\BiSperner$, in order to formally define the
computational problem $\HD$-$\BiSperner$ we need to define the coloring of the
$d$-dimensional grid $N \times \cdots \times N$ in a succinct way. The
fundamental difference compared to the definition of $\TD$-$\BiSperner$ is that
for the $\HD$-$\BiSperner$ we assume that $N$ is only
\textit{polynomially large}. This difference will enable us to exclude
algorithms for $\gdaFixed$ that run in time $\poly(d, 1/\alpha, G, L)$. The
input to $\HD$-$\BiSperner$ is a coloring via a binary circuit $\calC_l$ that
takes as input the coordinates of a vertex of the grid and outputs the sequence
of colors that are used to color this vertex. Each one of $d$ coordinates is
given via the binary representation of a number in $\nm{N}$. Setting
$N = 2^\ell$, where here $\ell$ is a logarithmically in $d$ small number, we
have that the representation of each coordinate is a member of $\{0, 1\}^\ell$.
In the rest of the section we abuse the notation and we use a coordinate
$i \in \{0, 1\}^\ell$ both as a binary string and as a number in $\nm{2^{\ell}}$
and which of the two we use it is clear from the context. The output of
$\calC_l$ should be a sequence of $d$ colors, where the $i$th member of this
sequence is one of the colors $\{i^-, i^+\}$. We represent this sequence as a
member of $\{-1, +1\}^d$, where the $i$th coordinate refers to the choice of
$i^-$ or $i^+$.

  In the definition of the computational problem $\HD$-$\BiSperner$ the input is
a circuit $\calC_l$, as we described above. As we discussed above in the
$\HD$-$\BiSperner$ problem we are asking for a panchromatic cubelet of the grid.
One issue with this high-dimensional setting is that in order to check whether a
cubelet is panchromatic or not we have to query all the $2^d$ corners of this
cubelet which makes the verification problem inefficient and hence a containment
to the $\PPAD$ class cannot be proved. For this reason as a solution to the
$\HD$-$\BiSperner$ we ask not just for a cubelet but for $2 \cdot d$ vertices
$\vecv^{(1)}$, $\dots$, $\vecv^{(d)}$ $\vecu^{(1)}$, $\dots$, $\vecu^{(d)}$, not
necessarily different, such that they all belong to the same cubelet and the
$i$th output of $\calC_l$ with input $\vecv_i$ is $-1$, i.e. corresponds to the
color $i^-$, whereas the $i$th output of $\calC_l$ with input $\vecu_i$ is $+1$,
i.e. corresponds to the color $i^+$. This way we have a certificate of size
$2 \cdot d$ that can be checked in polynomial time. Another possible solution of
$\HD$-$\BiSperner$ is a vertex whose coloring violates the aforementioned
boundary conditions 2. and 3.. of a proper coloring. For notational convenience
we refer to the $i$th coordinate of $\calC_l$ by $\calC_l^i$. The formal
definition of $\HD$-$\BiSperner$ is then the following.
\smallskip

\begin{nproblem}[\HD$-$\BiSperner]\label{d:Bi-Sperner}
  \textsc{Input:} A boolean circuit
  $\calC_l : \underbrace{\{0, 1\}^\ell \times \dots \times \{0, 1\}^\ell}_{d \text{ times}} \to \{-1, 1\}^d$
  \smallskip

  \noindent \textsc{Output:} One of the following:
  \begin{Enumerate}
    \item Two sequences of $d$ vertices $\vecv^{(1)}$, $\dots$, $\vecv^{(d)}$ an
    $\vecu^{(1)}$, $\dots$, $\vecu^{(d)}$ with
    $\vecv^{(i)}, \vecu^{(i)} \in \p{\{0, 1\}^\ell}^d$ such that
    $\calC_l^i(\vecv^{(i)}) = -1$ and $\calC_l^i(\vecu^{(i)}) = +1$.
    \item A vertex $\vecv \in \p{\{0, 1\}^\ell}^d$ with $v_i = \vec{0}$ such
    that $\calC_l^i(\vecv) = -1$.
    \item A vertex $\vecv \in \p{\{0, 1\}^\ell}^d$ with $v_i = \vec{1}$ such
    that $\calC_l^i(\vecv) = +1$.
  \end{Enumerate}
\end{nproblem}
\medskip

\noindent Our first step is to establish the $\PPAD$-hardness of
$\HD$-$\BiSperner$ in Theorem~\ref{t:Bi_Sperner_PPAD}. To prove this we use a
stronger version of the $\textsc{Brouwer}$ problem that is called
\textsc{$\gamma$-SuccinctBrouwer} and was first introduced in \cite{R16}.
\medskip

\begin{nproblem}[\textsc{$\gamma$-SuccinctBrouwer}]\label{d:Brouwer}
  \textsc{Input:} A polynomial-time Turing machine $\calC_M$ evaluating a
  $1/\gamma$-Lipschitz continuous vector-valued function
  $M : [0, 1]^d \to [0, 1]^d$.
  \smallskip

  \noindent \textsc{Output:} A point $\vecx^{\star} \in [0, 1]^d$ such that
  $\norm{M(\vecx^{\star}) - \vecx^{\star}}_2 \le \gamma$.
\end{nproblem}
\medskip

\begin{theorem}[\cite{R16}] \label{t:Brouwer}
    \textsc{$\gamma$-SuccinctBrouwer} is $\PPAD$-complete for any fixed constant
  $\gamma > 0$.
\end{theorem}

\begin{theorem} \label{t:Bi_Sperner_PPAD}
    There is a polynomial time reducton from any instance of the
  \textsc{$\gamma$-SuccinctBrouwer} problem to an instance of $\HD$-$\BiSperner$
  with $N = \Theta(d/\gamma^2)$.
\end{theorem}

\begin{proof}
    Consider the function $g(\vecx) = M(\vecx) - \vecx$. Since $M$ is
  $1/\gamma$-Lipschitz, $g : [0, 1]^d \to [-1, 1]^d$ is also
  $(1 + 1/\gamma)$-Lipschitz. Additionally $g$ can be easily computed via a
  polynomial-time Turing machine $\calC_g$ that uses $\calC_M$ as a subroutine.
  We construct the coloring sequences of every vertex of a $d$-dimensional grid
  with $N = \Theta(d/\gamma^2)$ points in every direction using $g$. Let
  $g_{\eta} : [0, 1]^2 \to [-1, 1]^2$ be the function that the Turing Machine
  $\calC_g$ evaluate when the requested accuracy is $\eta > 0$. For each vertex
  $\vecv = (v_1, \ldots, v_n) \in \p{\nm{N}}^d$ of the $d$-dimensional grid its
  coloring sequence $\calC_l(\vecv) \in \{-1, 1\}^d$ is constructed as follows: For each coordinate $j=1,\ldots,d$,
    \[ \calC_l^j( \vecv) = \begin{cases}

                        1  &  v_j = 0 \\
                                              -1  &  v_j = 2^n - 1\\
                                              \mathrm{sign}\left(g_j\p{\frac{v_1}{N - 1}, \ldots, \frac{v_n}{N - 1}}\right)  & \text{otherwise}
                       \end{cases}, \]

  \noindent where $\mathrm{sign} : [-1, 1] \mapsto \{-1, 1\}$ is the sign
  function and $g_{\eta, j}(\cdot)$ is the $j$-th coordinate of $g_{\eta}$.
  Observe that since $M : [0, 1]^d \to [0, 1]^d$, for any vertex $\vecv$ with
  $v_j = 0$ it holds that $\calC_l^j(\vecv) = + 1$ and respectively for any
  vertex $\vecv$ with $v_j = N - 1$ it holds that $\calC_l^j(\vecv) = - 1$ due
  to the fact that the value of $M$ is always in $[0, 1]^d$ and hence there are
  no vertices in the grid satisfying the possible outputs 2. or 3. of the
  $\HD$-$\BiSperner$ problem. Thus the only possible solution of the above
  $\HD$-$\BiSperner$ instance is a sequence of $2 d$ vertices $\vecv^{(1)}$,
  $\ldots$, $\vecv^{(d)}$, $\vecu^{(1)}$, $\ldots$, $\vecu^{(d)}$ on the same
  cubelet that certify that the corresponding cubelet is panchromatic, as per
  possible output 1. of the $\HD$-$\BiSperner$ problem. We next prove that any
  vertex $\vecv$ of that cubelet it holds that
  \[ \abs{g_j\p{\frac{\vecv}{N - 1}}} \leq \frac{2\sqrt{d}}{\gamma N} \text{ ~~ for all coordinates } j = 1, \ldots, d. \]
  \noindent Let $\vecv$ be any vertex on the same cubelet with the output
  vertices $\vecv^{(1)}$, $\ldots$, $\vecv^{(d)}$, $\vecu^{(1)}$, $\ldots$,
  $\vecu^{(d)}$. From the guarantees of colors of the sequences $\vecv^{(1)}$,
  $\ldots$, $\vecv^{(d)}$, $\vecu^{(1)}$, $\ldots$, $\vecu^{(d)}$ we have that
  either $\calC_l^j(\vecv) \cdot \calC_l^j(\vecv^{(j)}) = -1$ or
  $\calC_l^j(\vecv) \cdot \calC_l^j(\vecu^{(j)}) = -1$, let $\bar{\vecv}^{(j)}$
  be the vertex $\vecv^{(j)}$ or $\vecu^{(j)}$ depending on which one the $j$th
  color has product equal to $-1$ with $\calC_l^j(\vecv)$. Now let
  $\eta = \frac{2 \sqrt{d}}{\gamma N}$ if
  $g_j\p{\frac{\vecv}{N - 1}} \in [-\eta, \eta]$ then the wanted inequality
  follows. On the other hand if $g_j\p{\frac{\vecv}{N - 1}} \in [-\eta, \eta]$
  then using the fact that
  $\norm{g\p{\frac{\vecv}{N - 1}} - g_{\eta}\p{\frac{\vecv}{N - 1}}}_{\infty} \le \eta$
  and that from the definition of the colors we have that either
  $g_{\eta, j}\p{\frac{\vecv}{N - 1}} \ge 0$,
  $g_{\eta, j}\p{\frac{\bar{\vecv}^{(j)}}{N - 1}} < 0$ or
  $g_{\eta, j}\p{\frac{\vecv}{N - 1}} < 0$,
  $g_{\eta, j}\p{\frac{\hat{\vecv}^{(j)}}{N - 1}} \geq 0$ we conclude that
  $g_j\p{\frac{\vecv}{N - 1}} \ge 0$,
  $g_j\p{\frac{\bar{\vecv}^{(j)}}{N - 1}} < 0$ or
  $g_j\p{\frac{\vecv}{N - 1}} < 0$,
  $g_j\p{\frac{\hat{\vecv}^{(j)}}{N - 1}} \geq 0$ and thus,
  \[ \abs{g_j\p{\frac{\vecv}{N - 1}}} \leq \abs{g_j\p{\frac{\vecv}{N - 1}} - g_j\p{\frac{\bar{\vecv}^{(j)}}{N - 1}}} \le \p{1 + \frac{1}{\gamma}} \cdot \norm{\frac{\vecv}{N - 1} - \frac{\bar{\vecv}^{(j)}}{N - 1}}_2 \le \frac{2 \sqrt{d}}{\gamma N} \]
  \noindent where in the second inequality we have used the
  $(1 + 1/\gamma)$-Lipschitzness of $g$. As a result, the point
  $\hat{\vecv} = \vecv/(N - 1) \in [0,1]^d$ satisfies
  $\norm{M(\hat{\vecv}) - \hat{\vecv}}_2 \leq 2 d/(\gamma N)$ and thus for if we
  pick $N = \Theta(d/\gamma^2)$ then any vertex $\vecv$ of the panchromatic
  cell is a solution for \textsc{$\gamma$-SuccinctBrouwer}.
\end{proof}

\noindent Now that we have established the $\PPAD$-hardness of
$\HD$-$\BiSperner$ we are ready to present our main result of this section which
is a reduction from the problem $\HD$-$\BiSperner$ to the problem $\gdaFixed$
with the additional constraints that the scalars $\alpha$, $G$, $L$ in the input
satisfy $1/\alpha = \poly(d)$, $G = \poly(d)$, and $L = \poly(d)$.

\subsection{From High Dimensional Bi-Sperner to Fixed Points of Gradient Descent/Ascent} \label{sec:hardness:mainConstruction}

  Given the binary circuit $\calC_l : \p{\nm{N}}^d \to \{-1, +1\}^d$ that is an
instance of $\HD$-$\BiSperner$, we construct a $G$-Lipschitz and $L$-smooth
function $f_{\calC_l} : [0, 1]^d \times [0, 1]^d \to \R$. To do so, we divide
the $[0, 1]^d$ hypercube into cubelets of length $\delta = 1/(N - 1)$. The
corners of such cubelets have coordinates that are integer multiples of
$\delta = 1/(N - 1)$ and we call them \textit{vertices}. Each vertex can be
represented by the vector $\vecv = (v_1, \ldots, v_d) \in \p{\nm{N}}^d$ and
admits a coloring sequence defined by the boolean circuit
$\calC_l : \p{\nm{N}}^d \to \{-1, +1\}^d$. For every $\vecx \in [0, 1]^d$, we
use $R(\vecx)$ to denote the cubelet that contains $\vecx$, formally
\[ R(\vecx) = \left[\frac{c_1}{N - 1}, \frac{c_1 + 1}{N - 1} \right] \times \cdots \times \left[\frac{c_d}{N - 1}, \frac{c_d + 1}{N - 1} \right] \]
\noindent where $\vecc \in \p{\nm{N - 1}}^d$ such that
$\vecx \in \left[\frac{c_1}{N - 1}, \frac{c_1 + 1}{N - 1} \right] \times \cdots \times \left[\frac{c_d}{N - 1}, \frac{c_d + 1}{N - 1} \right]$
and if there are multiple corners $\vecc$ that satisfy this condition then we
choose $R(\vecx)$ to be the cell that corresponds to the $\vecc$ that is
lexicographically first among those that satisfy the condition. We also define
$R_c(\vecx)$ to be the set of vertices that are corners of the cublet
$R(\vecx)$, namely
\[ R_c(\vecx) = \set{c_1, c_1 + 1} \times \cdots \times \set{c_d, c_d + 1} \]
\noindent where $\vecc \in \p{\nm{N - 1}}^d$ such that
$R(\vecx) = \left[\frac{c_1}{N - 1}, \frac{c_1 + 1}{N - 1} \right] \times \cdots \times \left[\frac{c_d}{N - 1}, \frac{c_d + 1}{N - 1} \right]$
Every $\vecy$ that belongs to the cubelet $R(\vecx)$ can be written as a convex
combination of the vectors $\vecv/(N - 1)$ where $\vecv \in R_c(\vecx)$. The value of
the function $f_{\calC_l}(\vecx, \vecy)$ that we construct in this section is
determined by the coloring sequences $\calC_l(\vecv)$ of the vertices
$\vecv \in R_c(\vecx)$. One of the main challenges that we face though is that
the size of $R_c(\vecx)$ is $2^d$ and hence if we want to be able to compute the
value of $f_{\calC_l}(\vecx, \vecy)$ efficiently then we have to find a
consistent rule to pick a subset of the vertices of $R_c(\vecx)$ whose coloring
sequence we need to define the function value $f_{\calC_l}(\vecx, \vecy)$.
Although there are traditional ways to overcome this difficulty using the
\textit{canonical simplicization} of the cubelet $R(\vecx)$, these technique
leads only to functions that are continuous and Lipschitz but they are not
enough to guarantee continuity of the gradient and hence the resulting functions
are not smooth.

\subsubsection{Smooth and Efficient Interpolation Coefficients} \label{sec:hardness:seic}

  The problem of finding a computationally efficient way to define a continuous
function as an interpolation of some fixed function in the corners of a cubelet
so that the resulting function is both Lischitz and smooth is surprisingly
difficult to solve. For this reason we introduce in this section the
\textit{smooth and efficient interpolation coefficients (SEIC)} that as we will
see in Section \ref{sec:hardness:definitionOfFunction}, is the main technical
tool to implement such an interpolation. Our novel interpolation coefficients
are of independent interest and we believe that they will serve as a main
technical tool for proving other hardness results in continuous optimization in
the future.

  In this section we only give a high level description of the smooth and
efficient interpolation coefficients via their properties that we use in Section
\ref{sec:hardness:definitionOfFunction} to define the function $f_{\calC_l}$.
The actual construction of the coefficients is very challenging and technical
and hence we postpone a detail exposition for Section \ref{sec:seic}.

\begin{restatable}[Smooth and Efficient Interpolation Coefficients]{definition}{seic} \label{def:seic}
    For every $N \in \N$ we define the set of
  \textit{smooth and efficient interpolation coefficients (SEIC)} as the family
  of functions, called \textit{coefficients},
  $\calI_{d, N} = \set{\P_{\vecv} : [0, 1]^d \to \R \mid \vecv \in \p{\nm{N}}^d}$
  with the following properties.
  \begin{enumerate}[label=(\Alph*)]
    \item For all vertices $\vecv \in \p{\nm{N}}^d$, the coefficient
          $\P_{\vecv}(\vecx)$ is a twice-differentiable function and satisfies
          \begin{enumerate}
            \item[$\blacktriangleright$] $\abs{\frac{\partial \P_{\vecv}(\vecx)}{\partial x_i}} \leq \Theta(d^{12}/\delta)$.
            \item[$\blacktriangleright$] $\abs{\frac{\partial^2 \mathrm{P}_{\vecv}(\vecx)}{\partial x_i ~\partial x_\ell}} \leq \Theta(d^{24}/\delta^2)$.
          \end{enumerate} \label{enu:properties:Coefficients:1}
    \item For all $\vecv \in \p{\nm{N}}^d$, it holds that
          $\P_{\vecv} (\vecx) \ge 0$ and
          $\sum_{\vecv \in \p{\nm{N}}^d}\P_{\vecv}(\vecx) = \sum_{\vecv \in R_c(\vecx)}\P_{\vecv}(\vecx) = 1$.
          \label{enu:properties:Coefficients:2}
    \item For all $\vecx \in [0, 1]^d$, it holds that all but $d + 1$ of the
          coefficients $\P_{\vecv} \in \calI_{d, N}$ satisfy
          $\P_{\vecv} (\vecx) = 0$, $\nabla \P_{\vecv} (\vecx) = 0$ and
          $\nabla^2 \P_{\vecv} (\vecx) = 0$. We denote this set of $d + 1$
          vertices by $R_+(\vecx)$. Furthermore, it holds that
          $R_+(\vecx) \subseteq R_c(\vecx)$ and given $\vecx$ we can compute the
          set $R_{+}(\vecx)$ it time $\poly(d)$.
          \label{enu:properties:Coefficients:3}
    \item For all $\vecx \in [0, 1]^d$, if $x_i \le 1/(N - 1)$ for some $i \in [d]$
          then there exists $\vecv \in R_+(\vecx)$ such that $v_i = 0$.
          Respectively, if $x_i \ge 1 - 1/(N - 1)$ then there exists
          $\vecv \in R_+(\vecx)$ such that $v_i = 1$.
          \label{enu:properties:Coefficients:4}
  \end{enumerate}
\end{restatable}

\noindent An intuitive explanation of the properties of the SEIC coefficients is
the following
\begin{description}
  \item[\ref{enu:properties:Coefficients:1} --] The coefficients $\P_{\vecv}$
       are both Lipschitz and smooth with Lipschitzness and smoothness
       parameters that depends polynomially in $d$ and $N = 1/\delta + 1$.
  \item[\ref{enu:properties:Coefficients:2} --] The coefficients
       $\P_{\vecv}(\vecx)$ define a convex combination of the vertices
       $R_c(\vecx)$.
  \item[\ref{enu:properties:Coefficients:3} --] For every $\vecx \in [0, 1]^d$,
       out of the $N^d$ coefficients $\P_{\vecv}$ only $d + 1$ have non-zero
       value, or non-zero gradient or non-zero Hessian when evaluated at the
       point $\vecx$. Moreover, given $\vecx \in [0, 1]^d$  we can identify
       these $d + 1$ coefficients efficiently.
  \item[\ref{enu:properties:Coefficients:4} --] For every $\vecx \in [0, 1]^d$
       that is in a cubelet that touches the boundary there is at least one of
       the vertices in $R_+(\vecx)$ that is on the boundary of the continuous
       hypercube $[0, 1]^d$.
\end{description}

  In Section \ref{sec:localNash} in the proof of Theorem
\ref{thm:globalMinOracle} we present a simple application of the existence of
the SEIC coefficients for proving very simple black box oracle lower bounds for
the global minimization problem.

  Based on the existence of these coefficients we are now ready to define the
function $f_{\calC_l}$ which is the main construction of our reduction.

\subsubsection{Definition of a Lipschitz and Smooth Function Based on a BiSperner Instance} \label{sec:hardness:definitionOfFunction}

  In this section our goal is to formally define the function $f_{\calC_l}$ and
prove its Lipschitzness and smoothness properties in Lemma \ref{l:smoothness_d}.

\begin{definition}[Continuous and Smooth Function from Colorings of Bi-Sperner] \label{d:d-payoff}
    Given a binary circuit $\calC_l : \p{\nm{N}}^d \to \{-1, 1\}^d$, we define
  the function $f_{\calC_l} : [0, 1]^d \times [0, 1]^d \to \R$ as follows
  \[ f_{\calC_l}(\vecx, \vecy) = \sum_{j = 1}^d (x_j - y_j)\cdot \alpha_j(\vecx) \]
  \noindent where
  $\alpha_j(\vecx) = - \sum_{\vecv \in \p{\nm{N}}^d} \P_{\vecv}(\vecx) \cdot \calC^j_l(\vecv)$,
  and $\P_{\vecv}$ are the coefficients defined in Definition \ref{def:seic}.
\end{definition}

  We first prove that the function $f_{\calC_l}$ constructed in Definition
\ref{d:d-payoff} is $G$-Lipschitz and $L$-smooth for some appropriately selected
parameters $G$, $L$ that are polynomial in the dimension $d$ and in the
discretization parameter $N$. We use this property to establish that
$f_{\calC_l}$ is a valid input to the promise problem $\gdaFixed$.

\begin{lemma} \label{l:smoothness_d}
    The function $f_{\calC_l}$ of Definition~\ref{d:d-payoff} is
  $O(d^{15}/\delta)$-Lipschitz and $O(d^{27}/\delta^2)$-smooth.
\end{lemma}

\begin{proof}
    If we take the derivative with respect to $x_i$ and $y_i$ and using property
  \ref{enu:properties:Coefficients:2} of the coefficients $\P_{\vecv}$ we get
  the following relations,
  \[ \frac{\partial f_{\calC_l}(\vecx, \vecy)}{\partial x_i} = \sum_{j = 1}^d (x_j - y_j) \cdot \frac{\partial \alpha_j(\vecx)}{\partial x_i} + \alpha_i(\vecx) ~~~ \text{ and } ~~~ \frac{\partial f_{\calC_l}(\vecx, \vecy)}{\partial y_i} = - \alpha_i(\vecx) \]
  where
  \[ \alpha_i(\vecx) = - \sum_{\vecv \in \p{\nm{N}}^d} \P_{\vecv}(\vecx) ~~~ \text{ and } ~~~ \frac{\partial \alpha_j(\vecx)}{\partial x_i} = - \sum_{\vecv \in \p{\nm{N}}^d} \frac{\partial \P_{\vecv}(\vecx)}{\partial x_i} \cdot \calC_l^j(\vecv). \]
  \noindent Now by the property \ref{enu:properties:Coefficients:3} of
  Definition \ref{def:seic} there are most $d + 1$ vertices $\vecv$ of
  $R_c(\vecx)$ with the property $\nabla \P_{\vecv}(\vecx) \neq 0$. Then if we
  also use property \ref{enu:properties:Coefficients:1} we get
  $\abs{\frac{\partial \alpha_j(\vecx)}{\partial x_i}} \leq \Theta(d^{13}/\delta)$
  and using the property \ref{enu:properties:Coefficients:2} we get
  $\abs{\alpha_i(\vecx)} \le 1$. Thus
  $\abs{\frac{\partial f_{\calC_l}(\vecx, \vecy)}{\partial x_i}} \leq \Theta(d^{14}/\delta)$
  and
  $\abs{\frac{\partial f_{\calC_l}(\vecx, \vecy)}{\partial y_i}} \leq \Theta(d)$.
  Therefore we can conclude that
  $\norm{\nabla f_{\calC_l}(\vecx, \vecy)}_2 \le \Theta(d^{15}/\delta)$ and
  hence this proves that the function $f_{\calC_l}$ is Lipschitz continuous with
  Lipschitz constant $\Theta(d^{15}/\delta)$.

    To prove the smoothness of $f_{\calC_l}$, we use the property
  \ref{enu:properties:Coefficients:2} of the Definition \ref{def:seic} and we
  have
  \begin{align*}
    & \frac{\partial^2 f_{\calC_l}(\vecx, \vecy)}{\partial x_i~\partial x_\ell} = \sum_{j = 1}^d (x_j - y_j)\cdot \frac{\partial^2 \alpha_j(\vecx)}{\partial x_i~\partial x_\ell} + \frac{\partial \alpha_{\ell}(\vecx)}{\partial x_i} + &\frac{\partial \alpha_i(\vecx)}{\partial x_\ell}, \\
    & \frac{\partial^2 f_{\calC_l}(\vecx, \vecy)}{\partial x_i~\partial y_\ell} = - \frac{\partial \alpha_{\ell}(\vecx)}{\partial x_i}, \quad \text{and} \quad
    \frac{\partial^2 f_{\calC_l}(\vecx, \vecy)}{\partial y_i~\partial y_\ell} = 0
  \end{align*}
  where
  \[ \frac{\partial^2 \alpha_j(\vecx)}{\partial x_i ~\partial x_\ell} = - \sum_{\vecv \in \p{\nm{N}}^d} \frac{\partial^2 \P_{\vecv}(\vecx)}{\partial x_i~\partial x_\ell} \cdot \calC_l^j(\vecv) \]

  \noindent Again using the property \ref{enu:properties:Coefficients:3} of
  Definition \ref{def:seic} we get that there are most $d + 1$ vertices $\vecv$
  of $R_c(\vecx)$ such that $\nabla^2 \P_{\vecv}(\vecx) \neq 0$. This together
  with the property \ref{enu:properties:Coefficients:1} of Definition
  \ref{def:seic} leads to the fact that
  $\abs{\frac{\partial^2 \alpha_j(\vecx)}{\partial x_i ~\partial x_\ell}} \leq \Theta(d^{25}/\delta^2)$.
  Using the later together with the bounds that we obtained for
  $\abs{\frac{\partial \alpha_j(\vecx)}{\partial x_i}}$ in the beginning of the
  proof we get that
  $\norm{\nabla^2 f_{\calC_l}(\vecx, \vecy)}_{F} \leq \Theta(d^{27}/\delta^2)$,
  where with $\norm{\cdot}_F$ we denote the Frobenious norm. Since the bound on
  the Frobenious norm is a bound to the spectral norm too, we get that the
  function $f_{\calC_l}$ is $\Theta(d^{27}/\delta^2)$-smooth.
\end{proof}

\subsubsection{Description and Correctness of the Reduction -- Proof of Theorem \ref{thm:localNashHardness}} \label{sec:hardness:Reduction}

  We start with a description of the reduction from $\HD$-$\BiSperner$ to
$\gdaFixed$. Suppose we have an instance of $\HD$-$\BiSperner$ given by boolean
circuit $\calC_l : \p{\nm{N}}^d \to \{-1, 1\}^d$, we construct an instance of
$\gdaFixed$ according to the following set of rules.
\smallskip

\noindent \textbf{$\boldsymbol(\star)$ Construction of Instance for Fixed Points of Gradient Descent/Ascent.} \label{lbl:constructionHighD}
\begin{enumerate}
  \item[$\blacktriangleright$] The payoff function is the real-valued function
  $f_{\calC_l}(\vecx, \vecy)$ from the Definition \ref{d:d-payoff}.
  \item[$\blacktriangleright$] The domain is the polytope $\calP(\matA, \vecb)$
  that we described in Section \ref{sec:computational}. The matrix $\matA$ and
  the vector $\vecb$ are computed so that the following inequalities hold
  \begin{equation} \label{eq:HDproof:domainDefinition}
    x_i - y_i \le \Delta, ~~ y_i - x_i \le \Delta ~~~~ \text{for all} ~~ i \in [d]
  \end{equation}
  where $\Delta = t \cdot \delta / d^{14}$, with $t \in \R_+$ be a constant such
  that
  $\abs{\frac{\partial \P_{\vecv}(\vecx)}{\partial x_i}} \cdot \frac{\delta}{d^{12}} t \le \frac{1}{2}$,
  for all $\vecv \in \p{\nm{N}}^d$ and $\vecx \in [0, 1]^d$. The fact that such
  a constant $t$ exists follows from the property
  \ref{enu:properties:Coefficients:1} of the smooth and efficient coefficients.
  \item[$\blacktriangleright$] The parameter $\alpha$ is set to be equal to
  $\Delta/3$.
  \item[$\blacktriangleright$] The parameters $G$ and $L$ are set to be equal to
  the upper bounds on the Lipschitzness and the smoothness of $f_{\calC_l}$
  respectively that we derived in Lemma \ref{l:smoothness_d}. Namely we have
  that $G = O(d^{15}/\delta)$ and $L = O(d^{27}/\delta^2)$.
\end{enumerate}

  The first thing to observe is that the afore-described reduction is
polynomial-time. For this observe that all of $\alpha$, $G$, $L$, $\matA$, and
$\vecb$ have representation that is polynomial in $d$ even if we use unary
instead of binary representation. So the only thing that remains is the
existence of a Turing machine $\calC_{f_{\calC_l}}$ that computes the function
and the gradient value of $f_{\calC_l}$ in time polynomial to the size of
$\calC_l$ and the requested accuracy. To prove this we need a detailed
description of the SEIC coefficients and for this reason we postpone the proof
of this to the Appendix \ref{sec:proof:TuringMachine}. Here we state the
formally the result that we prove in the Appendix \ref{sec:proof:TuringMachine}
which together with the discussion above proves that our reduction is indeed
polynomial-time.

\begin{theorem}\label{t:Turing_Machine}
    Given a binary circuit $\calC_l: \p{\nm{N}}^d \rightarrow \{-1, 1\}^d$ that
  is an input to the $\HD$-$\BiSperner$ problem. Then, there exists a
  polynomial-time Turing machine $\calC_{f_{\calC_l}}$, that can be constructed
  in polynomial-time from the circuit $\calC_l$ such that for all vector
  $\vecx, \vecy \in [0,1]^d$ and accuracy $\eps >0$, $\calC_{f_{\calC_l}}$
  computes both $z \in \R$ and $\vecw \in \R^d$ such that
  \[ \abs{z - f_{\calC_l}(\vecx, \vecy)} \le \eps, \quad  \norm{\vecw - \nabla f_{\calC_l}(\vecx, \vecy)}_2 \le \eps. \]
  Moreover the running time of $\calC_{f_{\calC_l}}$ is polynomial in the binary
  representation of $\vecx$, $\vecy$, and $\log(1/\eps)$.
\end{theorem}

  We also observe that according to Lemma \ref{l:smoothness_d}, the function
$f_{\calC_l}$ is both $G$-Lipschitz and $L$-smooth and hence the output of our
reduction is a valid input for the constructed instance of the promise problem
$\gdaFixed$. The next step is to prove that the vector $\vecx^{\star}$ of every
solution $(\vecx^{\star}, \vecy^{\star})$ of $\gdaFixed$ with input as we
described above, lies in a cubelet that is either panchromatic according to
$\calC_l$ or is a violation of the rules for proper coloring of the
$\HD$-$\BiSperner$ problem.

\begin{lemma} \label{c:conditions_d}
   Let $\calC_l$ be an input to the $\HD$-$\BiSperner$ problem, let
  $f_{\calC_l}$ be the corresponding $G$-Lipschitz and $L$-smooth function
  defined in Definition \ref{d:d-payoff}, and let $\calP(\matA, \vecb)$ be the
  polytope defined by \eqref{eq:HDproof:domainDefinition}. If
  $(\vecx^{\star}, \vecy^{\star})$ is any solution to the $\gdaFixed$ problem
  with input $\alpha$, $G$, $L$, $\calC_{f_{\calC_l}}$, $\matA$, and $\vecb$,
  defined in \hyperref[lbl:constructionHighD]{$(\star)$} then the following
  statements hold, where we remind that $\Delta = t \cdot \delta / d^{14}$.
  \begin{enumerate}
    \item[$\diamond$] If $x_i^\star \in (\alpha, 1 - \alpha)$ and
    $x_i^{\star} \in (y_i^\star - \Delta + \alpha, y_i^\star + \Delta - \alpha)$
    then
    $\abs{\frac{\partial f_{\calC_l}(\vecx^\star,\vecy^\star)}{\partial x_i}}
    \leq \alpha$.
    \item[$\diamond$] If $x^\star_i \le \alpha$ or
    $x^\star_i \le y^\star_i - \Delta + \alpha$ then
    $\frac{\partial f_{\calC_l}(\vecx^\star,\vecy^\star)}{\partial x_i} \ge - \alpha$.
    \item[$\diamond$] If $x^\star_i \ge 1 - \alpha$ or
    $x^\star_i \ge y^\star_i + \Delta - \alpha$ then
    $\frac{\partial f_{\calC_l}(\vecx^\star, \vecy^\star)}{\partial x_i} \le \alpha$.
  \end{enumerate}
  The symmetric statements for $y_i^{\star}$ hold.
  \begin{enumerate}
    \item[$\diamond$] If $y_i^\star \in  (\alpha, 1 - \alpha)$ and
    $y_i^{\star} \in (x_i^\star - \Delta + \alpha, x_i^\star + \Delta - \alpha)$
    then
    $\abs{\frac{\partial f_{\calC_l}(\vecx^\star, \vecy^\star)}{\partial y_i}} \le \alpha$.
    \item[$\diamond$] If $y^\star_i \le \alpha$ or
    $y^\star_i \le x^\star_i - \Delta + \alpha$ then
    $\frac{\partial f_{\calC_l}(\vecx^\star, \vecy^\star)}{\partial y_i} \le \alpha$.
    \item[$\diamond$] If $y^\star_i \ge 1 - \alpha$ or
    $y^\star_i \ge x^\star_i + \Delta - \alpha$ then
    $\frac{\partial f_{\calC_l}(\vecx^\star,\vecy^\star)}{\partial y_i} \ge - \alpha$.
  \end{enumerate}
\end{lemma}

\begin{proof}
    The proof of this lemma is identical to the proof of Lemma
  \ref{c:conditions} and for this reason we skip the details of the proof here.
\end{proof}

\begin{lemma} \label{l:main_d}
   Let $\calC_l$ be an input to the $\HD$-$\BiSperner$ problem, let
  $f_{\calC_l}$ be the corresponding $G$-Lipschitz and $L$-smooth function
  defined in Definition \ref{d:d-payoff}, and let $\calP(\matA, \vecb)$ be the
  polytope defined by \eqref{eq:HDproof:domainDefinition}. If
  $(\vecx^{\star}, \vecy^{\star})$ is any solution to the $\gdaFixed$ problem
  with input $\alpha$, $G$, $L$, $\calC_{f_{\calC_l}}$, $\matA$, and $\vecb$,
  defined in \hyperref[lbl:constructionHighD]{$(\star)$}, then none of the
  following statements hold for the cubelet $R(\vecx^{\star})$.
  \begin{enumerate}
    \item $x_i^{\star} \ge 1/(N - 1)$ and for any $\vecv \in R_+(\vecx^{\star})$, it
    holds that $\calC_l^i(\vecv) = -1$.
    \item $x_i^{\star} \le 1 - 1/(N - 1)$ and for any $\vecv \in R_+(\vecx^{\star})$,
    it holds that $\calC_l^1(\vecv) = +1$.
  \end{enumerate}
\end{lemma}

\begin{proof}
    We prove that there is no solution $(\vecx^\star, \vecy^\star)$ of
  $\gdaFixed$ that satisfies the statement 1. and the fact that
  $(\vecx^\star, \vecy^\star)$ cannot satisfy the statement 2. follows
  similarly. It is convenient for us to define
  $\hat{\vecx} = \vecx^\star - \nabla_x f_{\calC_l}(\vecx^{\star}, \vecy^{\star})$,
  $K(\vecy^{\star}) = \{\vecx \mid (\vecx,\vecy^\star) \in \calP (\matA,\vecb))\}$,
  $\vecz = \Pi_{K(\vecy^{\star})} \hat{\vecx}$, and
  $\hat{\vecy} = \vecy^\star - \nabla_y f_{\calC_l}(\vecx^{\star}, \vecy^{\star})$,
  $K(\vecx^{\star}) = \{\vecy \mid (\vecx^\star, \vecy) \in \calP (\matA,\vecb))\}$,
  $\vecw = \Pi_{K(\vecx^{\star})} \hat{\vecy}$.

    For the sake of contradiction we assume that there exists a solution of
  $(\vecx^\star,\vecy^\star)$ such that $x_1^{\star} \ge 1/(N - 1)$ and for any
  $\vecv \in R_+(\vecx^{\star})$ it holds that $\calC_l^i(\vecv) = -1$. Using
  this fact, we will prove that (1)
  $\frac{\partial f_{\calC_l}(\vecx^\star, \vecy^\star)}{\partial x_i} \geq 1/2$,
  and (2)
  $\frac{\partial f_{\calC_l}(\vecx^\star, \vecy^\star)}{\partial y_i} = -1$.

    Let
  $R(\vecx^{\star}) = \left[\frac{c_1}{N - 1}, \frac{c_1 + 1}{N - 1} \right] \times \cdots \times \left[\frac{c_d}{N - 1}, \frac{c_d + 1}{N - 1} \right]$,
  then since all the corners $\vecv \in R_+(\vecx^{\star})$ have
  $\calC_l^i(\vecv) = -1$, from the Definition \ref{d:d-payoff} we have that
  \begin{align*}
    f_{\calC_l}(\vecx^{\star}, \vecy^{\star})
      = & ~ (x_i^{\star} - y_i^{\star}) + \sum_{j = 1, j \neq i}^d (x_j^{\star} - y_j^{\star})\cdot \alpha_j(\vecx)
  \end{align*}
  If we differentiate this with respect to $y_i$ we immediately get that
  $\frac{\partial f_{\calC_l}(\vecx^{\star}, \vecy^{\star})}{\partial y_i} = -1$.
  On the other hand if we differentiate with respect to $x_i$ we get
  \begin{eqnarray*}
    \frac{\partial f_{\calC_l}(\vecx^{\star}, \vecy^{\star})}{\partial x_i}
    &  =  & 1 + \sum_{j = 1, j \neq i}^d (x_j - y_j) \cdot \frac{\partial \alpha_j(\vecx)}{\partial x_i} \\
    & \ge & 1 - \sum_{j \neq i} \abs{x_j - y_j} \cdot \abs{\frac{\partial \alpha_j(\vecx)}{\partial x_i}} \\
    & \ge & 1 - \Delta \cdot d \cdot \Theta\left(\frac{d^{13}}{\delta}\right) \\
    & \ge & 1/2
  \end{eqnarray*}
  where the above follows from the following facts: (1) that
  $\abs{\frac{\partial \alpha_j(\vecx)}{\partial x_l}} \leq \Theta(d^{13}/\delta)$,
  which is proved in the proof of Lemma~\ref{l:smoothness_d}, (2)
  $\abs{x_j - y_j} \leq \Delta$, and (3) the definition of $\Delta$. Now it is
  easy to see that the only way to satisfy both
  $\frac{\partial f_{\calC_l}(\vecx^\star, \vecy^\star)}{\partial x_i} \geq 1/2$
  and $\abs{z_i - x_i^{\star}} \le \alpha$ is that either $x_i^\star \le \alpha$
  or $x_i^\star \le y_i^\star - \Delta + \alpha$. The first case is excluded by
  the assumption of the first statement of our lemma and our choice of
  $\alpha = \Delta/3 < 1/(N - 1)$, thus it holds that
  $x_i^\star \le y_i^\star - \Delta + \alpha$. But then we can use the case 3.
  for the $y$ variables of Lemma \ref{c:conditions} and we get that
  $\frac{\partial f_{\calC_l}(\vecx^\star, \vecy^\star)}{\partial y_1} \ge - \alpha$,
  which cannot be true since we proved that
  $\frac{\partial f_{\calC_l}(\vecx^\star, \vecy^\star)}{\partial y_i} = -1$.
  Therefore we have a contradiction and the first statement of the lemma holds.
  Using the same reasoning we prove the second statement too.
\end{proof}

\noindent We are now ready to complete the proof that the our reduction from
$\HD$-$\BiSperner$ to $\gdaFixed$ is correct and hence we can prove Theorem
\ref{thm:localNashHardness}.

\begin{proof}[Proof of Theorem \ref{thm:localNashHardness}]
    Let $(\vecx^{\star}, \vecy^{\star})$ be a solution to the $\gdaFixed$
  problem with input a Turing machine that represents the function
  $f_{\calC_l}$, $\alpha = \Delta/3$, where $\Delta = t \cdot \delta / d^{14}$,
  $G = \Theta(d^{15}/\delta)$, $L = \Theta(d^{27}/\delta^2)$, and $\matA$,
  $\vecb$ as described in \hyperref[lbl:constructionHighD]{($\star$)}.

    For each coordinate $i$, there exist the following three mutually exclusive
  cases,
  \begin{enumerate}
    \item[$\triangleright$]
    $\boldsymbol{\frac{1}{N - 1} \le x_i^\star \le 1 - \frac{1}{N - 1}}$:
    Since $\abs{R_+(\vecx^\star)} \ge 1$, it follows directly from Lemma
    \ref{l:main_d} that there exists $\vecv \in R_+(\vecx^{\star})$ such that
    $\calC_l^i(\vecv) = -1$ and $\vecv' \in R_+(\vecx^{\star})$ such that
    $\calC_l^i(\vecv) = +1$.
    \item[$\triangleright$] $\boldsymbol{0 \le x_i^\star < \frac{1}{N - 1}}$:
    Let $\calC_l^i(\vecv) = -1$ for all $\vecv \in R_+(\vecx^\star)$. By the
    property \ref{enu:properties:Coefficients:4} of the SEIC coefficients, we
    have that there exists $\vecv \in R_+(x^\star)$ with $v_i = 0$. This node is
    hence a solution of type 2. for the $\HD$-$\BiSperner$ problem.
    \item[$\triangleright$]
    $\boldsymbol{1 - \frac{1}{N - 1} < x_i^\star \le 1}$: Let
    $\calC_l^i(\vecv) = +1$ for all $\vecv \in R_+(\vecx^\star)$. By the
    property \ref{enu:properties:Coefficients:4} of the SEIC coefficients, we
    have that there exists $\vecv \in R_+(x^\star)$ with $v_i = 1$. This node is
    hence a solution of type 3. for the $\HD$-$\BiSperner$ problem.
  \end{enumerate}

  \noindent Since $\R_+(\vecx^{\star})$ computable in polynomial time given
  $\vecx^{\star}$, we can easily check for every $i \in [d]$ whether any of the
  above cases hold. If at least for some $i \in [d]$ the $2$nd or the $3$rd case
  from above hold, then the corresponding vertex gives a solution to the
  $\HD$-$\BiSperner$ problem and therefore our reduction is correct. Hence we
  may assume that for every $i \in [d]$ the $1$st of the above cases holds. This
  implies that the cubelet $R(\vecx^{\star})$ is pachromatic and therefore it is
  a solution to the problem $\HD$-$\BiSperner$. Finally, we observe that the
  function that we define has range $[-d, d]$ and hence the Theorem
  \ref{thm:localNashHardness} follows using Theorem \ref{t:LocalMaxMin-GDA}.
\end{proof}

\section{Smooth and Efficient Interpolation Coefficients} \label{sec:seic}

  In this section we describe the construction of the smooth and efficient
interpolation coefficients (SEIC) that we introduced in Section
\ref{sec:hardness:seic}. After the description of the construction we present
the statements of the lemmas that prove the properties
\ref{enu:properties:Coefficients:1} - \ref{enu:properties:Coefficients:4} of
their Definition \ref{def:seic} and we refer to the Appendix \ref{s:Q_coef}. We
first remind the definition of the SEIC coefficients.

\seic*

\noindent Our main goal in this section is to prove the following theorem.
\begin{theorem} \label{thm:seic}
    For every $d \in \N$ and every $N = \poly(d)$ there exist a family of
  functions $\calI_{d, N}$ that satisfies the properties
  \ref{enu:properties:Coefficients:1} - \ref{enu:properties:Coefficients:4} of
  Definition \ref{def:seic}.
\end{theorem}

  One important component of the construction of the SEIC coefficients is the
\textit{smooth-step functions} which we introduce in Section
\ref{sec:seic:smoothStep}. These functions also provide a toy example of smooth
and efficient interpolation coefficients in $1$ dimension. Then in Section
\ref{sec:seic:construction} we present the construction of the SEIC coefficients
in multiple dimensions and in Section \ref{sec:seic:lemmas} we state the main
lemmas that lead to the proof of Theorem \ref{thm:seic}.

\subsection{Smooth Step Functions -- Toy Single Dimensional Example} \label{sec:seic:smoothStep}

  Smooth step functions are real-valued function $g : \R \to \R$ of a single
real variable with the following properties
\begin{description}
  \item[Step Value.] For every $x \le 0$ it holds that $g(x) = 0$, for every
    $x \ge 1$ it holds that $g(x) = 1$ and for every $x \in [0, 1]$ it holds
    that $S(x) \in [0, 1]$.
  \item[Smoothness.] For some $k$ it holds that $g$ is $k$ times continuously
    differentiable and its $k$th derivative satisfies $g^{(k)}(0) = 0$ and
    $g^{(k)}(1) = 0$.
\end{description}
\noindent The largest number $k$ such that the smoothness property from above
holds is characterizes the \textit{order of smoothness} of the smooth step
function $g$.

  In Section \ref{sec:hardness:2D} we have already defined and used the smooth
step function of order $1$. For the construction of the SEIC coefficients we use
the smooth step function of order $2$ and the smooth step function of order
$\infty$ defined as follows.

\begin{definition} \label{def:smoothStep}
    We define the smooth step function $S : \R \to \R$ of order $2$ as the
  following function
  \[ S(x) = \begin{cases}
              6 x^5 - 15 x^4 + 10 x^3 & ~~~~ x \in (0, 1) \\
              0             & ~~~~ x \le 0 \\
              1             & ~~~~ x \ge 1
            \end{cases}. \]
  We also define the smooth step function $S_{\infty} : \R \to \R$ of order
  $\infty$ as the following function
  \[ S_{\infty}(x) = \begin{cases}
                       \frac{2^{- 1/x}}{2^{- 1/x} + 2^{- 1/(1 - x)}} & ~~~~ x \in (0, 1) \\
                       0             & ~~~~ x \le 0 \\
                       1             & ~~~~ x \ge 1
                     \end{cases}. \]
  We note that we use the notation $S$ instead of $S_2$ for the smooth step
  function of order $2$ for simplicitly of the exposition of the paper.
\end{definition}

  We present a plot of these step function in Figure \ref{fig:stepFunctions}, and we
summarize some of their properties in Lemma \ref{lem:stepFunctions}. A more detailed lemma
with additional properties of $S_{\infty}$ that are useful for the proof of Theorem
\ref{thm:seic} is presented in Lemma \ref{l:infty} in the Appendix \ref{s:Q_coef}.

\begin{figure}[t]
  \centering
  \begin{subfigure}{0.45\textwidth}
    \centering
    \includegraphics[scale=0.2]{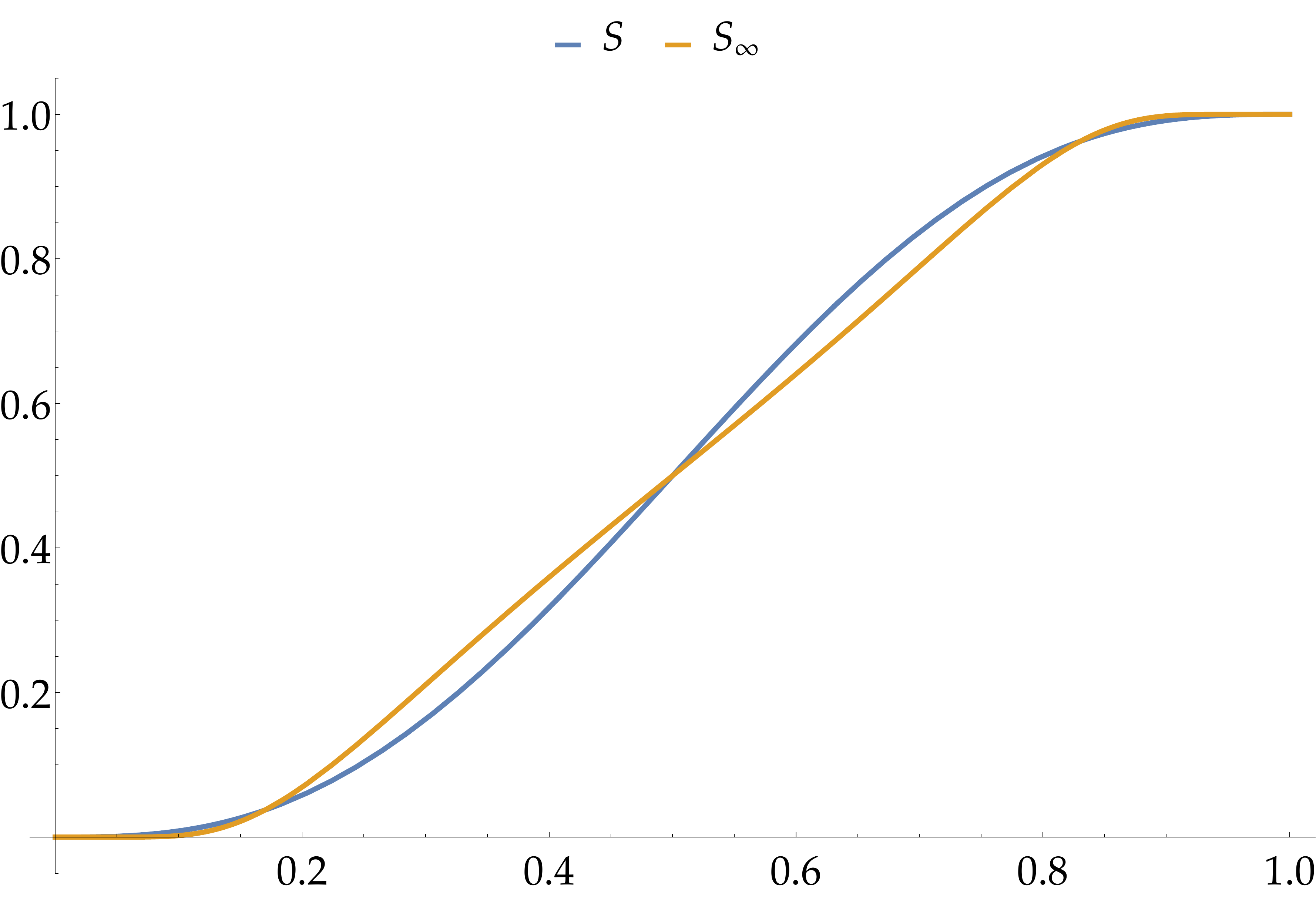}
    \caption{Functions $S$ and $S_{\infty}$.}
  \end{subfigure}
  ~
  \begin{subfigure}{0.45\textwidth}
    \centering
    \includegraphics[scale=0.2
    ]{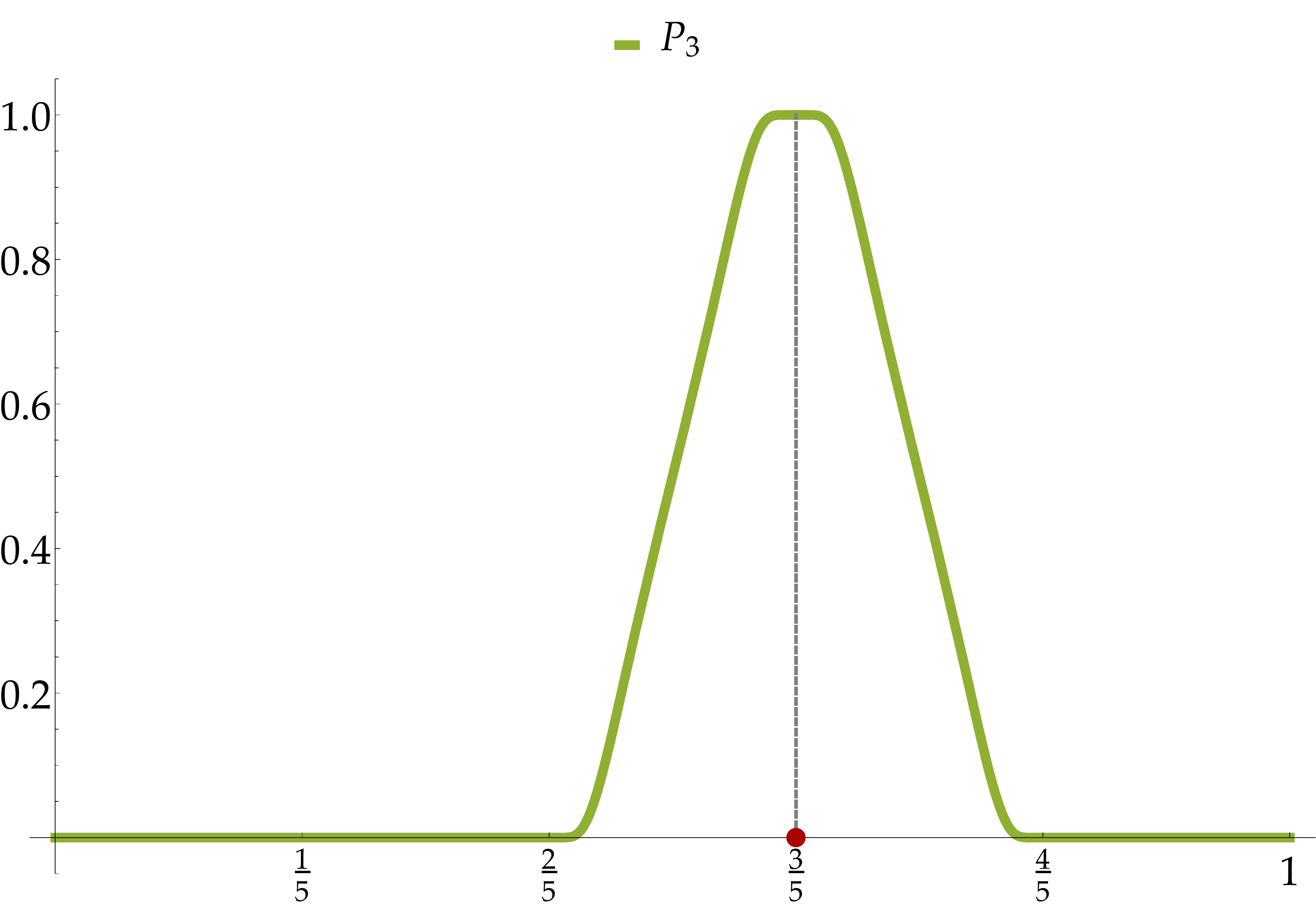}
    \caption{The function $P_3$ from Example \ref{exm:1DSeic}.}
  \end{subfigure}
  \caption{(a) The smooth step function $S$ of order $1$ and the smooth step function
  $S_{\infty}$ of order $\infty$. As we can see both $S$ and $S_{\infty}$ are continuous
  and continuously differentiable functions but $S_{\infty}$ is much more flat around $0$
  and $1$ since it has all its derivatives equal to $0$ both at the point $0$ and at the
  point $1$. This makes the $S_{\infty}$ function infinitely many times differentiable.
  (b) The constructed $P_3$ function of the family of SEIC coefficients for single
  dimensional case with $N = 5$. For details we refer to the Example \ref{exm:1DSeic}.}
  \label{fig:stepFunctions}
\end{figure}

\begin{lemma} \label{lem:stepFunctions}
    Let $S$ and $S_{\infty}$ be the smooth step functions defined in Definition
  \ref{def:smoothStep}. It holds both $S$ and $S_{\infty}$ are monotone increasing
  functions and that $S(0) = 0$, $S(1) = 1$ and also $S'(0) = S'(1) = S''(0) = S''(1) = 0$.
  It also holds that $S_{\infty}(0) = 0$, $S_{\infty}(1) = 1$ and also
  $S^{(k)}_{\infty}(0) = S^{(k)}_{\infty}(1) = 0$ for every $k \in \N$. Additionally it
  holds for every $x$ that $\abs{S'(x)} \le 2$, and $\abs{S''(x)} \le 6$ whereas
  $\abs{S'_{\infty}(x)} \le 16$, and $\abs{S''_{\infty}(x)} \le 32$.
\end{lemma}

\begin{proof}
    For the function $S$ we compute $S'(x) = 30 x^4 - 60 x^3 + 30 x^2$ for $x \in [0, 1]$
  and $S'(x) = 0$ for $x \not\in [0, 1]$. Therefore we can easily get that
  $\abs{S'(x)} \le 2$ for all $x \in \R$. We also have that
  $S''(x) = 120 x^3 - 180 x^2 + 60 x$ for $x \in (0, 1)$ and $S''(x) = 0$ for
  $x \not\in [0, 1]$ hence we can conclude that $\abs{S''(x)} \le 6$.

    The calculations for $S_{\infty}$ are more complicated. We have that
  \[ S'_{\infty}(x) = \ln(2) \frac{\exp\p{\frac{\ln(2)}{x (1 - x)}} (1 - 2 x (1 - x))}{\p{\exp\p{\frac{\ln(2)}{x}} + \exp\p{\frac{\ln(2)}{1 - x}}}^2 \p{1 - x}^2 x^2}. \]
  We set
  $h(x) \triangleq \p{\exp\p{\frac{\ln(2)}{x}} + \exp\p{\frac{\ln(2)}{1 - x}}} \p{1 - x}^2 x^2$
  for $x \in [0, 1]$ and doing simple calculations we get that for $x \le 1/2$ it holds
  that $h(x) \ge \frac{1}{4} \exp\p{\frac{\ln(2)}{x}} x^2$. But the later can be easily
  lower bounded by $1/4$. Applying the same argument for $x \ge 1/2$ we get that in general
  $h(x) \ge 1/4$. Also it is not hard to see that for $x \le 1/2$ it holds that
  $\exp\p{\frac{\ln(2)}{x (1 - x)}} \le 4 \exp\p{\frac{\ln(2)}{x}}$, whereas for
  $x \ge 1/2$ it holds that
  $\exp\p{\frac{\ln(2)}{x (1 - x)}} \le 4 \exp\p{\frac{\ln(2)}{1 - x}}$. Combining all
  these we can conclude that $\abs{S'_{\infty}(x)} \le 16$. Using similar argument we can
  prove that $\abs{S''_{\infty}(x)} \le 32$. For all the derivatives of $S_{\infty}$ we can
  inductively prove that
  \[ S^{(k)}_{\infty}(x) = \sum_{i = 0}^{k - 1} h_i(x) \cdot S^{(i)}_{\infty}(x),  \]
  where $h_0(1) = 0$ and all the functions $h_i(x)$ are bounded. Then the fact that all the
  derivatives of $S_{\infty}$ vanish at $0$ and at $1$ follows by a simple inductive
  argument.
\end{proof}

\begin{example}[Single Dimensional Smooth and Efficient Interpolation Coefficients] \label{exm:1DSeic}
    Using the smooth step functions that we described above we can get a construction of
  SEIC coefficients for the single dimensional case. Unfortunately the extension to
  multiple dimensions is substantially harder and invokes new ideas that we explore later
  in this section. For the single dimensional problem of this example we have the interval
  $[0, 1]$ divided with $N$ discrete points and our goal is to design $N$ functions
  $\P_1$ - $\P_N$ that satisfy the properties
  \ref{enu:properties:Coefficients:1} - \ref{enu:properties:Coefficients:4} of
  Definition \ref{def:seic}. A simple construction of such functions is the following
  \[ \P_i(x) = \begin{cases}
                 S_{\infty}\p{N \cdot x - (i - 1)}   & ~~~~ x \le \frac{i}{N - 1} \\
                 S_{\infty}\p{i + 1 - N \cdot x}     & ~~~~ x > \frac{i}{N - 1}
               \end{cases}. \]
  Based on Lemma \ref{lem:stepFunctions} it is not hard then to see that $\P_i$ is twice
  differentiable and it has bounded first and second derivatives, hence it satisfies
  property \ref{enu:properties:Coefficients:1} of Definition \ref{sec:seic}. Using the fact
  that $1 - S_{\infty}(x) = S_{\infty}(1 - x)$ we can also prove property
  \ref{enu:properties:Coefficients:2}. Finally properties
  \ref{enu:properties:Coefficients:3} and \ref{enu:properties:Coefficients:4} can be proved
  via the definition of the coefficient $\P_i$ from above. In Figure
  \ref{fig:stepFunctions} we can see the plot of $\P_3$ for $N = 5$. We leave the exact
  proofs of this example as an exercise for the reader.
\end{example}

\subsection{Construction of SEIC Coefficients in High-Dimensions} \label{sec:seic:construction}

  The goal of this section is to present the construction of the family
$\calI_{d, N}$ of smooth and efficient interpolation coefficients for every number
of dimensions $d$ and any discretization parameter $N$. Before diving into the
details of our construction observe that even the 2-dimensional case with $N = 2$
is not trivial. In particular, the first attempt would be to define the SEIC
coefficients based on the simple split of the square $[0, 1]^2$ to two triangles
divided by the diagonal of $[0, 1]^2$. Then using any soft-max function that is
twice continuously differentiable we define a convex combination at every triangle.
Unfortunately this approach cannot work since the resulting coefficients have
discontinuous gradients along the diagonal of $[0, 1]^2$. We leave the presice
calculations of this example as an exercise to the reader.

  We start with some definitions about the orientation and the representation of
the cubelets of the grid $\p{\nm{N}}^d$. Then we proceed with the definition of
the $Q_{\vecv}$ functions in Definition \ref{d:cof}. Finally using $Q_{\vecv}$
we can proceed with the construction of the SEIC coefficients.

\begin{definition}[Source and Target of Cubelets] \label{d:orientation}
    Each cubelet
  $\left[\frac{c_1}{N - 1}, \frac{c_1 + 1}{N - 1}\right] \times \cdots \times \left[\frac{c_d}{N - 1}, \frac{c_d + 1}{N - 1}\right]$,
  where $\vecc \in \p{\nm{N - 1}}^d$ admits a \textbf{source vertex}
  $\vecs^{\vecc} = (s_1, \ldots, s_d) \in \p{\nm{N}}^d$ and a
  \textbf{target vertex}
  $\vect^{\vecc} = (t_1, \ldots, t_d) \in \p{\nm{N}}^d$ defined as follows,
  \[ s_j = \left\{ \begin{array}{ll}
                     c_j + 1 & c_j \text{ is odd} \\
                     c_j     & c_j \text{ is even} \\
                   \end{array}
           \right.
     ~~~~~\text{and}~~~~~
     t_j = \left\{ \begin{array}{ll}
                     c_j     & c_j \text{ is odd} \\
                     c_j + 1 & c_j \text{ is even} \\
                   \end{array}
           \right. \]
\end{definition}

\noindent Notice that the source $\vecs^{\vecc}$ and the target $\vect^{\vecc}$ are
vertices of the cubelet whose down-left corner is $\vecc$.

\begin{definition}(Canonical Representation) \label{d:canonical_representation}
    Let $\vecx \in [0,1]^d$ and
  $R(\vecx) = \left[\frac{c_1}{N - 1}, \frac{c_1 + 1}{N - 1}\right] \times \cdots \times \left[\frac{c_d}{N - 1}, \frac{c_d + 1}{N - 1}\right]$ where
  $\vecc \in \p{\nm{N - 1}}^d$. The \textbf{\textit{canonical representation} of
  $\vecx$ under cubelet with down-left corner $\vecc$}, denoted by
  $\vecp_{\vecx}^{\vecc} = (p_1,\ldots,p_d)$ is defined as follows,
  \[ p_j = \frac{x_j - s_j}{t_j - s_j} \]
  where $\vect^{\vecc} = (t_1, \ldots, t_d)$ and
  $\vecs^{\vecc} = (s_1,\ldots,s_d)$ are respectively the \textit{target} and
  the \textit{source} of $R(\vecx)$.
\end{definition}

\begin{definition}[Defining the functions $Q_{\vecv}(\vecx)$] \label{d:cof}
    Let $\vecx \in [0,1]^d$ lying in the cublet
  \[ R(\vecx) = \left[\frac{c_1}{N - 1}, \frac{c_1 + 1}{N - 1}\right] \times \cdots \times \left[\frac{c_d}{N - 1}, \frac{c_d + 1}{N - 1}\right],\]
  with corners
  $R_c(\vecx) = \{c_1, c_1 + 1\} \times \cdots \times \{c_d, c_d + 1\}$, where
  $\vecc \in \p{\nm{N - 1}}^d$. Let also $\vecs^{\vecc} = (s_1,\ldots,s_d)$ be
  the source vertex of $R(\vecx)$ and $\vecp_{\vecx}^{\vecc} = (p_1,\ldots,p_d)$
  be the canonical representation of $\vecx$. Then for each vertex
  $\vecv \in R_c(\vecx)$ we define the following partition of the set of
  coordinates $[d]$,
  \[ A_{\vecv}^{\vecc}=\{j:~~|v_j - s_j|=0\} \text{  and  } B_{\vecv}^{\vecc}=\{j:~~|v_j - s_j|=1\} \]
  If there exist $j \in A_{\vecv}^{\vecc}$ and $\ell \in B_{\vecv}^{\vecc}$ such
  that $p_{j} \geq p_{\ell}$ then $Q_{\vecv}^{\vecc}(\vecx) = 0$. Otherwise we
  define\footnote{We note that in the following expression $\prod$ denotes the
  product symbol and should not be confused with the projection operator used in
  the previous sections.}
  \begin{equation*}
    Q_{\vecv}^{\vecc}(\vecx) =
      \begin{cases}
        \prod_{j \in A_{\vecv}^{\vecc}} \prod_{\ell \in B_{\vecv}^{\vecc}}S_{\infty}(S(p_\ell) - S(p_j)) & A_{\vecv}^{\vecc},B_{\vecv}^{\vecc} \neq \varnothing \\
        \prod_{\ell=1}^d S_{\infty}(1- S(p_\ell))& B_{\vecv}^{\vecc} = \varnothing \\
        \prod_{j=1}^d S_{\infty}(S(p_j)) & A_{\vecv}^{\vecc} = \varnothing
      \end{cases}
  \end{equation*}
  where $S_{\infty}(x)$ and $S(x)$ are the smooth step function defined in Definition
  \ref{def:smoothStep}.
\end{definition}

\noindent To provide a better understanding of the Definitions
\ref{d:orientation}, \ref{d:canonical_representation}, and \ref{d:cof} we
present the following $3$-dimensional example.

\begin{example}
    We consider a case where $d = 3$ and $N = 3$. Let
  $\vecx = (1.3/3, 2.5/3, 0.3/3)$ lying in the cubelet
  $R(\vecx) = \left[\frac{1}{3}, \frac{2}{3}\right] \times \left[\frac{2}{3}, 1\right] \times \left[0, \frac{1}{3}\right]$,
  and let $\vecc = (1, 2, 0)$. Then the source of $R(\vecx)$ is
  $\vecs^{\vecc} = (2, 2, 0)$ and the target $\vect^{\vecc} = (1, 3, 1)$
  (Definition~\ref{d:orientation}). The canonical representation of $\vecx$ is
  $\vecp_{\vecx}^{\vecc} = (0.7 , 0.5, 0.3)$
  (Definition~\ref{d:canonical_representation}). The only vertices with no-zero
  coefficients $Q_{\vecv}^{\vecc}(\vecx)$ are those belonging in the set
  $R_+(\vecx) = \{(1, 3, 1), (1, 3, 0), (1, 2, 0), (2, 2, 0)\}$ and again by
  Definition~\ref{d:cof} we have that
  \begin{enumerate}
    \item[$\triangleright$] $Q_{(1, 3, 1)}(\vecx) = S_{\infty}(S(0.3)) \cdot S_{\infty}(S(0.5)) \cdot S_{\infty}(S(0.7))$,
    \item[$\triangleright$] $Q_{(1, 3, 0)}(\vecx) = S_{\infty}(S(0.5) - S(0.3)) \cdot S_{\infty}(S(0.7) - S(0.3))$,
    \item[$\triangleright$] $Q_{(1, 2, 0)}(\vecx) = S_{\infty}(S(0.7) - S(0.3)) \cdot S_{\infty}(S(0.7) - S(0.5))$,
    \item[$\triangleright$] $Q_{(2, 2, 0)}(\vecx) = S_{\infty}(1-S(0.3)) \cdot S_{\infty}(1-S(0.5)) \cdot S_{\infty}(1-S(0.7))$.
  \end{enumerate}
\end{example}

\noindent Now based on the Definitions \ref{d:orientation},
\ref{d:canonical_representation}, and \ref{d:cof} we are ready to present the
construction of the smooth and efficient interpolation coefficients.

\begin{definition}[Construction of SEIC Coefficients] \label{d:_coeff}
    Let $\vecx \in [0, 1]^d$ lying in the cubelet
  $R(\vecx) = \left[\frac{c_1}{N - 1}, \frac{c_1 + 1}{N - 1}\right] \times \cdots \times \left[\frac{c_d}{N - 1}, \frac{c_d + 1}{N - 1}\right]$.
  Then for each vertex $\vecv \in \p{\nm{N}}^d$ the coefficient
  $\P_{\vecv}(\vecx)$ is defined as follows,
  \[ \P_{\vecv}(\vecx) = \left\{ \begin{array}{ll}
                                 Q^{\vecc}_{\vecv}(\vecx)/(\sum_{\vecv \in R_c(\vecx)}Q^{\vecc}_{\vecv}(\vecx)) & \text{ if }  \vecv \in R_c(\vecx) \\
                                 0 & \text{ if } \vecv \notin R_c(\vecx)
                               \end{array} \right. \]
  where the functions $Q_{\vecv}^{\vecc}(\vecx) \ge 0$ are defined in Definition
  \ref{d:cof} for any $\vecv \in R_c(\vecx)$.
\end{definition}

\subsection{Sketch of the Proof of Theorem \ref{thm:seic}} \label{sec:seic:lemmas}

  First it is necessary to argue that $\P_{\vecv}(\vecx)$ is a continuous
function since it could be the case that
$Q^{\vecc}_{\vecv}(\vecx)/(\sum_{\vecv \in R_{\vecc}(\vecx)}Q^{\vecc}_{\vecv}(\vecx)) \neq Q^{\vecc'}_{\vecv}(\vecx)/(\sum_{\vecv \in V_{\vecc'}}Q^{\vecc'}_{\vecv}(\vecx))$
for some point $\vecx$ that lies in the boundary of two adjacent cubelets with
down-left corners $\vecc$ and $\vecc'$ respectively. We specifically design the
coefficients $Q_v^{\vecc}(\vecx)$ such as the latter does not occur and this is
the main reason that the definition of the function $Q_{\vecv}^{\vecc}(\vecx)$
is slightly complicated. For this reason we prove the following lemma.

\begin{lemma} \label{l:well_defined}
    For any vertex $\vecv \in \p{\nm{N}}^d$, $\P_{\vecv}(\vecx)$ is a continuous
  and twice differentiable function and for any $\vecv \notin R_c(\vecx)$ it
  holds that
  $\P_{\vecv}(\vecx) = \nabla \P_{\vecv}(\vecx) = \nabla^2 \P_{\vecv}(\vecx) = 0$.
  Moreover, for every $\vecx \in [0, 1]^d$ the set $R_+(\vecx)$ of vertices
  $\vecv \in \p{\nm{N}}^d$ such that $\P_{\vecv}(\vecx) > 0$ satisfies
  $\abs{R_+(\vecx)} = d + 1$.
\end{lemma}

\noindent Based on Lemma \ref{l:well_defined} and the expression of $\P_{\vecv}$
we can prove that the $\P_{\vecv}$ coefficients defined in Definition
\ref{d:_coeff} satisfy the properties \ref{enu:properties:Coefficients:2} and
\ref{enu:properties:Coefficients:3} of the definition \ref{def:seic}. To prove
the properties \ref{enu:properties:Coefficients:1} and
\ref{enu:properties:Coefficients:4} we also need the following two lemmas.

\begin{restatable}{lemma}{bgradients} \label{l:bounding_gradients}
    For any vertex $\vecv \in \p{\nm{N}}^d$, it holds that
  \begin{enumerate}
    \item $\abs{\frac{\partial \mathrm{P}_{\vecv}(\vecx)}{\partial x_i}} \leq \Theta(d^{12}/\delta)$,
    \item $\abs{\frac{\partial^2 \mathrm{P}_{\vecv}(\vecx)}{\partial x_i~\partial x_j}} \leq \Theta(d^{24}/\delta^2)$.
  \end{enumerate}
\end{restatable}

\begin{restatable}{lemma}{pcorners} \label{l:positive_corners_boundary}
    Let a point $\vecx \in [0,1]^d$ and $R_+(\vecx)$ the set of vertices with
  $\P_{\vecv}(\vecx) > 0$, then we have that
  \begin{enumerate}
    \item If $0 \leq x_i < 1/(N - 1)$ then there always exists a vertex
      $\vecv \in R_+(\vecx)$ such that $v_i = 0$.
    \item If $1 - 1/(N - 1) < x_i \leq 1$ then there always exists a vertex
      $\vecv \in R_+(\vecx)$ such that $v_i = 1$.
  \end{enumerate}
\end{restatable}

\noindent The proofs of Lemmas \ref{l:well_defined},
\ref{l:bounding_gradients}, and \ref{l:positive_corners_boundary} can be found
in Appendix \ref{s:Q_coef}. Based on Lemmas \ref{l:well_defined},
\ref{l:bounding_gradients}, and \ref{l:positive_corners_boundary} we are now
ready to prove Theorem \ref{thm:seic}.

\begin{proof}[Proof of Theorem \ref{thm:seic}]
    The fact that the coefficients $\P_{\vecv}$ satisfy the property
  \ref{enu:properties:Coefficients:1} follows directly from Lemma
  \ref{l:bounding_gradients}. Property \ref{enu:properties:Coefficients:2}
  follows directly from the definition of $\P_{\vecv}$ in Definition
  \ref{d:_coeff} and the simple fact that $Q^{\vecc}_{\vecv}(\vecx) \ge 0$.
  Property \ref{enu:properties:Coefficients:3} follows from the second part of
  Lemma \ref{l:well_defined}. Finally Property
  \ref{enu:properties:Coefficients:4} follows directly from Lemma
  \ref{l:positive_corners_boundary}.
\end{proof}

\section{Unconditional Black-Box Lower Bounds} \label{sec:blackBox}

  In this section our goal is to prove Theorem
\ref{thm:localNashBlackBoxLowerBound} based on the Theorem
\ref{thm:localNashHardness} that we proved in Section \ref{sec:hardness} and
the known black box lower bounds that we know for $\PPAD$ by \cite{HPV88}. In
this section we assume that all the real number operation are performed with
infinite precision.

\begin{theorem}[\cite{HPV88}] \label{thm:HPVLowerBound}
    Assume that there exists an algorithm $A$ that has black-box oracle access to
  the value of a  function $M : [0, 1]^d \to [0, 1]^d$ and outputs
  $\vecw^{\star} \in [0, 1]^d$. There exists a universal constant $c > 0$ such
  that if $M$ is $2$-Lipschitz and
  $\norm{M(\vecw^{\star}) - \vecw^{\star}}_2 \le 1/(2 c)$, then $A$ has to make
  at least $2^d$ different oracle calls to the function value of $M$.
\end{theorem}

  It is easy to observe in the reduction in the proof of Theorem
\ref{t:Bi_Sperner_PPAD} is a black-box reduction and in every evaluation of the
constructed circuit $\calC_l$ only requires one evaluation of the input function
$M$. Therefore the proof of Theorem \ref{t:Bi_Sperner_PPAD} together with the
Theorem \ref{thm:HPVLowerBound} imply the following corollary.

\begin{figure}[t]
  \centering
  \includegraphics[scale=1.2]{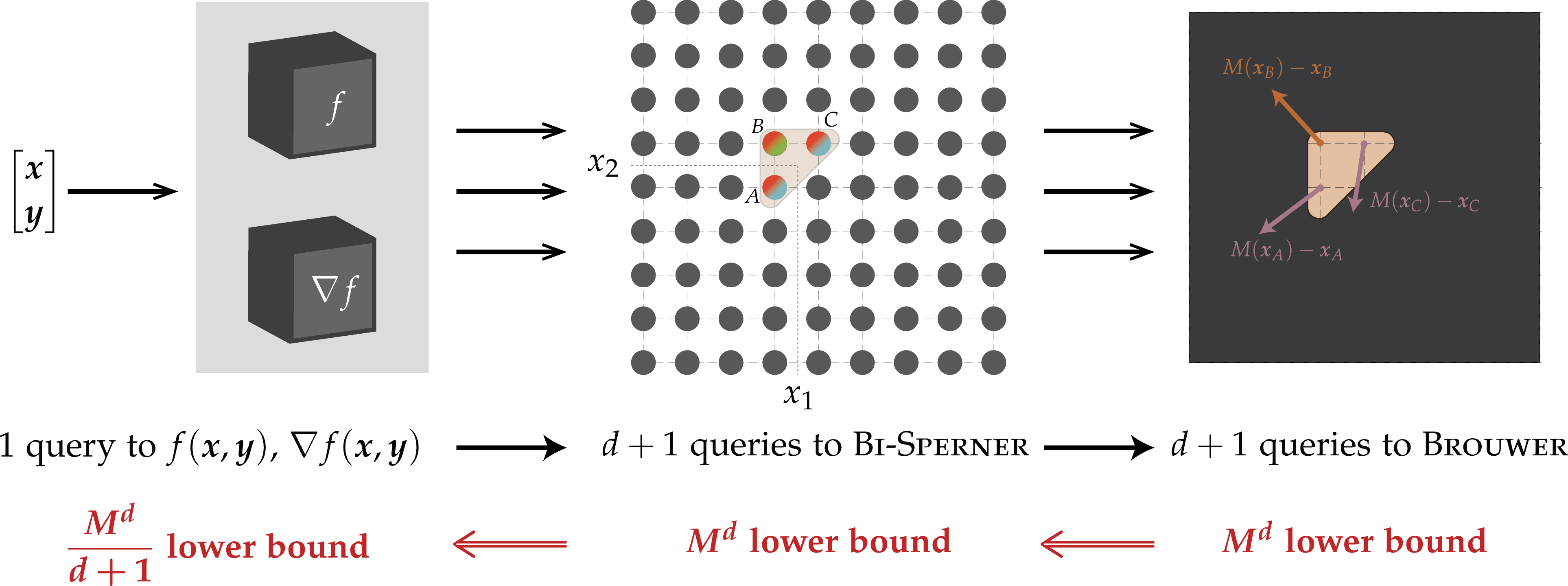}
  \caption{Pictorial representation on the way the black box lower bound
  follows from the white box $\PPAD$-completeness that presented in Section
  \ref{sec:hardness} and the known black box lower bounds for the
  $\textsc{Brouwer}$ problem by \cite{HPV88}. In the figure we can see the four
  dimensional case of Section \ref{sec:hardness:2D} that corresponds to the
  $\TD$-$\BiSperner$ and the 2-dimensional Brouwer. As we can see, in that case
  $1$ query to $\calO_f$ can be implemented with $3$ queries to
  $\TD$-$\BiSperner$ and each of these can be implemented with $1$ query to
  2-dimensional $\textsc{Brouwer}$. In the high dimensional setting of Section
  \ref{sec:hardness}, every query $(\vecx, \vecy)$ to the oracle $\calO_f$ to
  return the values $f(\vecx, \vecy)$ and $\nabla f(\vecx, \vecy)$ can be
  implemented via $d + 1$ oracles to an $\HD$-$\BiSperner$ instance. Each of
  these oracles to $\HD$-$\BiSperner$ can be implemented via $1$ oracle to a
  \textsc{Brouwer} instance. Therefore an $M^d$ query lower bound for
  \textsc{Brouwer} implies an $M^d$ query lower bound for $\HD$-$\BiSperner$
  which in turn implies an $M^d/(d + 1)$ query lower bound for our $\gdaFixed$
  and $\lrlNash$ problems.}
  \label{fig:blackBox}
\end{figure}

\begin{corollary}[Black-Box Lower Bound for Bi-Sperner] \label{cor:blackBoxBiSperner}
    Let $\calC_l : \p{\nm{N}}^d \to \set{-1, 1}^d$ be an instance of the
  $\HD$-$\BiSperner$ problem with $N = O(d)$. Then any algorithm that has
  black-box oracle access to $\calC_l$ and outputs a solution to the
  corresponding $\HD$-$\BiSperner$ problem, needs $2^d$ different oracle calls
  to the value of $\calC_l$.
\end{corollary}

  Based on Corollary \ref{cor:blackBoxBiSperner} and the reduction that we
presented in Section \ref{sec:hardness}, we are now ready to prove Theorem
\ref{thm:localNashBlackBoxLowerBound}.

\begin{proof}[Proof of Theorem \ref{thm:localNashBlackBoxLowerBound}]
    This proof follows the steps of Figure \ref{fig:blackBox}. The last part of
  that figure is established in Corollary \ref{cor:blackBoxBiSperner}. So what
  is left to prove Theorem \ref{thm:localNashBlackBoxLowerBound} is that for
  every instance of $\HD$-$\BiSperner$ we can construct a function $f$ such that
  the oracle $\calO_f$ can be implemented via $d + 1$ queries to the instance of
  $\HD$-$\BiSperner$ and also every solution of $\gdaFixed$ with oracle access
  $\calO_f$ to $f$ and $\nabla f$ reveals one solution of the starting
  $\HD$-$\BiSperner$ instance.

    To construct this oracle $\calO_f$ we follow exactly the reduction that we
  described in Section \ref{sec:hardness}. The correctness of the reduction that
  we provide in Section \ref{sec:hardness} suffices to prove that every solution
  of the $\gdaFixed$ with oracle access $\calO_f$ to $f$ and $\nabla f$ gives a
  solution to the initial $\HD$-$\BiSperner$ instance. So the only thing that
  remains is to bound the number of queries to the $\HD$-$\BiSperner$ instance
  that we need in order to implement the oracle $\calO_f$. To do this consider
  the following definition of $f$ based on an instance $\calC_l$ of
  $\HD$-$\BiSperner$ from Definition \ref{d:d-payoff} with a scaling factor to make sure that the range of the function is $[-1, 1]$
  \[ f_{\calC_l}(\vecx, \vecy) = \frac{1}{d} \cdot \sum_{j = 1}^d (x_j - y_j)\cdot \alpha_j(\vecx) \]
  \noindent where
  $\alpha_j(\vecx) = - \sum_{\vecv \in \p{\nm{N}}^d} \P_{\vecv}(\vecx) \cdot \calC^j_l(\vecv)$,
  and $\P_{\vecv}$ are the coefficients defined in Definition \ref{def:seic}.
  From the property \ref{enu:properties:Coefficients:3} of the coefficients
  $\P_{\vecv}$ we have that to evaluate $a_j(\vecx)$ we only need the values
  $\calC_l^j(\vecv)$ for $d + 1$ coefficients $\vecv$ and the same coefficients
  are needed to evaluate $a_j(\vecx)$ for every $j$. This implies that for every
  $(\vecx, \vecy)$ we need $d + 1$ oracle calls to the instance $\calC_l$ of
  $\HD$-$\BiSperner$ so that $\calO_f$ returns the value of
  $f_{\calC_l}(\vecx, \vecy)$. If we take the gradient of $f_{\calC_l}$ with
  respect to $(\vecx, \vecy)$ then an identical argument implies that the same
  set of $d + 1$ queries to $\HD$-$\BiSperner$ are needed so that $\calO_f$
  returns the value of $\nabla f_{\calC_l}(\vecx, \vecy)$ too. Therefore every
  query to the oracle $\calO_f$ can be implemented via $d + 1$ queries to
  $\calC_l$. Now we can use Corollary \ref{cor:blackBoxBiSperner} to get that
  the number of queries that we need in order to solve the $\gdaFixed$ with
  oracle access $\calO_f$ to $f$ and $\nabla f$ is at least $2^d/(d + 1)$.
  Finally observe that the proof of the Theorem \ref{t:LocalMaxMin-GDA} applies
  in th the black box model too. Hence finding solution of $\gdaFixed$ in when
  we have black box access $\calO_f$ to $f$ and $\nabla f$ is equivalent to
  finding solutions of $\lrlNash$ when we have exactly the same black box access
  $\calO_f$ to $f$ and $\nabla f$. Therefore to find solutions of $\lrlNash$
  with black box access $\calO_f$ to $f$ and $\nabla f$ we need at least
  $2^d/(d + 1)$ queries to $\calO_f$ and the theorem follows by observing that
  in our proof the only parameters that depend on $d$ are $L$, $G$, $\eps$, and
  possibly $\delta$ but $1/\delta = O(\sqrt{L/\eps})$ and hence the dependence
  of $\delta$ can be replaced by dependence on $L$ and $\eps$.
\end{proof}

\section{Hardness in the Global Regime} \label{sec:localNash}

  In this section our goal is to prove that the complexity of the problems
$\lNash$ and $\lmin$ is significantly increased when $\eps$, $\delta$ lie
outside the local regime, in the global regime. We start with the following
theorem where we show that $\FNP$-hardness of $\lNash$.

\begin{theorem} \label{t:local_nash}
    $\lNash$ is $\FNP$-hard even when $\eps$ is
  set to any value $\le 1/384$, $\delta$ is set to any value $\ge 1$, and even
  when $\calP(\matA, \vecb) = [0, 1]^d$, $G = \sqrt{d}$, $L = d$, and $B = d$.
\end{theorem}

\begin{proof}
    We now present a reduction from \textsc{3-SAT(3)} to $\lNash$ that
  proves Theorem~\ref{t:local_nash}. First we remind the definition of the
  problem \textsc{3-SAT(3)}.

  \smallskip
  \begin{nproblem}[\textsc{3-SAT(3)}]
    \textsc{Input:} A boolean CNF-formula $\phi$ with boolean variables
    $x_1, \dots, x_n$ such that every clause of $\phi$ has at most $3$ boolean
    variables and every boolean variable appears to at most $3$ clauses.
    \smallskip

    \noindent \textsc{Output:} An assignment $\vecx \in \{0, 1\}^n$ that
    satisfies $\phi$, or $\bot$ if no such assignment exists.
  \end{nproblem}
  \medskip

    Given an instance of \textsc{3-SAT(3)} we first construct a polynomial
  $P_j(\vecx)$ for each clause $\phi_j$ as follows: for each boolean variable
  $x_i$ (there are $n$ boolean variables $x_i$) we correspond a respective real-valued variable $x_i$. Then for each
  clause $\phi_j$ (there are $m$ such clauses), let $\ell_i,\ell_k,\ell_m$ denote the literals participating
  in $\phi_j$,
  $P_j(\vecx) = P_{ji}(\vecx)\cdot P_{jk}(\vecx) \cdot P_{jm}(\vecx)$ where
  \[ P_{ji}(\vecx) = \left\{
     \begin{array}{ll}
       1-x_i & \text{if } \ell_i = x_i \\
       x_i   & \text{if } \ell_i = \overline{x_i}\\
     \end{array}
     \right. \]
  \noindent Then the overall constructed function is
  \[ f(\vecx, \vecw, \vecz) = \sum_{j=1}^m P_j(\vecx) \cdot (w_j - z_j)^2 \]
  \noindent where each $w_j, z_j$ are additional variables associated with
  clause $\phi_j$. The player that wants to minimize $f$ controls $\vecx, \vecw$
  vectors while the maximizing player controls the $\vecz$ variables.

  \begin{lemma} \label{lem:proof:t:local_nash}
      The formula $\phi$ admits a satisfying assignment if and only if there
    exist an $(\eps, \delta)$-local min-max equilibrium of $f(\vecx,\vecw)$ with
    $\eps \leq 1/384$, $\delta = 1$ and $(\vecx,\vecw) \in [0,1]^{n + 2m}$.
  \end{lemma}
  \begin{proof}
    Let us assume that there
    exists a satisfying assignment. Given such a satisfying assignment we will construct $\left((\vecx^\star,\vecw^\star), \vecz^\star\right)$ that is a $(0, 1)$-local min-max equilibrium of $f$.
    We set each variable
    $x_i^\star \triangleq 1$ if and only if the respective boolean variable is
    true. Observe that this implies that $P_j(\vecx^\star) = 0$ for all $j$,
    meaning that the strategy profile
    $\left((\vecx^\star,\vecw^\star),\vecz^\star\right)$ is a global Nash
    equilibrium no matter the values of $\vecw^\star, \vecz^\star$.

    On the opposite direction, let us assume that there exists an
    $(\eps, \delta)$-local min-max equilibrium of $f$ with $\eps = 1/384$ and $\delta = 1$. In
    this case we first prove that for each $j = 1, \ldots, m$
    \[ P_j(\vecx^\star) \leq 16 \cdot \eps. \]
    \noindent Fix any clause $j$. In case $\abs{w_j^\star - z_j^\star} \geq 1/4$
    then the minimizing player can further decrease $f$ by at least $P_j(x)/16$
    by setting $w_j^\star \triangleq z_j^\star$. On the other hand in case
    $\abs{w_j^\star - z_j^\star} \leq 1/4$ then the maximizing player can
    increase $f$ by at least $P_j(x^\star)/16$ by moving $z_j^\star$ either to
    $0$ or to $1$. We remark that both of the options are feasible since
    $\delta = 1$.

      Now consider the probability distribution over the boolean assignments
    where each boolean variable $x_i$ is independently selected to be true with
    probability $x_i^\star$. Then,
    \[ \Prob\left( \text{clause } \phi_j \text{ is not satisfied} \right) = P_j(\vecx^\star) \leq 16 \cdot \eps = 1/24 \]
    \noindent Since each $\phi_j$ shares variables with at most $6$ other
    clauses, the event of $\phi_j$ not being satisfied is dependent with at most
    $6$ other events. By the Lov\'asz Local Lemma \cite{ErdosL1973}, we get that
    the probability none of these events occur is positive. As a result, there
    exists a satisfying assignment.
  \end{proof}

    Hence the formula $\phi$ is satisfiable if and only if $f$ has a
  $(1/384, 1)$-local min-max equilibrium point. What is left to prove the
  $\FNP$-hardness is to show how we can find a satisfying assignment of $\phi$
  given an approximate stationary point of $f$. This can be done using the
  celebrated results that provide constructive proofs of the Lov\'asz Local
  Lemma \cite{Moser'09, MoserT'10}. Finally to conclude the proof observe that
  since the $f$ that we construct is a polynomial of degree $6$ which can
  efficiently be described as a sum of monomials, we can trivially construct a
  Turing machine that computes the values of both $f$ and $\nabla f$ in the
  polynomial time in the requested number of bits accuracy. The constructed function $f$ is $\sqrt{d}$-Lipschitz and $d$-smooth,
  where $d$ is the number of variables that is equal to $n + 2 m$. More precisely since each variable $x_i$ participates in at most $3$ clauses, the real-valued variable $x_i$ appears in at most $3$ monomials $P_j$. Thus
  $-3 \leq \frac{\partial  f(\vecx,\vecw,\vecx)}{\partial x_i} \leq 3$. Similarly it is not hard to see that
  $-2 \leq \frac{\partial f(\vecx,\vecw,\vecx)}{\partial w_j} ,
  \frac{\partial f(\vecx,\vecw,\vecx)}{\partial z_j}
  \leq 2$. All the latter imply that $\norm{\nabla f(\vecx,\vecw,\vecz)}_2 \leq \Theta(\sqrt{n+m})$, meaning that $f(\vecx,\vecw,\vecz)$ is $\Theta(n+m)$-Lipschitz. Using again the fact that each $x_i$ participates in at most $3$ monomials $P_j(\vecx)$, we get that all the terms
  $\frac{\partial ^2 f(\vecx,\vecw,\vecz)}{\partial^ 2 x_i},\frac{\partial ^2 f(\vecx,\vecw,\vecz)}{\partial^ 2 w_j},\frac{\partial ^2 f(\vecx,\vecw,\vecz)}{\partial^ 2 z_j},
  \frac{\partial ^2 f(\vecx,\vecw,\vecz)}{\partial x_i~\partial w_j},
    \frac{\partial ^2 f(\vecx,\vecw,\vecz)}{\partial x_i~\partial z_j},
      \frac{\partial ^2 f(\vecx,\vecw,\vecz)}{\partial w_j~\partial z_j} \in [-6,6]$. Thus the absolute value of each entry of $\nabla^2 f(\vecx,\vecw,\vecz)$ is bounded by $6$ and thus $\norm{\nabla^2 f(\vecx,\vecw,\vecz)}_2 \leq \Theta(n + m)$, which implies the $\Theta(n+m)$-smoothness.
  Therefore our
  reduction produces a valid instance of $\lNash$ and hence the theorem follows.
\end{proof}

  Next we show the $\FNP$-hardness of $\lmin$. As we can see there is a gap
between Theorem \ref{t:local_nash} and Theorem \ref{t:local_min}. In particular,
the $\FNP$-hardness result of $\lNash$ is stronger since it holds for any
$\delta \ge 1$ whereas for the $\FNP$-hardness of $\lmin$ our proof needs
$\delta \ge \sqrt{d}$ when the rest of the parameters remain the same.

\begin{theorem} \label{t:local_min}
    $\lmin$ is $\FNP$-hard even when $\eps$ is set to any value $\le 1/24$,
  $\delta$ is set to any value $\ge \sqrt{d}$, and even
  when $\calP(\matA, \vecb) = [0, 1]^d$, $G = \sqrt{d}$, $L = d$, and $B = d$.
\end{theorem}

\begin{proof}
    We follow the same proof as in the proof of Theorem \ref{t:local_nash} but
  we instead set $f(\vecx) = \sum_{j = 1}^m P_j(\vecx)$ where $\vecx \in [0,1]^n$ (the number of variables is $d:= n$). We then get that if the initial formula is satisfiable then there exist
  $\vecx \in \calP(\matA, \vecb)$, such that $f(\vecx) = 0$. On
  the other hand if there exist $\vecx \in \calP(\matA, \vecb)$ such
  that $f(\vecx) \le 1/24$ then the formula is satisfiable due to the
  Lov\'asz Local Lemma \cite{ErdosL1973}. Therefore the $\FNP$-hardness follows
  again from the constructive proof of the Lov\'asz Local Lemma
  \cite{Moser'09, MoserT'10}. Setting $\delta \geq \sqrt{n}$ which equals the
  diameter of the feasibility set implies that in case there exists
  $\hat{\vecx}$ with $f(\hat{\vecx}) =0$ then all
  $(\eps,\delta)$-\textsc{LocalMin} $\vecx^\ast$ must admit value
  $f(\vecx^\ast) \le 1/24$ and thus a satisfying assignment is implied.
\end{proof}

  Next we prove a black box lower bound for minimization in the global regime.
The proof of following lower bound illustrates the strength of the SEIC
coefficients presented in Section \ref{sec:seic}. The next Theorem can also be
used to prove the $\FNP$-hardness of $\lmin$ in the global regime but with worse
Lipschitzness and smoothness parameters than the once at Theorem
\ref{t:local_min} and for this reason we present both of them.

\begin{theorem} \label{thm:globalMinOracle}
    In the worst case, $\Omega\left(2^d/d\right)$ value/gradient black-box
  queries are needed to determine a $(\eps,\delta)$-\textsc{LocalMin} for
  functions $f(\vecx):[0,1]^d \to [0, 1]$ with $G = \Theta(d^{15})$,
  $L = \Theta(d^{22})$, $\eps < 1$, $\delta = \sqrt{d}$.
\end{theorem}

\begin{proof}

  The proof is based on the fact that given just \textit{black-box access} to a
boolean formula $\phi:\{0,1\}^d \mapsto \{0,1\}$, at least $\Omega(2^d)$ queries
are needed in order to determine whether $\phi$ admits a satisfying assignment.
The term \textit{black-box access} refers to the fact that the clauses of the
formula are not given and the only way to determine whether a specific boolean
assignment is satisfying is by quering the specific binary string.

  Given such a black-box oracle for a satisfying assignment $d$, we construct
the function $f_{\phi}(\vecx):[0,1]^d \mapsto [0,1]$ as follows:

\begin{enumerate}
    \item for each corner $\vecv \in V$ of the $[0,1]^d$ hypercube, i.e. $\vecv \in \{0,1\}^d$, we set $f_{\phi}(\vecv) := 1 - \phi(\vecv)$.

    \item for the rest of the points $\vecx \in [0,1]^d / V$,
    $f_{\phi}(\vecx) := \sum_{\vecv \in V} P_{\vecv}(\vecx) \cdot f_{\phi}(\vecv)$
    where $P_{\vecv}$ are the coefficients of Definition~\ref{d:_coeff}.
\end{enumerate}
  We remind that by Lemma~\ref{l:bounding_gradients}, we get that
$\norm{\nabla f_{\phi}(\vecx)}_2 \leq \Theta(d^{12})$ and
$\norm{\nabla^2 f_{\phi}(\vecx)}_2 \leq \Theta(d^{25})$, meaning that
$f_{\phi}(\cdot)$ is $\Theta(d^{12})$-Lipschitz and $\Theta(d^{25})$-smooth.
Moreover by Lemma~\ref{d:cof} , for any $\vecx \in [0,1]^n$ the set
$V(x) = \{\vecv \in V: P_{\vecv}(\vecx) \neq 0\}$ has cardinality at most
$d + 1$, while at the same time $\sum_{\vecv \in V}P_{\vecv}(\vecx) = 1$.

  In case $\phi$ is not satisfiable then $f_{\phi}(\vecx) = 1$ for all
$\vecx \in [0,1]^d$ since $f_{\phi}(\vecv) = 1$ for all $\vecv \in V$. In case
there exists a satisfying assignment $\vecv^\ast$ then
$f_{\phi}(\vecv^\ast) = 0$. Since $\delta \geq \sqrt{d}$ that is the diameter of
$[0, 1]^d$, any $(\eps,\delta)$-\textsc{LocalMin} $\vecx^\ast$ must have
$f_{\phi}(\vecx) \leq \eps < 1$. Since
$f_{\phi}(\vecx^\ast) \triangleq \sum_{\vecv \in V(\vecx^\ast)}P_{\vecv}(\vecx^\ast) \cdot f_{\phi}(\vecv^\ast) < 1$,
there exists at least one vertex $\hat{\vecv} \in V(\vecx)$ with
$f_{\phi}(\hat{\vecv}) = 0$, meaning that $\phi(\vecv^\ast) = 1$. As a result,
given an $(\eps,\delta)$-$\textsc{LocalMin}$ $\vecx^\ast$ with
$f_{\phi}(\vecx^\ast) < 1$, we can find a satisfying $\hat{\vecv}$ by querying
$\phi(\vecv)$ for each vertex $\vecv \in V(\vecx^\ast)$. Since
$\abs{V(\vecx^\ast)} \le d + 1$, this will take at most $d+1$ additional
queries.

  Up next, we argue that in case an $(\eps,\delta)$-$\textsc{LocalMin}$ could be
determined with less than $O(2^d/d)$ value/gradient queries, then determining
whether $\phi$ admits a satisfying assignment could be done with less that
$O(2^d)$ queries on $\phi$ (the latter is obviously impossible). Notice that any
value/gradient query both $f_{\phi}(\vecx)$ and $\nabla f_{\phi}(\vecx)$ can be
computed by querying the value $f_{\phi}(\vecv)$ of the vertices
$\vecv \in V(\vecx)$. Since $\abs{V(\vecx)} \le d + 1$, any value/gradient query
of $f_{\phi}$ can be simulated by $d+1$ queries on $\phi$.
\end{proof}

\section*{Acknowledgements}
This work was supported by NSF Awards IIS-1741137, CCF-1617730 and CCF-1901292,
by a Simons Investigator Award, by the DOE PhILMs project (No.
DE-AC05-76RL01830), and by the DARPA award HR00111990021. M.Z. was also
supported by Google Ph.D. Fellowship. S.S. was supported by NRF 2018 Fellowship
NRF-NRFF2018-07.

\bibliographystyle{alpha}
\bibliography{refs}

\appendix

\newpage

\section{Proof of Theorem \ref{thm:approximateStationaryPointsHardness}} \label{sec:proof:approximateStationaryPointsHardness}

We first remind the definition of the \textsc{3-SAT(3)} problem that we will use for our
reduction.

\begin{nproblem}[\textsc{3-SAT(3)}]
  \textsc{Input:} A boolean CNF-formula $\phi$ with boolean variables
  $x_1, \dots, x_n$ such that every clause of $\phi$ has at most $3$ boolean
  variables and every boolean variable appears to at most $3$ clauses.
  \smallskip

  \noindent \textsc{Output:} An assignment $\vecx \in \{0, 1\}^n$ that satisfies
  $\phi$, or $\bot$ if no such assignment exists.
\end{nproblem}

  It is well known that \textsc{3-SAT(3)} is $\FNP$-complete, for details see $\S 9.2$
of \cite{Papadimitriou1994CC}. To prove Theorem
\ref{thm:approximateStationaryPointsHardness}, we reduce \textsc{3-SAT(3)} to
$\eps$-$\stat$.
\medskip

  Given an instance of \textsc{3-SAT(3)} we construct the function
$f : [0, 1]^{n + m} \to [0, 1]$, where $m$ is the number of clauses of $\phi$. For each
literal $x_i$ we assign a real-valued variable which by abuse of notation we also denote
$x_i$ and it would be clear from the context whether we refer to the literal or the
real-valued variable. Then for each clause $\phi_j$ of $\phi$, we construct a polynomial
$P_j(x)$ as follows: if $\ell_i, \ell_k, \ell_m$ are the literals participating in
$\phi_j$, then $P_j(x) = P_{ji}(x)\cdot P_{jk}(x) \cdot P_{jm}(x)$ where
\[ P_{ji}(x) = \left\{
     \begin{array}{ll}
       1-x_i & \text{if } \ell_i = x_i \\
       x_i&  \text{if } \ell_i = \overline{x_i}\\
   \end{array}
  \right. \]
\noindent The overall constructed function is
$f(\vecx, \vecw) = \sum_{j=1}^m w_j \cdot P_j(\vecx)$, where each $w_j$ is an additional
variable associated with clause $\phi_j$.
Notice that $0 \leq \frac{\partial f(\vecx,\vecw)}{\partial w_j} \leq 1$ and
$ -3 \leq \frac{\partial f(\vecx,\vecw)}{\partial x_i} \leq 3$ since the boolean variable $x_i$ participates in at most $3$ clauses. As a result, $\norm{\nabla f(\vecx,\vecw)}_{2} \leq \Theta(\sqrt{n+m})$, meaning that $f(\vecx,\vecw)$ is $G$-Lipschitz with $G = \Theta(\sqrt{n+m})$. Also notice that all the entries of $\nabla^2 f(\vecx,\vecw)$, i.e.
$\frac{\partial^2 f(\vecx,\vecw)}{\partial^2 x_i} = \frac{\partial^2 f(\vecx,\vecw)}{\partial^2 w_j}, \frac{\partial^2 f(\vecx,\vecw)}{\partial x_i ~ \partial w_j} , \frac{\partial^2 f(\vecx,\vecw)}{\partial x_i ~ \partial x_m} , \frac{\partial^2 f(\vecx,\vecw)}{\partial w_k ~ \partial w_j} \in [-3,3]$. As a result, $\norm{\nabla^2 f(\vecx,\vecw)}_{2} \leq \Theta(n+m)$, meaning that $f(\vecx,\vecw)$ is $L$-smooth with $L = \Theta(n+m)$.

\begin{lemma} \label{lem:proof:thm:approximateStationaryPointsHardness}
There exists a satisfying assignment for the clauses $\phi_1,\ldots,\phi_m$
if and only if there solution of the constructed \textsc{StationaryPoint}
with $\eps = 1/24$ a admits solution  $(\vecx^\star,\vecw^\star) \in [0,1]^{n+m}$ such that $\norm{\nabla f(\vecx^\star,\vecw^\star)}_2 < 1/24$.
\end{lemma}

\begin{proof}
By the definition of \textsc{StationaryPoint}, in case there exists a pair of points $(\hat{\vecx},\hat{\vecw}) \in [0,1]^{n+m}$ with $\norm{\nabla f(\hat{\vecx},\hat{\vecw})}_2 < \eps / 2 = 1 /48$, then a pair of points $(\vecx^\star,\vecw^\star)$ with $\norm{\nabla f(\vecx^\star,\vecw^\star)}_2 < \eps = 1 /24$ must be returned. In case $\norm{\nabla f(\vecx,\vecw)}_2 > \eps = 1 /24$ for all $(\vecx,\vecw) \in [0,1]^{n+m}$, the null symbol $\bot$ is returned.

Let us assume that there exists a satisfying assignment of $\phi$. Consider the solution $(\hat{\vecx},\hat{\vecw})$ constructed as follows: each
variable $\hat{x}_i$ is set to $1$ iff the respective boolean variable is true and $\hat{w}_j = 0$ for all $j=1, \ldots, m$. Since the assignment satisfies the CNF-formula $\phi$,
there exists at least one true literal in each clause $\phi_j$ which means that
$P_j(x) = 0$ for all $j = 1, \ldots, m$. As a result $\frac{\partial f(\hat{\vecx},\hat{\vecw})}{\partial w_j} = P_j(\hat{\vecx}) = 0$ for all $j=1,\ldots, m$. At the same time, $\frac{\partial f(\hat{\vecx},\hat{\vecw})}{\partial x_i} = 0$ since $\hat{w}_j =0$ for all $j=1,\ldots,m$.
Overall we have that $\nabla f(\hat{\vecx},\hat{\vecw}) = 0 < 1/48 = \eps / 2 $. As a result, the constructed \textsc{StationaryPoint} instance must return a solution $(\vecx^\star,\vecw^\star)$ with $\norm{\nabla f(\vecx^\star,\vecw^\star)}_{2} < \frac{1}{24} = \eps$.

On the opposite direction, the existence of a pair of points $(\vecx^\star,\vecw^\star)$ with
$\norm{\nabla f(\vecx^\star,\vecw^\star) }_2 < 1/24$ implies $P_j(\vecx^\ast) < 1/24$ for all $j = 1\ldots m$. Consider the probability distribution over the boolean assignments
in which \textit{each boolean variable $x_i$ is independently selected to be true with
probability $x_i^\star$.} Then,
\[\Prob \left( \text{clause } \phi_j
\text{ is not satisfied}\right) = P_j(\vecx^\star) < 1/24\]
\noindent Since $\phi_j$ shares variables with at most $6$ other clauses, the bad event of
$\phi_j$ not being satisfied is dependent with at most $6$ other bad events. By Lov\'asz
Local Lemma \cite{ErdosL1973}, we get that the probability none of the events occurs is positive. As a
result, there exists a satisfying assignment.
\end{proof}

\noindent Using Lemma \ref{lem:proof:thm:approximateStationaryPointsHardness} we
can conclude that $\phi$ is satisfiable if and only if $f$ has a
$1/24$-approximate stationary point. What is left to prove the $\FNP$-hardness
is to show how we can find a satisfying assignment of $\phi$ given an
approximate stationary point of $f$. This can be done using the celebrated
results that provide constructive proofs of the Lov\'asz Local Lemma
\cite{Moser'09, MoserT'10}. Finally, we remind that the
constructed function $f$ is $\Theta\left(\sqrt{d} \right)$-Lipschitz and $\Theta\left(d \right)$-smooth, where $d$ is the number
of variables that is equal to $n + m$.

\section{Missing Proofs from Section \ref{sec:existence}} \label{sec:app:existence}

  In this section we give  proofs for the statements presented in Section~\ref{sec:existence}. These statements establish the totality and  inclusion to $\PPAD$ of  $\lrlNash$ and $\gdaFixed$.

\subsection{Proof of Theorem \ref{t:LocalMaxMin-GDA}} \label{sec:proof:LocalMaxMin-GDA}

  We start with establishing claim ``1.'' in the statement of the theorem. It
will be clear that our proof will provide a polynomial-time reduction from
\lrlNash\ to \gdaFixed. Suppose that $(\vecx^\star,\vecy^\star)$ is an
$\alpha$-approximate fixed point of $F_{GDA}$, where $\alpha$ is the specified
in the theorem statement function of $\delta$, $G$ and $L$. To simplify our
proof, we abuse notation and define $f(\vecx) \triangleq f(\vecx, \vecy^\star)$,
$\nabla f(\vecx) \triangleq \nabla_x f(\vecx,\vecy^\star)$,
$K \triangleq \{\vecx \mid (\vecx, \vecy^{\star}) \in \calP(\matA,\vecb)\}$
and $\hat{\vecx} \triangleq \Pi_{K}(\vecx^\star - \nabla f(\vecx^\star))$.
Because $(\vecx^{\star}, \vecy^{\star})$ is an $\alpha$-approximate fixed point
of $F_{FDA}$, it follows that $\norm{\hat{\vecx} - \vecx^\star}_2 < \alpha$.

\begin{claim} \label{clm:proofOfExistence:1}
  $\langle \nabla f(\vecx^\star), \vecx^\star - \vecx \rangle < (G + \delta + \alpha) \cdot  \alpha, \text{ for all } \vecx \in K \cap B_{d_1}(\delta; \vecx^{\star})$.
\end{claim}

\begin{proof}
    Using the fact that
  $\hat{\vecx} = \Pi_{K}(\vecx^\star - \nabla f(\vecx^\star))$ and that $K$ is a
  convex set we can apply Theorem 1.5.5 (b) of \cite{facchinei2007finite} to get
  that
  \begin{equation} \label{eq:proof:clm:proofOfExistence:1:1}
    \langle \vecx^\star - \nabla f(\vecx^\star)- \hat{\vecx} ,\vecx - \hat{\vecx}\rangle\leq 0, \forall \vecx\in K.
  \end{equation}
  Next, we do some simple algebra to get that, for all
  $\vecx\in K \cap B_{d_1}(\delta; \vecx^{\star})$,
  \begin{align*}
    \langle \nabla f(\vecx^\star), \vecx^\star - \vecx \rangle & = \langle \vecx^\star - \nabla f(\vecx^\star) - \hat{\vecx}, \vecx - \hat{\vecx}\rangle + \langle \vecx - \hat{\vecx} - \nabla f(\vecx^\star), \hat{\vecx} - \vecx^\star \rangle \\
    & \stackrel{\eqref{eq:proof:clm:proofOfExistence:1:1}}{\le} \langle \vecx - \hat{\vecx} - \nabla f(\vecx^\star), \hat{\vecx} - \vecx^\star \rangle \\
    & \le \p{\norm{\vecx - \hat{\vecx}}_2 + \norm{\nabla f(\vecx^{\star})}_2} \norm{\hat{\vecx} - \vecx^{\star}}_2 < (G + \delta+\alpha) \cdot \alpha,
  \end{align*}
  \noindent where the second to last inequality follows from Cauchy–Schwarz
  inequality and the triangle inequality, and the last inequality follows from
  the triangle inequality and the following facts: (1)
  $\norm{\vecx^\star - \hat{\vecx}}_2 < \alpha$, (2)
  $\vecx \in B_{d_1}(\delta; \vecx^{\star})$, and (3)
  $\norm{\nabla f(\vecx, \vecy)}_2 \le G$ for all
  $(\vecx, \vecy) \in \calP(\matA, \vecb)$.
\end{proof}

\noindent For all $\vecx \in K \cap B_{d_1}(\delta; \vecx^{\star})$, from the
$L$-smoothness of $f$ we have that
\begin{align}
  \abs{f(\vecx)- (f(\vecx^\star) + \langle \nabla f(\vecx^*), \vecx - \vecx^\star \rangle)} \le   \frac{L}{2} \norm{\vecx - \vecx^\star}_2^2. \label{eq:costas}
\end{align}

\noindent We distinguish two cases:
\begin{enumerate}
  \item $f(\vecx^\star) \le f(\vecx)$: In this case we stop, remembering that
    \begin{align}
      f(\vecx^\star) \le f(\vecx).  \label{eq:costas0}
    \end{align}
  \item $f(\vecx^\star) > f(\vecx)$: In this case, we consider two further
    sub-cases:
    \begin{enumerate}
      \item $\langle \nabla f(\vecx^*), \vecx - \vecx^\star \rangle \ge 0$: in
        this sub-case, Eq~\eqref{eq:costas} gives
        \begin{align*}
            f(\vecx^\star)- f(\vecx) + \langle \nabla f(\vecx^*), \vecx - \vecx^\star \rangle
            \le   \frac{L}{2} \norm{\vecx - \vecx^\star}_2^2
        \end{align*}
        Thus
        \begin{align}
            f(\vecx^\star) \le f(\vecx) + \frac{L}{2} \norm{\vecx - \vecx^\star}_2^2 \le f(\vecx) + \frac{L}{2} \delta^2 < f(\vecx)+\eps,  \label{eq:costas1}
        \end{align}
        where for the last inequality we used that $\vecx \in B_{d_1}(\delta; \vecx^{\star})$, and that $\delta < \sqrt{2\eps/L}$.
      \item $\langle \nabla f(\vecx^*), \vecx - \vecx^\star \rangle < 0$: in
        this sub-case, Eq~\eqref{eq:costas} gives
        \begin{align*}
            f(\vecx^\star)- f(\vecx) - \langle \nabla f(\vecx^*), \vecx^\star - \vecx \rangle
            \le   \frac{L}{2} \norm{\vecx - \vecx^\star}_2^2.
        \end{align*}
        Thus
        \begin{align}
            f(\vecx^\star) &\le f(\vecx) + \langle \nabla f(\vecx^*), \vecx^\star - \vecx \rangle + \frac{L}{2} \norm{\vecx - \vecx^\star}_2^2 \notag\\
            & \le f(\vecx) + \langle \nabla f(\vecx^*), \vecx^\star - \vecx \rangle + \frac{L}{2} \cdot \delta^2 \notag\\
             & < f(\vecx) + (G + \delta + \alpha) \cdot \alpha + \frac{L}{2} \cdot \delta^2 \notag\\
             & \le  f(\vecx) + \eps,  \label{eq:costas2}
        \end{align}
        where the second inequality follows from the fact that
        $\vecx \in B_{d_1}(\delta; \vecx^{\star})$, the third inequality follows
        from Claim \ref{clm:proofOfExistence:1}, and the last inequality follows
        from the constraints $\delta < \sqrt{2\eps/L}$ and
        $\alpha \le \frac{\sqrt{(G+\delta)^2+4(\eps - \frac{L}{2}\delta^2)} - (G+\delta)}{2}$.
    \end{enumerate}
\end{enumerate}
In all cases, we get from~\eqref{eq:costas0},~\eqref{eq:costas1}
and~\eqref{eq:costas2} that $f(\vecx^\star) < f(\vecx) +\eps$, for all
$x \in K \cap B_{d_1}(\delta; \vecx^{\star})$. Thus, lifting our abuse of
notation, we get that
$f(\vecx^{\star}, \vecy^{\star}) < f(\vecx, \vecy^\star) + \eps$, for all
$\vecx \in \{\vecx \mid \vecx \in B_{d_1}(\delta; \vecx^{\star}) \text{ and } (\vecx, \vecy^\star) \in \calP(\matA,\vecb)\}$.
Using an identical argument we can also show that
$f(\vecx^{\star}, \vecy^{\star}) > f(\vecx^{\star}, \vecy) - \eps$ for all
$\vecy \in \{\vecy \mid \vecy \in B_{d_2}(\delta; \vecy^{\star}) \text{ and } (\vecx^{\star}, \vecy) \in \calP(\matA,\vecb)\}$.
The first part of the theorem follows.
\bigskip

\noindent Now let us establish claim ``2.'' in the theorem statement. It will be
clear that our proof will provide a polynomial-time reduction from \gdaFixed\ to
\lrlNash. For the choice of parameters $\eps$ and $\delta$ described in the
theorem statement, we will show that, if $(\vecx^\star,\vecy^\star)$ is an
$(\eps,\delta)$-local min-max equilibrium of $f$, then
$\norm{F_{GDAx} (\vecx^\star, \vecy^\star) - \vecx^{\star}}_2 < \alpha/2$ and
$\norm{F_{GDAy} (\vecx^\star, \vecy^\star) - \vecy^{\star}}_2 < \alpha/2$. The
second part of the theorem will then follow. We only prove that
$\norm{F_{GDAx} (\vecx^\star, \vecy^\star) - \vecx^{\star}}_2 < \alpha/2$, as
the argument for $\vecy^\star$ is identical. In the argument below we abuse
notation in the same way we described earlier. With that notation we will show
that $\norm{\hat{\vecx} - \vecx^{\star}}_2 < \alpha/2$.
\smallskip

\noindent \textbf{Proof that
$\boldsymbol{\norm{\hat{\vecx} - \vecx^{\star}} < \alpha/2}$.}
From our choice of $\eps$ and $\delta$, it is easy to see that
$\delta = \alpha / (5 L + 2) < \alpha / 2$. Thus, if
$\norm{\hat{\vecx} - \vecx^{\star}} < \delta$, then we automatically get
$\norm{\hat{\vecx} - \vecx^{\star}} < \alpha/2$. So it remains to handle the
case $\norm{\hat{\vecx} - \vecx^{\star}} \ge \delta$. We choose
$\vecx_c \triangleq \vecx^\star + \delta \frac{\hat{\vecx} - \vecx^\star}{\norm{\hat{\vecx} - \vecx^\star}_2}$.
It is easy to see that $\vecx_c \in B_{d_1}(\delta; \vecx^{\star})$ and hence we
get that
\begin{align*}
  f(\vecx^\star) - \eps < f(\vecx_c) &\leq f(\vecx^\star) + \left \langle \nabla f(\vecx^\star),\vecx_c - \vecx^\star \right \rangle +\frac{L}{2}\norm{\vecx_c-\vecx^\star}^2 \\
  & \le  f(\vecx^\star) + \left \langle \nabla f(\vecx^\star),\vecx_c - \vecx^\star \right \rangle +\frac{\eps}{2},
\end{align*}
\noindent where the first inequality follows from the fact that
$(\vecx^\star,\vecy^\star)$ is an $(\eps,\delta)$-local min-max equilibrium, the
second inequality follows from the $L$-smoothness of $f$, and the third
inequality follows from $\norm{\vecx_c - \vecx^{\star}} \le \delta$ and our
choice of $\delta=\sqrt{\eps/L}$. The above implies:
\[ \left \langle \nabla f(\vecx^\star), \vecx^\star - \vecx_c \right \rangle < 3 \eps/2. \]
Since
$\hat{\vecx} - \vecx^{\star} = \p{\vecx_c - \vecx^{\star}} \cdot \norm{\hat{\vecx} - \vecx^{\star}}_2 / \delta$
we get that
$\left \langle \nabla f(\vecx^\star),\vecx^\star - \hat{\vecx} \right \rangle < \frac{3 \eps}{2 \delta} \norm{\vecx^\star - \hat{\vecx}}_2$. Therefore
\begin{eqnarray*}
  \norm{\vecx^\star - \hat{\vecx}}_2^2
  &  =  & \left \langle \vecx^\star - \nabla f(\vecx^\star) - \hat{\vecx}, \vecx^\star - \hat{\vecx}\right \rangle + \left \langle\nabla f(\vecx^\star),\vecx^\star - \hat{\vecx} \right \rangle \\
  & < & \frac{3 \eps}{2 \delta}\norm{\vecx^\star - \hat{\vecx}}_2
\end{eqnarray*}
\noindent where in the above inequality we have also used
\eqref{eq:proof:clm:proofOfExistence:1:1}. As a result,
$\norm{\vecx^\star- \hat{\vecx}}_2 < \frac{3\eps}{2\delta} < \alpha / 2$.

\subsection{Proof of Theorem \ref{t:PPAD_inclusion}} \label{sec:proof:t:PPAD_inclusion}

We provide a polynomial-time reduction from \gdaFixed\ to {\sc Brouwer}. This establishes both the totality of \gdaFixed\ and its inclusion to \PPAD, since {\sc Brouwer} is both total and lies in \PPAD, as per Lemma \ref{lem:BrouwerPPAD}. It also establishes the totality and inclusion to \PPAD of \lrlNash, since \lrlNash\ is polynomial-time reducible to \gdaFixed, as shown in Theorem~\ref{t:LocalMaxMin-GDA}.

We proceed to describe our reduction. Suppose that $f$ is the $G$-Lipschitz and $L$-smooth function provided as  input to \gdaFixed. Suppose also that $\alpha$ is the approximation parameter provided as input to \gdaFixed. Given $f$ and $\alpha$, we define function $M: \calP(\matA, \vecb) \rightarrow \calP(\matA, \vecb)$, which serves as input to {\sc Brouwer}, as follows:
\[ M(\vecx,\vecy) = \Pi_{\calP(\matA, \vecb)} \left[
(\vecx - \nabla_x f(\vecx,\vecy) , \vecy + \nabla_y f(\vecx,\vecy))
\right].
\]
\noindent Given that $f$ is $L$-smooth, it follows that $M$ is $(L + 1)$-Lipschitz. We set the approximation parameter provided as input to {\sc Brouwer} be $\gamma = \alpha^2/4(G + 2 \sqrt{d})$.

\smallskip To show the validity of the afore-described reduction, we prove that
every feasible point $(\vecx^\star,\vecy^\star) \in \calP(\matA, \vecb)$ that is
a $\gamma$-approximate fixed point of $M$, i.e.
$\norm{M(\vecx^\star,\vecy^\star) - (\vecx^\star,\vecy^\star)}_2 < \gamma$ is also an $\alpha$-approximate fixed point of $F_{GDA}$. Observe that since $\calP(\matA, \vecb) \subseteq [0, 1]^d$ it holds that
  $\norm{(\vecx, \vecy) - (\vecx', \vecy')}_2 \le \sqrt{d}$ for all
  $(\vecx, \vecy), (\vecx', \vecy') \in \calP(\matA, \vecb)$. Hence, if
  $\gamma > \sqrt{d}$, then finding $\gamma$-approximate fixed points of $M$ is trivial and the same is true for fiding $\alpha$-approximate fixed points of $F_{GDA}$, since $\gamma = \alpha^2/4(G + 2 \sqrt{d})$ which implies that, if $\gamma >\sqrt{d}$, then $\alpha > \sqrt{d}$. Thus, we
  may assume that $\gamma \le \sqrt{d}$.
\smallskip

Next, to simplify notation we define
$(\vecx_\Delta,\vecy_\Delta) = (x^\star - \nabla_x f(\vecx^\star,\vecy^\star) , \vecy^\star + \nabla_y f(\vecx^\star,\vecy^\star))$ and
$(\hat{\vecx},\hat{\vecy})=\argmin_{(\vecx,\vecy) \in \calP(\matA, \vecb)} \norm{(\vecx_\Delta,\vecy_\Delta) - (\vecx,\vecy)}_2$.
Given that $(\vecx^\star,\vecy^\star)$ is a $\gamma$-approximate fixed point of $M$, we have that
\begin{align}
    \norm{(\vecx^\star,\vecy^\star) - (\hat{\vecx},\hat{\vecy})}_2 < \gamma. \label{eq:observation 1}
\end{align}
Using Theorem 1.5.5 (b) of \cite{facchinei2007finite}, we get that
\begin{align}
    \left \langle (\vecx_\Delta,\vecy_\Delta) - (\hat{\vecx},\hat{\vecy}),
    (\vecx,\vecy) - (\hat{\vecx},\hat{\vecy}) \right \rangle \leq 0 \text{ for all } (\vecx,\vecy) \in \calP(\matA, \vecb). \label{eq:observation 2}
\end{align}
Next we show the following:
\begin{claim} \label{c:2}
    For all $(\vecx,\vecy) \in \calP(\matA,\vecb)$,
  $\left \langle (\vecx_\Delta,\vecy_\Delta) - (\vecx^\star,\vecy^\star), (\vecx,\vecy) - (\vecx^\star,\vecy^\star)\right \rangle < (G + 2 \sqrt{d}) \cdot \gamma$.
\end{claim}

\begin{proof}
 We have that:
  \begin{eqnarray*}
    \left \langle (\vecx_\Delta,y_\Delta) - (\vecx^\star,y^\star), (\vecx,y) - (\vecx^\star,y^\star) \right \rangle
    &  =            & \left \langle (\vecx_\Delta,y_\Delta) - (\hat{\vecx},\hat{y}),(\vecx,y) - (\vecx^\star,y^\star)\right \rangle\\
    &~~~~~~~~~~~~ + & \left \langle (\hat{\vecx},\hat{y}) - (\vecx^\star,y^\star), (\vecx,y) - (\vecx^\star,y^\star)\right \rangle\\
    & =             & \left \langle (\vecx_\Delta,y_\Delta) - (\hat{\vecx},\hat{y}), (\vecx,y) - (\hat{\vecx},\hat{y})\right \rangle\\
    &~~~~~~~~~~~~ + & \left \langle (\vecx_\Delta,y_\Delta) - (\hat{\vecx},\hat{y}),(\hat{\vecx},\hat{y}) - (\vecx^\star,y^\star) \right \rangle\\
    &~~~~~~~~~~~~ + & \left \langle (\hat{\vecx},\hat{y}) - (\vecx^\star,y^\star), (\vecx,y) - (\vecx^\star,y^\star)\right \rangle\\
    & <             & \norm{(\vecx_\Delta,y_\Delta) - (\hat{\vecx},\hat{y})}_2 \gamma + \gamma \cdot \sqrt{d} \\
    & \le           & \norm{(\vecx_\Delta,y_\Delta) - (\vecx^\star,y^\star)}_2 \gamma + \gamma^2 + \gamma \cdot \sqrt{d}\\
    & =            & \norm{\nabla f(\vecx^\star,y^\star)}_2 \gamma + \gamma^2 + \gamma \cdot \sqrt{d} \\
    & \leq         & (G + 2 \sqrt{d}) \cdot \gamma,
  \end{eqnarray*}
  where (1) for the first inequality we use \eqref{eq:observation 1}, \eqref{eq:observation 2}, the Cauchy-Schwarz inequality, and the fact that the $\ell_2$ diameter of $\calP(\matA,\vecb)$ is at most $\sqrt{d}$; (2) for the second inquality we use the triangle inequality and~\eqref{eq:observation 1}; (3) for the equality that follows we use the definition of $(\vecx_\Delta,y_\Delta)$; and (4) for the last inequality we use that $G$, the Lipschitzness of $f$, bounds the magnitude of its gradient, and that $\gamma \le \sqrt{d}$.
\end{proof}

\noindent Now let $\vecx' = \argmin_{\vecx \in K(y^\star)}\norm{\vecx-\vecx_{\Delta}}_2$
where $K(\vecy^{\star}) =\{\vecx \mid (\vecx,\vecy^\star)\in \calP (\matA,\vecb))\}$.
Using Theorem 1.5.5 (b) of \cite{facchinei2007finite} for $\vecx'$ we get that
$\left \langle \vecx_\Delta -\vecx', \vecx^\star - \vecx' \right \rangle \leq 0$. Using
Claim \ref{c:2} for  vector $(\vecx',y^\star) \in \calP(\matA,\vecb)$ we get that
$\left \langle \vecx^\star - \vecx_\Delta, \vecx^\star - \vecx'  \right \rangle < (G + 2 \sqrt{d}) \gamma$.
Adding the last two inequalities and using the fact that
$\gamma = \alpha^2/4(G + 2 \sqrt{d})$ we  get the following
\[ \norm{\vecx^\star - \Pi_{K(y^\star)}(\vecx^\star - \nabla_x f(\vecx^\star,y^\star))}_{2} < \sqrt{(G + 2 \sqrt{d}) \cdot \gamma} = \alpha/2. \]
\noindent Using the exact same reasoning we can also prove that
\[ \norm{\vecy^\star - \Pi_{K(x^\star)}(\vecy^\star - \nabla_y f(\vecx^\star, y^\star))}_{2} < \alpha/2 \]
where $K(\vecx^{\star}) =\{\vecy \mid (\vecx^{\star}, \vecy)\in \calP (\matA, \vecb))\}$.
Combining the last two inequalities we get that $(\vecx^{\star}, \vecy^{\star})$ is an $\alpha$-approximate fixed point of $F_{GDA}$.

\section{Missing Proofs from Section \ref{sec:seic}} \label{s:Q_coef}

  In this section we present the missing proofs from Section \ref{sec:seic} and
more precisely in the following sections we prove the Lemmas
\ref{l:well_defined}, \ref{l:bounding_gradients}, and
\ref{l:positive_corners_boundary}. For the rest of the proofs in this section we
define $L(\vecc)$ to be the cubelet which has the down-left corner equal to
$\vecc$, formaly
\[ L(\vecc) = \left[\frac{c_1}{N - 1}, \frac{c_1 + 1}{N - 1}\right] \times \cdots \times \left[\frac{c_d}{N - 1}, \frac{c_d + 1}{N - 1}\right] \]
\noindent and we also define $L_c(\vecc)$ to be the set of corners of the
cubelet $L(\vecc)$, or more formally
\[ L_c(\vecc) = \{c_1, c_1 + 1\} \times \cdots \times \{c_d, c_d + 1\}. \]

\addtocontents{toc}{\protect\setcounter{tocdepth}{1}}
\subsection{Proof of Lemma \ref{l:well_defined}} \label{sec:proof:l:well_defined}
\addtocontents{toc}{\protect\setcounter{tocdepth}{2}}

  We start with a lemma about the differentiability properties of the functions
$Q_{\vecv}^{\vecc}$ which we defined in Definition \ref{d:cof}.

\begin{lemma} \label{l:derivation1}
    Let $\vecx \in [0,1]^d$ lying in cublet
  $R(\vecx) = \left[\frac{c_1}{N - 1}, \frac{c_1 + 1}{N - 1}\right] \times \cdots \times \left[\frac{c_d}{N - 1}, \frac{c_d + 1}{N - 1}\right]$,
  where $\vecc \in \p{\nm{N}}^d$. Then for any vertex $\vecv \in R_c(\vecx)$,
  the function $Q_{\vecv}^{\vecc}(\vecx)$ is continuous and twice
  differentiable. Moreover if $Q_{\vecv}^{\vecc}(\vecx) = 0$ then also
  $\frac{d Q_{\vecv}^{\vecc}(\vecx)}{d x_i} = 0$ and
  $\frac{d^2 Q_{\vecv}^{\vecc}(\vecx)}{d x_i~d x_j} = 0$.
\end{lemma}

\begin{proof}
    \textbf{1st order differentiability:} We remind from the Definition
  \ref{d:cof} that if we let $\vecs^{\vecc} = (s_1,\ldots,s_d)$ be the source
  vertex of $R(\vecx)$ and $\vecp_{\vecx}^{\vecc} = (p_1,\ldots,p_d)$ be the
  canonical representation of $\vecx$. Then for each vertex
  $\vecv \in R_c(\vecx)$ we define the following partition of the set of
  coordinates $[d]$,
  \[ A_{\vecv}^{\vecc}=\{j:~~|v_j - s_j|=0\} \text{  and  } B_{\vecv}^{\vecc}=\{j:~~|v_j - s_j|=1\}. \]
  \noindent Now in case $B_{\vecv}^{\vecc} = \varnothing$, which corresponds to
  $\vecv$ being the source node $\vecs^{\vecc}$ then
  $Q_{\vecv}^{\vecc}(\vecx) = \prod_{j = 1}^d S_{\infty}(1 - S(p_j))$ which is
  clearly differentiable as product of compositions of differentiable functions.
  The exact same holds for $A_{\vecv}^{\vecc} = \varnothing$ which corresponds
  to $\vecv$ being the target vertex $\vect^{\vecc}$ of the cubelet $R(\vecx)$.
  We thus focus on the case where
  $A_{\vecv}^{\vecc}, B_{\vecv}^{\vecc} \neq \varnothing$. To simplify notation
  we denote
  $Q_{\vecv}^{\vecc}(\vecx)$ by $Q(\vecx)$, $A_{\vecv}^{\vecc}$ by $A$ and
  $B_{\vecv}^{\vecc}$ by $B$ for the rest of this proof. We prove that in case
  $i \in B$ then $\frac{\partial Q(\vecx)}{\partial x_i}$ always exits. The case
  $i \in A$ follows then symmetrically. We have the following cases
  \begin{enumerate}
    \item[$\blacktriangleright$] Let $j \in A$ and $\ell \in B \setminus \{i\}$
    such that $p_j \geq p_\ell$. By Definition \ref{d:cof}, if $\epsilon$ is
    sufficiently small then
    $Q(x_i - \epsilon, \vecx_{-i}) = Q(x_i + \epsilon, \vecx_{-i}) = Q(x_i, \vecx_{-i}) = 0$.
    Thus $\frac{\partial Q(\vecx)}{\partial x_i}$ exists and equals $0$.
    \item[$\blacktriangleright$] Let $p_\ell > p_j$ for all
    $\ell \in B \setminus \{i\}$ and $j \in A$. In this case we have the
    following subcases.
    \begin{enumerate}
      \item[$\triangleright$] $p_i > p_j$ for all $j \in A$: Then
      $\frac{\partial Q(\vecx)}{\partial x_i}$ exists since both
      $S_{\infty}(\cdot)$ and $S(\cdot)$ are differentiable.
      \item[$\triangleright$] $p_i < p_j$ for some $j \in A$: By
      Definition~\ref{d:cof}, if $\epsilon$ is sufficiently small then
      $Q(x_i - \epsilon, x_{-i}) = Q(x_i + \epsilon, \vecx_{-i}) = Q(x_i, \vecx_{-i}) = 0$.
      Thus $\frac{\partial Q(\vecx)}{\partial x_i}$ exists and equals $0$.
      \item[$\triangleright$] $p_i = p_j$ for some $j \in A$ and
      $p_i \ge p_{j'}$ for all $j' \in A \setminus \{j\}$: By Definition
      \ref{d:cof}, if $\eps$ is sufficiently small then
      $Q(x_i - \eps, \vecx_{-i}) = 0$ and also $Q(x_i, \vecx_{-i}) = 0$, thus
      \[ \lim_{\eps \rightarrow 0^+}\frac{Q(x_i, \vecx_{-i}) - Q(x_i - \eps, \vecx_{-i})}{\eps} = 0. \]
      At the same time
      \[ \lim_{\eps \rightarrow 0^+}\frac{Q(x_i + \eps,x_{-i}) - Q(x_i,x_{-i})}{\eps} = 0 \]
      since both $S_{\infty}(\cdot)$ and $S(\cdot)$ are differentiable
      functions, $S_{\infty}(S(p_i)- S(p_j)) = S_{\infty}(0) = 0$, and
      $S'_{\infty}(S(p_i)- S(p_j)) = S'_{\infty}(0) = 0$.
    \end{enumerate}
  \end{enumerate}

  \noindent \textbf{2nd order differentiability:} Let $Q'(\vecx)$ be equal to
  $\frac{\partial Q(\vecx)}{\partial x_k}$ for convenience. As in the previous analysis in
  case $A_{\vecv}^{\vecc} = \varnothing$ or $B_{\vecv}^{\vecc} = \varnothing$ then $Q'(x)$
  is differentiable with respect to $x_i$ since $S(\cdot), S_{\infty}(\cdot)$ are twice
  differentiable. Thus we again focus in the case where $A, B \neq \varnothing$. Notice
  that by the previous analysis $Q'(\vecx) = 0$ if there exists $\ell \in B$ and $j \in A$
  such that $p_\ell \geq p_j$. Without loss of generality we assume that $i \in B$ and we
  prove that
  $\frac{\partial Q'(\vecx)}{\partial x_i} \triangleq \frac{\partial^2 Q(\vecx)}{\partial x_i \partial x_k}$
  always exists.
  \begin{enumerate}
    \item[$\blacktriangleright$] Let $j \in A$ and $\ell \in B \setminus \{i\}$ such that
    $p_j \geq p_\ell$. By Definition \ref{d:cof},
    $Q'(x_i - \eps, \vecx_{-i}) = Q'(x_i + \eps, \vecx_{-i}) = Q'(x_i, \vecx_{-i}) = 0$.
    Thus
    $\frac{\partial Q'(\vecx)}{\partial x_i} \triangleq \frac{\partial^2 Q'(\vecx)}{\partial x_i \partial x_k}$
    exists and equals $0$.
    \item[$\blacktriangleright$] Let $p_\ell > p_j$ for all $\ell \in B \setminus \{i\}$ and
    $j \in A$.
    \begin{enumerate}
        \item[$\triangleright$] $p_i > p_j$ for all $j \in A$: Then
        $\frac{\partial Q'(\vecx)}{\partial x_i} \triangleq \frac{\partial^2 Q(\vecx)}{\partial x_i \partial x_k}$
        exists since both $S_{\infty}(\cdot)$ and $S(\cdot)$ are twice differentiable.
        \item[$\triangleright$] $p_i < p_j$ for some $j \in A$. By Definition \ref{d:cof},
        $Q'(x_i - \eps, \vecx_{-i}) = Q'(x_i + \eps, \vecx_{-i}) = Q'(x_i, \vecx_{-i}) = 0$.
        Thus
        $\frac{\partial Q'(\vecx)}{\partial x_i} \triangleq \frac{\partial^2 Q(\vecx)}{\partial x_i \partial x_k}$
        exists and equals $0$.
        \item[$\triangleright$] $p_i = p_j$ for some $j \in A$ and $p_i > p_{j'}$ for all
        $j' \in A \setminus \{j\}$. By Definition \ref{d:cof}, if $\eps$ is sufficiently
        small then $Q'(x_i - \eps, \vecx_{-i}) = 0$ and thus
        \[ \lim_{\eps \rightarrow 0^+}\frac{Q'(x_i, \vecx_{-i}) - Q'(x_i - \eps, \vecx_{-i})}{\eps} = 0. \]
        At the same time
        $\lim_{\eps \rightarrow 0^+}\frac{Q'(x_i + \eps, \vecx_{-i}) - Q'(x_i, \vecx_{-i})}{\eps}$
        exists since both $S_{\infty}(\cdot)$ and $S(\cdot)$ are twice differentiable.
        Moreover equals $0$ since $S_{\infty}(S(p_i) - S(p_j)) = S_{\infty}(0) = 0$ and
        $S'_{\infty}(S(p_i) - S(p_j)) = S'_{\infty}(0) = S''_{\infty}(0) = S(0) = 0$.
    \end{enumerate}
  \end{enumerate}

  \noindent In every step of the above proof where we use properties of $S_{\infty}$ and $S$
  we use Lemma \ref{lem:stepFunctions}.
\end{proof}

\noindent So far we have established the fact that the functions $Q^{\vecc}_{\vecv}(\vecx)$
are twice differentiable when $\vecx$ moves within the same cubelet. Next we will show that
when $\vecx$ moves from one cubelet to another then the corresponding $Q^{\vecc}_{\vecv}$
functions changes value smoothly.

\begin{lemma} \label{l:continuity1}
    Let $\vecx \in [0, 1]^d$ such that there exists a coordinate $i \in [d]$
  with the property
  $R(x_i + \eps, \vecx_{- i}) = \left[\frac{c_1}{N - 1}, \frac{c_1 + 1}{N - 1}\right] \times \cdots \times \left[\frac{c_d}{N - 1}, \frac{c_d + 1}{N - 1}\right]$ and
  $R(x_i - \eps, \vecx_{- i}) = \left[\frac{c'_1}{N - 1}, \frac{c'_1 + 1}{N - 1}\right] \times \cdots \times \left[\frac{c'_d}{N - 1}, \frac{c'_d + 1}{N - 1}\right]$,
  with $\vecc, \vecc' \in \p{\nm{N - 1}}^d$ and $\eps$ sufficiently small, i.e.
  $\vecx$ lies in the boundary of two cubelets. Then the following statements
  hold.
  \begin{enumerate}
    \item For all vertices
    $\vecv \in R_c(x_i + \eps, \vecx_{- i}) \cap R_c(x_i - \eps, \vecx_{- i})$,
    it holds that
    \begin{enumerate}
        \item $Q_{\vecv}^{\vecc}(\vecx) = Q_{\vecv}^{\vecc'}(\vecx)$,
        \item $\frac{\partial Q_{\vecv}^{\vecc}(\vecx)}{\partial x_j} = \frac{\partial Q_{\vecv}^{\vecc'}(\vecx)}{\partial x_i}$
        for all $i \in [d]$, and
        \item $\frac{\partial^2 Q_{\vecv}^{\vecc}(\vecx)}{\partial x_i~\partial x_j} = \frac{\partial Q_{\vecv}^{\vecc'}(\vecx)}{\partial x_i~\partial x_j}$
        for all $i, j \in [d]$.
    \end{enumerate}
    \item For all vertices
    $\vecv \in R_c(x_i + \eps, \vecx_{- i}) \setminus R_c(x_i - \eps, \vecx_{- i})$,
    it holds that
    $Q_{\vecv}^{\vecc}(\vecx) = \frac{\partial Q_{\vecv}^{\vecc}(\vecx)}{\partial x_i}=\frac{\partial^2 Q_{\vecv}^{\vecc}(\vecx)}{\partial x_i~\partial x_j}=0$.
    \item for all vertices
    $\vecv \in R_c(x_i - \eps, \vecx_{- i}) \setminus R_c(x_i + \eps, \vecx_{- i})$, it
    holds that
    $Q_{\vecv}^{\vecc'}(\vecx) = \frac{\partial Q_{\vecv}^{\vecc'}(\vecx)}{\partial x_i}=\frac{\partial^2 Q_{\vecv}^{\vecc'}(\vecx)}{\partial x_i~\partial x_j}=0$.
\end{enumerate}
\end{lemma}

\noindent Lemma \ref{l:continuity1} is crucial since it establishes that
$\P_{\vecv}(\vecx)$ is a continuous and twice differentiable even when $\vecx$
moves from  one cubelet to another. Since the proof of Lemma \ref{l:continuity1}
is very long and contains the proof of some sublemmas, we postpone it for the
end of this section in Section \ref{ss:proof_of_continuity1}. We now proceed
with the proof of Lemma \ref{l:well_defined}.

\begin{proof}[Proof of Lemma \ref{l:well_defined}]
    We first prove that $\P_{\vecv}(\vecx)$ is a continuous function. Let
  $\vecx \in [0, 1]^d$ lying on the boundary of the following cubelets
  \[ \left[\frac{c_1^{(1)}}{N - 1}, \frac{c_1^{(1)} + 1}{N - 1}\right] \times \cdots \times \left[\frac{c_d^{(1)}}{N - 1}, \frac{c_d^{(1)} + 1}{N - 1}\right] \]
  \[ \cdots \]
  \[ \left[\frac{c_1^{(i)}}{N - 1}, \frac{c_1^{(i)} + 1}{N - 1}\right] \times \cdots \times \left[\frac{c_d^{(i)}}{N - 1}, \frac{c_d^{(i)} + 1}{N - 1}\right] \]
  \[ \cdots \]
  \[\left[\frac{c^{(m)}_1}{N - 1}, \frac{c^{(m)}_1 + 1}{N - 1}\right] \times \cdots \times \left[\frac{c^{(m)}_d}{N - 1}, \frac{c^{(m)}_d + 1}{N - 1}\right]. \]
  \noindent where $\vecc^{(1)}, \dots, \vecc^{(m)} \in \p{\nm{N - 1}}^d$. This
  means that for every $i \in [m]$ there exists a coordinate $j_i \in [d]$ and a
  value $\eta_i \in \R$ with sufficiently small absolute value such that
  \[ R(x_{j_i} + \eta_i, \vecx_{- j_i}) = \left[\frac{c_1^{(i)}}{N - 1}, \frac{c_1^{(i)} + 1}{N - 1}\right] \times \cdots \times \left[\frac{c_d^{(i)}}{N - 1}, \frac{c_d^{(i)} + 1}{N - 1}\right]. \]

  \noindent We then consider the following cases.
  \begin{enumerate}
    \item[$\blacktriangleright$]
    $\vecv \notin \cup_{i = 1}^m R_c(x_{j_i} + \eta_i, \vecx_{- j_i})$. By
    Definition \ref{d:_coeff}, in all the $m$ aforementioned cubelets, the
    coefficient $\P_{\vecv}$ takes value $0$ and hence it is continuous in this
    part of the space.

    \item[$\blacktriangleright$]
    $\vecv \in \cap_{j \in U} R_c(x_{j_i} + \eta_i, \vecx_{- j_i})$ and
    $\vecv \notin \cup_{i \in \bar{U}} R_c(x_{j_i} + \eta_i, \vecx_{- j_i})$,
    for some $U \subseteq [m]$ with $\bar{U} = [m] \setminus U$. In this case
    $\P_{\vecv}(x_{j_i} + \eta_i, \vecx_{j_i})$ was computed according to a
    cubelet with $\vecv \in R_c(x_{j_i} + \eta_i, \vecx_{- j_i})$. Then Lemma
    \ref{l:continuity1} implies that $Q^{\vecc^{(i)}}_{\vecv}(\vecx) = 0$ since
    $\vecv \in R_c(x_{j_i} + \eta_i, \vecx_{- j_i}) \setminus R_c(x_{j_{i'}} + \eta_{i'}, \vecx_{- j_{i'}})$
    where $i' \in [m]$ and $i \neq i'$. Therefore we conclude that
    $\P_{\vecv}(\vecx) = 0$ and
    \[ \lim\limits_{\eta_i \to 0} \P_{\vecv}(x_{j_i} + \eta_i, \vecx_{-i}) = 0. \]
    \item[$\blacktriangleright$]
    $\vecv \in \cap_{i = 1}^m R_c(x_{j_i} + \eta_i, \vecx_{- j_i})$. By Lemma
    \ref{l:continuity1} for all $i \in [m]$ it holds that
    \[ \frac{Q_{\vecv}^{\vecc^{(i)}}(\vecx)}{\sum_{\vecv \in R_c(x_{j_i} + \eta_i, \vecx_{- j_i})} Q_{\vecv}^{\vecc^{(i)}}(\vecx)} =
    \frac{Q_{\vecv}^{\vecc^{(i)}}(\vecx)}{\sum_{\vecv \in \cap_{i = 1}^m R_c(x_{j_i} + \eta_i, \vecx_{- j_i})} Q_{\vecv}^{\vecc^{(i)}}(\vecx)} \]
    \[ = \frac{Q_{\vecv}^{\vecc^{(i')}}(\vecx)}{\sum_{\vecv \in \cap_{i = 1}^m R_c(x_{j_i} + \eta_i, \vecx_{- j_i})} Q_{\vecv}^{\vecc^{(i')}}(\vecx)}
    = \frac{Q_{\vecv}^{\vecc^{(i')}}(\vecx)}{\sum_{\vecv \in R_c(x_{j_i} + \eta_i, \vecx_{- j_i})} Q_{\vecv}^{\vecc^{(i')}}(\vecx)} \]
    which again implies the continuity of $\P_{\vecv}(\vecx)$ at $\vecx$.
  \end{enumerate}

    Next we prove that $\P_{\vecv}(\vecx)$ is differentiable for all
  $\vecv \in \p{\nm{N}}^d$. Fix some $i \in [d]$ we will prove that
  $\frac{\partial \mathrm{P}(\vecx)}{\partial x_i}$ always exists. Let $C^+$ be
  the set of down-left corners of the cubelets in which
  $\lim_{\eps \rightarrow 0^+} (x_i + \eps, \vecx_{-i})$ belongs to and $C^{-}$
  be the set of down-left corners of the cubelets in which
  $\lim_{\eps \rightarrow 0^+} (x_i - \eps, \vecx_{-i})$ belongs to. It easy to
  see that $C^+$ and $C^-$ are non-empty and fixed for $\epsilon > 0$ and
  sufficiently small.

   To prove that
  $\frac{\partial \P_{\vecv}(\vecx)}{\partial x_i}$ always exists, we consider
  the following $3$ mutually exclusive cases.
  \begin{enumerate}
    \item[$\blacktriangleright$]
    \underline{$\vecv \in L_c(\vecc^{(1)})$ for $\vecc^{(1)} \in C^+$ and
    $\vecv \in L_c(\vecc^{(2)})$ for $\vecc^{(2)} \in C^-$.} Since the
    coefficient $\P_{\vecv}(\vecx)$ is a continuous function, we have that
    \begin{enumerate}
      \item[$\triangleright$]
      $\lim_{\eps \rightarrow 0^+} \frac{\P_{\vecv}(x_i + \epsilon, \vecx_{- i}) - \P_{\vecv}(x_i, \vecx_{- i})}{\epsilon} = \frac{\frac{\partial Q_{\vecv}^{\vecc^{(1)}}(\vecx)}{\partial x_i} \sum_{\vecv' \in L_c(\vecc^{(1)})} Q_{\vecv'}^{\vecc^{(1)}}(\vecx) - Q_{\vecv}^{\vecc^{(1)}}(\vecx) \sum_{\vecv' \in L_c(\vecc^{(1)})} \frac{\partial Q_{\vecv'}^{\vecc^{(1)}}(\vecx)}{\partial x_i}}{\left(\sum_{\vecv' \in L_c(\vecc^{(1)})} Q_{\vecv'}^{\vecc^{(1)}}(\vecx)\right)^2}$
      \item[$\triangleright$]
      $\lim_{\eps \rightarrow 0^+} \frac{\P_{\vecv}(x_i, \vecx_{- i}) - \P_{\vecv}(x_i - \epsilon, \vecx_{- i})}{\epsilon} = \frac{\frac{\partial Q_{\vecv}^{\vecc^{(2)}}(\vecx)}{\partial x_i} \sum_{\vecv' \in L_c(\vecc^{(2)})} Q_{\vecv'}^{\vecc^{(2)}}(\vecx) - Q_{\vecv}^{\vecc^{(2)}}(\vecx) \sum_{\vecv' \in L_c(\vecc^{(2)})}\frac{\partial Q_{\vecv'}^{\vecc^{(2)}}(\vecx)}{\partial x_i}}{\left(\sum_{\vecv' \in L_c(\vecc^{(2)})} Q_{\vecv'}^{\vecc^{(2)}}(\vecx)\right)^2}$
    \end{enumerate}
    Both of the above limits exists due to the fact that $Q^{\vecc}_{\vecv}(\vecx)$ is
    differentiable (Lemma~\ref{l:derivation1}). Moreover, since
    $\vecv \in L_c(\vecc^{(1)}) \cap L_c(\vecc^{(2)})$, Case $1$ of Lemma
    \ref{l:continuity1} implies that the two limits above have exactly the same value and
    hence $\P_{\vecv}$ is differentiable at $\vecx$.
    \item[$\blacktriangleright$]
    \underline{$\vecv \notin L_c(\vecc^{(1)})$ for all $\vecc^{(1)} \in C^+$.} In the case
    where $\vecv \notin L_c(\vecc)$ for all the down-left corners $\vecc$ of the cubelets at
    which $\vecx$ lies, then by Definition \ref{d:_coeff}
    $\P_{\vecv}(x_i, \vecx_{-i}) = \P_{\vecv}(x_i + \eps, \vecx_{-i}) = \P_{\vecv}(x_i - \eps, \vecx_{-i}) = 0$.
    Thus $\frac{\partial \P_{\vecv}(\vecx)}{\partial x_i}$ exists and equals $0$. Therefore
    we may assume that $\vecv \in L_c(\vecc)$ for some down-left corner $\vecc$ of a cubelet
    at which $\vecx$ lies. Due to the fact that $\P_{\vecv}(\vecx)$ is a continuous function
    and that $\vecv \notin L_c(\vecc^{(1)})$ for all $\vecc^{(1)} \in C^+$, we get that
    \[ \P_{\vecv}(x_i + \eps, \vecx_{-i}) = 0 ~~~\text{ and }~~~ \P_{\vecv}(x_i, \vecx_{-i}) = 0. \]
    \noindent We also have that $\vecv \in L_c(\vecc) / L_c{\vecc^{(1)}}$ where $\vecc$,
    $\vecc^{(1)}$ are down-left corners of cubelets at which $\vecx$ lies and
    $(x_i + \eps, \vecx_{-i})$ lies respectively. Therefore we get by Case $1$ of
    Lemma \ref{l:continuity1} that $Q_{\vecv}^{\vecc}(\vecx) = 0$ implying that
    $\P_{\vecv}(x_i, \vecx_{-i}) = 0$. As a result,
    \[ \lim_{\eps \to 0^+} \frac{\P_{\vecv}(x_i + \eps, \vecx_{- i}) - \P_{\vecv}(x_i, x_{- i})}{\eps} = 0 \]
    We now need to argue that
    $\lim_{\eps \to 0^+} \frac{\P_{\vecv}(x_i, \vecx_{-i}) - \P_{\vecv}(x_i - \eps, x_{-i})}{\eps}$
    exists and equals $0$. At first observe that $0 \leq x_i - c_i \leq \delta$ since
    $\vecx$ lies in the cubelet with down-left corner $\vecc$. In case $x_i - c_i < \delta$
    then $(x_i + \eps, \vecx_{-i})$ lies in $\vecc$ for arbitrarily small $\eps$, meaning
    that $\vecc \in C^+$. The latter contradicts the fact that
    $\vecv \notin L_c{\vecc^{(1)}}$ for all $\vecc^{(1)} \in C^+$. As a result,
    $x_i - c_i = \delta$ which implies that $\vecc \in C^-$ and hence
    \[ \lim_{\eps \to 0^+} \frac{\P_{\vecv}(x_i, \vecx_{- i}) - \P_{\vecv}(x_i - \eps, \vecx_{- i})}{\eps} = \frac{\frac{\partial Q_{\vecv}^{\vecc}(\vecx)}{\partial x_i} \sum_{\vecv' \in L_c(\vecc)} Q_{\vecv'}^{\vecc}(\vecx) - Q_{\vecv}^{\vecc}(\vecx) \sum_{\vecv' \in L_c(\vecc)}\frac{\partial Q_{\vecv'}^{\vecc}(\vecx)}{\partial x_i}}{\left(\sum_{\vecv' \in L_c(\vecc)} Q_{\vecv'}^{\vecc}(\vecx) \right)^2}. \]
    The above limit equals to $0$ since
    $Q_{\vecv}^{\vecc}(\vecx) = \frac{\partial Q_{\vecv}^{\vecc}(\vecx)}{\partial x_i} = 0$
    by applying Lemma \ref{l:continuity1} due to the fact that
    $\vecv \in L_c(\vecc) \setminus L_c(\vecc^{(1)})$.
    \item[$\blacktriangleright$]
    \underline{$\vecv \notin L_c(\vecc^{(2)})$ for all $\vecc^{(2)} \in C^-$.} Symmetrically
    with the previous case.
  \end{enumerate}
  The second order differentiability of $\P_{\vecv}(\vecx)$ can be established using exactly
  the same arguments for computing the following limit
  \[ \lim_{\eps, \eps' \to 0} \frac{\P_{\vecv}(x_i + \eps, x_j + \eps', \vecx_{- i, j}) - \P_{\vecv}(\vecx)}{\eps^2}. \]

    The last thing that we need to show to prove Lemma \ref{l:well_defined} is that the set
  $R_+(\vecx)$ has cardinality at most $d + 1$ and that it can be computed in $\poly(d)$
  time. Let $p_{\vecx}^{\vecc} \in [0,1]^d$ be the canonical representation of $\vecx$ with
  the respect to a cubelet $L(\vecc)$ in which $\vecx$ belongs to. We define the source
  vertex $\vecs^{\vecc} = (s_1,\ldots,s_d)$ and the target vertex
  $\vect^{\vecc} = (t_1,\ldots,t_d)$ of $L(\vecc)$. Once this is done the vertices in
  $R_+(\vecv)$ are exactly the vertices of $L_c(\vecc)$ for which it holds that
  \[ p_{\ell} > p_j ~~~~\text{ for all } \ell \in A_{\vecv}^{\vecc}, j \in B_{\vecv}^{\vecc} \]
  \noindent since for all the others $\vecv \in \p{\nm{N}}^d$ it holds that
  $Q_{\vecv}^{\vecc} (\vecx) = 0$, $\nabla Q_{\vecv}^{\vecc} (\vecx) = 0$, and
  $\nabla^2 Q_{\vecv}^{\vecc} (\vecx) = 0$. These vertices $\vecv \in R_+(\vecx)$ can be
  computed in polynomial time as follows: \textbf{i)} the coordinates $p_1,\ldots,p_d$ are
  sorted in increasing order, and \textbf{ii)} for each $m=0, \ldots, d$ compute the vertex
  $\vecv^{(m)} \in L_c(\vecc)$,
  \[ \vecv_{j}^m = \left\{
     \begin{array}{ll}
       s_j & \text{if coordinate $j$ belongs in the first } m \text{ coordinates wrt  the order of } \vecp_{\vecx}^{\vecc}\\
       t_j & \text{if coordinate $j$ belongs in the last } d-m \text{ coordinates wrt  the order of } \vecp_{\vecx}^{\vecc}\\
     \end{array}
     \right. \]
  \noindent By Definition~\ref{d:cof} it immediately follows that
  $R_+(\vecx) \subseteq \{\vecv^{(1)}, \ldots, \vecv^{(m)}\}$ from which we get that
  $\abs{R_+(\vecx)} \leq d + 1$ and also they can be computed in $\poly(d)$ time.
\end{proof}

\noindent To finish the proof of Lemma \ref{l:well_defined} we only need the proof of Lemma
\ref{l:continuity1} which we present in the following section.

\subsubsection{Proof of Lemma \ref{l:continuity1}} \label{ss:proof_of_continuity1}

\begin{lemma}\label{l:1}
    Let a point $\vecx \in [0,1]^d$ lying in the boundary of the cubelets with down-left
  corners $\vecc = (c_1, \ldots, c_{m - 1}, c_m, c_{m + 1}, \ldots, c_d)$ and
  $\vecc' = (c_1, \ldots, c_{m-1}, c_m + 1, c_{m + 1}, \ldots, c_d)$. Then the canonical
  representation of $\vecx$ in the cubelet $L(\vecc)$ is the same with the the canonical
  representation of $\vecx$ in the cubelet $L(\vecc')$. More precisely,
  $\vecp_{\vecx}^{\vecc} = \vecp_{\vecx}^{\vecc'}$.
\end{lemma}

\begin{proof}
    Let $c_m$ be even. By the definition of the canonical representation in Definition
  \ref{d:canonical_representation}, the source and target of the cubelets $L(\vecc)$ and
  $L(\vecc')$ are respectively,
  \begin{enumerate}
    \item[$\diamond$] $\vecs^{\vecc} = (s_1, \ldots, s_{m - 1}, c_m, s_{m + 1}, \ldots, s_d)$,
    \item[$\diamond$] $\vect^{\vecc} = (t_1, \ldots, s_{m - 1}, c_m + 1, t_{m + 1}, \ldots, t_d)$,
    \item[$\diamond$] $\vecs^{\vecc'} = (s_1, \ldots, s_{m - 1}, c_m + 2, s_{m + 1}, \ldots, s_d)$,
    \item[$\diamond$] $\vect^{\vecc'} = (t_1, \ldots, t_{m - 1}, c_m + 1, t_{m + 1}, \ldots, t_d)$.
  \end{enumerate}
  \noindent Hence we get that $p_j = p_j'$ for $j \neq m$. Since $\vecx$ belongs to the
  boundary of both cublets $L(\vecc)$ and $L(\vecc')$ we get that $x_m = c_m + 1$ which
  implies that $p_m = p_m' = 1$. In case $c_m$ is odd we get that
  $\vecp_{\vecx}^{\vecc} = \vecp_{\vecx}^{\vecc'}$ but with $p_m = p_m'= 0$.
\end{proof}

\begin{lemma} \label{l:continuity}
    Let $\vecx \in [0,1]^d$ lying at the intersection of the cubelets $L(\vecc)$,
  $L(\vecc')$ with down-left corners
  $\vecc = (c_1, \ldots, c_{m - 1}, c_m, c_{m + 1}, \ldots, c_d)$, and
  $\vecc' = (c_1, \ldots, c_{m-1}, c_m + 1, c_{m + 1}, \ldots, c_d)$. Then the following
  statements are true.
  \begin{enumerate}
    \item For all vertices $\vecv \in L_c(\vecc) \cap L_c(\vecc')$ it holds that
    \begin{enumerate}
        \item $Q_{\vecv}^{\vecc}(\vecx) = Q_{\vecv}^{\vecc'}(\vecx)$,
        \item $\frac{\partial Q_{\vecv}^{\vecc}(\vecx)}{\partial x_i} = \frac{\partial Q_{\vecv}^{\vecc'}(\vecx)}{\partial x_i}$,
        \item $\frac{\partial^2 Q_{\vecv}^{\vecc}(\vecx)}{\partial x_i~\partial x_j} = \frac{\partial^2 Q_{\vecv}^{\vecc'}(\vecx)}{\partial x_i~\partial x_j}$.
    \end{enumerate}
    \item For all vertices
    $\vecv \in L_c(\vecc) \setminus L_c(\vecc')$ it holds that
    $Q_{\vecv}^{\vecc}(\vecx) = \frac{\partial Q_{\vecv}^{\vecc}(\vecx)}{\partial x_i} = \frac{\partial^2 Q_{\vecv}^{\vecc}(\vecx)}{\partial x_i~\partial x_j} = 0$.
    \item For all vertices $\vecv \in L_c(\vecc') / L_c(\vecc)$ it holds that
    $Q_{\vecv}^{\vecc'}(\vecx) = \frac{\partial Q_{\vecv}^{\vecc'}(\vecx)}{\partial x_i} = \frac{\partial^2 Q_{\vecv}^{\vecc'}(\vecx)}{\partial x_i~\partial x_j} = 0$.
  \end{enumerate}
\end{lemma}

\begin{proof}
  \begin{enumerate}
    \item Let $\vecv \in L_c(\vecc) \cap L_c(\vecc')$ then we have that
    \begin{enumerate}
      \item \underline{$Q_{\vecv}^{\vecc}(\vecx) = Q_{\vecv}^{\vecc'}(\vecx)$.}
      By Lemma \ref{l:1} we get that the canonical representation
      $\vecp_{\vecx}^{\vecc} = \vecp_{\vecx}^{\vecc'}$. Since $Q_{\vecv}^{\vecc}(\vecx)$ is
      a function of the canonical representation $\vecp_{\vecx}^{\vecc}$
      (see Definition~\ref{d:_coeff}), it holds that
      $Q_{\vecv}^{\vecc}(\vecx) = Q_{\vecv}^{\vecc'}(\vecx)$ for all vertices
      $\vecv \in L_c(\vecc) \cap L_c(\vecc')$.
      \item
      \underline{$\frac{\partial Q_{\vecv}^{\vecc}(\vecx)}{\partial x_i} = \frac{\partial Q_{\vecv}^{\vecc'}(\vecx)}{\partial x_i}$.}
      For $i \neq m$, we get that
      $\frac{\partial Q_{\vecv}^{\vecc}(\vecx)}{\partial x_i} = \frac{1}{t_i -s_i} \frac{\partial Q_{\vecv}^{\vecc}(\vecx)}{\partial p_i} = \frac{1}{t'_i -s'_i} \frac{\partial Q_{\vecv}^{\vecc'}(\vecx)}{\partial p'_i} = \frac{\partial Q_{\vecv}^{\vecc'}(\vecx)}{\partial x_i}$
      since $t_i = t_i'$ and $s_i = s_i'$ for all $i \neq m$. The latter argument cannot be
      applied for the $m$-th coordinate since $t_m - s_m = -(t_m' - s_m')$. However since
      $\vecx$ belongs to the boundary of both the cubelets $L(\vecc)$ and $L(\vecc')$ it is
      implied that $p_m = p_m'$ is either $0$ or $1$, meaning that
      $\frac{\partial Q_{\vecv}^{\vecc}(\vecx)}{\partial x_m} = \frac{\partial Q_{\vecv}^{\vecc'}(\vecx)}{\partial x_m} = 0$
      since $S'(0) = S'(1) = 0$ from Lemma \ref{lem:stepFunctions}.
      \item
      \underline{$\frac{\partial^2 Q_{\vecv}^{\vecc}(\vecx)}{\partial x_i~\partial x_j} = \frac{\partial^2 Q_{\vecv}^{\vecc'}(\vecx)}{\partial x_i~\partial x_j}$.}
      For $i, j \neq m$, we get that
      $\frac{\partial^2 Q_{\vecv}^{\vecc}(\vecx)}{\partial x_i~\partial x_j} = \frac{1}{t_i -s_i} \frac{1}{t_j -s_j} \frac{\partial^2 Q_{\vecv}^{\vecc}(\vecx)}{\partial  p_i~\partial p_j} = \frac{1}{t'_i - s'_i}\frac{1}{t'_j - s'_j} \frac{\partial Q_{\vecv}^{\vecc'}(\vecx)}{\partial p'_i~\partial p'_j} = \frac{\partial^2 Q_{\vecv}^{\vecc'}(\vecx)}{\partial x_i~\partial x_j}$
      since $t_i = t_i'$ and $s_i = s_i'$ for all $i \neq m$. As in the previous case,
      $p_m = p_m'$ equals either $0$ or $1$. As a result,
      $\frac{\partial^2 Q_{\vecv}^{\vecc}(\vecx)}{\partial x_m ~\partial x_j} = \frac{\partial^2 Q_{\vecv}^{\vecc'}(\vecx)}{\partial x_m~\partial x_j} = 0$ since
      $S'(0) = S'(1) = S''(0) = S''(1) = 0$ by Lemma \ref{lem:stepFunctions}.
    \end{enumerate}
    \item Since $\vecv \in L_c(\vecc) \setminus L_c(\vecc')$, we get that $v_{m} = c_m$. In
    case $c_m$ is even, we get that $s_m = c_m = v_m$ and thus the coordinate the coordinate
    $m$ belongs in the set $A_{\vecv}^{\vecc}$. Since $\vecx$ coincides with one of the
    corners in $L_c(\vecc) \setminus L_c(\vecc')$ we get that $p_m = 1$ which combined with
    the fact that $m \in A_{\vecv}^{\vecc}$ implies that $Q_{\vecv}^{\vecc}(\vecx) = 0$
    (see Definition~\ref{d:cof}). Then by Lemma \ref{l:derivation1},
    $\frac{\partial Q_{\vecv}^{\vecc'}(\vecx)}{\partial x_i} = \frac{\partial^2 Q_{\vecv}^{\vecc'}(\vecx)}{\partial x_i~\partial x_j} = 0$.
    In case is odd, we get that $s_m = c_m + 1$. The latter combined with the fact that
    $v_m = c_m$ implies that the $m$-th coordinate belongs in $B_{\vecv}^{\vecc}$. Now
    $p_m = 0$ and by Definition~\ref{d:cof}, $Q_{\vecv}^{\vecc}(\vecx) = 0$. Then again by
    Lemma \ref{l:derivation1},
    $\frac{\partial Q_{\vecv}^{\vecc'}(\vecx)}{\partial x_i} = \frac{\partial^2 Q_{\vecv}^{\vecc'}(\vecx)}{\partial x_i~\partial x_j} = 0$.
    \item This case follows with the same reasoning with previous case $2$.
  \end{enumerate}
\end{proof}

  We are now ready to prove Lemma \ref{l:continuity1}.

\begin{proof}[Proof of Lemma \ref{l:continuity1}]
  \begin{enumerate}
    \item Let $\vecv \in L_c(\vecc) \cap L_c(\vecc')$. There exists a sequence of corners
    \[ \vecc = \vecc^{(1)}, \ldots, \vecc^{(m)} = \vecc'\]
    \noindent such that $\norm{\vecc^{(j)} - \vecc^{(j + 1)}}_1 = 1$ and $\vecv \in L_c(\vecc^j)$
    for all $j \in [m]$. By Lemma \ref{l:continuity} we get that,
    \begin{enumerate}
      \item $Q_{\vecv}^{\vecc^{(j)}}(\vecx) = Q_{\vecv}^{\vecc^{(j + 1)}}(\vecx)$.
      \item $\frac{\partial Q_{\vecv}^{\vecc^{(j)}}(\vecx)}{\partial x_i} = \frac{\partial Q_{\vecv}^{\vecc^{(j + 1)}}(\vecx)}{\partial x_i}$.
      \item $\frac{\partial^2 Q_{\vecv}^{\vecc^{(j)}}(\vecx)}{\partial x_i~\partial x_j} = \frac{\partial Q_{\vecv}^{\vecc^{(j + 1)}}(\vecx)}{\partial x_i~\partial x_j}$.
    \end{enumerate}
    which implies Case $1$ of Lemma \ref{l:continuity1}.
    \item Let $\vecv \in L_c(\vecc) \setminus L_c(\vecc')$. There exists a sequence of
    corners $\vecc = \vecc^{(1)} \ldots, \vecc^{(i)}$ such that
    $\norm{\vecc^{(j)} - \vecc^{(j + 1)}}_1 = 1$ and $\vecv \notin L_c{\vecc^{(i)}}$ and
    $\vecv \in L_c(\vecc^{(j)})$ for all $j < i$. By case $2$ of Lemma \ref{l:continuity} we
    get that
    $Q_{\vecv}^{\vecc^{(i - 1)}}(\vecx) = \frac{\partial Q_{\vecv}^{\vecc^{(i - 1)}}(\vecx)}{\partial x_i} = \frac{\partial^2 Q_{\vecv}^{\vecc^{(i - 1)}}(\vecx)}{\partial x_i~\partial x_j} = 0$.
    Then case $2$ of Lemma \ref{l:continuity1} follows by case $1$ of Lemma
    \ref{l:continuity}.
    \item Similarly with case $2$.
  \end{enumerate}
\end{proof}

\addtocontents{toc}{\protect\setcounter{tocdepth}{1}}
\subsection{Proof of Lemma \ref{l:bounding_gradients}} \label{ss:gradinet_bounds}
\addtocontents{toc}{\protect\setcounter{tocdepth}{2}}

  We start this section with some fundamental properties of the smooth step function
$S_{\infty}$ that are more fine-grained than the properties we presented in Lemma
\ref{lem:stepFunctions}.

\begin{lemma} \label{l:infty}
    For $d \ge 10$ there exists a universal constant $c > 0$ such that the following
  statements hold.
  \begin{enumerate}
    \item If $x \ge 1/d$ then $S_{\infty}(x) \ge c \cdot 2^{-d}$.
    \item If $x \le 1/d$ then $S_{\infty}'(x) \le c \cdot d^2 \cdot 2^{-d}$.
    \item If $x \ge 1/d$ then $\frac{S_{\infty}'(x)}{S_{\infty}(x)} \le c \cdot d^2$.
    \item If $x \le 1/d$ then $\abs{S''_{\infty}(x)} \le c \cdot d^4 \cdot 2^{-d}$.
    \item If $x \ge 1/d$ then $\frac{\abs{S''_{\infty}(x)}}{S_{\infty}(x)} \le c \cdot d^4$.
  \end{enumerate}
\end{lemma}

\begin{proof}
    We compute the derivative of $S_{\infty}$ and we have that
  \[ S'_{\infty}(x) = \ln(2) S_{\infty}(x) S_{\infty}(1 - x) \p{\frac{1}{x^2} + \frac{1}{(1 - x)^2}} \]
  from which we immediately get $S'_{\infty}(x) \ge 0$. Then we can compute the second derivative
  of $S_{\infty}$ as follows
  \[ S''_{\infty}(x) = \ln(2) S_{\infty}(x) S_{\infty}(1 - x) \cdot \]
  \[ \cdot \p{\ln(2) \p{S_{\infty}(1 - x) - S_{\infty}(x)}\p{\frac{1}{x^2} + \frac{1}{(1 - x)^2}}^2 - 2 \p{\frac{1}{x^3} - \frac{1}{\p{1 - x}^3}}}. \]
  We next want to prove that $S''_{\infty}(x) \ge 0$ for $x \le 1/10$. To see this observe that
  $1 - 2 \cdot S_{\infty}(x) \ge 1/2$ for $x \le 1/d$ and therefore
  \[ S''_{\infty}(x) \ge \frac{\ln(2)}{x^3} S_{\infty}(x) S_{\infty}(1 - x) \p{\frac{\ln(2)}{2 x} - 2} \]
  hence for $x \le 4 / \ln(2)$ it holds that $S''_{\infty}(x) \ge 0$. By similar but more tedious
  calculations we can conclude that $S'''_{\infty}(x) \ge 0$ for $x \le 1/10$. Hence in the
  interval $x \in [0, 1/10]$ all the functions $S_{\infty}$, $S'_{\infty}$, $S''_{\infty}$ are
  all increasing functions of $x$.

    Next we show that the function $h(x) = 2^{-1/x} + 2^{-1/(1-x)}$ is upper and lower bounded.
  First observe that $h(x) \ge \max\{2^{-1/x}, 2^{-1/(1-x)}\}$. Now if we set $t(x) = 2^{-1/x}$
  then $t'(x) = \ln(2) t(x) / x^2$ and hence $t(x) \ge t(1/2) = 1/4$ for $x \ge 1/2$. The same
  way we can prove that $ 2^{-1/(1-x)} \ge 1/4$ for $x \le 1/2$. Therefore $h(x) \ge 1/4$ for all
  $x \in [0, 1]$. Also it is not hard to see that $2^{-1/x} \le 1/2$ and $2^{-1/(1-x)} \le 1/2$
  which implies $h(x) \le 1$. Hence overall we have that $h(x) \in [1/4, 1]$ for all
  $x \in [0, 1]$. We are now ready to prove the statements.
  \begin{enumerate}
    \item We have shown that $S'_{\infty}(x) \ge 0$ for all $x \in [0, 1]$. Hence $S_{\infty}$ is
          an increasing function and therefore $S_{\infty}(x) \ge S_{\infty}(1/d)$ for
          $x \ge 1/d$. Now we have that $S_{\infty}(1/d) = 2^{-d}/h(1/d) \ge 2^{-d}$.
    \item Since $S'_{\infty}(x)$ is increasing for $x \in [0, 1/10]$, we have that
          $S'_{\infty}(x) \le S'_{\infty}(1/d)$ for $x \le 1/d$ and therefore
          \[ S'_{\infty}(x) \le \ln(2) S_{\infty}(1 - 1/d) S_{\infty}(1/d) \p{d^2 + \frac{1}{\p{1 - \frac{1}{d}}^2}} \]
          \[ \le 2 \ln(2) \frac{2^{-d}}{h(1/d)} \le 8 \ln(2) 2^{-d}. \]
    \item We have that for $x \le 1/d$
          \[ \frac{S'_{\infty}(x)}{S_{\infty}(x)} = \ln(2) S_{\infty}(1 - x) \p{\frac{1}{x^2} + \frac{1}{(1 - x)^2}} \le 2 \ln(2) \frac{1}{x^2} \le 2 \ln(2) d^2. \]
    \item Follows directly from the statement 1., the fact that $S''_{\infty}(x)$ is increasing
          for $x \in [0, 1/10]$ and the above expression of $S''_{\infty}$ this statement
          follows.
    \item This statement follows using the same reasoning with statement 3.
  \end{enumerate}
\end{proof}

In this section we establish the bounds
on the gradient and the hessian of $\mathrm{P}_{\vecv}(\vecx)$. These bounds are formally stated in Lemma \ref{l:bounding_gradients} the proof of which is the main goal of the section.

\bgradients*

\noindent In order to prove Lemma \ref{l:bounding_gradients}. We first introduce several technical lemmas.

\begin{lemma}\label{l:bounding_one_gradient}
    Let $\vecx \in [0,1]^d$ lying in cublet $L(\vecc)$, with
  $\vecc \in \p{\nm{N}}^d$ and let $\vecp_{\vecx}^{\vecc} = (p_1, \ldots, p_d)$
  be the canonical representation of $\vecx$. Then for all vertices
  $\vecv \in L_c(\vecc)$, it holds that
  \[  \abs{\frac{\partial Q_{\vecv}^{\vecc}(\vecx)}{\partial p_i}} \leq \Theta(d^{11}) \cdot \sum_{\vecv \in V_c} Q_{\vecv}^{\vecc}(\vecx). \]
\end{lemma}

\begin{proof}
    To simplify notation we use $Q_{\vecv}(\vecx)$ instead of $Q_{\vecv}^{\vecc}(\vecx)$,
  $A$ instead of $A_{\vecv}^{\vecc}$ and $B$ instead of $B_{\vecv}^{\vecc}$ for the rest of the
  proof. Without loss of generality we assume that for all $j \in A$ and $\ell \in B$,
  $p_\ell > p_j$ since otherwise $\frac{\partial Q_{\vecv}^{\vecc}(\vecx)}{\partial p_i} = 0$
  trivially by the Definition \ref{d:cof}. Let $i \in B$ (symmetrically for $i \in A$) then,
  \begin{align*}
    & \abs{\frac{\partial Q_{\vecv}^{\vecc}(\vecx)}{\partial p_i}} = \\
    & ~~~~~~~~~ =   ~~ \prod_{\ell \neq i} \prod_{j \in A} S_\infty(S(p_{\ell}) - S(p_{j})) \cdot \left[ \sum_{j \in A} \abs{S'_{\infty}(S(p_i) - S(p_j))} \prod_{j' \in A /\{j\}} S_{\infty}(S(p_i) - S(p_{j'})) \right]S'(p_i) \\
    & ~~~~~~~~~ \le ~~ 6 \sum_{j \in A} \abs{S'_{\infty}(S(p_i) - S(p_j))} \cdot \prod_{(j',\ell) \ne (j,i)} S_{\infty}(S(p_{\ell}) - S(p_{j'}))
  \end{align*}
  where the last inequality follows by the fact that $\abs{S'(\cdot)} \le 6$. Since
  $\abs{A} \leq d$ the proof of the lemma will be completed if we are able to show that for any
  $j \in A$, it holds that
  \[ \abs{S'_{\infty}(S(p_i) - S(p_j))} \cdot \prod_{(j',\ell) \neq (j,i)} S_{\infty}(S(p_{\ell}) - S(p_{j'})) \le \Theta(d^{10}) \cdot \sum_{\vecv' \in L_c(\vecc)}Q_{\vecv'}(\vecx)  \]
  In case $S(p_i) - S(p_j) \geq 1/d^5$ then by case $3.$ of Lemma~\ref{l:infty} we get that
  $\abs{S_{\infty}'(S(p_i) - S(p_j))} \leq c \cdot d^{10} \cdot S_{\infty}(S(p_i) - S(p_j))$,
  which implies gthe following
  \begin{align*}
    & \abs{S'_{\infty}(S(p_i) - S(p_j))} \cdot \prod_{(j',\ell) \neq (j,i)} S_{\infty}(S(p_{\ell}) - S(p_{j'})) \le \\
    & ~~~~~~~~~~ \le c \cdot d^{10} \cdot S_{\infty}(S(p_i) - S(p_j)) \cdot \prod_{(j',\ell) \neq (j, i)} S_{\infty}(S(p_{\ell}) - S(p_{j'})) \\
    & ~~~~~~~~~~ = c \cdot d^{10} \cdot Q_{\vecv}(\vecx) \\
    & ~~~~~~~~~~ \le c \cdot d^{10} \cdot \sum_{\vecv' \in L_c(\vecc)}Q_{\vecv'}(\vecx)
  \end{align*}
  \noindent Now consider the case where $S(p_i) - S(p_j) \le 1/d^5$. Using case $2.$ of
  Lemma \ref{l:infty}, we have that
  \[ \abs{S'_{\infty}(S(p_i) - S(p_j))} \cdot \prod_{(j',\ell) \neq (j,i)} S_{\infty}(S(p_{\ell}) - S(p_{j'})) \le \abs{S'_{\infty}(S(p_i) - S(p_j))} \le \Theta(d^{10} \cdot 2^{-d^5}) \]
  \noindent Consider the sequence of points in the $[0, 1]$ interval $0, p_1, \ldots, p_d, 1$.
  There always exist two consecutive points with distance greater that $1/(d + 1)$. As a result,
  there exists $\vecv^\ast \in L_c(\vecc)$ such that $p_\ell - p_j \ge 1/(d + 1)$ for all
  $\ell \in B_{\vecv^{\ast}}$ and $j \in A_{\vecv^{\ast}}$. Then
  $S(p_\ell) - S(p_j) \ge 1/(d + 1)^2$ and by case $1.$ of Lemma \ref{l:infty},
  $S_{\infty}(S(p_\ell) - S(p_j)) \ge c 2^{-(d+1)^2}$. If we also use the fact that
  $\abs{A_{\vecv^{\ast}}} \cdot \abs{B_{\vecv^{\ast}}} \le d^2$, we get that
  \[ Q_{\vecv^\ast}(\vecx) \ge (c \cdot 2^{-(d + 1)^2})^{d^2} = c^{d^2} 2^{-(d+1)^2 \cdot d^2}. \]
  Then it holds that
  \[ \frac{1}{Q_{\vecv^\ast}(\vecx)} \cdot \abs{S'_{\infty}(S(p_i) - S(p_j))} \cdot  \prod_{(j',\ell) \neq (j,i)} S_{\infty}(S(p_{\ell}) - S(p_{j'})) \le \]
  \[ \le \Theta\p{d^{10} \cdot \left( (1/c) \cdot 2^{-d^3 + (d + 1)^2}  \right)^{d^2}} \le \Theta(d^{10}). \]
  Combining the later with the discussion in the rest of the proof the lemma follows.
\end{proof}

\begin{lemma}\label{l:derivation}
    For any vertex $\vecv \in \p{\nm{N}}^d$ it holds that
  $\abs{\frac{\partial \P_{\vecv}(\vecx)}{\partial x_i}} \leq \Theta\p{d^{12}/\delta}$.
\end{lemma}

\begin{proof}
    To simplify notation we use $Q_{\vecv}(\vecx)$ instead of $Q_{\vecv}^{\vecc}(\vecx)$ for the
  rest of the proof. Without loss of generality we assume that $\vecx$ lies on a cubelet
  $L(\vecc)$ with $\vecc \in \p{\nm{N}}^d$ and $\vecv \in L_c(\vecc)$, since
  otherwise $\frac{\partial \P_{\vecv}(\vecx)}{\partial x_i} = 0$. Let
  $\vecp_{\vecx}^{\vecc} = (p_1, \ldots, p_d)$ be the canonical representation
  of $\vecx$ in the cubelet $L(\vecc)$. Then it holds that
  \begin{eqnarray*}
    \abs{\frac{\partial \mathrm{P}_{\vecv}(\vecx)}{\partial p_i}}
    &  =  & \frac{\abs{\frac{\partial Q_{\vecv}(\vecx)}{\partial p_i} \cdot \left[\sum_{\vecv' \in L_c(\vecc)}Q_{\vecv'}(\vecx) \right] - Q_{\vecv}(\vecx) \cdot \left[\sum_{\vecv' \in L_c(\vecc)} \frac{\partial Q_{\vecv'}(\vecx)}{\partial p_i} \right]}}{(\sum_{\vecv' \in L_c(\vecc)} Q_{\vecv'}(\vecx))^2} \\
    & \le & \frac{\abs{\frac{\partial Q_{\vecv}(\vecx)}{\partial p_i}}}{\sum_{\vecv' \in L_c(\vecc)} Q_{\vecv'}(\vecx)} + \frac{ \sum_{\vecv' \in L_c(\vecc)} \abs{\frac{\partial Q_{\vecv'}(\vecx)}{\partial p_i}}}{\sum_{\vecv' \in L_c(\vecc)} Q_{\vecv'}(\vecx)} \\
    & \le & (d + 2) \cdot \Theta(d^{11}) = \Theta(d^{12})
  \end{eqnarray*}
  \noindent where the last inequality follows by Lemma \ref{l:bounding_one_gradient} and the fact
  that at most $d + 1$ vertices $\vecv$ of $L_c(\vecc)$ have non-zero gradient as we have proved
  in Lemma \ref{l:well_defined}. Then the proof of Lemma \ref{l:derivation} follows by the fact
  that $p_i = \frac{x_i - s_i}{t_i - s_i}$.
\end{proof}

\begin{lemma}\label{l:2nd-derivative}
    Let $\vecc \in \p{\nm{N}}^d$ and $\vecv \in L_c(\vecc)$ then it holds that
  $\abs{\frac{\partial^2 \mathrm{Q}_{\vecv}^{\vecc}(\vecx)}{\partial p_i~\partial p_j} } \le \Theta(d^{22}) \cdot \sum_{\vecv \in R_{\vecc}(\vecx)} Q_{\vecv}^{\vecc}(\vecx)$.
\end{lemma}

\begin{proof}
    To simplify the notation we use $CS(p_\ell - p_m)$ to denote
  $S_{\infty}(S(p_\ell) - S(p_m))$, $CS'(p_\ell - p_m)$ to denote
  $\abs{S'_{\infty}(S(p_\ell) - S(p_m))}$, $A$ to denote $A_{\vecv}^{\vecc}$ and $B$ to denote
  $B_{\vecv}^{\vecc}$ for the rest of the proof. As in Lemma \ref{l:derivation}, we assume that
  $p_{\ell} > p_{m}$ for all $\ell \in B$ and $m \in A$ since otherwise
  $\frac{\partial^2 \mathrm{Q}_{\vecv}(\vecx)}{\partial p_i~\partial p_j} = 0$. We have the
  following cases for the indices $i$ and $j$
  \begin{enumerate}
    \item[$\blacktriangleright$] If $i,j \in B$ then
    \begin{align*}
      & ~~ \abs{\frac{\partial^2 \mathrm{Q}_{\vecv}(\vecx)}{\partial p_i~\partial p_j}} = \\
      & =   \sum_{m_1, m_2 \in A} CS'(p_i - p_{m_1}) CS'(p_j - p_{m_2}) \cdot \prod_{(m,\ell) \neq \{(m_1,i),(m_2,j)\}} CS(p_\ell - p_m) \cdot S'(p_i) S'(p_j) \\
      & \le 36 \sum_{m_1,m_2 \in A} \underbrace{CS'(p_i - p_{m_1}) CS'(p_j - p_{m_2}) \cdot \prod_{(m,\ell)\neq\{(m_1,i),(m_2,j)\}} CS(p_\ell - p_m)}_{\triangleq U(i,j)}.
    \end{align*}
    If additionally it holds that $S(p_i) - S(p_{m_1}) \le 1/d^5$ or
    $S(p_j) - S(p_{m_2}) \leq 1/d^5$, then by the case $2.$ of Lemma \ref{l:infty}, we have that
    \[ U(i,j) \leq CS'(p_i - p_{m_1}) \cdot CS'(p_j - p_{m_2}) \leq \Theta( d^{10} e^{-d^5}). \]
    The latter follows from the fact that the function $S'_{\infty}(\cdot)$ is bounded in the
    $[0, 1]$ interval and that $CS(p_\ell - p_m) \leq 1$. With the exact same arguments as in
    Lemma \ref{l:bounding_one_gradient}, we hence get that
    \[ CS'(p_i - p_{m_1}) CS'(p_j - p_{m_2}) \cdot\Pi_{(m,\ell)\neq\{(m_1,i),(m_2,j)\}} CS(p_\ell - p_m) \leq \Theta(d^{10}) \sum_{\vecv' \in L_c(\vecc)} Q_{\vecv'}^{\vecc}(\vecx). \]
    Thus $\abs{\frac{\partial^2 Q_{\vecv}(\vecx)}{\partial p_i~\partial p_j}} \leq \Theta(d^{12})\sum_{\vecv' \in L_c(\vecc)} Q_{\vecv'}^{\vecc}(\vecx)$.

    On the other hand if $S(p_i) - S(p_{m_1}) \geq 1/d^5$ and $S(p_j) - S(p_{m_2}) \geq 1/d^5$
    then by case $1.$ of Lemma \ref{l:infty},
    $CS'(p_i - p_{m_1}) \leq c \cdot d^{10} \cdot CS(p_i - p_{m_1})$ and
    $CS'(p_j - p_{m_2}) \leq c \cdot d^{10} \cdot CS(p_j - p_{m_2})$ and thus
    $U(i,j) \leq \Theta(d^{20})\cdot Q_{\vecv}^{\vecc}(\vecx)$. Overall we get that
    $\abs{\frac{\partial^2 Q_{\vecv}(\vecx)}{\partial p_i~\partial p_j}} \le \Theta(d^{22}) \cdot \sum_{\vecv' \in R_{\vecc}(\vecx)}Q_{\vecv'}^{\vecc}(\vecx)$.
    \item[$\blacktriangleright$] If $i \in B$ and $j \in A$ then
    \begin{align*}
      & \abs{\frac{\partial^2 \mathrm{Q}_{\vecv}(\vecx)}{\partial p_i~\partial p_j}} \le \\
      & \le \sum_{m_1\in A,\ell_2 \in B} CS'(p_i - p_{m_1}) CS'(p_{\ell_2} - p_{j}) \cdot \prod_{(m,\ell)\neq\{(i,m_1),(\ell_2,j)\}} CS(p_\ell - p_m) \cdot S'(p_i)S'(p_j) \\
      & ~~~~~ + \abs{CS^{''}(p_i - p_{j}) \cdot \prod_{(m,\ell)\neq(i,j)} CS(p_\ell - p_m) \cdot S'(p_i) S'(p_j)} \\
      & \le \Theta(d^{22}) \sum_{\vecv \in L_c(\vecc)}Q_{\vecv}^{\vecc}(\vecx) + 36 \underbrace{\abs{CS^{''}(p_i - p_{j}) \cdot \prod_{(m,\ell)\neq(i,j)} CS(p_\ell - p_m)}}_{Q^{''}(\vecx)}.
    \end{align*}
    In case $S(p_i) - S(p_j) \geq 1/d^5$ then by case $4.$ of Lemma \ref{l:infty}, we get that
    $\abs{CS^{''}(p_i - p_{j})} \leq c d^{20} \cdot CS(p_i - p_{j})$ which implies that
    $Q^{''} \leq \Theta(d^{20}) \cdot Q_{\vecv}^{\vecc}(\vecx)$.

    On the other hand if $S(p_i) - S(p_j) \leq 1/d^5$ then by case $5.$ of Lemma \ref{l:infty},
    we get that $Q^{''} \leq \abs{CS^{''}(p_i - p_j)} \leq c \cdot d^{20}e^{-d^5}$. As in the
    proof of Lemma~\ref{l:bounding_one_gradient}, there exists a vertex
    $\vecv^\ast \in R_{\vecc}(\vecx)$ such that
    $Q_{\vecv^{\ast}}^{\vecc}(\vecx) \geq c^{d^2}e^{-(d+1)^2d^2}$ and thus
    $Q^{''} \leq \Theta(d^{20})\sum_{\vecv \in L_c(\vecc)}Q_{\vecv}^{\vecc}(\vecx)$. Overall we
    get that
    \[ \abs{\frac{\partial^2 \mathrm{Q}_{\vecv}(\vecx)}{\partial p_i~\partial p_j}} \leq \Theta(d^{22})\sum_{\vecv \in L_c(\vecc)} Q_{\vecv}^{\vecc}(\vecx). \]
    \item[$\blacktriangleright$] If $i = j \in B$ then
    \begin{align*}
      & \abs{\frac{\partial^2 \mathrm{Q}_{\vecv}(\vecx)}{\partial^2 p_i}} \le \\
      & \le \sum_{m_1, m_2 \in A}\abs{ CS'(p_i - p_{m_1}) CS'(p_i - p_{m_2}) \cdot \prod_{(m,\ell)\neq\{(m_1,i), (m_2,i)\}} CS(p_\ell - p_m) \cdot S'(p_i) S'(p_i)} \\
      & ~~~~ + \sum_{m_1 \in A}\abs{CS''(p_i - p_{m_1}) \cdot \prod_{(m,\ell)\neq(m_1,\ell)}CS(p_\ell - p_m)S'(p_i)S'(p_i)} \\
      & \le \Theta(d^{22} + d \cdot d^{20}) \cdot \sum_{\vecv \in L_c(\vecc)}Q_{\vecv}^{\vecc}(\vecx).
    \end{align*}
  \end{enumerate}
  If we combine all the above cases then the Lemma follows.
\end{proof}

\begin{lemma}  \label{l:derivation2}
    For any vertex $\vecv \in \p{\nm{N}}^d$, it holds that
  $\abs{\frac{\partial^2 \P_{\vecv}(\vecx)}{\partial x_i~\partial x_j}} \leq \Theta(d^{24}/\delta^2)$.
\end{lemma}

\begin{proof}
    Without loss of generality we assume that $\vecv \in L_c(\vecc)$, where $\vecc \in \p{\nm{N - 1}}^d$
  such that $\vecx \in L(\vecc)$, since otherwise
  $\frac{\partial^2 \P_{\vecv}(\vecx)}{\partial x_i~\partial x_j} = 0$.
  \begin{eqnarray*}
    \frac{\partial^2 \mathrm{P}_{\vecv}(\vecx)}{\partial p_i~\partial p_j}
    &  =  & \frac{\partial^2 Q_{\vecv}(\vecx)}{\partial p_i~\partial p_j} \left(\sum_{\vecv' \in L_c(\vecc)} Q_{\vecv'}(\vecx)\right)^3 \cdot \frac{1}{\left (\sum_{\vecv' \in L_c(\vecc)} Q_{\vecv'}(\vecx) \right)^4} \\
    & ~~~~~ +  & \frac{\partial Q_{\vecv}(\vecx)}{\partial p_i} \sum_{\vecv' \in L_c(\vecc)}\frac{\partial Q_{\vecv'}(\vecx)}{\partial p_j} \left( \sum_{\vecv' \in L_c(\vecc)} Q_{\vecv'}(\vecx)\right)^2\cdot \frac{1}{\left (\sum_{\vecv' \in L_c(\vecc)} Q_{\vecv'}(\vecx) \right)^4} \\
    & ~~~~~ -  & \frac{\partial Q_{\vecv'}(\vecx)}{\partial p_j} \sum_{\vecv' \in L_c(\vecc)}\frac{\partial Q_{\vecv'}(\vecx)}{\partial p_i} \left( \sum_{\vecv' \in L_c(\vecc)} Q_{\vecv'}(\vecx)\right)^2\cdot \frac{1}{\left (\sum_{\vecv' \in L_c(\vecc)} Q_{\vecv'}(\vecx) \right)^4} \\
    & ~~~~~ -  & Q_{\vecv}(\vecx)\sum_{\vecv' \in L_c(\vecc)}\frac{\partial^2 Q_{\vecv'}(\vecx)}{\partial p_i~\partial p_j}\left(\sum_{\vecv' \in L_c(\vecc)}Q_{\vecv'}(\vecx)\right)^2\cdot \frac{1}{\left (\sum_{\vecv' \in L_c(\vecc)}Q_{\vecv'}(\vecx) \right)^4}\\
    & ~~~~~ -  & \frac{\partial Q_{\vecv}(\vecx)}{\partial p_i} \sum_{\vecv' \in L_c(\vecc)} Q_{\vecv'}(\vecx) \cdot 2 \sum_{\vecv' \in L_c(\vecc)} Q_{\vecv'}(\vecx) \sum_{\vecv' \in L_c(\vecc)}\frac{\partial Q_{\vecv'}(\vecx)}{\partial p_j}\cdot \frac{1}{\left (\sum_{\vecv' \in L_c(\vecc)} Q_{\vecv'}(\vecx) \right)^4}\\
    & ~~~~~ +  & Q_{\vecv}(\vecx) \sum_{\vecv' \in L_c(\vecc)} \frac{\partial Q_{\vecv'}(\vecx)}{\partial p_i} \cdot 2 \sum_{\vecv' \in L_c(\vecc)} Q_{\vecv'}(\vecx) \sum_{\vecv' \in L_c(\vecc)} \frac{\partial Q_{\vecv'}(\vecx)}{\partial p_j} \cdot \frac{1}{\left (\sum_{\vecv' \in L_c(\vecc)} Q_{\vecv'}(\vecx) \right)^4}
  \end{eqnarray*}
  Using Lemma \ref{l:2nd-derivative} and Lemma \ref{l:bounding_one_gradient} we can bound every
  term in the above expression and hence we get that
  $\abs{\frac{\partial^2 \mathrm{P}_{\vecv}(\vecx)}{\partial p_i~\partial p_j}} \leq \Theta(d^{24})$. Then the lemma follows from the fact that
  $\frac{\partial p_i}{\partial x_i} = 1/\delta$.
\end{proof}

\noindent Finally using Lemma \ref{l:derivation} and Lemma \ref{l:derivation2} we get the proof
of Lemma \ref{l:bounding_gradients}.

\addtocontents{toc}{\protect\setcounter{tocdepth}{1}}
\subsection{Proof of Lemma \ref{l:positive_corners_boundary}} \label{sec:proof:l:positive_corners_boundary}
\addtocontents{toc}{\protect\setcounter{tocdepth}{2}}

  Let $0 \leq x_i < 1/(N - 1)$ and $\vecc = (c_1,\ldots,c_i,\ldots,c_d)$ denote down-left corner of the
cubelet $R(\vecx)$ at which $\vecx \in [0, 1]^d$ lies, i.e. $\vecx \in L(\vecc)$. Since
$\vecx \le 1/(N - 1)$, this means that $c_i = 0$. By the definition of \textit{sources and targets} in
Definition \ref{d:canonical_representation}, we have that $s_i = 0$ and $t_i = 1/(N - 1)$, where
$s_i$, $t_i$ are respectively the $i$-th coordinate of the source $\vecs_{\vecc}$ and the target
$\vect_{\vecc}$ vertex. Let the canonical representation $p_{\vecx}^{\vecc} = (p_1, \ldots, p_d)$
of $\vecx$ in the cubelet $L(\vecc)$. Now partition the coordinates $[d]$ in the following sets
\[ A = \set{j \mid p_j \le p_i} ~~~\text{ and }~~~ B = \set{j \mid p_i < p_j}. \]
If $B = \varnothing$ then notice that $\P_{\vecs_{\vecc}}(\vecx) > 0$, since $p_i < 1$, by the
fact that $x_i < 1/(N - 1)$. Thus the lemma follows since $s_i = 0$. So we may assume that
$B \neq \varnothing$. In this case consider the corner $\vecv = (v_1,\ldots,v_d)$ defined as
follows
\[ v_j =  \left\{ \begin{array}{ll}
                    s_j & j \in A\\
                    t_j & j \in B \\
                  \end{array} \right.. \]
\noindent Observe that $Q_{\vecv}^{\vecc}(\vecx) > 0$ and thus $\vecv \in R_+(\vecx)$. Moreover
the coordinate $i\in A$ and therefore it holds that $v_i = s_i = 0$. This proves the first
statement of the Lemma.

  For the second statement let $1 - 1/(N - 1) \leq x_i \leq 1/(N - 1)$ and
$\vecc = (c_1,\ldots,c_i,\ldots,c_d)$ denote down-left corner of the cubelet $R(\vecx)$ at which
$\vecx \in [0, 1]^d$ lies, i.e. $\vecx \in L(\vecc)$. This means that $c_i = \frac{N - 2}{N - 1}$.
\begin{enumerate}
    \item[$\blacktriangleright$] Let $N$ be odd. In this case by the definition of sources and
    targets in Definition \ref{d:canonical_representation}, we have that $s_i = 1 - 1/(N - 1)$ and
    $t_i = 1$, where $s_i$, $t_i$ are respectively the $i$-th coordinate of the source and target
    vertex. Let $p_{\vecx}^{\vecc} = (p_1,\ldots,p_d)$ be the canonical representation of $\vecx$
    under in the cubelet $L(\vecc)$. Now partition the coordinates $[d]$ as follows,
    \[ A = \set{j \mid p_j < p_i} ~~~\text{ and }~~~ B = \set{j \mid p_i \leq p_j} \]
    If $A = \varnothing$ then notice that for the target vertex $\vect_{\vecc}$,
    $\P_{\vect_{\vecc}}(\vecx) > 0$, since $p_i > 0$, by the fact that $x_i > 1 - 1/(N - 1)$. Thus the
    lemma follows since $t_i = 1$. So we may assume that $A \neq \varnothing$. In this case
    consider the corner $\vecv = (v_1, \ldots, v_d)$ defined as follows,
    \[ v_j = \left\{ \begin{array}{ll}
                       s_j & j \in A\\
                       t_j & j \in B \\
                     \end{array} \right. \]
    \noindent Observe that $Q_{\vecv}^{\vecc}(\vecx) > 0$ and thus $\vecv \in R_+(\vecx)$.
    Moreover the coordinate $i \in B$ and thus $v_i = t_i = 1$.
    \item[$\blacktriangleright$] Let $N$ be even. In this case we have that $t_i = 1 - 1/(N - 1)$ and
    $s_i = 1$. Now partition the coordinates $[d]$ as follows,
    \[ A = \set{j \mid p_j \leq p_i} ~~~\text{ and }~~~ B = \set{j \mid p_i < p_j} \]
    If $B = \varnothing$ then notice that for the source vertex $\vecs_{\vecc}$,
    $\P_{\vecs_{\vecc}}(\vecx) > 0$, since $p_i < 1$, by the fact that $x_i > 1 - 1/(N - 1)$. Thus the
    lemma follows since $s_i = 1$. In case $B \neq \varnothing$ consider the corner
    $\vecv = (v_1, \ldots, v_d)$ defined as follows,
    \[ v_j =  \left\{ \begin{array}{ll}
                        s_j & j \in A \\
                        t_j & j \in B \\
                      \end{array} \right. \]
    \noindent Observe that $Q_{\vecv}^{\vecc}(\vecx) > 0$ and thus $\vecv \in R_+(\vecx)$.
    Moreover the coordinate $i \in A$ and thus $v_i = s_i = 1$.
\end{enumerate}
\noindent If we put together the last two cases then this implies the second statement of the
lemma.

\section{Constructing the Turing Machine -- Proof of Theorem \ref{t:Turing_Machine}} \label{sec:proof:TuringMachine}

  In this section we prove Theorem~\ref{t:Turing_Machine} establishing that both
the function $f_{\calC_l}(\vecx, \vecy)$ of Definition~\ref{d:d-payoff} and its
gradient, is computable by a polynomial-time Turing Machine. We prove
Theorem~\ref{t:Turing_Machine} through a series of Lemmas. To simplify notation
we set $b \triangleq \log 1/\eps$.

\begin{definition}
    For a $x \in \R$, we denote by $\left [ x \right ]_{b} \in \R$, a value
  represented by the $b$ bits such that
  \[ \abs{\left[ x\right]_b - x} \leq 2^{-b}.\]
\end{definition}

\begin{lemma}\label{l:computing_S_infty}
    There exist Turing Machines $M_{S_{\infty}}$, $M_{S'_{\infty}}$ that given
  input $x \in [0,1]$ and $\eps$ in binary form, compute
  $\left[S_{\infty}(x) \right]_{b}$ and $\left[S'_{\infty}(x) \right]_b$
  in time polynomial in $b = \log(1/\eps)$ and the binary representation of $x$.
\end{lemma}

\begin{proof}
    The Turing Machine $M_{S_{\infty}}$ outputs the fist $b$ bits of the
  following quantity,
  \[ W(x) = \left[\frac{1}{1 + \left[ 2^{\left[ -\frac{1}{x} + \frac{1}{x-1}\right]_{b'} } \right]_{b'} } \right]_{b'} \]
  where $b'$ will be selected sufficiently large. Notice it is possible to compute
  the above quantity due to the fact that all functions
  $\frac{1}{\gamma} + \frac{1}{\gamma - 1}$, $2^\gamma$ and
  $\frac{1}{1 + \gamma}$ can be computed with accuracy $2^{-b'}$ in polynomial
  time with respect to $b'$ and the binary representation of $\gamma$
  \cite{Brent1976}. Moreover,
  \begin{align*}
    & \abs{ \left[\frac{1}{1 + \left[ 2^{\left[ -\frac{1}{x} + \frac{1}{x-1}\right]_{b'} } \right]_{b'} } \right]_{b'} - \frac{1}{1 + 2^{ -\frac{1}{x} + \frac{1}{x-1} } } } \\
    & ~~~~ \le \abs{ \left[\frac{1}{1 + \left[ 2^{\left[ -\frac{1}{x} + \frac{1}{x-1}\right]_{b'} } \right]_{b'} } \right]_{b'} - \frac{1}{1 + \left[ 2^{\left[ -\frac{1}{x} + \frac{1}{x-1}\right]_{b'} } \right]_{b'} } } \\
    & ~~~~~~~~~~~~~~~ +  \abs{ \frac{1}{1 + \left[ 2^{\left[ -\frac{1}{x} + \frac{1}{x-1}\right]_{b'} } \right]_{b'} } - \frac{1}{1 +  2^{\left[ -\frac{1}{x} + \frac{1}{x-1}\right]_{b'} } } } \\
    & ~~~~~~~~~~~~~~~ + \abs{ \frac{1}{1 +  2^{\left[ -\frac{1}{x} + \frac{1}{x-1}\right]_{b'} } } - \frac{1}{1 + 2^{ -\frac{1}{x} + \frac{1}{x-1} } } } \\
    & ~~~~ \le 2^{-b'} + \abs{ \left[ 2^{\left[ -\frac{1}{x} + \frac{1}{x-1}\right]_{b'} } \right]_{b'} - 2^{\left[ -\frac{1}{x} + \frac{1}{x-1}\right]_{b'} }} \\
    & ~~~~~~~~~~~~~~~ + \ln 2 \abs{\left[ -\frac{1}{x} + \frac{1}{x-1}\right]_{b'} -\left( -\frac{1}{x} + \frac{1}{x-1}\right)} \\
    & ~~~~ \le  4 \cdot  2^{-b'}
  \end{align*}
  where the first inequality follows from triangle inequality and the second
  follows from the facts that $1/(1 + \gamma)$ is a $1$-Lipschitz function of
  $\gamma$ for $\gamma \ge 0$, and $1/(1 + 2^{\gamma})$ is an $\ln(2)$-Lipschitz
  function of $\gamma$ for $\gamma \ge 0$. The last inequality follows from the
  definition of $\b{\cdot}_{b'}$. Hence $W(x)$ is indeed equal to
  $\b{S_{\infty}(x)}_b$ if we choose $b' = b + 2$.
  \smallskip

    Next we explain how $M_{S'_{\infty}}$ computes
  $\left[ S'_{\infty}(x)\right]_b$. First notice that $S'_{\infty}(x)$ is equal
  to
  \[ S'_{\infty}(x) = \ln2 \cdot \frac{\frac{1}{x^2} 2^{-\frac{1}{x} + \frac{1}{x-1}} - \frac{1}{(x-1)^2} 2^{-\frac{1}{x} + \frac{1}{x-1}}}{\left(2^{-\frac{1}{x}} + 2^{\frac{1}{x-1}}\right)^2}.\]

  \noindent To describe how to compute $S'_{\infty}(x)$ we first assume that we
  have computed the following quantities. Then based on these quantities we show
  how $S'_{\infty}(x)$ can be computed and finally we consider the computation
  of these quantities.
  \begin{itemize}
    \item[$\triangleright$] $\left [\ln 2 \right]_{b'}$,
    \item[$\triangleright$] $A \leftarrow \left[\frac{1}{x^2} 2^{-\frac{1}{x} + \frac{1}{x-1}} \right]_{b'}$,
    \item[$\triangleright$] $B \leftarrow \left[\frac{1}{(x-1)^2} 2^{-\frac{1}{x} + \frac{1}{x-1}} \right]_{b'}$,
    \item[$\triangleright$] $C \leftarrow \left[ \left(2^{-\frac{1}{x}} + 2^{\frac{1}{x-1}}\right)^2 \right]_{b'}$.
  \end{itemize}
  \noindent Then $M_{S'_{\infty}}$ outputs the fist $b$ bits of the quantity
  $\left[\left[\ln2 \right]_{b'} \cdot\left[ \frac{A+B}{C}\right]_{b'} \right]_{b'}$.
  We now prove that
  \[ \abs{[\ln2]_{b'} \left[\frac{A+B}{C}\right]_{b'} - \underbrace{\ln2 \frac{A+B}{C}}_{S'_{\infty}(x)}} \leq \Theta \left( 2^{-b'}\right) \]
  \noindent Consider the function
  $g(\alpha,\beta,\gamma) = \frac{\alpha + \beta}{\gamma}$
  where $\abs{\alpha},\abs{\beta} \leq c_1$ and $\abs{\gamma} \geq c_2$ where
  $c_1, c_2$ are universal constants. Notice that $g(\alpha,\beta,\gamma)$ is
  $c$-Lipschitz for $c = \sqrt{\frac{2}{c_2^2} + \frac{2c_1}{c_2^2}}$. Since for
  sufficiently large $b'$ all the quantities
  $\abs{A}, \abs{B}, \abs{\frac{1}{x^2} 2^{-\frac{1}{x} + \frac{1}{x-1}}}, \abs{\frac{1}{(x-1)^2} 2^{-\frac{1}{x} + \frac{1}{x-1}}} \leq c_1$
  and $\abs{C}, \left(2^{-\frac{1}{x}} + 2^{\frac{1}{x-1}}\right)^2 \geq c_2$
  where $c_1, c_2$ are universal constants we get that
  \[ \abs{\left[\frac{A+B}{C} \right]_{b'} - \frac{A+B}{C}} \leq \Theta \left(2^{-b'}\right). \]
  \noindent Now consider the function $g(\alpha,\beta) = \alpha \cdot \beta$
  where $\abs{\alpha},\abs{\beta} \leq c$ where $c$ is a universal constant. In
  this case $g(\alpha,\beta)$ is $\sqrt{2}c$-Lipschitz continuous. Since for
  $b'$ sufficiently large all the quantities
  $\abs {[\ln2]_{b'}}, \abs{\left[\frac{A + B}{C}\right]_{b'}}, \ln2, \abs{\frac{A + B}{C}}$
  are bounded by a universal constant $c$, we have that,
  \[ \abs{[\ln2]_{b'} \left[\frac{A + B}{C}\right]_{b'} - \ln2 \frac{A + B}{C}} \leq \Theta \left( 2^{-b'}\right) \]

  \noindent Next we explain how the values $A,B$ and $C$ are computed while
  $\b{\ln(2)}_b'$ can easily be computed via standard techniques
  \cite{Brent1976}.
  \begin{itemize}
    \item[$\blacktriangleright$] \textbf{Computation of $\boldsymbol{A}$.} The
      Turing Machine $M_{S'_{\infty}}$ will compute $A$ by taking the first $b'$
      bits of the following quantity,
      \[ \left[2^{\left[ -\frac{1}{x} + \frac{1}{x-1} + 2\ln x / \ln 2 \right]_{b''}}\right]_{b''} \]
      where $b''$ will be taken sufficiently large. We remark that both where
      both the exponentiation and the natural logarithm can be computed in
      polynomial-time with respect to the number of accuracy bits and the binary
      representation of the input \cite{Brent1976}. The function
      $\frac{1}{x^2}2^{-\frac{1}{x} + \frac{1}{x-1}} = 2^{-\frac{1}{x} + \frac{1}{x-1} + 2\ln x / \ln 2}$
      is $c$-Lipschitz where $c$ is a universal constant. Thus,
      \[ \abs{\left[2^{\left[ -\frac{1}{x} + \frac{1}{x-1} + 2\ln x / \ln 2 \right]_{b''}}\right]_{b''} - \frac{1}{x^2} 2^{-\frac{1}{x} + \frac{1}{x-1}} } \leq \Theta(2^{-b''}). \]
    \item[$\blacktriangleright$] \textbf{Computation of $\boldsymbol{B}$.} Using
      the same arguments as for $A$.
    \item[$\blacktriangleright$] \textbf{Computation of $\boldsymbol{C}$.} To
      compute $C$ we first compute $b''$ bits of the following quantity,
      \[ \left[\frac{1}{\left[ 2^{-\left[ \frac{1}{x} \right]_{b''}} \right]_{b''} + \left[ 2^{\left[ \frac{1}{x-1}   \right]_{b''}} \right]_{b''} } \right]_{b''}^2 \]
      \noindent We first argue that
      \begin{eqnarray*}
        \abs{\left[\frac{1}{\left[ 2^{-\left[ \frac{1}{x} \right]_{b''}} \right]_{b''} + \left[ 2^{\left[ \frac{1}{x-1}   \right]_{b''}}\right]_{b''}} \right]_{b''}^2 - \left(\frac{1}{2^{-\frac{1}{x}} + 2^{ \frac{1}{x-1}}} \right)^2} &\le & \Theta\left( 2^{-b''} \right)
      \end{eqnarray*}
      The latter follows by applying the triangle inequality and the following
      $3$ inequalities.
      \begin{enumerate}
        \item
          \begin{eqnarray*}
            \abs{\left[\frac{1}{\left[ 2^{-\left[ \frac{1}{x} \right]_{b''}} \right]_{b''} + \left[ 2^{\left[ \frac{1}{x-1} \right]_{b''}} \right]_{b''}} \right]_{b''}^2 - \left(\frac{1}{\left[2^{-\left[ \frac{1}{x} \right]_{b''}}\right]_{b''} + \left[2^{\left[ \frac{1}{x - 1} \right]_{b''}} \right]_{b''} } \right)^2} & \le \Theta(2^{-b''})\\
          \end{eqnarray*}
          this holds since for $b'' > 1$ we have
          \[ \left[\frac{1}{\left(\left[ 2^{-\left[ \frac{1}{x} \right]_{b''}} \right]_{b''} + \left[ 2^{\left[ \frac{1}{x-1}   \right]_{b''}} \right]_{b''}\right)}\right]_{b''} \quad \text{ and } \quad
          \frac{1}{\left(\left[ 2^{-\left[ \frac{1}{x} \right]_{b''}} \right]_{b''} + \left[ 2^{\left[ \frac{1}{x-1}   \right]_{b''}} \right]_{b''}\right)} \]
          are both upper-bounded by $2$ while the function
          $g(\alpha) = \alpha^2$ is $4$-Lipschitz for $\abs{\alpha} \le 2$.
        \item
          \begin{eqnarray*}
            \abs{ \left(\frac{1}{\left[ 2^{-\left[ \frac{1}{x} \right]_{b''}} \right]_{b''} + \left[ 2^{\left[ \frac{1}{x-1}   \right]_{b''}} \right]_{b''}} \right)^2 - \left(\frac{1}{2^{-\left[ \frac{1}{x} \right]_{b''}} + 2^{\left[ \frac{1}{x - 1} \right]_{b''}}} \right)^2} & \le \Theta \left( 2^{-b''} \right)
          \end{eqnarray*}
          \noindent The latter follows since for $b''$ larger than a universal
          constant, both
          $\left[2 ^{-\left[ \frac{1}{x} \right]_{b''}} \right]_{b''} + \left[ 2^{\left[ \frac{1}{x-1}   \right]_{b''}} \right]_{b''}$
          and
          $2^{-\left[ \frac{1}{x}\right]_{b''}} + 2^{\left[ \frac{1}{x-1} \right]_{b''}}$
          are greater than a universal constant $c$, while the function
          $g(\alpha,\beta) = 1/(\alpha + \beta)^2$ is
          $\Theta\left(c^3\right)$-Lipschitz for $\alpha + \beta \geq c$.
        \item
          \begin{eqnarray*}
            \abs{ \left(\frac{1}{2^{-\left[ \frac{1}{x} \right]_{b''}} + 2^{\left[ \frac{1}{x-1}   \right]_{b''}}} \right)^2 - \left(\frac{1}{ 2^{-\frac{1}{x}} + 2^{\frac{1}{x-1}} } \right)^2} \le \Theta \left( 2^{-b''} \right)
          \end{eqnarray*}
          \noindent The latter follows since for $b''$ larger than a universal
          constant it holds that both the quantities in the left hand side are
          greater than a positive universal constant $c$, while the function
          $g(\alpha, \beta) = 1/(2^{-\alpha} + 2^{\beta})$ for
          $2^{-\alpha} + 2^{\beta} \geq c$, $\alpha \geq 0$, and $\beta \leq 0$
          is $\Theta \left(1/c^3 \right)$-Lipschitz.
      \end{enumerate}
  \end{itemize}
  This concludes the proof of the lemma.
\end{proof}

\begin{lemma}\label{l:computing_Q}
    There exist Turing Machines $M_Q$ and $M_{Q'}$ that given
  $\vecx \in [0, 1]^d$ and $\eps > 0$ in binary form, respectively compute
  $\b{Q_{\vecv}^{\vecc}(\vecx)}_b$ and $\b{\nabla Q_{\vecv}^{\vecc}(\vecx)}_b$
  for all vertices $\vecv \in \p{\nm{N}}^d$ with $Q_{\vecv}^{\vecc}(\vecx) > 0$,
  where $b = \log(1/\eps)$. These vertices are most $d + 1$. Moreover both $M_Q$
  and $M_{Q'}$ run in polynomial time with respect to $b$, $d$ and the binary
  representation of $\vecx$.
\end{lemma}

\begin{proof}
    Both $M_Q$, $M_{Q'}$ firsts compute the canonical representation
  $p_{\vecx}^{\vecc} \in [0,1]^d$ with the respect to the cell $R(\vecx)$ in
  which $\vecx$ lies. Such a cell $R(\vecx)$ can be computed by taking the first
  $(\log N + 1)$-bits at each coordinate of $\vecx$. The source vertex
  $\vecs^{\vecc} = (s_1, \ldots, s_d)$ and the target vertex
  $\vect^{\vecc} = (t_1, \ldots, t_d)$ with respect to $R(\vecx)$ are also
  computed. Once this is done we are only interested in vertices
  $\vecv \in R_{\vecc}(\vecx)$ for which
  \[ p_{\ell} > p_j~~~~\text{ for all } \ell \in A_{\vecv}^{\vecc}, j \in B_{\vecv}^{\vecc}\]
  \noindent since for all the other $\vecv \in \p{\nm{N}}^d$ both
  $Q_{\vecv}^{\vecc} (\vecx) = 0$ and $\nabla Q_{\vecv}^{\vecc} (\vecx) = 0$.
  These vertices, that are denoted by $R_{+}(\vecx)$, are at most $d + 1$ and
  can be computed in polynomial time.

  The vertices $\vecv \in R_+(\vecx)$ can be computed in polynomial time as
  follows: \textbf{(i)} the coordinates $p_1,\ldots,p_d$ are sorted in
  increasing order \textbf{ii)} for each $m=0, \ldots, d$ compute the vertex
  $\vecv^m \in R_{\vecc}(\vecx)$,
  \[ \vecv_{j}^m = \left\{
     \begin{array}{ll}
        s_j & \text{if coordinate $j$ belongs in the first } m \text{ coordinates wrt  the order of } \vecp_{\vecx}^{\vecc}\\
        t_j & \text{if coordinate $j$ belongs in the last } d-m \text{ coordinates wrt  the order of } \vecp_{\vecx}^{\vecc}\\
     \end{array} \right. \]

  \noindent By Definition~\ref{d:cof} it immediately follows that
  $R_{+}(\vecx) \subseteq \bigcup_{m =0}^{d}\{\vecv^m\}$ which also establish
  that $\abs{R_{+}(\vecx)} \leq d+1$.

    Once $R_{+}(\vecx)$ is computed, $M_Q$ computes for each pair
  $(\ell,j) \in B_{\vecv}^{\vecc} \times A_{\vecv}^{\vecc}$ the value of the
  number $\left[S_{\infty}(S(p_\ell) - S(p_j))\right]_{b'}$ for some accuracy
  $b'$ that we determine later but depends polynomially on $b$, $d$ and the input
  accuracy of $\vecx$. Then each $\vecv \in R_{+}(\vecx)$, $M_Q$ outputs as
  $\left[Q_{\vecv}^{\vecc}(\vecx)\right]_{b}$ the fist $b$ bits of the following
  quantity
  \[ \left[ \prod_{\ell \in B_{\vecv}^{\vecc},j \in A_{\vecv}^{\vecc}} \left[S_{\infty}(S(p_\ell) - S(p_j))\right]_{b'}\right]_{b'}\]
  where $b'$ is selected sufficiently large. We next prove that this computation
  indeed outputs $\left[Q_{\vecv}^{\vecc}(\vecx)\right]_{b}$ accurately.
  \smallskip

  \noindent To simplify notation let $S_{\infty}(S(p_\ell) - S(p_j))$ be denoted
  by $S_{\ell j}$, $A_{\vecv}^{\vecc}$ denoted by $A$ and $B_{\vecv}^{\vecc}$
  denoted by $B$. Then,
  \begin{eqnarray*}
    \abs{\left[ \Pi_{\ell \in B,j \in A}\left[S_{\ell j}\right]_{b'}\right]_{b'} - \Pi_{\ell \in B,j \in A}S_{\ell j}}
    & \le & \abs{\left[ \Pi_{\ell \in B,j \in A}\left[S_{\ell j}\right]_{b'}\right]_{b'} - \Pi_{\ell \in B,j \in A}\left[S_{\ell j}\right]_{b'}} \\
    &  +  & \abs{\Pi_{\ell \in B,j \in A}\left[S_{\ell j}\right]_{b'} - \Pi_{\ell \in B,j \in A}S_{\ell j}} \\
    & \le & 2^{-b'} + \abs{\Pi_{\ell \in B,j \in A}\left[S_{\ell j}\right]_{b'} - \Pi_{\ell \in B,j \in A}S_{\ell j}}
  \end{eqnarray*}
  \noindent Consider the function
  $g(\vecy) = \prod_{\ell \in B,j \in A} y_{\ell j}$. For $\vecy \in [0,1 + 1/d^2]^{|A| \times |B|}$,
  $\norm{\nabla g (\vecy)}_2 \leq \Theta(d)$.
  As a result, for all $\vecy, \vecz \in [0,1 + 1/d^2]^{|A| \times |B|}$,
  \[ \abs{g(\vecy) - g(\vecz)} \leq \Theta(d) \cdot \left[\sum_{\ell \in B,j \in A}(y_{\ell j} - z_{\ell j})\right]^{1/2} \]
  \noindent In case the accuracy $b' \geq \Theta(\log d)$ then
  $\left[S_{\ell j}\right]_{b'} \leq S_{\ell j} + 1/d^2 \leq 1 + 1/d^2$ and the
  above inequality applies. Thus,
  \begin{eqnarray*}
    \abs{\prod_{\ell \in B,j \in A}\left[S_{\ell j}\right]_{B'} - \Pi_{\ell \in B,j \in A}S_{\ell j}}
    & \le & \Theta(d)  \left[\sum_{\ell \in B,j \in A}\left(\left[S_{\ell j}\right]_{B'} - S_{\ell j}\right)\right]^{1/2} \\
    & \le & \Theta(d^2) \cdot 2^{-b'}
  \end{eqnarray*}
  \noindent Overall,
  $\abs{\left[ \Pi_{\ell \in B,j \in A}\left[S_{\ell j}\right]_{b'}\right]_{b'} - \Pi_{\ell \in B,j \in A}S_{\ell j}} \leq \Theta(d^2) \cdot 2^{-b'}$
  which concludes the proofof the corrected of
  $\left[Q_{\vecv}^{\vecc}(\vecx)\right]_{b}$ by selecting
  $b' = b + \Theta(\log d)$. \\
  \medskip

    In order to compute
  $\frac{\partial Q_{\vecv}^{\vecc}(\vecx)}{\partial x_\ell}$ where
  $\ell \in B_{\vecv}^{\vecc}$ (symmetrically for $j\in A_{\vecv}^{\vecc}$),
  $M_{Q'}$ additionally computes the
  $\left [S'_{\infty}(S(p_\ell) - S(p_j))\right]_{b'}$ with accuracy $b'$. To
  simplify notation we denote with $S'_{\infty}(S(p_\ell) - S(p_j))$ with
  $S'_{\ell j}$ and $S'(p_i)$ by $S'_i$. Then $M_{Q'}$ outputs,
  \[ \left[\frac{\partial Q_{\vecv}^{\vecc}(\vecx)}{\partial x_i}\right]_{b'} \leftarrow \left[\frac{1}{t_i - s_i} \cdot \left[\frac{\partial Q_{\vecv}^{\vecc}(\vecx)}{\partial p_i}\right]_{b'} \right]_{b'} \]
  \[ \text{where }~~~\left[\frac{\partial Q_{\vecv}^{\vecc}(\vecx)}{\partial p_i}\right]_{b'}\leftarrow\left[ \sum_{j \in A} \left[S'_{i j}\right]_{b'} \cdot \left[ S'_i \right]_{b'} \Pi_{m\in A/{j},\ell \in B}\left[S_{\ell m}\right]_{b'} \right]_{b'} \]
  \noindent Observe that $t_i - s_i = \frac{\mathrm{sign}(t_i- s_i)}{N - 1}$ and
  thus
  $\frac{1}{t_i - s_i} \cdot \left[\frac{\partial Q_{\vecv}^{\vecc}(\vecx)}{\partial p_i}\right]_{b'}$
  can be exactly computed. We next prove that these computations of
  $\left[\frac{\partial Q_{\vecv}^{\vecc}(\vecx)}{\partial x_i}\right]_{b'}$ and
  $\left[\frac{\partial Q_{\vecv}^{\vecc}(\vecx)}{\partial p_i}\right]_{b'}$ are
  correct.

    We first bound
  $\abs{\left [S'_{i j}\right]_{b'} \cdot \left[ S'_i \right]_{b'} \cdot \Pi_{m\in A/\{j\},\ell \in B}\left[S_{\ell m}\right]_{b'} - S'_{i j} \cdot S'_i \cdot \Pi_{m\in A/\{j\},\ell \in B}S_{\ell m}}$.
  \medskip

  \noindent Consider the function
  $g(y_1,y_2,\vecy)= y_1 \cdot y_2 \cdot \prod_{m \in A/\{j\},\ell \in B}y_{\ell m}$.
  As previously done, for $y_1, y_2 \in [0, 6]$ and
  $\vecy \in [0, 1+1/d^2]^{|A| \times |B| - 1}$ we have that,
  $\norm{\nabla g(y_1, y_2, \vecy)}_2 \leq \Theta(d)$.
  If $b' \leq \Theta(\log d)$ then $\abs{S'_{ij}},S'_i \leq 6$ and
  $S_{\ell m} \in [0,1 + 1/d^2]$. As a result,
  \begin{eqnarray*}
    \abs{\left [S'_{i j}\right]_{b'} \cdot \left[ S'_i \right]_{b'} \cdot \Pi_{m\in A/\{j\},\ell \in B}\left[S_{\ell m}\right]_{b'} - S'_{i j} \cdot S'_i \cdot \Pi_{m\in A/\{j\},\ell \in B}S_{\ell m} }
    & \le & \Theta(d^2) \cdot  2^{-b'}.
  \end{eqnarray*}

  \noindent We can now use the above inequality to bound
  $\abs{\left[\frac{\partial Q_{\vecv}^{\vecc}(\vecx)}{\partial p_i}\right]_{b'} - \frac{\partial Q_{\vecv}^{\vecc}(\vecx)}{\partial p_i}}$.
  More precisely,
  \begin{align*}
    & \abs{\left[\frac{\partial Q_{\vecv}^{\vecc}(\vecx)}{\partial p_i}\right]_{b'} - \frac{\partial Q_{\vecv}^{\vecc}(\vecx)}{\partial p_i}} \\
    & ~~~~~~ \le 2^{-b} + \abs{\sum_{j \in A} \left[S'_{i j}\right]_{b'} \cdot \left[ S'_i \right]_{b'} \cdot \prod_{m\in A/\{j\},\ell \in B}\left[S_{\ell m}\right]_{b'} - \sum_{j \in A} S'_{i j} \cdot S'_i \cdot \prod_{m\in A/\{j\},\ell \in B}S_{\ell m}} \\
    & ~~~~~~ \le \Theta(d^3) \cdot 2^{-b'}
  \end{align*}
  We finally get that
  \begin{align*}
    \abs{\left[\frac{\partial Q_{\vecv}^{\vecc}(\vecx)}{\partial x_i}\right]_{b'} - \frac{\partial Q_{\vecv}^{\vecc}(\vecx)}{\partial x_i}} \leq 2^{-b'} + N\abs{\left[\frac{\partial Q_{\vecv}^{\vecc}(\vecx)}{\partial p_i}\right]_{b'} - \frac{\partial Q_{\vecv}^{\vecc}(\vecx)}{\partial p_i}}
    & \le\Theta(N d^3) \cdot 2^{-b'}.
  \end{align*}
  \noindent Thus the analysis is completed by selecting
  $b' = b + \Theta(\log d)$ + $\Theta(\log N)$.
\end{proof}

\begin{lemma}\label{l:computing_P}
    There exist Turing Machines $M_P$ and $M_{P'}$ that given
  $\vecx \in [0,1]^d$ and $\eps > 0$ in binary form compute
  $\b{\mathrm{P}_{\vecv}(\vecx)}_b$ and $\b{\nabla \mathrm{P}_{\vecv}(\vecx)}_b$
  respectively for all vertices $\vecv \in \p{\nm{N}}^d$ with
  $\mathrm{P}_{\vecv}(\vecx) > 0$, where $b = \log(1/\eps)$. These vertices are
  most $d + 1$. Moreover both $M_P$ and $M_{P'}$ run in polynomial time with
  respect to $b$, $d$ and the binary representation of $\vecx$.
\end{lemma}

\begin{proof}
    $M_{P}$ first runs $M_Q$ of Lemma~\ref{l:computing_Q} to find the
  coefficients $Q_{\vecv}^{\vecc}(\vecx) > 0$. We remind that these vertices are
  denoted with $R_{+}(\vecx)$ and $\abs{R_{+}(\vecx)} \le d + 1$. Then for each
  $\vecv \in R_{+}(\vecx)$, $M_P$ outputs as
  $\left[\mathrm{P}_{\vecv}(\vecx) \right]_{b}$ the fist $b$ bits of the
  quantity,
  \[ \left[ \frac{\left[Q_{\vecv}^{\vecc}(x)\right]_{b'}}{\sum_{\vecv' \in R_{+}(\vecx)}\left[Q_{\vecv'}^{\vecc}(x)\right]_{b'}} \right]_{b'} \]
  where we determine the value of $b'$ later in the proof but it is chosen to be
  polynomial in $b$ and $d$. We next present the proof that the above expression
  correctly computes $\left[\mathrm{P}_{\vecv}(\vecx) \right]_{b}$.
  \smallskip

  \noindent For accuracy $b' \geq \Theta(d^2 \log d)$ we get that,
  \begin{align*}
    \sum_{\vecv' \in R_{+}(\vecx)}\left[Q_{\vecv'}^{\vecc}(x)\right]_{b'}
    & \ge \sum_{\vecv' \in R_{+}(\vecx)}Q_{\vecv'}^{\vecc}(x) - \Theta(d)\cdot 2^{-b'} \\
    & = \sum_{\vecv' \in R_{\vecc}(\vecx)}Q_{\vecv'}^{\vecc}(x) - \Theta(d)\cdot 2^{-b'} \\
    & \ge \Theta \left(1/d)^{d^2} \right) - \Theta(d)\cdot 2^{-b'} \\
    & \ge \Theta \left( (1/d)^{d^2} \right)
  \end{align*}

  \noindent Consider the function $g(\vecy) = y_i/(\sum_{j=1}^{d+1} y_j)$.
  Notice that for $\vecy \in [0,1]^{d+1}$ and $\sum_{j=1}^{d+1} y_j \geq \mu$
  then $\norm{\nabla g(\vecy)}_2 \le \Theta(d^{3/2}/\mu^2)$. The latter implies
  that for $\vecy, \vecz \in [0,1]^{d+1}$ such that
  $\sum_{j=1}^{d+1}y_j \ge \mu$ and that $\sum_{j=1}^{d+1}z_j \ge \mu$, it holds
  that
  \[ \abs{\frac{y_i}{\sum_{j=1}^{d+1} y_j} - \frac{z_i}{\sum_{j=1}^{d+1} z_j}} \le \Theta\p{\frac{d^{3/2}}{\mu^2}} \cdot \norm{\vecy - \vecz}_2. \]
  \noindent Since there are at most $d+1$ vertices $\vecv' \in R_{+}(\vecx)$
  while both the term
  $\sum_{\vecv' \in R_{+}(\vecx)} \left[Q_{\vecv'}^{\vecc}(\vecx)\right]_{b'}$
  and the term $\sum_{\vecv' \in R_{+}(\vecx)}Q_{\vecv'}^{\vecc}(\vecx)$ are
  greater than $\Theta\left((1/d)^{d^2}\right)$, we can apply the above
  inequality with $\mu = \Theta\left( (1/d)^{d^2}\right)$ and we get the
  following
  \begin{align*}
    & \abs{\frac{\left[Q_{\vecv}^{\vecc}(x)\right]_{b'}}{\sum_{\vecv' \in R_{+}(\vecx)}\left[Q_{\vecv'}^{\vecc}(x)\right]_{b'}} - \frac{Q_{\vecv}^{\vecc}(x)}{\sum_{\vecv' \in R_{+}(\vecx)}Q_{\vecv'}^{\vecc}(x)}} \\
    & ~~~~~~~~~~ \le \Theta\left(d^{2d^2 + 3/2}\right) \cdot \left[\sum_{\vecv' \in R_{+}(\vecx)}\left(\left[Q_{\vecv'}^{\vecc}(x)\right]_{b'} -  Q_{\vecv'}^{\vecc}(x)\right)^2\right]^{1/2} \\
    & ~~~~~~~~~~ \le \Theta\left(d^{2d^2 + 2}\right) \cdot 2^{-b'}
  \end{align*}
  \noindent Overall, we have that
  \begin{align*}
    & \abs{\left[\frac{\left[Q_{\vecv}^{\vecc}(x)\right]_{b'}}{\sum_{\vecv' \in R_{+}(\vecx)}\left[Q_{\vecv'}^{\vecc}(x)\right]_{b'}}\right]_{b'} - \frac{Q_{\vecv}^{\vecc}(x)}{\sum_{\vecv' \in R_{\vecc}(\vecx)}Q_{\vecv'}^{\vecc}(x)}} \\
    & ~~~~~~~~~~~~ \le \abs{\left[\frac{\left[Q_{\vecv}^{\vecc}(x)\right]_{b'}}{\sum_{\vecv' \in R_{+}(\vecx)}\left[Q_{\vecv'}^{\vecc}(x)\right]_{b'}}\right]_{b'}- \frac{\left[Q_{\vecv}^{\vecc}(x)\right]_{b'}}{\sum_{\vecv' \in R_{+}(\vecx)}\left[Q_{\vecv'}^{\vecc}(x)\right]_{b'}}} \\
    & ~~~~~~~~~~~~~~~~~ + \abs{\frac{\left[Q_{\vecv}^{\vecc}(x)\right]_{b'}}{\sum_{\vecv' \in R_{+}(\vecx)}\left[Q_{\vecv'}^{\vecc}(x)\right]_{b'}} - \frac{Q_{\vecv}^{\vecc}(x)}{\sum_{\vecv' \in R_{+}(\vecx)}Q_{\vecv'}^{\vecc}(x)}} \\
    & ~~~~~~~~~~~~ \le \Theta\left(  d^{2d^2 + 1}\right) 2^{-b'}
  \end{align*}
  \noindent The proof is completed via selecting $b' = b + \Theta(d^2\log d)$.
  \medskip

  In order to compute $\frac{\partial \mathrm{P}_{\vecv}(\vecx)}{\partial x_i}$
  the Turing machine $M_{P'}$ computes all vertices $R_{+}(\vecx)$ the
  coefficients $\frac{\partial Q_{\vecv}^{\vecc}(\vecx)}{\partial x_i}$ with
  accuracy $b'$. Then for each $\vecv \in R_{+}(\vecx)$ the Turing Machine
  $M_{P'}$ outputs,
  \[  \left[\frac{\partial \mathrm{P}_{\vecv}(\vecx)}{\partial x_i} \right]_{b'} \leftarrow \left[ \frac{1}{t_i -s_i} \cdot \left[\frac{\partial \mathrm{P}_{\vecv}(\vecx)}{\partial p_i} \right]_{b'} \right]_{b'} \]
  \[ \text{where  }~~~ \left[\frac{\partial \mathrm{P}_{\vecv}(\vecx)}{\partial p_i} \right]_{b'} \leftarrow \left[\frac{ \left[\frac{\partial Q_{\vecv}(\vecx)}{\partial p_i}\right]_{b'} \cdot \sum_{\vecv' \in R_{+}(\vecx)} \left[Q_{\vecv'}(\vecx)\right]_{b'} -\left[Q_{\vecv}(\vecx)\right]_{b'} \cdot \sum_{\vecv' \in R_{+}(\vecx)}\left[\frac{\partial Q_{\vecv'}(\vecx)}{\partial p_i}\right]_{b'} }{\left(\sum_{\vecv' \in R_{+}(\vecx)} \left[Q_{\vecv'}(\vecx)\right]_{b'}\right)^2}\right]_{b'} \]
  \noindent Similarly as above and as in Lemma \ref{l:computing_Q} we can prove
  that if $b' \geq b + \Theta(d^2\log d) + \Theta(\log N)$,
  $\abs{\left[\frac{\partial \mathrm{P}_{\vecv}(\vecx)}{\partial p_i} \right]_{b'} - \frac{\partial \mathrm{P}_{\vecv}(\vecx)}{\partial p_i}} \leq 2^{-b}$.
\end{proof}

\begin{proof}[Proof of Theorem~\ref{t:Turing_Machine}]
    Let $R(\vecx)$ be the cell at which $\vecx$ lies. The Turing Machine
  $M_{f_{\calC_l}}$ initially calculates the vertices
  $\vecv \in R_{\vecc}(\vecx)$ with coefficient $\mathrm{P}_{\vecv}(\vecx) > 0$.
  We remind that this set is denoted by $R_{+}(\vecx)$ and
  $\abs{R_{+}(\vecx)} \leq d + 1$. Then $M_{f_{\calC_l}}$ outputs the first
  $b$ bits of the following quantity,
  \[ \left[f_{\calC_l(\vecx,\vecy)}\right]_{b'}=\sum_{j=1}^d \left[\alpha(\vecx,j)\right]_{b'} \cdot(x_j - y_j)~~~~\text{where}~~ \left[\alpha(\vecx,j)\right]_{b'} = \sum_{\vecv' \in R_{+}(\vecx)} \calC_l(\vecv,j) \cdot \left[\mathrm{P}_{\vecv}(\vecx)\right]_{b'} \]
  \noindent we next prove that the above computation is correct.

  \begin{align*}
    \abs{\left[f_{\calC_l(\vecx,\vecy)}\right]_{b'} - f_{\calC_l(\vecx,\vecy)}}
    & =  \abs{\sum_{j=1}^d \left[\alpha(\vecx,j)\right]_{b'} \cdot(x_j - y_j) - \sum_{j=1}^d \alpha(\vecx,j) \cdot(x_j - y_j) } \\
    & \le \sum_{j=1}^d \abs{\left [ \alpha(\vecx,j) \right] - \alpha(\vecx,j)} \\
    & = \sum_{j=1}^d\abs{\sum_{\vecv' \in R_{+}(\vecx)} \calC_l(\vecv,j) \cdot \left[ \mathrm{P}_{\vecv}(\vecx) \right]_{b'} - \sum_{\vecv' \in R_{+}(\vecx)} \calC_l(\vecv,j) \cdot \mathrm{P}_{\vecv}(\vecx)} \\
    & \le \sum_{j=1}^d\sum_{\vecv' \in R_{+}(\vecx)} \abs{\left[ \mathrm{P}_{\vecv}(\vecx)\right]_{b'} - \mathrm{P}_{\vecv}(\vecx)} \\
    & \le d \cdot (d + 1) \cdot 2^{-b'}
  \end{align*}
  Setting $b' = b + \Theta\left(\log d\right)$ we get the desired result.
  Similarly for $\frac{\partial f_{\calC_l(\vecx,\vecy)}}{\partial x_i}$ and
  $\frac{\partial f_{\calC_l(\vecx,\vecy)}}{\partial y_i}$.
\end{proof}

\section{Convergence of PGD to Approximate Local Minimum}
\label{sec:gdStationary}

  In this section we present for completeness the folklore result that the
Projected Gradient Descent with convex projection set converges fast to a first
order stationary point. Using the same ideas that we presented in Section
\ref{sec:existence} this result implies that Projected Gradient Descent solves
the $\lmin$ problem in time $\poly(1/\eps, L, G, d)$ when $(\eps, \delta)$ in
the input are in the local regime. Also observe that although the following
proof assumes access to the exact value of the gradient $\nabla f$ it is very
simple to adapt the proof to the case where we only have access to $\nabla f$
with accuracy $\eps^3$. We leave this as an exercise to the reader.

\begin{theorem} \label{thm:gdStationary}
    Let $f : K \to \R$ be an $L$-smooth function and $K \subseteq \R^d$ be a
  convex set. The projected gradient descent algorithm started at $\vecx_0$, with
  step size $\eta$, after at most
  $T \ge \frac{2 L \p{f(\vecx_0) - f(\vecx^{\star})}}{\eps^2}$ steps outputs a
  point $\hat{\vecx}$ such that
  \[ \norm{\hat{\vecx} - \Pi_K \p{\hat{\vecx} - \eta \nabla f(\hat{\vecx})}}_2 \le \eta \cdot \eps \]
  \noindent where $\eta = 1/L$ and $\vecx^{\star}$ is a global minimum of $f$.
\end{theorem}

\begin{proof}
    If we run the Projected Gradient Descent algorithm on $f$ then we have
  \[ \vecx_{t + 1} \leftarrow \Pi_K \p{\vecx_t - \eta \nabla f(\vecx_t)} \]
  then due to the $L$-smoothness of $f$ we have that
  \[ f(\vecx_{t + 1}) \le f(\vecx_t) + \langle \nabla f(\vecx_t), \vecx_{t + 1} - \vecx_t \rangle + \frac{L}{2}\norm{\vecx_{t + 1} - \vecx_t}_2^2. \]
  \noindent We can now apply Theorem 1.5.5 (b) of \cite{facchinei2007finite} to
  get that
  \[ \langle \eta \cdot \nabla f(\vecx_t), \vecx_{t + 1} - \vecx_t \rangle \le - \norm{\vecx_{t + 1} - \vecx_t}_2^2 \implies \]
  \[ \langle \nabla f(\vecx_t), \vecx_{t + 1} - \vecx_t \rangle \le - \frac{1}{\eta} \cdot \norm{\vecx_{t + 1} - \vecx_t}_2^2 \]
  \noindent If we combine these then we have that
  \[ f(\vecx_{t + 1}) \le f(\vecx_t) - \p{\frac{1}{\eta} - \frac{L}{2}} \norm{\vecx_{t + 1} - \vecx_t}_2^2. \]
  \noindent So if we pick $\eta = 1/L$ then we get
  \[ f(\vecx_{t + 1}) \le f(\vecx_t) - \frac{L}{2} \norm{\vecx_{t + 1} - \vecx_t}_2^2. \]
  \noindent If sum all the above inequalities and divide by $T$ then we get
  \[ \frac{1}{T} \sum_{t = 0}^{T - 1} \norm{\vecx_{t + 1} - \vecx_t}_2^2 \le \frac{2}{T \cdot L} \p{f(\vecx_0) - f(\vecx_T)} \]
  \noindent which implies that
  \[ \min_{0 \le t \le T - 1} \norm{\vecx_{t + 1} - \vecx_t}_2 \le \sqrt{\frac{2}{T \cdot L} \p{f(\vecx_0) - f(\vecx_T)}} \]
  \noindent Therefore for
  $T \ge \frac{2 L \p{f(\vecx_0) - f(\vecx^{\star})}}{\eps^2}$ we have that
  \[ \min_{0 \le t \le T - 1} \norm{\vecx_{t + 1} - \vecx_t}_2 \le \eta \cdot \eps  = \eps / L. \]
\end{proof}

\end{document}